\documentclass[11pt, a4paper]{ociamthesis}
\pdfoutput=1

\usepackage[sumlimits]{amsmath}
\usepackage{amsfonts,amssymb, amsthm, mathrsfs, mathtools, amscd}
\usepackage{graphicx, longtable, caption, afterpage, float}
\usepackage{xspace}
\usepackage[latin1]{inputenc}
\usepackage[usenames,dvipsnames]{color}
\usepackage{rotating}
\usepackage{tikz}
\usetikzlibrary{arrows,shapes}
\usepackage{setspace}
\usepackage{epsfig,cancel,cite}
\usepackage[linktoc=page]{hyperref}

\usepackage{array}

\font\twentyonerm=cmr12 at 21pt

\definecolor{dg}{rgb}{.1,0.5,.5}

\newcommand{\varstr}[2]{\vrule height #1 depth #2 width0pt}

\newcommand{\m}{\mu}

\newcommand{\cH}{\mathcal{H}}

\newcommand{\cN}{\mathcal{N}}

\newcommand{\ba}{\begin{array}}
\newcommand{\ea}{\end{array}}

\newcommand{\beqn}{\begin{equation*}}
\newcommand{\eeqn}{\end{equation*}}

\theoremstyle{plain}

\theoremstyle{definition}

\theoremstyle{remark}

\newcommand{\beq}{\begin{equation}}
\newcommand{\eeq}{\end{equation}}
\newcommand{\beqnn}{\begin{equation*}}
\newcommand{\eeqnn}{\end{equation*}}

\def\place#1#2#3{\vbox to0pt{\kern-\parskip\kern-7pt
                             \kern-#2truein\hbox{\kern#1truein #3}
                             \vss}\nointerlineskip}

\newcommand{\Bigcheck}{\lower2pt\hbox{\smash{\hbox{{\twentyonerm \v{}}}}}}
\newcommand{\Bighat}{\lower3.8pt\hbox{\smash{\hbox{{\twentyonerm \^{}}}}}}

\newcommand{\Xcheck}{\kern2pt\hbox{\Bigcheck\kern-15pt{$\mathscr{X}$}}}
\newcommand{\Xhat}{\kern2pt\hbox{\Bighat\kern-12pt{$\mathscr{X}$}}}

\numberwithin{equation}{section}

\usepackage[all, cmtip]{xy}

\usepackage[width=16cm, left=2.5cm, height=24.8cm, top=2cm]{geometry}
\usepackage{mathrsfs}
\usepackage{eufrak}
\usepackage{soul}
\usepackage{wasysym}
\usepackage[nottoc]{tocbibind}
\usepackage{amsmath, amsthm, amssymb}
\usepackage{thmtools,thm-restate}
\usepackage[toc,page]{appendix}
\theoremstyle{plain}

\renewenvironment{proof}{{\noindent \bfseries Proof:}}{}

\newtheorem{theorem}{Theorem}[section]
\newtheorem{lemma}[theorem]{Lemma}

\newtheorem{proposition}[theorem]{Proposition}
\newtheorem{definition}[theorem]{Definition}
\newtheorem{remark}[theorem]{Remark}
\newtheorem{example}[theorem]{Example}
\newtheorem*{conjecture*}{Conjecture}
\newtheorem*{thma}{Theorem 4}

\renewcommand{\qed}{\hfill \ensuremath{\Box}}

\newcommand\ket[1]{\left| #1\right.\!\rangle}
\newcommand\bra[1]{\langle #1\left.\!\right|}

\newcommand\mc[1]{\mathcal{#1}}
\newcommand\BL{\mc{B}(L)}
\renewcommand\cH{\mc{H}}
\newcommand\PH{\mc{P}(\cH)}
\newcommand\BH{\mc{B}(\cH)}
\newcommand\VH{\mc{V}(\cH)}
\newcommand\PV{\mc{P}(V)}
\newcommand\bbC{\mathbb{C}}
\newcommand\join{\vee}

\newcommand\ra{\rightarrow}

\newcommand{\IMP}{\; \Rightarrow \;}
\newcommand{\plusminus}{\pm}

\newcommand{\Pn}{\mathcal{P}_n}
\newcommand{\PP}{\mathcal{P}}
\newcommand{\ie}{\textit{i.e.}~}
\newcommand{\HH}{\mathcal{H}}
\newcommand{\KK}{\mathcal{K}}
\newcommand{\AND}{\; \wedge \;}
\newcommand{\Un}{\mathbf{U}}
\newcommand{\Tr}{\mathsf{Tr}}

\title{Sheaf$\,$-Theoretic Methods \\ in Quantum Mechanics \\ and Quantum Information Theory}
\author{Carmen Maria Constantin}
\college{University College}
\degree{Doctor of Philosophy}
\degreedate{Trinity Term 2015}

\begin{document}

\maketitle

\pagenumbering{gobble}
\setstretch{1.5}

\newpage
$ $
\newpage

\begin{table}[c]
\begin{flushright}
{\emph{``All models are wrong, but some are useful."}} \\
$ $
\\
$ $
\\
George E. P. Box
\end{flushright}

\end{table}
\newpage
$ $
\newpage
\begin{abstract}

In this thesis we use the language of sheaf theory in an attempt to develop a deeper understanding of some of the fundamental differences - such as entanglement, contextuality and non-locality - which separate quantum from classical physics. 

We first present, based on the work of Abramsky and Brandenburger \cite{AbrBra11}, how sheaves, defined over certain posets of physically meaningful contexts, give a natural setting for capturing and analysing important quantum mechanical phenomena, such as quantum non-locality and contextuality. We also describe how this setting naturally leads to a three level hierarchy of quantum contextuality: weak contextuality, logical non-locality and strong contextuality. One of the original contributions of this thesis is to use these insights in order to classify a particular class of multipartite entangled states, which we have named balanced states with functional dependencies. Almost all of these states turn out to be at least logically non-local, and a number of them even turn out to be strongly contextual. We then further extend this result by showing that in fact all $n$-qubit entangled states, with the exception of tensor products of single-qubit and bipartite maximally-entangled states, are logically non-local. Moreover, our proof is constructive: given any $n$-qubit state, we present an algorithm which produces $n+2$ local observables witnessing its logical non-locality. 

In the second half of the thesis we use the same basic principle of sheaves defined over physically meaningful contexts, in order to present an elegant mathematical language, known under the name of the Topos Approach \cite{disham}, in which many quantum mechanical concepts, such as states, observables, and propositions about these, can be expressed. This presentation is followed by another original contribution in which we show that the language of the Topos Approach is as least as expressive, in logical terms, as traditional quantum logic. Finally, starting from a topos-theoretic perspective, we develop the construction of contextual entropy in order to give a unified treatment of classical and quantum notions of information theoretic entropy.

\end{abstract}
\newpage
$ $
\newpage
\begin{acknowledgements}

The research presented in this thesis has been conducted under the guidance of my supervisors, Prof. Samson Abramsky and Dr. Andreas D\"{o}ring. Their enthusiasm for research has been contagious and a real source of inspiration for me, throughout my DPhil studies. I am especially grateful for their constant encouragement and support, and for allowing me to enjoy a very large degree of flexibility for the entire duration of this project.

My contact with the other members of the university has both shaped and nurtured my growth as a researcher and as a human being. I would like to thank my colleagues and professors for their friendship and their open attitude towards my ideas and questions, especially Prof. Bob Coecke, Prof. Kobi Kremnitzer, Prof. Boris Zilber, Prof. Prakash Panangaden, Rui Soares Barbosa, Shane Mansfield, Nadish de Silva, Aleks Kissinger, Ray Lal, Jamie Vicary, Chris Heunen, Alex Merry, Brendan Fong, Miriam Backens, Hugo Nava and Ondrej Rypacek. 

I am indebted to my husband, Andrei Constantin, far more than I can express in these few lines. He has supported me in every way possible, and has even patiently guided me through some of the more puzzling conceptual aspects of quantum mechanics. He has proofread this thesis, and in fact all my work, and has always helped me organize and express my ideas. 

I am also grateful to my teachers and mentors from university and high school which have kindled and nurtured my love for mathematics in general and for sheaf theory in particular: Dr. Antony Maciocia, Enache P\v{a}tra\c{s}cu and Constantin T\v{a}n\v{a}sescu.

Special thanks are also due to my formidable grandmother Nona, and to my aunt, Rodica. Without their kind and timely offer to help with child care related matters, this thesis would have taken much longer to complete.

Finally, I would like to acknowledge the financial support I have received from the EPSRC in the form of a generous scholarship including maternity support.

\end{acknowledgements}
\newpage
$ $
\newpage
\begin{dedication}
To my late parents, Dana-Maria and Tudor-Alin
\end{dedication}
\newpage
$ $
\newpage
\begin{romanpages}
\tableofcontents
\end{romanpages}

\pagenumbering{arabic}
\begin{chapter}{Introduction}

The emergence of Quantum mechanics at the beginning of the last century has shaken many of the intuitions underlying classical physics. One of its most peculiar aspects comes from the fact that quantum mechanical states induce, in general, only statistical restrictions on the results of measurements, instead of definite outcomes. 
One could be tempted to draw, as Einstein did, the conclusion that these states are therefore incomplete descriptions of quantum systems.  For a while physicists have wondered whether quantum mechanics could be supplemented with a more complete description involving hidden variables. Thus quantum probabilities could be naturally interpreted as epistemic probabilities, of the sort that arise in ordinary statistical mechanics. 

The existence of hidden variable theories has been refuted by two powerful theorems. The more famous of these is Bell's theorem which uses correlations between entangled states to show that, assuming the premise of locality (i.e. the idea that spatially separated systems cannot instantaneously influence each other), any hidden variable model can be used to derive an inequality which is violated by the predictions of quantum mechanics and these predictions have been experimentally confirmed. 

The second important no-go theorem is that of Kochen-Specker, which shows that it is not possible to assign values to observables under the premise of non-contextuality - which is the assumption that if a quantum system possesses a property (value of an observable), then it does so independently of any measurement context, that is, independently of which other observables are measured alongside the one under consideration.

The more recent rise of Quantum Information Theory is now posing new challenges to our understanding of both computation and physics. Entanglement, non-locality and contextuality, which have proved so problematic to our understanding of quantum theory, have turned out, when seen from a different viewpoint, to offer exciting new possibilities, thus challenging the usual assumptions of classical computation. Quantum systems have been shown to be able to perform information-theoretic tasks beyond the capabilities of classical systems. Famous examples include secure key commitment \cite{BenBra, Ekert:1991zz}, quantum teleportation \cite{Bennett:1992tv} and factoring primes in polynomial time \cite{Shor:1994jg}. 

The work presented in this thesis has been guided and motivated by several recent attempts to use the powerful mathematical tools of sheaf and category theory in order to obtain deeper structural insights into the nature of physical reality. These insights would hopefully allow us to fully harness the resources offered to information processing by quantum physics.

\section{The Topos Approach}

The topos approach to the formulation of physical theories was initiated by Chris Isham and Jeremy Butterfield\cite{Ish97,b1,b2,b3,b4}, expanded by Andreas D\"{o}ring and Chris Isham \cite{d1,d2,d3,d4,ms,Doe09b,Doe11a,Doe11b,disham,DoeIsh12, Doe12,Doe12e,Doe12c}, and further developed by Heunen, Landsman, Spitters and Wolters \cite{HLS09a,HLS09b,Wol10,HLS11}, Flori and collaborators \cite{Flo11,BGeFlo12,BreFlo12} and others \cite{Nak11,Nui11,Woo11}. One of the initial ideas which motivated the programme was that by choosing a suitable universe of mathematical discourse, that is to say a suitable topos, one could express quantum mechanical concepts in a way which would render them structurally similar to their classical mechanical counterparts. It was expected that such a characterisation would bring to light new potentialities of the quantum world.

While recent results, including the approach to non-commutative $K$-theory developed by Nadish de Silva \cite{Nadish}, the new perspective on quantum probability developed by Dewitt and D\"{o}ring \cite{DoeDew1, DoeDew2} and the non-commutative generalisation of Gelfand duality obtained by D\"{o}ring \cite {Doe12e, Doe12c, Doe14} give full justice to the claim that the formalism underpinning the topos approach is in itself rich enough to be of mathematical interest, the choice of formulating real world concepts within generalized mathematical universes appears rather counter-intuitive. 

Mathematics is the language in which physical theories are formulated and perhaps one of the reasons why the universe of sets and functions and the Boolean logic associated with it have provided the foundation of most of the current mathematical discourse is the fact that mathematics has been itself shaped by our experience of reality which, until relatively recently, has only been directly perceived at the macroscopic level. 

It is possible that, if we were living in a quantum world, our primary intuitions would be different and our mathematics would not necessarily be based on set theory. It is at least an interesting coincidence that the change in our basic perception of the natural world brought about by the development of quantum mechanics has taken place in the same century in which category theory, through the notion of topos, has succeeded in axiomatising set-theory, thus bringing about an entirely new, categorical foundation of mathematics. 

With these considerations in mind, one finds it less surprising that the same concepts which have had such an impact on the foundation of mathematics have eventually found their way into the realm of foundations of physics. 

However, apart from these general arguments, there is of course no a priori reason why quantum mechanical concepts, for example, should find a formulation in topos theoretic language. The fact that an elegant formulation does exist is in itself rather remarkable and often comes as a result of powerful theorems. 

\subsection{A Quantum State Space}

Both classical mechanics and quantum mechanics have been traditionally formulated in the topos of sets. Within this topos, the state space of a classical system is a set, which additionally has the structure of a symplectic manifold. The points of this set are the pure states in which the system can be. Physical observables can be thought of as real-valued functions on the state space and, as such, they form a commutative algebra under pointwise multiplication. This commutative algebra and the state space of the system effectively determine each other. This is a consequence of a mathematical result known as Gelfand duality, which establishes a bijective correspondence between compact Hausdorff spaces and unital commutative $C^*$-algebras. 

This correspondence does not hold for quantum mechanical systems which, in the usual Hilbert space formalism, have non-commutative algebras of observables given by the self-adjoint operators of the Hilbert space corresponding to the system. Non-commuting operators correspond to incompatible physical observables, that is, to observables whose values cannot be measured at the same time such as, for example, position and momentum.

Several approaches \cite{Bor,GilKum,Ake,Mul,Doe11b} to extending Gelfand duality to non-commutative algebras were compared in \cite{Ist}. One of these has been further developed by D\"{o}ring in \cite{Doe12c, Doe12e, Doe14}. The underlying idea is to consider all commutative sub-algebras of the non-commutative algebra of bounded operators on the Hilbert space associated with the system. Each of these subalgebras can be interpreted as a classical `perspective' on the physical system, as it only contains pairwise commutative observables. We call it a \textit{classical context}. Since two self-adjoint operators with discrete spectra commute only if they have a joint set of eigenspaces, in finite dimensions a classical context in effect corresponds to a set of pairwise orthogonal projections which add up to identity. 

Indeed, such projections generate, via the von Neumann double commutant construction, a commutative von Neumann algebra. We can use Gelfand duality to associate a classical state space to this algebra. The collection of all such classical state spaces forms a presheaf over the base category given by the set of commutative subalgebras, partially ordered by inclusion. The spectral presheaf is hence an object in the topos of presheaves over this base category. 

The analogy between the spectral presheaf and a classical state space is further justified by the fact that the set of quantum states is equivalent to the set of measures on the spectral presheaf, just like in classical mechanics mixed states are given by probability measures on state space. Pure states are given by `minimal' measures, but it must be noted that unlike their classical counterparts, these are not concentrated at points. A point of a presheaf, in the category theoretical sense, is given by a global section, but the spectral presheaf has no global sections. The lack of global sections of the spectral presheaf is equivalent \cite{b1} to the Kochen-Specker theorem which asserts that it is not possible to simultaneously assign values to all observables on a quantum system. The important insight provided by this elegant reformulation  \cite{b1} was one of the sources of inspiration behind both the topos programme as well as the later sheaf-theoretic approach to characterising contextuality and non-locality developed in Abramsky and Brandenburger's paper \cite{AbrBra11}.

\subsection{From Boolean to Intuitionistic Logic and Back}

We have mentioned that topoi include, but are more general than sets. Moreover each topos comes equipped with its own logical calculus. Just like in the universe of sets and functions the principles of classical logic are represented by operations on a certain set, namely the two-element Boolean algebra, each topos carries an analogue of this algebra. This is called a Heyting algebra, and the logical principles which hold in a topos turn out to be precisely those of intuitionistic logic.  

The defining characteristic of intuitionistic logic is that the law of excluded middle does not necessarily hold. Hence the main difference between theorems proved using Heyting logic and those using Boolean logic is that proofs by contradiction cannot be used in the former. D\"{o}ring and Isham argue in \cite{disham} that this is not a major restriction and intuitionistic logic is a viable alternative to classical logic for the purpose of building physical theories. 

In particular, the internal, multivalued topos logic allows us to assign (generalised) truth values to all propositions about a quantum system. As mentioned above, Kochen and Specker proved \cite{KocSpe} that such assignments are impossible if one works with classical two-valued logic. However, one must carefully distinguish between the internal logic which is adequate to the system in itself, and the meta-logic in which we argue \textit{about} the system. Indeed we, as macroscopic entities, use a meta-logic, typically Boolean, to reason about the world, and to define mathematical structures. 

Although interesting work has been done by Heunen, Landsman, Spitters and Wolters \cite{HLS09a,HLS09b,HLS11,Wol10} and Fauser, Raynaud and Vickers \cite{FRV} using the internal topos logic, it is also possible to reason about a topos in an external fashion. As D{\" o}ring argues in \cite{Doe13}, in order to provide an objective report of some phenomena outside ourselves, we must separate ourselves from the system we are considering, hence in our description and mathematical arguments about a physical system we are free to use the metalogic. In technical terms, this is equivalent to working in an ambient topos which is the familiar topos of sets and functions. This external perspective has been adopted in the present work as well.

\section{Sheaves and Contextuality}

We have mentioned above that one of the basic ingredients of the topos approach is given by the poset of contexts. The contexts capture the physical idea of measurements which can be performed jointly. One can achieve an even higher level of generality by replacing the formalism of operator algebras with abstractly defined families of maximal sets of commuting observables. This is done by Samson Abramsky and Adam Brandenburger in \cite{AbrBra11} in order to study the key information theoretic resources of contextuality and non-locality. Their approach covers $n$-partite Bell-type scenarios as well as Kochen-Specker configurations, and many other examples relevant to quantum information theory. 

The basic ingredients of the mathematical formalism used by Abramsky and Brandenburger include a set $X$ of labels for observables, a cover $\mathcal{U}$ of $X$ consisting of subsets $U\subseteq X$ (called measurement contexts) which correspond to the different combinations of observables which can be measured together, and a set $O$ which labels the possible outcomes. A joint outcome for any compatible set of measurements $U$ is specified by a function $s:U \rightarrow O$, which is called a section over $U$. From these ingredients one can build a sheaf of events $\mathcal{E}$ which associates to each $U$ the set $O^U$ of sections over $U$, together with restriction maps given by function restriction. In a quantum mechanical setting, this sheaf is closely related to Isham and Butterfield's spectral presheaf.

This sheaf of events is then composed with a certain distribution functor to obtain a presheaf of $R$-valued distributions, $\mathcal D_R\mathcal E$, which assigns to each $U\subseteq X$  the set $\mathcal D_R(\mathcal E(U))$ of distributions on $\mathcal E(U)$. In this thesis the ring $R$ will be taken to be either the positive reals or the Boolean ring.

Measurement covers are defined as maximal sets of compatible measurements. An empirical model over a given measurement cover $\mathcal{M}\subset \mathcal{P}(X)$ is then defined as a compatible family $\{e_C\}_{C\in \mathcal{M}}$ of probability distributions $e_C$ on $\mathcal{E}(C)$.

Using these basic ingredients, a novel three-level characterisation of multipartite quantum states in terms of their degree of contextuality is obtained. This characterisation is linked to the lack of global sections of the preseheaf $\mathcal D_R\mathcal E$, compatible with a given empirical model.

\subsection{Three Levels of contextuality}

The first, and lowest, degree of contextuality described by Abramsky and Brandenburger - weak, or probabilistic contextuality - generalizes the original argument used in Bell's theorem. Bell's argument essentially relies on the probabilistic predictions of quantum mechanics, which are inconsistent with the predictions of any local realistic theory. Abramsky and Brandenburger show that the content of Bell's argument can be summarised by an empirical model over the positive reals, that is by a probability table containing the empirical predictions for all allowed combinations of measurements. The information encoded by this empirical model can be used to derive a proof of Bell's theorem, without recourse to inequalities, but based on the non-existence of a certain joint distribution, or global section.

The second, intermediate degree of contextuality - logical, or possibilistic contextuality - generalizes an argument used by Hardy \cite{Hardy92, Hardy 93}, who showed that an inequality-free proof of Bell's theorem could be given for almost all 2-qubit systems. Hardy's construction works for all bipartite entangled states, except for the maximally entangled states. The content of Hardy's argument can again be summarised by an empirical model, this time over the Boolean ring, as Hardy's argument essentially only relies on the possibility (probability greater than $0$) and respectively impossibility (probability~$0$) of certain measurement outcomes, that is, it only relies on the support of the probability distributions. Abramsky and Brandenburger show that the possibilistic contextuality of the Hardy model is a stronger property than the probabilistic contextuality used in the Bell argument (i.e. the Bell model is not logically contextual, but any probabilistic model whose support coincides with the Hardy model must be both weakly and logically contextual).  

The third degree of contextuality - strong contextuality - generalizes an argument used by Greenberger, Horne, Zeilinger and Shimony \cite{GHZ89, GHZ90}, who used the non-classical properties of certain quantum states to give a strengthened inequality-free proof of Bell's theorem. Their proof, which predates Hardy's argument, needs -- on the other hand -- systems of at least three qubits, such as the entangled three-qubit GHZ state. The GHZ argument, like the Hardy one, only relies on the possibility or impossibility of certain measurement outcomes. Yet it is shown by Abramsky and Brandenburger that the three qubit GHZ model satisfies a stronger property than the logical contextuality of the Hardy model, which is logically but not strongly contextual. It is also interesting to note that strong contextuality can be exhibited in the usual $n$-qubit multipartite Bell-type scenarios only if we have a scenario with three or more parties.

This line of work has been further developed by Abramsky and Hardy in \cite{AbrHar12} where they introduce a notion of logical Bell inequality, based on logical consistency conditions. Logical Bell inequalities can be used to obtain proofs of Bell's theorem without probabilities, but also to derive testable inequalities with provable violations for a wide variety of situations. It can also be shown that measurement models achieve maximal violations of logical Bell inequalities if and only if they are strongly contextual. Non-maximal violations are achieved by measurement models which are possibilistically contextual - that is, they occupy the middle level of the hierarchy.



\section{Outline}

According to the type of methodology used, the material presented in this thesis has been divided into two parts. Each part begins with a background chapter which lays out the necessary mathematical and physical terminology. 

Thus Part I starts with Chapter \ref{SSC} in which the reader is first given a brief overview of some of the main information theoretic concepts, such as qubits and quantum gates. We also present some of the intriguing properties which distinguish quantum mechanical systems from classical ones, such as entanglement, non-locality and contextuality. 

We then proceed to introduce some sheaf theoretic constructions which can be used to describe general experimental scenarios. These are modelled in great generality as compatible families of measurements. We also show how, for such scenarios, compatibility in a sheaf theoretic sense corresponds to a physical condition motivated by Special Relativity, which is known as no-signalling. In particular, the fact that quantum mechanics satisfies no-signalling is taken to indicate a basic consistency between quantum mechanics and relativity. This core mathematical structure is then used to analyse contextuality and non-locality in a unified way. In particular, we show that these phenomena can be characterised precisely in terms of obstructions to the existence of global sections of a certain distribution presheaf.  

The chapter ends by distinguishing the three strengths of degree of contextuality: standard probabilistic non-locality, exhibited by the original example introduced by Bell; possibilistic contextuality, exemplified in our presentation by the Hardy model and by permutationally symmetric states; and strong contextuality, exemplified by GHZ states. 

In chapter \ref{CC} we use this hierarchy to classify balanced quantum states with functional dependencies. These are non-permutationally symmetric entangled states which are described by Boolean functions. The classification yields a large number of strongly contextual states. Moreover, it turns out that all the states considered, except those which are equivalent to tensor products of pure states and Bell pairs, are at least logically contextual.

Chapter \ref{LNL} extends the results of the previous chapter's classification by showing that in fact all $n$-qubit quantum states, with the exception of tensor products of pure states and Bell pairs, are logically contextual for suitable choices of measurements. We moreover show that $n+2$ observables suffice to witness the logical contextuality of any $n$-qubit state: two observables each for two of the parties, and one each for the remaining $n-2$ parties. Our proof is constructive, so we are also able to present an algorithm which returns the witnessing local observables for any $n$-qubit state.
$$*\ *\ *\ $$

The second part of this thesis begins with Chapter \ref{TA}, in which the reader is first guided through the category theoretical notions which are necessary for understanding the definition of an elementary topos. 

This is followed by a short presentation of the algebraic-geometric duality (Gelfand's duality) which is used in defining the main construction of the topos approach, the spectral presheaf. Finally the reader is introduced to the way in which several basic physical concepts, such as the state space, real number objects, and the collection of pure and mixed states of a system are formulated within the topos approach. These will be the main building blocks used in the construction of contextual entropy in Chapter \ref{CE}. We also show how projections, which are used to represent propositions about a system in the Hilbert space formulation of quantum mechanics, can be represented in the topos formalism.

In Chapter \ref{LA} we discuss whether the topos-based form of logic for quantum systems is (at least) as rich as the traditional quantum logic derived from the Hilbert space formulation of quantum mechanics. Our discussion revolves around the question whether the orthomodular lattice of projections on a Hilbert space, which is traditionally used in quantum logic to represent propositions, can be reconstructed from the poset of contexts which underlies the constructions in the topos approach. This question has a positive answer, according to a result obtained by Harding and Navara in \cite{navara}. Our contribution in this chapter seeks to answer an open problem stated at the end of \cite{navara}, which asks for an explicit reconstruction of an orthomodular lattice from its associated poset of distributive sub-lattices. We show that for atomic orthomodular lattices a reasonably direct reconstruction is indeed possible.

In Chapter \ref{CE} we take a look at the concept of entropy, one of the cornerstones of information theory, from the perspective of the topos approach. Since our reformulation of entropy does not directly depend on the interpretation of states as density matrices within the usual Hilbert space formalism of quantum mechanics, this endeavour is particularly relevant for the topos programme.
  
The chapter starts with a brief outline of some of the properties of Shannon entropy, its quantum counterpart, von Neumann entropy and the more general family of Renyi entropies. 

Using D{\" o}ring's reformulation of quantum states as probability measures \cite{ms} on the spectral presheaf, we next show how a classical probability distribution, and hence also its corresponding Shannon entropy, can be naturally associated to each context in the base category over which the spectral presheaf is constructed. The numerical values of these Shannon entropies form a global section of a presheaf of real values over the same base category, the poset of contexts. This global section is the state's contextual entropy. A powerful result known as the Schur-Horn Lemma can be used to show that this construction unifies the Shannon and von Neumann entropies by encoding the von Neumann entropy of the state in a distinguished way.  

A comparison is then drawn between Shannon, von Neumann and contextual entropies. This allows us to observe, for example, that one of the differences between Shannon and von Neumann entropies (the property of being monotone) is precisely due to contextuality, an idea which is not immediately obvious from the definition of these two entropies. 

Perhaps the main result of this chapter is the informatic-theoretic characterisation of quantum states provided by contextual entropy which can be shown to be rich enough to allow one to reconstruct the quantum state from which it originated. This also provides us with a new insight into Gleason's theorem.

At the end of the chapter we show that it is possible to adapt other classical entropies within the formalism of the topos approach, given that they satisfy a certain weak recursivity property.  We explicitly describe how the contextual entropy construction can be generalised to include Renyi entropies, and we observe that contextual Renyi entropies also encode sufficient information to allow for a complete state reconstruction. 

This approach has been inspired by one of the key ideas of the topos programme, namely that one can hope to obtain a complete description of a quantum system by looking at it from all possible classical perspectives and keeping track of the resulting information. 


The final chapter, namely Chapter \ref{gls}, contains a summary of the main results presented in this thesis, as well as a number of concluding remarks and directions for further work.

The original material presented in this thesis is drawn from four research papers. The material presented in Chapter \ref{CC} is based on:

\begin{enumerate}
\item S. Abramsky, C. M. Constantin, \textit{A Classification of Multipartite States by Degree of Non-locality}, EPCS, Proceedings of QPL 2013, arXiv:1412.5213, \cite{AbrCon}

Chapter \ref{LNL} is based on:
\item S. Abramsky, C. M. Constantin, S. Ying, \textit{Hardy is (almost) everywhere: Non-locality without inequalities for almost all entangled multipartite states}, arXiv:1506.01365, \cite{AbrConYing}

Chapter \ref{LA} is based on:
\item C. M. Constantin, A. D\"{o}ring, \textit{Reconstructing an Atomic Orthomodular Lattice from the Poset of its Boolean Sublattices}, Houston Journal of Mathematics, arXiv:1306.1950, \cite{ConDoe}

Finally, Chapter \ref{CE} is based on:
\item C. M. Constantin, A. D\"{o}ring, \textit{Contextual Entropy and Reconstruction of Quantum States}, arXiv:1208.2046, \cite{ConDoe1}
\end{enumerate}

\end{chapter}


\part{First part}
\newpage
$ $
\newpage

\begin{chapter}{The Sheaf Theoretic Structure of Contextuality and non-Locality}\label{SSC}

We begin this chapter by reviewing a number of fundamental notions in quantum information theory, such as entanglement, non-locality, hidden variables, Bell inequalities and logic gates. We show that these notions can be re-expressed in an elegant way using sheaf-theoretic language. The presentation is based on the textbook \cite{NieChu} and the original work of Abramsky, Brandenburger and Hardy \cite{AbrBra11, AbrHar12}.

\section{Entanglement, Non-locality and Contextuality}
\subsection{Qubits and qubit states}
Classical information theory relies on the concept of {\itshape bit}, understood as the basic unit of information, with possible values $0$ and $1$. Its quantum analogue, the {\itshape quantum bit}, or {\itshape qubit} for short, is understood as a quantum system with two states $\ket0$ and $\ket1$. More formally, a qubit is an element of a $2$-dimensional Hilbert space\footnote{Other frequent notations for a $2$-dimensional orthonormal basis, which we will also use in this work, are $\{\ket-, \ket+\}$ and $\{\ket{\,\downarrow\,}, \ket{\,\uparrow\,}\}$.} (usually over $\mathbb C$). While a classical bit can only have one value at a time, $0$ or $1$, its quantum counterpart is, in general, a superposition of $\ket 0$ and $\ket 1$:
$$\ket\psi = \alpha \ket0+\beta\ket1 $$
The qubit $\ket\psi$ is said to be in a pure state if $|\alpha|^2 +|\beta|^2=1$ and in a mixed state if it is a statistical mixture of different pure qubit states. 

Pure qubit states can be represented as points on the so-called Bloch sphere parametrized by $0\leq \theta \leq\pi$ and $0\leq\phi<2\pi$, by setting $\alpha = e^{i\phi}\cos(\theta/2)$ and $\beta=\sin(\theta/2)$. On the Bloch sphere, the North and the South poles correspond to the points accessible to a classical bit, while the interior of the Bloch sphere corresponds to mixed states.

A system of $n$-qubits represents a finite-dimensional Hilbert space over the complex numbers of dimension $2^n$. A state $\ket\psi$ in this Hilbert space is represented as a linear combination 
$$ \ket\psi\ =\sum_{j_1,j_2,\ldots,j_n=0,1} c_{j_1,j_2,\ldots,j_n} \ket{j_1}\otimes \ket{j_2}\otimes \dots \otimes \ket{j_n}$$
where $\left| 0\right.\!\rangle$ and $\ket{1}$ are the two states of a qubit. The vectors $\ket{j_1}\otimes\dots\otimes\ket{j_n}$ form an orthonormal basis and are usually labelled by binary strings of length $n$, $\ket{j_1j_2\dots j_n}$. 

A quantum computation process on $n$ qubits represents the following sequence of steps. First, we assemble $n$ qubits and prepare them in a standard initial state. Then we apply a unitary transformation $U$, usually written as a product of {\itshape quantum gates}, i.e.~unitary transformations that act on a small number of qubits. Finally, we measure all the qubits, by projecting on the $\left\{\ket0,\ket1\right\}$ basis, which represents the (probabilistic) outcome of the computation. 


\subsection{Entanglement}
Multiple qubits can exhibit quantum entanglement. {\itshape Entanglement} is a striking feature of quantum mechanics, on which many of its information theoretic successes that go beyond classical physics, such as quantum teleportation~\cite{Bennett:1992tv}, superdense coding~\cite{Bennett:1992zzb, Shor:1994jg} and quantum cryptography~\cite{Ekert:1991zz}, rely. As a result of interactions, quantum systems (e.g.~qubits), can become {\itshape entangled}, giving rise to correlations between the properties of the constituent systems which persist even when these become spatially separated. The possibility of having entangled systems over large distances, leading to non-local correlations, made Albert Einstein, Boris Podolski and Nathan Rosen to conclude in 1935 that quantum mechanics must be an incomplete description of reality \cite{Einstein:1935rr}. However, nearly 30 years later, John Stewart Bell realised in his groundbreaking paper ``On the Einstein Podolsky Rosen Paradox" \cite{Bell:1964kc} that any physical theory which would satisfy the principle of locality (stating that an object is directly influenced only by its immediate surroundings) must necessarily satisfy certain inequalities. Bell then showed that the predictions of quantum mechanics violate these inequalities. We will examine this circle of ideas into more detail in Section \ref{NLB}. 


Since quantum entanglement appears as an indispensable resource in quantum information processing, we will give particular attention to understanding the ways in which it can be quantified. To begin with, let us recall the basic definition of entanglement for bipartite (2-qubit) systems. 
%

\begin{definition} 
Let $\mathcal H_1$ and $\mathcal H_2$ be two Hilbert spaces and $\ket\psi\in \mathcal H_1\otimes \mathcal H_2$. Then $\ket\psi$ is said to be {\itshape disentangled}, or {\itshape separable} or a {\itshape product state} if $\ket\psi = \ket{\psi_1}\otimes \ket{\psi_2}$, for some $\ket{\psi_1}\in \mathcal H_1$ and $\ket{\psi_2}\in \mathcal H_2$. Otherwise, $\ket\psi$ is said to be {\itshape entangled}.
\end{definition}

\begin{example}
The {\itshape EPR-state (Einstein-Podolsky-Rosen state)} is defined as 
\begin{equation}
\ket{\psi}_{\text{EPR}} = \frac{1}{\sqrt 2} \left(\ket{0}\otimes \ket1 - \ket1\otimes\ket0 \right)
\end{equation}
One can easily check that the EPR-state cannot be written as a product state. 
\end{example}

\begin{example}
In the 2-qubit Hilbert space $\mathbb C^2\otimes \mathbb C^2$, the {\itshape Bell states} are given by: 
\begin{equation}
\begin{aligned}
\ket{\Phi^+} &= \frac{1}{\sqrt 2} \left(\ket{0}\otimes \ket0 +\ket1\otimes\ket1 \right)\ \ \ \ \ \ \ \ket{\Phi^-} = \frac{1}{\sqrt 2} \left(\ket{0}\otimes \ket0 -\ket1\otimes\ket1 \right)\\
\ket{\Psi^+} &= \frac{1}{\sqrt 2} \left(\ket{0}\otimes \ket1 +\ket1\otimes\ket0 \right)\ \ \ \ \ \ \ \ket{\Psi^-} = \frac{1}{\sqrt 2} \left(\ket{0}\otimes \ket1 -\ket1\otimes\ket0 \right)
\end{aligned}
\end{equation}
where $\left\{\ket0,\ket1\right\}$ represents an arbitrary orthonormal basis of the 1-qubit Hilbert space $\mathbb C^2$. The Bell states are also entangled. 
\end{example}
\vspace{-4pt}

\subsection{Measures of entanglement for bipartite states}
The distinction between separable and entangled states is mathematically straightforward. However, in practice, it is generally difficult to discern. As such it is important to find operational criteria to test separability.
 
The most crude measure of entanglement is the {\itshape Schmidt number} (also called the {\itshape Schmidt rank}). For a state $\ket{\psi}\in \mathcal H_A\otimes \mathcal H_B$ the Schmidt number over $\mathcal H_A\otimes \mathcal H_B$ is the smallest number $\text{Sch}(\ket\psi, \mathcal H_A, \mathcal H_B)$ such that $\ket\psi$ can be written as
$$\ket\psi= \sum_{j=1}^{\text{Sch}(\ket\psi, \mathcal H_A, \mathcal H_B)}  \ket{u_j}\otimes \ket{v_j}$$ 
where $\ket{u_j}\in \mathcal H_A$ and $\ket{v_j} \in \mathcal H_B$. Thus a separable state has Schmidt number 1 and an entangled state has Schmidt number greater than 1. The Schmidt number can also be defined using the Schmidt decomposition of $\ket\psi$ over $\mathcal H_A\otimes \mathcal H_B$: 
$$\ket\psi = \sum_{j=1}^{\text{min}(\text{dim}\mathcal H_A, \text{dim}\mathcal H_B)} s_j \ket{u_j} \otimes \ket{v_j}$$
where $\ket{u_j}$ form an orthonormal basis of $\mathcal H_A$ and $\ket{v_j}$ are an orthonormal basis of $\mathcal H_B$ and the coefficients $s_i$ are non-negative. The strictly positive coefficients are called Schmidt coefficients and their number corresponds to the Schmidt number. Alternatively, form the rank 1 matrix $\rho = \ket\psi\bra\psi$ and take the partial trace with respect to either $A$ or $B$. This will be a diagonal matrix with non-zero elements $|s_i|^2$. Thus the Schmidt number can also be found by counting the non-zero eigenvalues of $\rho_A:=\text{tr}_B\rho$ or, equivalently, the non-zero eigenvalues of $\rho_B:=\text{tr}_A\rho$.

For bipartite pure states, the above comments lead to another measure of entanglement, the {\itshape entropy of entanglement}. Entangled states have Schmidt number greater than 1, therefore a bipartite pure state is entangled if and only if its reduced states are mixed states. Consequently, the von Neumann entropy of either reduced state gives a well defined measure of entanglement: 
$$ S = - \text{tr} (\rho_A \log \rho_A )= - \text{tr} (\rho_B \log \rho_B )$$
which can be written as the Shannon entropy 
$S = -\sum_i  p_i \log p_i$
where, e.g.~$\rho_A$ is written in terms of its eigenvectors $\rho_A = \sum_i p_i\ket{i}\bra{i}$. 

Many other entanglement measures for bipartite states exist, such as topological entanglement entropy, entropy of formation and dilution, squashed entanglement etc. For a review of these and other related measures, see~\cite{Plenio:2007zz}.

\subsection{Measures of entanglement for multipartite states}
For multiple qubit states, one has to distinguish between states that have all subsystems entangled and those in which only certain subsets of qubits are entangled. To this end, one defines the notion of {\itshape biseparability} as a property of $n$-qubit states for which there exists a partition of the qubits into two disjoint subsets $A$ and $B$, such that $|A|+|B|=n$ and the original state is a product state with respect to the partition $A|B$.

However, the notion of biseparability (or multi-separability) is only a crude measure of entanglement. To see that, consider the following two tripartite states, known as the Greenberger-Horne-Zeilinger (GHZ) state and the W-state: 
\begin{equation*}
\begin{aligned}
\ket{\text{GHZ}_\pm} &=\frac{1}{\sqrt 2}\left(\ket{000}\pm\ket{111} \right)\\
\ket{\text{W}} &= \frac{1}{\sqrt 3} \left( \ket{001}+\ket{010}+\ket{100}\right)
\end{aligned}
\end{equation*}
Clearly, both states have all three qubits entangled (non-biseparable). The essential difference between these states can be seen if a measurement is performed on one of the three qubits; after measurement, the state is separated in the case of the GHZ-state (it is either $\ket{000}$ or $\ket{111}$), while the W-state remains entangled. Put differently, if one of the three qubits is lost (traced out), the reduced 2-qubit state is separated for the GHZ-state, while for the W-state it is entangled. One says that the entanglement properties of the W-state are robust with respect to particle loss, and fragile for the GHZ state. 

It can be shown \cite{Dur:2000zz} that the W- and the GHZ-states cannot be transformed into one another by any protocol involving local quantum operations (LO), i.e.~transformations which factor out as tensor products of local operators on each qubit and classical communication systems (CC) between the three parties (LOCC). Moreover, D\"ur, Vidal and Cirac showed in \cite{Dur:2000zz} that any non-biseparable three-qubit state can be transformed into either the W- or the GHZ-state.

\subsection{Non-locality, Bell's theorem and hidden variables}\label{NLB}

One could, as Bell did, ask the following question: if qubits were replaced by local classical variables (i.e.~classical states with definite values for each qubit at every moment of time), would it then be possible to reproduce the outcome of any quantum computation process by replacing each quantum gate with classical operations which are randomly chosen from a set of possible transformations?

The idea of using classical variables instead of qubits is in accord with the principle of {\itshape local realism} (in the sense of Einstein, \cite{Einstein}), according to which any two objects $A$ and $B$ that are far apart in space must exhibit relative independence, such that any external influence on $A$ has no direct (i.e.~instantaneous or faster-than-light) influence on $B$ ({\itshape Principle of Local Action}) and moreover, material objects have properties independent of any observation, and the results of any possible measurement depend on these properties
({\itshape realism}). 

The proposal to describe the measurement process by randomly chosen classical operations, instead of quantum gates, is linked to the class of so-called {\itshape local hidden variable theories}. Such theories attempt to obtain the non-classical features of Quantum Mechanics as emergent phenomena of a Local Realist theory. This is based on the conjecture that quantum mechanics, as a theory which violates the principle of local realism, cannot be a complete theory. As such one needs to supplement quantum mechanics with additional variables that determine the results of individual measurements, in order to restore causality and locality. 
The answer to the above question is provided by Bell's theorem: 
\begin{theorem}{\bfseries{(Bell)}}
There is no {\itshape local probabilistic algorithm}, or equivalently, no {\itshape local hidden variable theory} that can reproduce the conclusions of quantum mechanics.
\end{theorem}

Bell's theorem expresses the fact that if one considers correlations between ideal measurements on entangled states, the predictions of quantum mechanics violate certain inequalities which are necessary conditions for local realism. Thus Bell's theorem imposes an incompatible alternative between any local hidden variable theory and quantum entanglement as described in quantum mechanics. 


Bell's theorem relies on the notion of {\itshape non-local correlations}. The existence of non-local correlations makes it impossible to decipher a generic (entangled) multiple-qubit state by dividing the system into parts and studying each part separately. Let us consider the following set-up, in which Alice and Bob have (each of them) access to one qubit of an entangled 2-qubit state. Suppose Alice can perform measurements $a$ and $a'$ and Bob can perform $b$ and $b'$ and assume that the possible outcomes for any of these measurements are 0 and 1. Then, assuming the existence of local hidden variables $s$, the probability that Alice and Bob obtain the outcomes which we denote by $\alpha$ and $\beta$, when measuring the observables $a$ and $b$, respectively, is given by: 
\begin{equation} \label{eq:Bell1}
P(a,\alpha; b,\beta) = \int \rho(s)\, P_1(s; a,\alpha)\, P_2(s; b,\beta)\, ds~, \ \ \ \alpha, \beta \in \{0,1\}
\end{equation}
Here $s$ represents one or several parameters which contain all the relevant information about the past interaction between the two qubits, $\rho(s)\geq 0$ is a probability density normalised to unity and $0\leq P_1(s; a,\alpha), P_2(s; b,\beta)\leq1$. The value of $P_1(s; a,\alpha)$ is assumed to be independent of the measurement performed on the second qubit (locality). Following \cite{Clauser:1969ny, Clauser:1974tg}, one can show that: 
\begin{equation*}
-1 \leq P(a,\alpha; b,\beta) -P(a,\alpha; b',\beta') + P(a',\alpha'; b,\beta)  + P(a',\alpha'; b',\beta')  - P(a',\alpha') -P(b,\beta) \leq 0 
\end{equation*}
where $\alpha',\beta'\in\{0,1\}$ and $P(a,\alpha) = P(a,\alpha; b,0)+ P(a,\alpha; b,1) = P(a,\alpha; b',0) + P(a,\alpha; b',1)$. The above relation is known as the CH74 inequality, obtained by Clauser and Horne in 1974 and we will see in the next section that it is violated by the predictions of quantum mechanics.

\subsection{Measurement contexts}\label{sec:MeasurementContexts}
\vspace{-10pt}
Suppose a 2-qubit system (or rather, a large collection of identically prepared 2-qubit systems) is prepared in one of the four Bell states, say $\ket{\Phi^+}$. Assume Alice measures in the basis $\{\ket0_A,\ket1_A\}$ with possible outcomes $0$ and $1$, and similarly for Bob. Then both Alice and Bob measure $0$ and $1$ with equal frequency. In fact, they would obtain the same probabilities for measurements done in any arbitrary orthonormal basis, e.g.~of the type:
$$\left\{\ket{\alpha_0} =  \cos\alpha\ket 0 +\sin\alpha\ket 1\; ,  \ \ket{\alpha_1=\alpha_0 + \pi/2} =  -\sin\alpha\ket 0 +\cos\alpha\ket 1\right\}$$
To be more precise, when Alice measures in the basis $\{\ket{\alpha},\ket{\alpha+\pi/2}\}$, she obtains the outcomes $0$ and $1$ with equal probabilities independently of Bob's choice of measurement $\{\ket{\beta},\ket{\beta+\pi/2}\}$. For our set-up, the statement can be easily verified by computing the required probabilities with the formula given by Equation \ref{ip} in Section \ref{CQS}. In general, the statement that Alice's outcome cannot be influenced by Bob's choice of measurement is known as the {\itshape no-signalling condition}. 

\vspace{12pt}
If Alice and Bob communicate to each other only the frequencies with which they obtain their different outcomes, they have no way of distinguishing between the four Bell states. In order to discern, they have to correlate their measurements, and communicate the frequencies with which the correlated outcomes are obtained, as summarised in the table below: 
$$\begin{array}{c|c|c|c}
\varstr{14pt}{9pt}
~~~(0,0) ~~~& ~~~(0,1)~~~ & ~~~(1,0)~~~ & ~~~(1,1)~~~ \\
\hline
\varstr{14pt}{9pt}~~~1/2~~~ & ~~~0 ~~~& ~~~0 ~~~& ~~~1/2~~~
\end{array}
 $$
 
With this information, they can decide that the original state can only be $\ket{\Phi^+}$ or $\ket{\Phi^-}$. In order to further distinguish between these, Alice and Bob have to measure other observables. For example, Alice and Bob could perform the following measurements:
$$\begin{array}{c|c|c}
\varstr{14pt}{9pt}
~~~\text{Alice / Bob}~~~ & ~~~\text{Eigenbasis}~~~ & ~~~\text{Outcomes}~~~\\
\hline
\varstr{14pt}{9pt} a & \ket0,~~\ket1& 0,1\\ \hline
\varstr{14pt}{9pt} a' & ~~\cos \frac{\pi}{6} \ket0 + \sin \frac{\pi}{6} \ket 1, ~~~ -\sin \frac{\pi}{6}\ket0+\cos \frac{\pi}{6}\ket1 ~~& 0,1\\\hline
\varstr{14pt}{9pt} b & \ket0,~~\ket1& 0,1\\ \hline
\varstr{14pt}{9pt} b' &~~ \cos \frac{\pi}{6} \ket0 - \sin \frac{\pi}{6} \ket 1,~~~ ~~\sin \frac{\pi}{6}\ket0+\cos \frac{\pi}{6}\ket1 ~~& 0,1\\
\end{array}
$$

In this scenario, Alice can choose between measurements $a$ and $a'$ and Bob between $b$ and~$b'$. A particular choice will be called, from now on, a {\itshape measurement context}. Thus there are four possible measurement contexts, 
$ \{a,b\}, \{a,b'\},\{a',b\},\{a',b'\}$.

By elementary operations, one can show that, for $\ket{\alpha} = \cos\alpha\ket 0 +\sin\alpha\ket 1$ and $\ket\beta = \cos\beta\ket 0 +\sin\beta\ket 1$, the following expressions hold: 
\begin{align*}
 \big|\!\left(\bra\alpha \otimes\bra\beta \right) \ket {\Phi^+} \big|^2 &= \frac{1}{2}\cos^2(\alpha - \beta) \\
  \big|\!\left(\bra\alpha \otimes\bra\beta \right) \ket {\Phi^-} \big|^2 &= \frac{1}{2}\cos^2(\alpha + \beta) 
\end{align*}

Thus, if the entangled state corresponds to $\Phi^+$, Alice and Bob would report the following set of correlated probabilities:
$$\begin{array}{c|c|c|c|c}
\varstr{14pt}{9pt}
~~~A~~~~~B~~~&~~~(0,0) ~~~& ~~~(0,1)~~~ & ~~~(1,0)~~~ & ~~~(1,1)~~~ \\
\hline
\varstr{14pt}{9pt}~~~a~~~~~b~~~ &~~~1/2~~~ & ~~~0 ~~~& ~~~0 ~~~& ~~~1/2~~~\\\hline
\varstr{14pt}{9pt}~~~a~~~~~b'~~~ &~~~3/8~~~ & ~~~1/8 ~~~& ~~~ 1/8 ~~~& ~~~3/8~~~\\\hline
\varstr{14pt}{9pt}~~~a'~~~~~b~~~ &~~~3/8~~~ & ~~~1/8 ~~~& ~~~1/8 ~~~& ~~~3/8~~~\\\hline
\varstr{14pt}{9pt}~~~a'~~~~~b'~~~ &~~~1/8~~~ & ~~~3/8 ~~~& ~~~3/8 ~~~& ~~~1/8~~~\\
\end{array}
 $$
while if the state was $\Phi^-$, they would obtain:
$$\begin{array}{c|c|c|c|c}
\varstr{14pt}{9pt}
~~~A~~~~~B~~~&~~~(0,0) ~~~& ~~~(0,1)~~~ & ~~~(1,0)~~~ & ~~~(1,1)~~~ \\
\hline
\varstr{14pt}{9pt}~~~a~~~~~b~~~ &~~~1/2~~~ & ~~~0 ~~~& ~~~0 ~~~& ~~~1/2~~~\\\hline
\varstr{14pt}{9pt}~~~a~~~~~b'~~~ &~~~3/8~~~ & ~~~1/8 ~~~& ~~~ 1/8 ~~~& ~~~3/8~~~\\\hline
\varstr{14pt}{9pt}~~~a'~~~~~b~~~ &~~~3/8~~~ & ~~~1/8 ~~~& ~~~1/8 ~~~& ~~~3/8~~~\\\hline
\varstr{14pt}{9pt}~~~a'~~~~~b'~~~ &~~~1/2~~~ & ~~~0 ~~~& ~~~0 ~~~& ~~~1/2~~~\\
\end{array}
 $$

Now let us consider the CH74 inequality for the $\Phi^+$ state:
\begin{equation*}
P(a,\alpha; b,\beta) -P(a,\alpha; b',\beta') + P(a',\alpha'; b,\beta)  + P(a',\alpha'; b',\beta') \leq P(a',\alpha') + P(b,\beta) 
\end{equation*}
In our case, the right hand side is always $1/2+1/2$. In order to check the inequality, we need to take one entry from each row of the above table, in total 16 choices. Among these, some violate the CH74 inequality, e.g.:
$$P(a,0; b,0) -P(a,0; b',1) + P(a',0; b,0)  + P(a',0; b',1) = \frac{1}{2}-\frac{1}{8}+\frac{3}{8}+\frac{3}{8} = \frac{9}{8} \nleq 1  $$

\subsection{Contextuality and the Kochen-Specker theorem}

We will refer to a {\itshape measurement context} as a set of compatible measurements, i.e.~measurements that can be jointly performed in a certain experimental setup. In quantum mechanics, contexts correspond to commutative algebras of self-adjoint operators, which can be simultaneously diagonalised. This is, indeed, the use of contextuality that will be employed in the next chapters on the topos approach to quantum physics. 
The notion of contextuality discussed in the following section is more general, as it does not rely on the standard operator-algebra formulation of quantum mechanics. As such, the results obtained in this picture will be independent of the particular way in which quantum mechanics has been traditionally formulated and interpreted.

The notion of contextuality arises naturally from the question: is it possible to assign definite values to all (hidden) variables at any given moment of time, independently of the device used to measure them? The Kochen-Specker theorem shows that this is not the case, and for quantum systems with more than two levels, the value of an observable depends on which commuting set of observables is being measured along with it.  

The Kochen-Specker theorem asserts that quantum mechanics forbids the simultaneous existence of definite values for observables which cannot be measured together, that is it forbids non-contextual hidden variables.

\section{Non-locality and Contextuality in Sheaf-Theoretic Language}

In a series of papers \cite{AbrBra11, AbrShaneRui11, AbrHar12, AbrBra14}, Abramsky et al.~formulated a mathematical description of contextuality phenomena in terms of sheaves over a poset of contexts. Their work was partly inspired by \cite{b1}, in which Butterfield and Isham realised that the Kochen-Specker theorem can be elegantly reformulated in terms of the non-existence of global sections of a certain presheaf. 

While the spectral presheaf used by Butterfield and Isham was based on an operator algebra and, in this sense, it relied on concepts specific to quantum mechanics, the work of Abramsky and Brandenburger was carried out at a much higher level of generality, whithout using any of the characteristic mathematical structures of quantum mechanics. It moreover introduced many new key structures in order to be able to capture and analyse probability tables such as those used in Section~\ref{sec:MeasurementContexts} in the discussion on non-locality and contextuality. We shall describe some of these key structures in the rest of this section.

Let $X$ be a set of measurements and $O$ the set of possible outcomes for each measurement. For the present discussion, it suffices to consider the same $O$ for all measurements, although one could in general allow for different outcomes for each individual measurement. 

Suppose now that a subset $U\subset X$ of measurements are simultaneously performed. The outcomes obtained in this particular context are specified by a {\itshape section}: $ s:U\longrightarrow O $, which associates to a measurement $m\in U$ an outcome $s(m)$. The space of sections over $U$ is denoted by $\mathcal E(U)$ and is a subset of $O^U$, where $O^U$ denotes the set of all functions from $U$ to $O$. If one considers $\mathcal E(U)$, together with the natural restriction maps 
\begin{equation*}
\begin{aligned}
\ \ \ \ \ \ \ \ {\rm res}_U^{U'}: \mathcal E(U')& \longrightarrow \mathcal E(U) \ \ \ \ {\text{for any }} U\subseteq U'\\
s &\longmapsto s|_U
\end{aligned}
\end{equation*}
one immediately has ${\rm res}_U^U = {\rm id}_U$ and ${\rm res}_U^{U'} \circ{\rm res}_{U'}^{U''} = {\rm res}_{U}^{U''}$ for any $U\subseteq U'\subseteq U''$, thus $\mathcal E$ is a pre-sheaf, namely the pre-sheaf of sections over the poset  $\mathcal P(X)$. In other words, $\mathcal E$ is a functor $\mathcal E:\mathcal P(X)^{\rm op} \longrightarrow {\mathbf{Set}}$. 

The construction, $\mathcal E$ is in fact a sheaf since local sections can be uniquely glued together. That is, given a cover $\{U_i\}_{i\in I}$ of $U$ and a family of sections $\{s_i \in \mathcal E(U_i)\}_{i\in I}$ which are compatible on every intersection, i.e.~$s_i|_{U_i\cap U_j} = s_j|_{U_i\cap U_j}$, there is a unique section $s\in \mathcal E(U)$ such that its restriction to $U_i$ is $s_i$ for all $i\in I$. The sheaf condition is trivially satisfied, as one can always glue together partial functions on a discrete space which agree on overlaps by taking the union of their graphs. $\mathcal E$ will be called the {\itshape sheaf of events}.

\subsection{Measurement covers, the distribution functor and no-signalling}
So far we imposed no restrictions on the poset $\mathcal P(X)$, which we think of as a set of labels for different basic measurements. However, in quantum mechanics only certain measurements can be performed jointly, thus it makes sense to introduce the notion of {\itshape measurement cover} $\mathcal M\subset \mathcal P(X)$ of $X$, composed of measurement contexts only, with the property that the union of all contexts contained in $\mathcal M$ equals $X$. Additionally, we will require that $\mathcal M$ contains only maximal contexts (maximal sets of compatible measurements), that is if $C, C'\in \mathcal M$ and $C\subseteq C'$ then $C=C'$ (anti-chain condition on $\mathcal M$). 

\begin{example}
{\itshape Bell-type scenarios.} In the formulation of Bell-type theorems on non-locality, one refers to composed systems, whose parts may be space-like separated, in a fashion similar to our discussion in Section~\ref{sec:MeasurementContexts}. The class of measurement covers used in Bell-type scenarios and scenarios involving other non-local devices such as PR-boxes, can be described as follows. Let $I$ denote a set of indices labelling the different parts of a composed system and for each $i\in I$ let $X_i$ be the set of basic measurements that can be performed on the part labelled by $i$. Then $X$ is constructed as the disjoint union of the family $\{X_i\}_{i\in I}$. The measurement cover $\mathcal M$ is defined as the set of contexts containing exactly one measurement from each part, i.e.~we consider as compatible any two measurements performed on different parts of the system, but we do not allow for compatible measurements on the same part.

We shall refer to a Bell-type scenario which involves $n$ parts, each of which has $k$ possible choices of measurement, each choice with $l$ possible outcomes, as being of $(n,k,l)$-type. Note that for a system of $(n,k,l)$-type, there are $k^n$ measurement contexts, for each of which there are $l^n$ possible assignments of outcomes. Thus there are $(kl)^n$ sections over the contexts. The set of all measurements is of size $kn$, and there are $l^{kn}$ global assignments.
\end{example}


The last ingredient that we need refers to the probabilistic behaviour of quantum systems in their interaction with classical measuring apparatus. Standard probability distributions are represented by non-negative reals between $0$ and $1$, for which the usual addition and multiplication rules apply. Put differently, standard probability distributions are valued in the semiring $(\mathbb R_{\geq 0}, +, 0, \cdot, 1)$, where multiplication distributes over addition and $0$ and $1$ denote the units in the commutative monoids $(\mathbb R_{\geq 0},+,0)$ and $(\mathbb R_{\geq 0},\cdot, 1)$, respectively. If one allows for negative probabilities, the semiring $\mathbb R_{\geq 0}$ has to be extended to the reals $\mathbb R$. Interestingly, the results of Bell's theorem do not hold if one allows for negative probabilities. A third type of semiring of interest to us is the boolean semiring, $\mathbb B=(\{0,1\},\vee,0,\wedge,1)$. 

Fixing a semiring $R$, we define an $R$-distribution over a set $S$ as a function $d:S\rightarrow R$ with finite support (i.e.~$d$ is non-zero only on a finite subset of $S$, its support), satisfying the condition $$\sum_{x\in S} d(x) = 1$$ 
We denote by $\mathcal D_R(S)$ the set of $R$-distributions over $S$. In the case of the semiring $\mathbb R_{\geq 0}$, this is the set of probability distributions with finite support on $S$. In the case of the booleans $\mathbb B$, it is the set of non-empty finite subsets of $S$.
Furthermore, given a function $f:S\longrightarrow T$, we define: 
\begin{align*}
\mathcal D_R(f):\ \mathcal D_R(S) &\ \longrightarrow \ \mathcal D_R(T)\\
d &\ \longmapsto\  \bigg( t\mapsto \sum_{f(s)=t} d(s) \bigg)
\end{align*}
It can be easily seen that $\mathcal D_R$ is functorial, $\mathcal D_R(g\circ f) = \mathcal D_R(g)\circ \mathcal D_R(f)$ and $\mathcal D_R({\rm id}_S) = {\rm id}_{\mathcal D_R(S)}$. Thus we have defined a functor $\mathcal{D}_R:\mathbf{Set}\longrightarrow \mathbf{Set}$. 

By composing the distribution functor with the event sheaf $\mathcal E:\mathcal P(X)^{\rm op} \longrightarrow {\mathbf{Set}}$, we obtain a presheaf $\mathcal D_R\mathcal E: \mathcal P(X)^{{\rm op}}\longrightarrow \mathbf{Set}$. This assigns to a set of measurements $U$ the set $\mathcal D_R(\mathcal E(U))$ of distributions on sections over $U$. The above arrow function, applied to the restriction maps ${\rm res}_U^{U'}$, with $U\subseteq U'$, gives: 
\begin{align*}
\mathcal D_R\mathcal E(U')&\  \longrightarrow\ \mathcal D_R\mathcal E(U) \\
d& \ \longmapsto \ d|_U
\end{align*}
where for each $s\in \mathcal E(U)$,  
\begin{equation}\label{eq:resd}
d|_U :=\sum_{{s'\in\mathcal E(U');\ }{s'|_U=s}} d(s')
\end{equation}
Thus $d|_U$ is the {\itshape marginal} of the distribution $d$, in the sense that it associates to each section $s$ in the smaller context $U$ the sum of the weights of all sections $s'$ in the larger context which restrict to $s$.

\begin{example}
Let us consider a $(2,2,2)$ Bell-type scenario, in which Alice can perform measurements $a$ and $a'$ and Bob can perform $b$ and $b'$. There are two possible outcomes, $0$ and $1$, for each measurement.  The relevant measurement cover $\mathcal M$ consists of four maximal contexts: 
$$\mathcal M = \{ \{a,b\}, \{a',b\}, \{a,b'\},\{a',b'\} \}$$
Over each maximal context $C\in \mathcal M$, there are four sections, e.g.~for $C=\{a,b\}$ we have: 
$$\mathcal E(\{a,b\}) = \{\{a\rightarrow0,b\rightarrow0\},\{a\rightarrow0,b\rightarrow1\},\{a\rightarrow1,b\rightarrow0\},\{a\rightarrow1,b\rightarrow1\}\}$$

Further, assume that we can associate a distribution $e_C\in\mathcal D_R\mathcal E(C)$ to each context $C\in\mathcal M$. This gives a number $e_C(s)$ for each section $s\in \mathcal E(C)$. The distributions $e_C$ form the rows of the following table:
$$\begin{array}{c|c|c|c|c}
\varstr{14pt}{9pt}
~~~A~~~~~B~~~&~~~(0,0) ~~~& ~~~(0,1)~~~ & ~~~(1,0)~~~ & ~~~(1,1)~~~ \\
\hline
\varstr{14pt}{9pt}~~~a~~~~~b~~~ &~~~p_1~~~ & ~~~p_2 ~~~& ~~~ p_3 ~~~& ~~~p_4~~~\\\hline
\varstr{14pt}{9pt}~~~a~~~~~b'~~~ &~~~ p_5 ~~~ & ~~~ p_6 ~~~& ~~~ p_7 ~~~& ~~~ p_8 ~~~\\\hline
\varstr{14pt}{9pt}~~~a'~~~~~b~~~ &~~~ p_9 ~~~ & ~~~ p_{10} ~~~& ~~~ p_{11} ~~~& ~~~ p_{12} ~~~\\\hline
\varstr{14pt}{9pt}~~~a'~~~~~b'~~~ &~~~ p_{13} ~~~ & ~~~ p_{14} ~~~& ~~~ p_{15} ~~~& ~~~ p_{16} ~~~\\
\end{array}
 $$
\end{example}

The data given by $\mathcal M$, the sheaf of events $E$ and a family of distributions $\{e_C\}_{C\in\mathcal M}$, will be called an {\itshape empirical model}. In the standard case of probability tables, the numbers $p_i$ are non-negative integers, and the values along each row sum up to $1$; hence the distributions $e_C$ are probability distributions. One advantage of using empirical models is that they are formulated in an elegant, robust and general mathematical language which does not depend in particular on the Hilbert space formalism of quantum mechanics. 

\begin{example} One could also consider the weights $p_i$ to take values in other semirings, for example the boolean semiring. The following model is a possibilistic version of a non-local Hardy model \cite{Hardy93}, which can be viewed as specifying the {\itshape support} of a standard probabilistic Hardy model. 
$$\begin{array}{c|c|c|c|c}
\varstr{14pt}{9pt}
~~~A~~~~~B~~~&~~~(0,0) ~~~& ~~~(0,1)~~~ & ~~~(1,0)~~~ & ~~~(1,1)~~~ \\
\hline
\varstr{14pt}{9pt}~~~a~~~~~b~~~ &~~~1~~~ & ~~~1 ~~~& ~~~ 1 ~~~& ~~~1~~~\\\hline
\varstr{14pt}{9pt}~~~a~~~~~b'~~~ &~~~ 0 ~~~ & ~~~ 1 ~~~& ~~~ 1 ~~~& ~~~ 1 ~~~\\\hline
\varstr{14pt}{9pt}~~~a'~~~~~b~~~ &~~~ 0 ~~~ & ~~~ 1 ~~~& ~~~ 1 ~~~& ~~~ 1 ~~~\\\hline
\varstr{14pt}{9pt}~~~a'~~~~~b'~~~ &~~~ 1 ~~~ & ~~~ 1 ~~~& ~~~ 1 ~~~& ~~~ 0 ~~~\\
\end{array}
 $$

\end{example}

The {\itshape no-signalling condition} (the statement that the outcome obtained by Alice cannot be influenced by Bob's choice of a measurement) can be elegantly formulated in terms of the presheaf $\mathcal D_R\mathcal E$. We define a no-signalling empirical model for a measurement cover $\mathcal M$ to be an empirical model for which the family of distributions $\{e_C\}_{C\in\mathcal M}$ is compatible in the sense of the sheaf condition: for any $C,C'\in\mathcal M$,
\begin{equation}\label{eq:compatibility}
 e_C|_{C\cap C'} = e_{C'}|_{C\cap C'} 
\end{equation}
where the restrictions $e_C|_{C\cap C'}$ and $e_{C'}|_{C\cap C'}$ are defined as in Eq.~\ref{eq:resd}.

\begin{example}
For our previous example, involving two parties, Alice and Bob, consider the contexts $C=\{m_a,m_b\}$ and $C'=\{m_a,m_b'\}$. Fixing a section $s_0\in\mathcal E(\{m_a\})$, e.g.~$s_0=\{m_a\rightarrow 0\}$, the no-signalling condition can be expressed as
$$\sum_{\stackrel{s\in\mathcal E(C)}{s|_{m_a}=s_0}} e_C(s) = \sum_{\stackrel{s'\in\mathcal E(C')}{s'|_{m_a}=s_0}} e_{C'}(s')$$

For the above table, with $m_a=a, m_b=b,m_b'=b'$ and $s_0=\{a\rightarrow 0\}$, we obtain the condition: 
$$p_1+p_2 = p_5+p_6 $$
\end{example}

\subsection{Global Sections}
Let us recapitulate the mathematical structure introduced above. We started with a set of measurements $X$ and we defined the measurement sheaf of events $\mathcal E$, which assigns to each $U\subseteq X$ the set of sections $\{s:U\rightarrow O\}$. A particular section over $U$ specifies a particular set of outcomes, one outcome for each measurement in $U$. (Remember we assumed that all measurements have the same set of possible outcomes $O$).  Further, we defined the presheaf of $R$-valued distributions, $\mathcal D_R\mathcal E$, which assigns to each $U\subseteq X$  the set $\mathcal D_R(\mathcal E(U))$ of distributions on $\mathcal E(U)$, which is the set of sections over $U$. The distribution presheaf satisfied the important property~\ref{eq:resd}. 

Further, we had to deal with the factual requirement that, in general, not all measurements can be performed together. This introduces the notion of contexts and we formed the set $\mathcal M$ of maximal contexts (i.e.~maximal sets of compatible measurements). 
Since $\mathcal M\subset \mathcal P(X)$, passing to the measurement cover does not affect the definitions and properties of the event sheaf $\mathcal E$ and the distribution presheaf $\mathcal D_R\mathcal E $ on $X$.  
Each context $C\in\mathcal M$ comes with a set of sections~$O^C$. An empirical model is specified by providing, for each $C\in \mathcal M$, a probability distribution $e_C\in \mathcal D_R\mathcal E(C)$ over~$\mathcal E(C)=O^{C}$. These probabilities may or may not be compatible, in the sense of Eq.~\ref{eq:resd}, with a probability distribution over $\mathcal E(X)=O^X$. If such a `global' probability distribution exists, the model corresponds to a {\itshape deterministic hidden variable} scenario, in the sense that hidden variables determine which section $s\in\mathcal E(X)$ is being chosen in any certain situation and context\footnote{The unfamiliar reader might find helpful the following remarks: we have in mind the preparation of a state in the same conditions over and over again. If one had access to the complete set of variables that characterise this state, he should obtain the same set of outcomes each time he performs a `complete' measurement, as the state was prepared in the same conditions. The fact that, in practice (i.e.~in quantum mechanics), this does not happen, is explained in theories invoking hidden variables, by the existence of unaccessible variables which can alter the values of the accessible variables.}. The sheaf condition implies that the assignment of outcomes is independent of the measurement context, thus the existence of global probability distributions in the above sense is a form of {\itshape non-contextuality}. 

For no-signalling empirical models, the compatibility condition \ref{eq:compatibility} satisfied by the family of distributions $\{e_C\}_{C\in\mathcal M}$ implies that a global section $d\in \mathcal D_R\mathcal E(X)$, compatible with $\{e_C\}$, if it exists, will satisfy the sheaf condition with respect to $\mathcal M$. 

The existence of a global section $d\in \mathcal D_R\mathcal E(X)$ which restricts to yield the probabilities specified by the empirical model on each context $C\in\mathcal M$, i.e.~$d|_C=e_C$, has very important implications. Thus given a section $t\in \mathcal E(C)$,
\begin{equation}\label{eq:factorisation1}
e_C(t) = d|_C(t) = \sum_{\stackrel{s\in\mathcal E(X)}{s|_C=t}} d (s) = \sum_{s\in\mathcal E(X)} \delta_{s}|_C (t)\cdot d(s)
\end{equation}
where, for a global section $s\in\mathcal E(X)$, $\delta_{s}\in\mathcal D_R\mathcal E(X)$ is the distribution defined globally by $\delta_s(s)=1$ and $\delta_s(s')=0$ for $s\neq s'$. Hence, the presence of a global section $d\in\mathcal D_R\mathcal E(X)$, such that $d|_C=e_C$ turns out to be equivalent with the statement that the empirically observed probabilities $e_C(t)$ can be obtained by {\itshape averaging} over the hidden variables with respect to the distribution $d$. Further, from the definition of the globally induced distributions $\delta_s$, it follows that 
\begin{equation}\label{eq:factorisation2}
\delta_s|_C (t) = \prod_{x\in C} \delta_{s|_{\{x\}}} (t|_{\{x\}})
\end{equation}
which means that the probability distribution induced by $s$ {\itshape factorizes} as a product of probabilities assigned to individual measurements, independent of the context in which they appear. Note that this is precisely the notion of {\itshape locality} employed in Bell-type scenarios, e.g.~in Equation~\ref{eq:Bell1} we assumed that the probability  $P(s; a,\alpha; b,\beta)$ of the joint outcome $\{\alpha\mapsto a, \beta\mapsto b\}$ determined by the hidden variable $s$ is the product of the probabilities that $s$ determines for the outcomes $\{\alpha\mapsto a\}$ and $\{\beta\mapsto b\}$, i.e.~$P(s; a,\alpha; b,\beta)=P_1(s; a,\alpha)\, P_2(s; b,\beta)$ leading to:
\begin{equation*} 
P(a,\alpha; b,\beta) = \int \rho(s)\, P_1(s; a,\alpha)\, P_2(s; b,\beta)\, ds
\end{equation*}
in agreement with Equation~\ref{eq:factorisation1}.

We summarise the above discussion through the following:
\begin{proposition}
The existence of a global section for an empirical model implies the existence of a {\itshape local (non-contextual) deterministic hidden-variable model} which realises it. 
\end{proposition}

It is interesting to note that the converse of the above proposition also holds, as Abramsky and Brandenburger have shown in Theorem 8.1 \cite{AbrBra11}.

\subsection{Existence of Global Sections}
The above discussion leads to the following problem: {\itshape given an empirical model, determine it if has a global section}. Let us refine the expression of this problem. An empirical model consists of a measurement cover $\mathcal M$ (which covers a set of measurements $X$), together with the choice of a (local) section of the distribution presheaf $e_C\in\mathcal D_R\mathcal E(C)$, one over each context $C\in M$ -- more precisely, one distribution $e_C$ over each set of local sections (events) $\mathcal E(C)$. In addition, for no-signalling empirical models the local distributions $e_C$ satisfy the compatibility condition~\ref{eq:compatibility}. What we mean by a global section for an empirical model $(\mathcal M,\{e_C\})$ is a global section of the presheaf $\mathcal D_R\mathcal E$, whose restriction to each context $C\in\mathcal M$ agrees with the local sections $e_C$. Of course, global sections for the presheaf $\mathcal D_R\mathcal E$ always exist, since $\mathcal D_R\mathcal E(X)$ represents the set of distributions on the set $\mathcal E(X) = O^X$. 

We shall give a general linear-algebraic method of finding global sections for a given empirical model $(\mathcal M,\{e_C\})$. The main step involves constructing a \textit{incidence matrix} for the empirical model. This matrix is defined using only $\mathcal{M}$ and the event sheaf $\mathcal{E}$. It can be applied to any empirical model $(\mathcal M,\{e_C\})$ with respect to any distribution functor $\mathcal{D}_R$.

To define the incidence matrix, we first form the disjoint union $\coprod_{C\in\mathcal{M}} \mathcal{E}(C)$ of all the sections over the contexts in $\mathcal{M}$, and specify an enumeration $t_1,\ldots,t_p$ of this set. We also specify an enumeration $s_1,\ldots,s_q$ of all the global sections of the sheaf $\mathcal{E}$. We then form the $(p\times q)$-matrix $\mathbf{M}$, with entries defined as follows:
\begin{equation*}
\mathbf{M}[i,j]=\left\lbrace \begin{array}{cc} 1, &  s_j|_C=t_i , \text{ where }t_i\in \mathcal{E}(C)\\
0, &  \text{otherwise} 
\end{array}
\right.
\end{equation*}

It acts by matrix multiplication on distributions in $\mathcal{D}_R\mathcal{E}(X)$, seen as row vectors:
$d \mapsto (d|_C)_{C\in\mathcal{M}}$. The result of all multiplications of this form will be the set of families $\{e_C\}_{C\in \mathcal{M}}$ which arise from global sections.

A given empirical model assigns an element in the semiring $R$ to each section $t_i\in \coprod_{C\in\mathcal{M}} \mathcal{E}(C)$. Thus it can be specified by a vector $\mathbf{V}$ of length $p$, where $\mathbf{V}[i]=e_C(t_i)$. We can also introduce a vector $\mathbf{X}$ of length $q$ of `unknowns', one for each global section $s_j\in \mathcal{E}(X)$. Now a solution for the linear system $\mathbf{MX}=\mathbf{V}$ will be a vector of values in $R$, one for each $s_j$. 

In the case of Bell-type scenarios of $(n,k,l)$-type, such solutions are automatically distributions - due to the regular structure of the incidence matrix for these cases - and they are in bijective correspondence with global sections of the empirical model. For more general scenarios, one can simply augment $\mathbf{M}$ with an extra row, every entry in which is $1$, and similarly to augment $\mathbf{V}$ with an extra element, also equal to $1$. A solution for this augmented system enforces the constraint 
$$\mathbf{X}[1]+\ldots+\mathbf{X}[q]=1$$
and hence ensures that the first $q$ entries of $\mathbf{X}$ define a distribution on $\mathcal{E}(X)$.

\begin{example}
Consider a general $(2,2,2)$ Bell-type scenario:
$$\begin{array}{c|c|c|c|c}
\varstr{14pt}{9pt}
~~~A~~~~~B~~~&~~~(0,0) ~~~& ~~~(0,1)~~~ & ~~~(1,0)~~~ & ~~~(1,1)~~~ \\
\hline
\varstr{14pt}{9pt}~~~a~~~~~b~~~ &~~~p_1~~~ & ~~~p_2 ~~~& ~~~ p_3 ~~~& ~~~p_4~~~\\\hline
\varstr{14pt}{9pt}~~~a~~~~~b'~~~ &~~~ p_5 ~~~ & ~~~ p_6 ~~~& ~~~ p_7 ~~~& ~~~ p_8 ~~~\\\hline
\varstr{14pt}{9pt}~~~a'~~~~~b~~~ &~~~ p_9 ~~~ & ~~~ p_{10} ~~~& ~~~ p_{11} ~~~& ~~~ p_{12} ~~~\\\hline
\varstr{14pt}{9pt}~~~a'~~~~~b'~~~ &~~~ p_{13} ~~~ & ~~~ p_{14} ~~~& ~~~ p_{15} ~~~& ~~~ p_{16} ~~~\\
\end{array}
 $$

Here $X={a,b,a',b'}$ and the measurement cover is $\mathcal M = \{ \{a,b\}, \{a',b\}, \{a,b'\},\{a',b'\} \}$. Assume $O=\{0,1\}$ is the set of possible outcomes for all observables. Assume also that $R=\mathbb R_{\geq 0}$, i.e.~the numbers $p_1,\ldots,p_{16}$ are probabilities in the usual sense. Note that any measurement performed at at part $A$ is compatible with any measurement performed by part $B$. The event sheaf $\mathcal E$ has 16 global sections corresponding to:  
\begin{equation}
\begin{aligned}
(a,b,a',b')\in &\left\{(0,0,0,0),(0,0,0,1),(0,0,1,0),(0,0,1,1)\right.\\
&~~(0,1,0,0),(0,1,0,1),(0,1,1,0),(0,1,1,1)\\
&~~(1,0,0,0),(1,0,0,1),(1,0,1,0),(1,0,1,1)\\
&~\left.(1,1,0,0),(1,1,0,1),(1,1,1,0),(1,1,1,1)\right\}
\end{aligned}
\end{equation}
Suppose now there is a probability distribution $d$ over $\mathcal E(X)$ which associates probabilities $q_1,\ldots,q_{16}$ to the above global sections in the given order. Then one must have: 
\begin{equation}\label{eq:linearSystem}
\begin{array}{ll}
p_1=q_1+q_2+q_3+q_4 & ~~~~~~~p_2=q_5+q_6+q_7+q_8 \\
p_3=q_9+q_{10}+q_{11}+q_{12}& ~~~~~~~p_4=q_{13}+q_{14}+q_{15}+q_{16}\\
p_5=q_1+q_3+q_5+q_7& ~~~~~~~p_6=q_2+q_4+q_6+q_8\\
p_7=q_9+q_{11}+q_{13}+q_{15}& ~~~~~~~p_8=q_{10}+q_{12}+q_{14}+q_{16}\\
p_9=q_1+q_2+q_9+q_{10}&  ~~~~~~~p_{10}=q_3+q_4+q_{11}+q_{12}\\
p_{11}=q_5+q_6+q_{13}+q_{14}&  ~~~~~~~p_{12}=q_7+q_8+q_{15}+q_{16}\\
p_{13}=q_1+q_5+q_9+q_{13}&  ~~~~~~~p_{14}=q_2+q_6+q_{10}+q_{14}\\
p_{15}=q_3+q_7+q_{11}+q_{15}&  ~~~~~~~p_{16}=q_4+q_8+q_{12}+q_{16}
\end{array}
\end{equation}

and the $(16\times 16)$ incidence matrix is 

\begin{equation}\label{eq:incidenceMatrix}
\mathbf{M}=\left(
\begin{array}{cccccccccccccccc}
 1 & 1 & 1 & 1 & 0 & 0 & 0 & 0 &
   0 & 0 & 0 & 0 & 0 & 0 & 0 & 0
   \\
 0 & 0 & 0 & 0 & 1 & 1 & 1 & 1 &
   0 & 0 & 0 & 0 & 0 & 0 & 0 & 0
   \\
 0 & 0 & 0 & 0 & 0 & 0 & 0 & 0 &
   1 & 1 & 1 & 1 & 0 & 0 & 0 & 0
   \\
 0 & 0 & 0 & 0 & 0 & 0 & 0 & 0 &
   0 & 0 & 0 & 0 & 1 & 1 & 1 & 1
   \\
 1 & 0 & 1 & 0 & 1 & 0 & 1 & 0 &
   0 & 0 & 0 & 0 & 0 & 0 & 0 & 0
   \\
 0 & 1 & 0 & 1 & 0 & 1 & 0 & 1 &
   0 & 0 & 0 & 0 & 0 & 0 & 0 & 0
   \\
 0 & 0 & 0 & 0 & 0 & 0 & 0 & 0 &
   1 & 0 & 1 & 0 & 1 & 0 & 1 & 0
   \\
 0 & 0 & 0 & 0 & 0 & 0 & 0 & 0 &
   0 & 1 & 0 & 1 & 0 & 1 & 0 & 1
   \\
 1 & 1 & 0 & 0 & 0 & 0 & 0 & 0 &
   1 & 1 & 0 & 0 & 0 & 0 & 0 & 0
   \\
 0 & 0 & 1 & 1 & 0 & 0 & 0 & 0 &
   0 & 0 & 1 & 1 & 0 & 0 & 0 & 0
   \\
 0 & 0 & 0 & 0 & 1 & 1 & 0 & 0 &
   0 & 0 & 0 & 0 & 1 & 1 & 0 & 0
   \\
 0 & 0 & 0 & 0 & 0 & 0 & 1 & 1 &
   0 & 0 & 0 & 0 & 0 & 0 & 1 & 1
   \\
 1 & 0 & 0 & 0 & 1 & 0 & 0 & 0 &
   1 & 0 & 0 & 0 & 1 & 0 & 0 & 0
   \\
 0 & 1 & 0 & 0 & 0 & 1 & 0 & 0 &
   0 & 1 & 0 & 0 & 0 & 1 & 0 & 0
   \\
 0 & 0 & 1 & 0 & 0 & 0 & 1 & 0 &
   0 & 0 & 1 & 0 & 0 & 0 & 1 & 0
   \\
 0 & 0 & 0 & 1 & 0 & 0 & 0 & 1 &
   0 & 0 & 0 & 1 & 0 & 0 & 0 & 1
   \\
\end{array}
\right)
\end{equation}

One can solve the linear system $\mathbf{MX}=\mathbf{V}$ and obtain $\mathbf{X}=[q_1,\ldots, q_{16}]$ in terms of the given probabilities $\mathbf{V}=[p_1,\ldots,p_{16}]$ and require that $\mathbf{X}\geq\mathbf{0}$. If this system has a non-negative solution, then the model admits a realisation in terms of local deterministic hidden variables and is non-contextual, otherwise, the model is said to be {\itshape weakly contextual}. 
\end{example}

So far we have assumed that the distributions in $\mathcal D_R(\mathcal E (C))$ are real valued. We made this assumption explicitly or implicitly, by talking about probability distributions. It is both interesting and important to discuss models involving $\mathbb B$-valued distributions (where $\mathbb B$ is the boolean semiring). In fact, any {\itshape probabilistic model} (i.e.~a model involving $\mathbb R_{\geq 0}$-distributions) can be turned into a {\itshape possibilistic model} (i.e.~a model involving $\mathbb B$-distributions) by replacing all non-zero probabilities with~``1'' (True) and all null probabilities with ``0'' (False). Many quantum information-theoretic results, such as Bell's theorem, generalise to possibilistic models (see \cite{AbrHar12}). These considerations motivate the hierarchy of contextuality discussed below.

\section{A Hierarchy of Contextuality} \label{bckgc}

A probabilistic model $(X,\mathcal M, \{e_C\}, \mathbb R_{\geq 0})$ is {\itshape non-contextual} if there exists a global probability distribution $d\in \mathcal D_{\mathbb R_{\geq 0}}\mathcal E(X)$ compatible with the local probability distributions $\{e_C\}$. The compatibility condition can be expressed by the equivalence between $(X,\mathcal M, \{e_C\}, \mathbb R_{\geq 0})$ and $(X,\mathcal M, \{d|_C\}, \mathbb R_{\geq 0})$. If no such global distributions exist, the model is {\itshape weakly contextual}. 

The probabilistic model $(X,\mathcal M, \{e_C\}, \mathbb R_{\geq 0})$ induces a possibilistic model $(X,\mathcal M, \{e^{\mathbb{B}}_C\}, \mathbb B)$ with boolean-valued distributions. The boolean-valued distributions $\{e^{\mathbb{B}}_C\}$ assign the boolean value ``1'' to those sections which are in the support of the corresponding probability distributions $\{e_C\}$.
If the model $(X,\mathcal M, \{e^{\mathbb{B}}_C\}, \mathbb B)$ admits a global boolean distribution $d\in \mathcal D_{\mathbb B}\mathcal E(X)$ compatible with $\{e^{\mathbb{B}}_C\}$, then $(X,\mathcal M, \{e_C\}, \mathbb R_{\geq 0})$ is either non-contextual or weakly contextual. Otherwise, it is said to be {\itshape logically contextual}. 

\begin{remark}\label{rk:LW}
A logically contextual model is automatically weakly contextual, since a compatible global probability distribution $d\in \mathcal D_{\mathbb R_{\geq 0}}\mathcal E(X)$ induces a compatible global boolean distribution $d\in \mathcal D_{\mathbb B}\mathcal E(X)$. The converse, however, is not true: there are weakly contextual models which are not logically contextual. In general, we say that an empirical model is \textbf{probabilistically non-extendable} if it has no global section over $\mathcal{D}_{\mathbb{R}\geq 0}$, and \textbf{possibilistically non-extendable} if  it has no global section over $\mathcal{D}_{\mathbb{B}}$.
\end{remark}

Working over the boolean ring comes with a number of simplifications. For instance, the set of distributions $\mathcal D_{\mathbb B}\mathcal E(C)$ can be identified with $\mathcal P(\mathcal E(C))$ for any subset $C\subseteq X$. Indeed, a boolean distribution $d\in \mathcal D_{\mathbb B}\mathcal E(X)$ is a way of distinguishing between `possible' sets of outcomes and `impossible' sets of outcomes. Moreover, the boolean distribution $d$ induces a possibilistic model $(X,\mathcal M, \{d|_C\}, \mathbb B)$. If this model reproduces $(X,\mathcal M, \{e^{\mathbb{B}}_C\}, \mathbb B)$, or, equivalently, the support of the model $(X,\mathcal M, \{e_C\}, \mathbb R_{\geq 0})$, the model is not logically contextual. 

There is a further refinement to this hierarchy of contextuality. Let $\{s_1,\ldots,s_n\}=\mathcal E(X)$ be the set of global sections of the events sheaf. The global sections of the presheaf $\mathcal D_{\mathbb B}\mathcal E$ are boolean distributions over $\mathcal E(X)$. Let $d_1,\ldots,d_n$ be the global distributions with minimal support $\{s_1\},\ldots,\{s_n\}$, respectively. These distributions induce possibilistic models $(X,\mathcal M, \{d_1|_C\}, \mathbb B),$ $\ldots,$ $(X,\mathcal M, \{d_n|_C\}, \mathbb B)$. If all these models are such that their support falls outside of the support of the original model $(X,\mathcal M, \{e_C\}, \mathbb B)$, then any other distribution with a larger support $\{s_{i_1},s_{i_2},\ldots\}$ will produce a possibilistic model whose support will also fall outside of the support of the original model $(X,\mathcal M, \{e_C\}, \mathbb B)$. In this case, the model $(X,\mathcal M, \{e_C\}, \mathbb B)$ is said to be {\itshape strongly contextual}. In particular, the model is also logically contextual. However, there are logically contextual models which are not strongly contextual.

In the case of dichotomic measurements (i.e. measurements with only two outcomes) we can make a connection to logic by interpreting one outcome as \textit{true} and the other outcome as \textit{false}. This allows us to think of the set of measurements $X$ as a set of boolean variables. If $C$ is a finite context, any subset of $\mathbf{2}^C$ can be defined by a propositional formula. For example, each joint outcome $s:C\rightarrow \mathbf{2}$ determines a propositional formula 
$$\varphi_s \; = \; \bigwedge_{x\in C,\, s(m)=True}m \; \AND \; \bigwedge_{x\in C,\,s(m)=False} \neg m$$
The only satisfying assignment of $\varphi_s$ in $\mathbf{2}^C$ is $s$.

The propositional formula whose set of satisfying assignments is the support of $C$ is $$\varphi_C:=\bigvee_{s\in S(C)} \varphi_s.$$

An empirical model is \emph{logically contextual} if there exists some $C \in \mathcal{C}$ and some $s\in S(C)$ such that the formula 
$$\Phi \; = \; \varphi_{s} \; \AND \; \bigwedge_{V\in \mathcal{C}\backslash C} \varphi_{V}$$ 
is not satisfiable. 
This says that there is a possible joint outcome $s$ which cannot be accounted for by any valuation on all the variables in $X$ which is consistent with the support of the model. This immediately implies that there is no joint distribution on all the observables which marginalizes to yield the empirically observable probabilities as in Remark \ref{rk:LW}.
As originally shown in the bipartite case by Fine \cite{fine1982hidden}, and in a very general form in \cite{AbrBra11}, the non-existence of a joint distribution is equivalent to the usual definition of non-locality as given by Bell \cite{Bell:1964kc}.

\subsection{Strong contextuality as a CSP}

A constraint satisfaction problem (CSP) \cite{32,33} is specified by a triple $(V,K,\mathcal{R})$ where $V$ is a finite set of variables, $K$ is a finite set of values, and $\mathcal{R}$ is a finite set of constraints. A constraint is a pair $(C,S)$ where $C\subseteq V$ and $S\subseteq K^C$. This formulation is equivalent to the more common one where a constraint is specified as a list of $k$ variables together with a set of $k$-tuples of values. An assignment $s:V\rightarrow K$ satisfies a constraint $(C,S)$ if $s|_C\in S$. A solution of the CSP $(V,K,\mathcal{R})$ is an assignment $s:V\rightarrow K$ which satisfies every constraint in $\mathcal{R}$.

To any empirical model $e$ defined over a cover $\mathcal{M}$, with outcome set $O$, we can associate the CSP $(X,O,\{\mathrm{supp}(e_C)\}_{C\in\mathcal{M}})$ with $e$. An empirical model $e$ is maximally contextual if and only if the corresponding CSP has no solution.

In the case of dichotomic measurements (i.e. measurements with only two possible outcomes), the CSP reduces to a boolean satisfiability problem. In this case global sections correspond precisely to satisfying assignments for the formula
$$\varphi_e=\bigwedge_{C\in\mathcal{M}} \varphi_C$$
Hence an empirical model $e$ is strongly contextual if and only if the above formula is unsatisfiable.




\subsection{Contextuality for Quantum States}\label{CQS}

The hierarchy of contextuality which has so far been defined for abstract empirical models can be naturally lifted to apply to quantum states. If we fix observables for each party, a $n$-qubit quantum state gives rise to an empirical model. We shall mainly be concerned with $(n,2,2)$ scenarios where each of the $n$ parties has access to one qubit of a quantum state and can choose to perform one out of two available measurements on their qubit, each with two possible outcomes. The two available measurements can be represented by $2\times 2$ self-adjoint unitaries, $A$ and $B$. The two possible outcomes of a given measurement correspond to the two eigenvalues, $+1$ and $-1$ of the corresponding unitary. We usually associate the Boolean value $True$ to the $+$ eigenvalue and the Boolean value $False$ to the $-$ eigenvalue.

A compatible set of measurements is given by a choice of either $A$ or $B$ at each of the $n$ measurement sites. A given quantum state $|\Psi\rangle$ determines a probability distribution on the set of possible outcomes for each compatible set of measurements. Thus if $X_1X_2\ldots X_n$ is a compatible set of measurements (i.e. $X_i\in\{A,B\}$ for all $i$) the probability of obtaining the outcome $\sigma_1\sigma_2\ldots \sigma_n$ where $\sigma_i\in\{+,-\}$ is given by the squared norm of the inner product
\begin{equation}\label{ip}
 P(\Psi, \sigma_1\sigma_2\ldots \sigma_n)=|\langle e_1|\otimes \langle e_2|\otimes\ldots \otimes\langle e_n|\Psi\rangle |^2
\end{equation}

where $e_i$ is the eigenvector corresponding to the $\sigma_i$ eigenvalue of the unitary $X_i$ representing the measurement performed by the $i^{th}$ party. 


Thus for fixed observables $A$ and $B$ as above, a quantum state gives rise to an empirical model with rows indexed by $n$-tuples $X_1X_2\ldots X_n$, $X_i\in\{A,B\}$ and columns indexed by $n$-tuples $\sigma_1\sigma_2\ldots\sigma_n$. The section on a given row and column is given by formula (\ref{ip}) above. The section is in the support of the model if and only if the inner product (\ref{ip}) is non-zero.

The observables most frequently used in this thesis are the three Pauli matrices:

\begin{equation}
X=\left(\begin{array}{cc} 0&1\\ 1&0\end{array}\right), \ \ \ Y=\left(\begin{array}{cc} 0&-i\\ i&0\end{array}\right), \ \ \ Z=\left(\begin{array}{cc} 1&0\\ 0&-1\end{array}\right)
\end{equation}


\subsection{The Bell state}

We look again at the $(2,2,2)$-type scenario which can be represented in quantum mechanics by measuring the $\ket{\Phi^+}$ Bell state, as discussed in Section \ref{sec:MeasurementContexts}. This scenario is specified by an empirical model summarized in the following table:

$$\begin{array}{c|c|c|c|c}
\varstr{14pt}{9pt}
~~~A~~~~~B~~~&~~~(0,0) ~~~& ~~~(0,1)~~~ & ~~~(1,0)~~~ & ~~~(1,1)~~~ \\
\hline
\varstr{14pt}{9pt}~~~a~~~~~b~~~ &~~~1/2~~~ & ~~~0 ~~~& ~~~0 ~~~& ~~~1/2~~~\\\hline
\varstr{14pt}{9pt}~~~a~~~~~b'~~~ &~~~3/8~~~ & ~~~1/8 ~~~& ~~~ 1/8 ~~~& ~~~3/8~~~\\\hline
\varstr{14pt}{9pt}~~~a'~~~~~b~~~ &~~~3/8~~~ & ~~~1/8 ~~~& ~~~1/8 ~~~& ~~~3/8~~~\\\hline
\varstr{14pt}{9pt}~~~a'~~~~~b'~~~ &~~~1/8~~~ & ~~~3/8 ~~~& ~~~3/8 ~~~& ~~~1/8~~~\\
\end{array}
 $$

We are interested in finding a probability distribution on the global assignments $\mathcal{E}(X)$. This amounts to solving $\mathbf{MX}=\mathbf{V}$ over the reals, subject to the constraint $\mathbf{X}\geq \mathbf{0}$, which is in fact a linear programming problem. 

However, it is easier in this case to give a direct argument, as in \cite{AbrBra11}, showing that there is no such solution. This implies that the above model has no hidden-variable realization, which proves Bell's theorem \cite{Bell:1964kc}. The argument starts by considering $4$ out of the 16 equations, corresponding to rows 1,6, 11 and 13 of the incidence matrix $\mathbf{M}$ given in Equation \ref{eq:incidenceMatrix}. For ease, we write $X_i$ instead of $\mathbf{X}[i]$.

\begin{align*}
X_1 +X_2 +X_3 +X_4  = \mathbf{V}[1] & =1/2 \\
X_2 +X_4 +X_6 +X_8  = \mathbf{V}[6]& =1/8 \\
X_3 +X_4 +X_{11}+X_{12} =  \mathbf{V}[11] & =1/8 \\
X_1 +X_5 +X_9 +X_{13}  = \mathbf{V}[13] & =1/8 
\end{align*}


Adding the last three equations yields
$$X_2 +2X_4 +X_6 +X_8 + X_3 +X_{11}+X_{12}+ X_1 +X_5 +X_9 +X_{13} = 3/8 $$
Since all these terms must be non-negative, the left hand side of this equation must be greater than or equal to the left-hand side of the first equation, yielding the required contradiction.

In terms of our hierarchy, it is easy to see that, although this model is weakly contextual, it is not logically contextual. If we exclude from the set $\mathcal{E}(X)$ those global sections which map $a$ and $b$ to different outcomes, we are left with one example of a set of global sections (or equivalently, with one example of a boolean distribution over the set of global sections) whose restriction to each context is equal to the model's support.

\subsection{The possibilistic Hardy model}

The original purpose of the Hardy model, introduced by Lucien Hardy in \cite{Hardy92, Hardy93}, was to give a `logical' proof of Bell's theorem in the bipartite case. Its support is specified by the following table:

$$\begin{array}{c|c|c|c|c}
\varstr{14pt}{9pt}
~~~A~~~~~B~~~&~~~(0,0) ~~~& ~~~(0,1)~~~ & ~~~(1,0)~~~ & ~~~(1,1)~~~ \\
\hline
\varstr{14pt}{9pt}~~~a~~~~~b~~~ &~~~1~~~ & ~~~1 ~~~& ~~~1 ~~~& ~~~1~~~\\\hline
\varstr{14pt}{9pt}~~~a~~~~~b'~~~ &~~~0~~~ & ~~~1 ~~~& ~~~ 1 ~~~& ~~~1~~~\\\hline
\varstr{14pt}{9pt}~~~a'~~~~~b~~~ &~~~0~~~ & ~~~1 ~~~& ~~~1 ~~~& ~~~1~~~\\\hline
\varstr{14pt}{9pt}~~~a'~~~~~b'~~~ &~~~1~~~ & ~~~1 ~~~& ~~~1 ~~~& ~~~0~~~
\end{array}
 $$

It has been shown in \cite{AbrBra11} that this model is logically contextual, as it admits no solutions over the Boolean semiring. As for the previous example, this is can be done by a direct argument. This time the equations specified by the incidence matrix are equations over $\mathbb{B}$ which can also be interpreted as logical clauses. For example, the equation specified by the fifth row of the incidence matrix is 
$$X_1\vee X_3\vee X_5 \vee X_7 = \mathbf{V}[5] = 0$$
This yields the equivalent logical formula
$$\neg X_1\wedge \neg X_3\wedge \neg X_5 \wedge \neg X_7 $$
We focus on the formulas corresponding to rows 1, 5, 9 and 16 of the incidence matrix:
\begin{align*}
X_1\vee X_2&\vee X_3 \vee X_4\\
\neg X_1\wedge \neg X_3&\wedge \neg X_5 \wedge \neg X_7 \\
\neg X_1\wedge \neg X_2&\wedge \neg X_9 \wedge \neg X_{10} \\
\neg X_4\wedge \neg X_8&\wedge \neg X_{12} \wedge \neg X_7{16}
\end{align*}
Since every disjunct in the first formula appears as a negated conjunct in one of the other three formulas, there is no satisfying argument. 

According to Remark \ref{rk:LW}, the Hardy model satisfies a stronger non-locality property than the Bell model. So far we have proved this directly, but it also follows from the general results in \cite{ManFri}, which show that models which rely on Hardy's paradox are complete for the $(2,2,2)$-type cases, and in particular that there must be at least three null sections in the support in order for non-locality to hold, while the Bell model has only two zero entries.

The Hardy model on the other hand is not strongly contextual, as the global assignment $\{a\mapsto 1, b\mapsto 1, a'\mapsto 0, b'\mapsto 0\}$ is one example of a global section whose restriction to each context is compatible with the model's support.


\subsection{Dicke states}

A permutation-symmetric $n$-qubit state is one which is invariant under the action of the full symmetry group $S_n$.
A natural basis for the permutation-symmetric states is provided by the  \emph{Dicke states} \cite{dicke1954coherence}, which are also physically significant. 
For each $n \geq 2$, $0 < k < n$ we define:
\[ S(n, k) \; := \; K \sum_{\mbox{{\small perm}}} \ket{0^{k}1^{n-k}} . \]
Here $K = {n \choose k}^{-1/2}$ is a normalization constant, and we sum over all products of $k$ $0$-kets and $n-k$ $1$-kets.

Note that the $W$ state is the $S(3,2)$ Dicke state in the above notation. For each $n > 2$, and $0 < k < n$, the Dicke state $S(n,k)$ is logically contextual.

We have excluded the cases $k=0$ and $k=n$, since in these cases $S(n, k)$ = $\ket{0^n}$ or $\ket{1^n}$, and these are obviously product states. We have also excluded the bipartite case, for which $S(2,1)$ is the Bell state $\frac{\ket{01} + \ket{10}}{\sqrt{2}}$.

If each party is allowed choose to measure one of the Pauli observables $X$ and $Z$, a Dicke state $S(n, k)$ gives rise to an $(n, 2, 2)$ empirical model. The model is specified by a table with $2^n$ rows, corresponding to the possible choices of an observable at each site. We shall focus firstly on the $\frac{n(n-1)}{2}$ rows $R_{ij}$, where $X$ observables are selected at sites $i$ and $j$, and $Z$ observables at the remaining sites. Let $T_{ij}$ be the support of the model at row $R_{ij}$.

Now consider any joint outcome $t$ for this row in which there are $k$ outcomes corresponding to the $+$ eigenvalue and $(n-k)$ outcomes corresponding to the $-$ eigenvalue, and the outcome for $X^i$ is different to the outcome for $X^j$. We claim that $t$ is not in $T_{ij}$. If we compute the inner product whose squared norm gives the probability for $t$, we see that there are two terms, of the form $+1/c$ and $-1/c$ respectively. Thus the probability of $t$ is $0$, and it is not in the support.
We can express this in logical terms by saying that $T_{ij}$ satisfies the formula
\begin{equation}
\label{XXZeq}
\bigwedge_{k \neq i,j, t(k) = {+}} Z^k \; \wedge \; \bigwedge_{k \neq i,j, t(k) = {-}} \neg Z^k \;\; \IMP \; \; (X^i \leftrightarrow X^j) . 
\end{equation}
We now consider the row where $Z$ measurements are selected by every party. The support of this row is described by the formula
\begin{equation}
\label{Zeq}
\bigvee_{\pi\in S_n} \; [\bigwedge_{i=1}^k Z^{\pi(i)} \; \wedge \; \bigwedge_{j=k+1}^n \neg Z^{\pi(j)} ]. 
\end{equation}
This is the logical counterpart of the description of $S(n, k)$ in the $Z$-basis.

From each disjunct $D$ of~(\ref{Zeq}) together with the relevant instances of~(\ref{XXZeq}), we can prove that $X^i \leftrightarrow X^j$ for all $i$, $j$ such that $Z^i$ and $Z^j$ appear with opposite polarity in $D$.
Note that, by the conditions on $k$ and $n$, both polarities do appear in $D$. By the transitivity of logical equivalence, it follows that $X^i \leftrightarrow X^j$ can be derived for all $i$, $j$. Thus $D$, together with the formulas~(\ref{XXZeq}), implies the formula
\begin{equation}
\label{Xeq}
\bigwedge_{i,j} \; X^i \leftrightarrow X^j . 
\end{equation}
Thus~(\ref{Zeq}) together with the conjunction of all instances of~(\ref{XXZeq}) implies~(\ref{Xeq}).

It follows that any global section which satisfies these formulas must restrict to just two joint outcomes in the row where $X$ measurements are selected by every party, namely those with the same outcome at every part.

To complete the argument, it suffices to show that these two outcomes form a proper subset of the support at that row. If we calculate the probability for each of these events, we obtain
\[ \left(\frac{{n \choose k}}{(\sqrt{2})^n \sqrt{{n \choose k}}} \right)^2 \;\;\; = \;\;\; \frac{{n \choose k}}{2^n} . \]
Thus we must show that
\[ \frac{{n \choose k}}{2^n} \; < \; \frac{1}{2} , \]
or equivalently
\[ {n \choose k} \;\; < \;\; 2^{n-1} \;\; = \;\; \sum_{l=0}^{n-1} {n-1 \choose l} \]
which follows from Pascal's rule:
\[ {n \choose k} \; = \; {n-1 \choose k-1} + {n-1 \choose k} . \]
Note however that to obtain a strict inequality, we need the assumption that $n>2$; the argument for the Bell state $S(2,1)$ fails at exactly this point.


We also note that logical contextuality is preserved by the action of local unitaries $U_1 \otimes \cdots \otimes U_n$. If a state $\ket{\psi}$ is logically contextual with respect to measurement bases 
\[ \eta_1^{+}, \eta_1^{-}, \ldots , \eta_n^{+}, \eta_n^{-} , \]
then $U_1 \otimes \cdots \otimes U_n \ket{\psi}$ is logically contextual with respect to the measurement bases
\[ U_1 \eta_1^{+}, U_1 \eta_1^{-}, \ldots , U_n \eta_n^{+}, U_n \eta_n^{-} . \]
This follows since inner products and hence probabilities are preserved:
\[ \begin{array}{lcl}
\langle U_1 \eta_1^{\plusminus}\otimes \cdots \otimes U_n \eta_n^{\plusminus} \mid (U_1 \otimes \cdots \otimes U_n) \ket{\psi} \rangle  & = & \langle (U_1  \otimes \cdots \otimes U_n) \eta_1^{\plusminus} \otimes \cdots \otimes \eta_n^{\plusminus} \mid (U_1 \otimes \cdots \otimes U_n) \ket{\psi} \rangle \\
& = & \langle (U_1 \otimes \cdots \otimes U_n)^{\dagger} (U_1  \otimes \cdots \otimes U_n) \eta_1^{\plusminus} \otimes \cdots \otimes \eta_n^{\plusminus}   \ket{\psi}  \\
& = & \langle  \eta_1^{\plusminus} \otimes \cdots \otimes  \eta_n^{\plusminus} \ket{\psi}  .
\end{array} \]
Thus the orbits of the Dicke states under the actions of local unitaries are all logically contextual.

\subsection{Permutationally symmetric states}

In \cite{wang2012nonlocality} it is shown that all permutation-symmetric states \emph{except} the unitary orbit of the Dicke states admit a Hardy argument, making use of the Majorana representation of permutation-symmetric states. This is easily converted into a proof of logical contextuality. By combining this result with the one presented above, we can conclude that all permutation-symmetric $n$-partite entangled states, for $n>2$, are logically contextual.

\subsection{GHZ models}

A GHZ model of type $(n,2,2)$ can be specified abstractly as follows. Each part can choose between two measurements labeled as $X^{i}$ and $Y^{i}$, respectively, with $1\leq i\leq n$. The outcomes are labeled as $0$ and $1$. For each context $C$, every section $t$ in the support of the model must satisfy two conditions. First, if the number of $Y$ measurements in $C$ is a multiple of $4$, the number of $1$'s in the outcomes allowed by $t$ is even. Second, if the number of $Y$ measurements in $C$ is $4k+2$, the number of $1$s in the outcomes allowed under the assignment is odd. 

In quantum mechanics, such models can be obtained by measuring the $X$ and $Y$ Pauli observables for the $n$-partite GHZ state 
$$\ket{GHZ} =\frac{ \ket{0 \ldots 0} +\ket{1\ldots 1}}{\sqrt{2}}$$

To show that these are strongly contextual, we proceed by contradiction.

We first consider the case when $n=4k$, $k\geq 1$. Assume that $s$ is a global section compatible with the support of the $GHZ$ model. If $Y$ measurements are taken at every part, the number of $1$ outcomes under the assignment is even. 

If we replace any two $Y$'s by $X$'s, we must have the opposite parity for the number of $1$ outcomes under the assignment. Thus for any $Y^{i},Y^{j}$ which have been assigned the same outcome, if we substitute $X$'s in those positions, they must receive different values. Similarly, for any $Y^{i},Y^{j}$ assigned different values, the corresponding $X^{i},X^{j}$ must receive the same value.

Suppose firstly that not all $Y^{i}$ are assigned the same value by $s$. Then there exist some $i,j,k$ such that $Y^{i}$ is assigned the same value as $Y^{j}$, and $Y^{j}$ is assigned a different value to $Y^{k}$. Then $X^{i}$ is assigned the same value as $X^{k}$, and $X^{j}$ is assigned the same value as $X^{k}$. By transitivity, $X^{i}$ is assigned the same value as $X^{j}$, yielding a contradiction.

The remaining cases are those where all $Y$'s receive the same value. Then any pair of $X$'s must receive different values. But by taking any 3 $X$'s, this yields a contradiction, since there are only two values, so some pair must receive the same value.

The case when $n=4k+2$, $k\geq 1$, is proved in the same fashion, interchanging the parities. When $n\geq 5$ is odd, we start with a context containing one $X$, and again proceed similarly.

The most familiar case, for $n=3$ does not admit this argument, which relies on having at least $4$ $Y$'s in the initial configuration. However, this case is also strongly contextual, \cite{AbrBra11}. The proof uses a case analysis to show that there are $8$ possible global sections satisfying the parity constraint on the $3$ measurement combinations with two $Y$s and one $X$, and all of these violate the constraint for the $XXX$ measurement.

\end{chapter}
\newpage
$ $
\newpage
\begin{chapter}{Classifying a Class of Permutationally Asymetric Quantum States}\label{CC}

In Chapter \ref{SSC} we have seen that there is a hierarchy of forms of non-locality or contextuality which empirical models may satisfy. This hierarchy has three levels:
\begin{itemize}
 \item A model is \textbf{strongly contextual} if its support has no global section; that is, there is no simultaneous assignment of outcomes to all the measurements whose restriction to each compatible set of measurements is in the support.
 \item A model is \textbf{logically contextual} if there are events in the support of some compatible family of measurements $F_j$ which are not consistent with the supports of the other measurement contexts.
 \item Finally, a model is \textbf{weakly contextual} if it is contextual, but neither logically nor strongly contextual.
\end{itemize}

These notions form a proper hierarchy. Strong contextuality implies logical contextuality, which implies contextuality in the usual sense. On the other hand, there are weakly contextual models which are not logically contextual, and logically contextual models which are not strongly contextual. 

The contextual characterisation of empirical models can also be used to characterise $n$-partite quantum states shared by $n$ observers, as long as we allow each observer to choose one out of a finite set of measurements, each with a fixed set of possible outcomes. For example, when the set of allowed measurements contains two elements, each with two possible outcomes, the resulting empirical model corresponds to a $(n,2,2)$ Bell scenario. 



Note that the bipartite case seems to be anomalous within the landscape of multipartite entangled states. For example, the only strongly contextual bipartite models are given by super-quantum devices known as PR-boxes  \cite{AbrBra11, popescu1994quantum}. These are however not quantum realizable. By contrast, for all $n>2$, the $n$-partite GHZ states are strongly contextual \cite{AbrBra11}. Moreover, it is known that in the bipartite case, all entangled states \emph{except} the maximally entangled ones admit Hardy arguments, and hence are logically contextual \cite{Hardy93}


We have seen at the end of the previous chapter that (almost) all permutationally symmetric entangled states are at least logically contextual. We will now look at a class of highly non-permutation-symmetric entangled states, the \emph{balanced states with functional dependencies}. These states are described by Boolean functions, and have a rich structure, allowing a detailed analysis. We provide logical contextuality witnesses (i.e. measurement choices) for almost all of these states and we also reveal a large collection of strongly contextual multipartite entangled states. 

\section{Balanced States with Functional Dependency}

For each $n\geq 2$, a $n$-ary Boolean function is a function $F:\{0,1\}^n\rightarrow \{0,1\}$. Each $n$-ary Boolean function can be expressed as a multivariate polynomial over $GL(2)$:
$$F(x_1,\ldots,x_n)=a_0 + \sum_{i}a_1^{i}x_i+ \sum_{i,j}a_2^{i,j}x_ix_j+\ldots + a_n^{1,2,\ldots,n}x_1x_2\ldots x_n$$
There are $2^n=1 + n + \binom{n}{2} + \ldots + \binom{n}{n}$ summands in the expression of the above polynomial, each of which containing a binary coefficient $a_t^{i_1,\ldots,i_t}$. Hence there are $2^{2^{n}}$ distinct $n$-variate polynomials over $GF(2)$. Alternatively, each $n$-ary Boolean function can be expressed as a propositional formula in the Boolean variables $x_1,\ldots,x_n$ \cite{CraHam11}.

We define a balanced $n+1$-qubit quantum state with a functional dependency given by a $n$-variate polynomial $F$ as above to be a state which has the form
$$\Psi_F(n+1)=\frac{1}{\sqrt{2^n}}\sum_{q_1q_2\ldots q_n=00\ldots0}^{11\ldots 1} |q_1q_2\ldots q_n F(q_1,q_2,\ldots,q_n) \rangle $$
when expressed in the $Z$-basis. 

We start by classifying the tripartite functionally dependent balanced states in terms of their contextuality properties. A classification of the $n+1$-qubit states for $n>2$ can then be obtained using the results from the tripartite scenarios.

\section{Contextuality classification for the tripartite case}

\subsection{Polynomials of degree zero}

There are $2^{2^{2}}=16$ tripartite balanced states with a functional dependency. Two of these, namely 
$$\frac{1}{2}|000\rangle+|010\rangle+|100\rangle+|110\rangle=\left(\frac{|0\rangle+|1\rangle}{\sqrt{2}}\right)^{\otimes2}\otimes|0\rangle$$
 and 
$$\frac{1}{2}|001\rangle+|011\rangle+|101\rangle+|111\rangle=\left(\frac{|0\rangle+|1\rangle}{\sqrt{2}}\right)^{\otimes2}\otimes|1\rangle$$ 
are obviously product states, and hence non-contextual. They correspond to the constant polynomials $F_0(q_1,q_2)=0$ and $F_1(q_1,q_2)=1$ respectively.

\subsection{Degree one polynomials}\label{dict}
There are six states whose corresponding polynomials have degree one. Two of these are given by the functional dependencies which correspond to the two-variable propositional formulas $XOR$ and $NXOR$. Another four states are given by dictatorships, i.e. the value of the last qubit is dictated either by the value of the first qubit or by the value of the second qubit. We shall look at these two classes of states below.

\noindent\textbf{XOR and NXOR}
The polynomials corresponding to the $XOR$ and $NXOR$ states have the form $$F^a_{XOR}(q_1,q_2)=a+q_1+q_2$$ with $a=0$ for $XOR$ and $a=1$ for $NXOR$. 

\begin{theorem}
 The $XOR$ state is strongly contextual if each party chooses between $Y$ and $Z$ measurements.
\end{theorem}

\begin{proof}
The support of the probability table for the $XOR$ state is 
\begin{center}
\begin{tabular}{l|cccccccc}
& $+++$ & $++-$ & $+-+$ & $+--$ & $-++$ & $-+-$ & $--+$ & $---$ \\ \hline
$YYY$ &  $1$ & $1$ & $1$ & $1$ & $1$ & $1$ & $1$ & $1$ \\
$YYZ$ &  $0$ & $1$ & $1$ & $0$ & $1$ & $0$ & $0$ & $1$ \\
$YZY$ &  $0$ & $1$ & $1$ & $0$ & $1$ & $0$ & $0$ & $1$ \\
$ZYY$ &  $0$ & $1$ & $1$ & $0$ & $1$ & $0$ & $0$ & $1$ \\
$YZZ$ &  $1$ & $1$ & $1$ & $1$ & $1$ & $1$ & $1$ & $1$ \\
$ZYZ$ &  $1$ & $1$ & $1$ & $1$ & $1$ & $1$ & $1$ & $1$  \\
$ZZY$ &  $1$ & $1$ & $1$ & $1$ & $1$ & $1$ & $1$ & $1$ \\
$ZZZ$ & $1$ & $0$ & $0$ & $1$ & $0$ & $1$ & $1$ & $0$ \\
\end{tabular}
\end{center}

One can simply inspect the table above and check that none of the sections in the support of the $ZZZ$ row can be extended to global sections (i.e. each possible global  assignment consistent with the support of the $ZZZ$ row will restrict to a section outside the support on at least one of the three rows $YYZ$, $YZY$ and $ZYY$). Thus there cannot be any global assignment of outcomes whose restriction to each set of compatible measurements is in the support of the model. 

It is worth at this point to give a more formal expression to this argument in order to gain a better understanding of what is actually going on. For this recall that the $+$ and $-$ eigenstates of the $Z$ observable are $|0\rangle$ and $|1\rangle$ respectively while for the $Y$ observable they are (modulo some normalization constant which does not play any role in our argument) $|Y^+\rangle:=|0\rangle + i|1\rangle$ and $|Y^-\rangle:=|0\rangle-i|1\rangle$ respectively

We start our argument by assuming that a global section does exist. Assume next that this global section makes the assignment $Z^3=+$. The probability of obtaining the outcome $Z^1Z^2+$ with $Z^i\in\{+,-\}$ is given by the squared norm of the inner product 
$$\langle e_{Z^1} e_{Z^2} 0|XOR\rangle=\langle e_{Z^1} e_{Z^2}0|\frac{|000\rangle+|011\rangle+|101\rangle+|110\rangle}{2}$$
where $e_+=0$ and $e_-=1$. If we regard each $e_{Z^i}$ as an element of $GF(2)$ then the inner product above is non-zero only if $$F^0_{XOR}(e_{Z^1},e_{Z^2})=e_{Z^1}+e_{Z^2}=0$$

So the sections in the support of the $ZZZ$ for which $Z^3= +$ must have $Z^1=Z^2$, as the table confirms. 

Next consider the $YYZ$ set of compatible measurements. The probability (modulo normalization constants) of obtaining the outcome $Y^1Y^2+$ with $Y^i\in\{+,-\}$ for this set of measurements is given by the squared norm of the inner product
\begin{equation}\label{yyz}
\langle Y^{Y^1} Y^{Y^2} 0|XOR\rangle =\langle Y^{Y^1} Y^{Y^2} 0|\frac{|000\rangle+|011\rangle+|101\rangle+|110\rangle}{2}
\end{equation}
We have
\begin{equation*}
 \begin{aligned}
  \langle Y^+Y^+|&=\langle 00|+i\langle 01|+i\langle 10|-\langle 11|\\
\langle Y^+Y^-|&=\langle 00|-i\langle 01|+i\langle 10|+\langle 11|\\
\langle Y^-Y^+|&=\langle 00|+i\langle 01|-i\langle 10|+\langle 11|\\
\langle Y^-Y^-|&=\langle 00|-i\langle 01|-i\langle 10|-\langle 11|
 \end{aligned}
\end{equation*}
and since $F^0_{XOR}(0,1)=F^0_{XOR}(1,0)\neq 0$ the imaginary part of the tensor products above will not bring any contribution towards the value of the inner product (\ref{yyz}). The only contribution will come from the real part of the tensor products above, and it is easy to see that the inner product (\ref{yyz}) will vanish when $Y^1=Y^2$. So we must have $Y^1\neq Y^2$ in any global assignment which sends $Z^3$ to $+$ in order to stay within the support of the $YYZ$ row.

On the other hand, the probabilities of obtaining the outcomes $Y^1zY^3$ and $zY^2Y^3$, where $z=Z^1=Z^2$, for the $YZY$ and $ZYY$ sets of compatible measurements are given by the inner products
\begin{equation}
 \begin{aligned}
  \langle Y^{Y^1} e_z Y^{Y^3} &|XOR\rangle =\left(\langle 0e_z0|+iY^3\langle 0e_z1|+iY^1\langle 1e_z0|-(Y^1Y^3)\langle 1e_z1|\right)\ |XOR\rangle\\ 
\langle  e_zY^{Y^2}Y^{Y^3} &|XOR\rangle =\left(\langle e_z00|+iY^3\langle e_z01|+iY^2\langle e_z10|-(Y^2Y^3)\langle e_z11|\right)\ |XOR\rangle \label{zyy}
 \end{aligned}
\end{equation}
If $e_z=0$ the imaginary part of the two expressions in (\ref{zyy}) will be equal to zero for all values of $Y^i$. If $e_z=1$ the real part of the two expressions in (\ref{zyy}) will vanish for all values of $Y^i$. In the first case the expressions are non-zero only if $Y^1=Y^2=-Y^3$ and in the second case they are non-zero only if $Y^1=Y^2=Y^3$. But both these assignments violate the previous requirement that $Y^1\neq Y^2$.

So far we have established the fact that no global section can assign the outcome $+$ to $Z^3$. If on the other hand the outcome $-$ is assigned to $Z^3$, we can construct a similar argument which yields a contradiction. This time the sections in the support of $ZZZ$ for which $Z^3=-$ must have $Z^1=-Z^2$. The sections in the support of $YYZ$ must have $Y^1=Y^2$, while those in the support of $YZY$ and $ZYY$ must either have $Y^1=-Y^3=-Y^2$ for $e_{Z^2}=0$, $e_{Z^3}=1$ or $Y^1=Y^3=-Y^2$ for $e_{Z^2}=1$ and $e_{Z^3}=0$. \qed

\end{proof}

\begin{theorem}
 The $NXOR$ state is also strongly contextual if each party chooses between $Y$ and $Z$ measurements.
\end{theorem}

\begin{proof}
The support of the probability table for the $NXOR$ state is 
\begin{center}
\begin{tabular}{l|cccccccc}
& $+++$ & $++-$ & $+-+$ & $+--$ & $-++$ & $-+-$ & $--+$ & $---$ \\ \hline
$YYY$ &  $1$ & $1$ & $1$ & $1$ & $1$ & $1$ & $1$ & $1$ \\
$YYZ$ &  $1$ & $0$ & $0$ & $1$ & $0$ & $1$ & $1$ & $0$ \\
$YZY$ &  $1$ & $0$ & $0$ & $1$ & $0$ & $1$ & $1$ & $0$ \\
$ZYY$ &  $1$ & $0$ & $0$ & $1$ & $0$ & $1$ & $1$ & $0$ \\
$YZZ$ &  $1$ & $1$ & $1$ & $1$ & $1$ & $1$ & $1$ & $1$ \\
$ZYZ$ &  $1$ & $1$ & $1$ & $1$ & $1$ & $1$ & $1$ & $1$  \\
$ZZY$ &  $1$ & $1$ & $1$ & $1$ & $1$ & $1$ & $1$ & $1$ \\
$ZZZ$ & $0$ & $1$ & $1$ & $0$ & $1$ & $0$ & $0$ & $1$ \\
\end{tabular}
\end{center}

The argument for strong contextuality follows the same pattern as for the $XOR$ state. We assume by contradiction that a global section exists, and that it makes the assignment $Z^3=+$. Then from the $ZZZ$ row we obtain the requirement that $Z^1\neq Z^2$. From the $YYZ$ row we obtain that $Y^1=Y^2$ and from the $YZY$ and $ZYY$ rows we obtain that $Y^1\neq Y^2$, which is a contradiction. 

Similarly, if $Z^3=-$ we must have $Z^1=Z^2$ and $Y^1\neq Y^2$ from the $ZZZ$ and $YYZ$ rows. This means we must also have $Y^1=Y^2$ from the $YZY$ and $ZYY$ rows, which again is a contradiction.

Note at this point that the similarity between these two arguments for strong contextuality is due to the similar structure of the tables for the $XOR$ an $NXOR$ states. Namely, the second table can be obtained from the first by interchanging the $+$ and $-$ signs which label the table columns. Thus the second argument is the same as the first, only with the $+$ and $-$ signs interchanged.
 \qed

\end{proof}

\noindent\textbf{Dictatorships}

The four degree one polynomials of the form $F^a_{1}(q_1,q_2)=a+q_1$ and $F^a_{1}(q_1,q_2)=a+q_2$ where $a\in\{0,1\}$ correspond to the dictatorship states, where the value of the last qubit is dictated by the value of either the first or of the second qubit. In the $Z$ basis these states are
$$\Delta^+_2:=\frac{|0\rangle+|1\rangle}{\sqrt{2}}\otimes \frac{|00\rangle+|11\rangle}{\sqrt{2}}$$
or
$$\Delta^-_2:= \frac{|0\rangle+|1\rangle}{\sqrt{2}}\otimes \frac{|01\rangle+|10\rangle}{\sqrt{2}}$$
if the dictatorship is given by the second qubit. Similarly, if the dictatorship is given by the first qubit, we have two possible states
$$\Delta^+_1:= \frac{|0_2\rangle+|1_2\rangle}{\sqrt{2}}\otimes \frac{|0_10_3\rangle+|1_11_3\rangle}{\sqrt{2}}$$
and
$$\Delta^-_1:=\frac{|0_2\rangle+|1_2\rangle}{\sqrt{2}}\otimes \frac{|0_a1_3\rangle+|1_10_3\rangle}{\sqrt{2}}$$
where the subscripts $1$, $2$ and $3$ indicate whether the qubit belongs to the first, second or third party respectively.

\begin{proposition}\label{bell}
The four dictatorship states are weakly contextual for suitable dichotomic choices of measurements.
\end{proposition}

\begin{proof}

Consider the general form of an observable, given in terms of angles $\theta$ and $\phi$ on the Bloch sphere
$$U(\theta,\phi):=\left(\begin{array}{cc}
                          \cos\theta& e^{-i\phi}\sin\theta\\
			  e^{i\phi}\sin\theta&-\cos\theta
                         \end{array}\right)$$

We will use the fact that the bell basis states $\Phi^+=\frac{|00\rangle+|11\rangle}{\sqrt{2}}$ and $\Phi^-=\frac{|00\rangle+|11\rangle}{\sqrt{2}}$ are weakly contextual with respect to suitable choices of measurements.

It can be machine checked that the state $\Phi^+$ is weakly contextual if we allow each party to choose between the measurements $A:=U \left(\frac{\pi}{2},\frac{\pi}{8}\right)$ and $B:=U\left(\frac{\pi}{2},\frac{5\pi}{8}\right)$, while the state  $\Phi^-$ is weakly contextual if we allow each party to choose between the measurements $C:=U\left(\frac{\pi}{8},\frac{\pi}{2}\right)$ and $D:=U\left(\frac{5\pi}{8},\frac{\pi}{2}\right)$

In fact, it can also be machine checked that this choice of measurements gives a maximal violation of Bell inequalities for both states.

The probability models of the dictatorship states can be obtained from the probability models of the states $\Phi^+$ and $\Phi^-$ in a straightforward way. Let $|+_A\rangle$ and $|-_A\rangle$ stand for the eigenstates of $A$ and $|+_B\rangle$ and $|-_B\rangle$ stand for the eigenstates of $B$. 

Define the two constants
\begin{align*}
a_+:&=\frac{1}{\sqrt{2}}(\langle+_A|0\rangle+\langle+_A|1\rangle)\\
a_-:&=\frac{1}{\sqrt{2}}(\langle-_A|0\rangle+\langle-_A|1\rangle) 
\end{align*}
Note that $a_++a_-=1$. 
Similarly, define the two constants
\begin{align*}
b_+:&=\frac{1}{\sqrt{2}}(\langle+_B|0\rangle+\langle+_B|1\rangle)\\
b_-:&=\frac{1}{\sqrt{2}}(\langle-_B|0\rangle+\langle-_B|1\rangle) 
\end{align*}
Up to two decimal points precision, the probability table of the $\Phi^+$ state for the observables $A$ and $B$ is 
\begin{center}
\begin{tabular}{l|cccc}
& $++$ & $+-$ & $-+$ & $--$ \\ \hline
$AA$ &  $0.43$ & $0.07$ & $0.07$ & $0.43$  \\
$AB$ &  $0.07$ & $0.43$ & $0.43$ & $0.07$  \\
$BA$ &  $0.07$ & $0.43$ & $0.43$ & $0.07$  \\
$BB$ &  $0.07$ & $0.43$ & $0.43$ & $0.07$  \\
\end{tabular}
\end{center}

The inner product formula (\ref{ip}) for computing probabilities implies that the probability table of the dictatorship state $\Delta^+_2$ can be expressed in terms of the constants $a_+$, $a_-$, $b_+$ and $b_-$ and the probability table of $\Phi^+$:

\begin{center}
\begin{tabular}{l|cccc|cccc}
& $+++$ & $++-$ & $+-+$ & $+--$ & $-++$ & $-+-$ & $--+$ & $---$ \\ \hline
$AAA$ &  $0.43a_+$ & $0.07a_+$ & $0.07a_+$ & $0.43a_+$ & $0.43a_-$ & $0.07a_-$ & $0.07a_-$ & $0.43a_-$  \\
$AAB$ &  $0.07a_+$ & $0.43a_+$ & $0.43a_+$ & $0.07a_+$ & $0.07a_-$ & $0.43a_-$ & $0.43a_-$ & $0.07a_-$\\
$ABA$ &  $0.07a_+$ & $0.43a_+$ & $0.43a_+$ & $0.07a_+$ & $0.07a_-$ & $0.43a_-$ & $0.43a_-$ & $0.07a_-$ \\
$ABB$ &  $0.07a_+$ & $0.43a_+$ & $0.43a_+$ & $0.07a_+$ & $0.07a_-$ & $0.43a_-$ & $0.43a_-$ & $0.07a_-$ \\
\hline
$BAA$ &  $0.43b_+$ & $0.07b_+$ & $0.07b_+$ & $0.43b_+$& $0.43b_-$ & $0.07b_-$ & $0.07b_-$ & $0.43b_-$  \\
$BAB$ &  $0.07b_+$ & $0.43b_+$ & $0.43b_+$ & $0.07b_+$ & $0.07b_-$ & $0.43b_-$ & $0.43b_-$ & $0.07b_-$\\
$BBA$ &  $0.07b_+$ & $0.43b_+$ & $0.43b_+$ & $0.07b_+$ & $0.07b_-$ & $0.43b_-$ & $0.43b_-$ & $0.07b_-$\\
$BBB$ &  $0.07b_+$ & $0.43b_+$ & $0.43b_+$ & $0.07b_+$ & $0.07b_-$ & $0.43b_-$ & $0.43b_-$ & $0.07b_-$\\
\end{tabular}
\end{center}

Note also that the table of the dictatorship state $\Delta^+_1$ will have the same values as the one above, but the rows will be indexed in the order $AAA$, $AAB$, $BAA$, $BAB$, $ABA$, $ABB$, $BBA$, $BBB$, since the coefficients $a_{+/-}$ and $b_{+/-}$ come from the second qubit's contribution to the inner product.

It is now straightforward to deduce that the states $\Delta^+_1$ and $\Delta^+_2$ are indeed weakly contextual for the same choice of measurements for which the $\Phi^+$  state is weakly contextual, since any probability distribution on the set of global sections of one of these two dictatorship states would restrict to a probability distribution on the set of global sections of the $\Phi^+$ state.

Next note that up to two decimal points precision, the probability table of the $\Phi^-$ state for the observables $C$ and $D$ is
\begin{center}
\begin{tabular}{l|cccc}
& $++$ & $+-$ & $-+$ & $--$ \\ \hline
$AA$ &  $0.43$ & $0.07$ & $0.07$ & $0.43$  \\
$AB$ &  $0.07$ & $0.43$ & $0.43$ & $0.07$  \\
$BA$ &  $0.07$ & $0.43$ & $0.43$ & $0.07$  \\
$BB$ &  $0.07$ & $0.43$ & $0.43$ & $0.07$  \\
\end{tabular}
\end{center}

and the probability tables of the $\Delta^-_1$ and $\Delta^-_2$ dictatorship states can be expressed in terms of the table above and four suitably defined constants $c_{+/-}$ and $d_{+/-}$, so by analogy with the $\Delta^+_1$ and $\Delta^+_2$ case, these states will also be weakly contextual.\qed

\end{proof}

\begin{theorem}\label{bellth}
None of the four dictatorship states is logically contextual, for \textbf{any} dichotomic choice of measurements.
\end{theorem}

\begin{proof}
The relationship between probability tables discussed in Proposition \ref{bell} allows us to reduce the problem to the bipartite scenario. Thus we seek to prove that neither of the two Bell basis states is logically contextual for any given choice of measurements. 

Let $A:=U(\theta_1,\phi_1)$ and $B:=U(\theta_2,\phi_2)$. Let $c$, $s$ and $f$ stand for $\cos\frac{\theta_1}{2}$, $\sin\frac{\theta_1}{2}$ and $e^{i\phi_1}$ respectively. Similarly, let $k$, $z$ and $v$ stand for $\cos\frac{\theta_2}{2}$, $\sin\frac{\theta_2}{2}$ and $e^{i\phi_2}$ respectively. Then the general form of the probability model of the $\Phi^+$ state is 

\begin{center}
\begin{tabular}{l|cccc}
& $++$ & $+-$ & $-+$ & $--$ \\ \hline
$AA$ &  $|c^2+f^2\cdot s^2|^2$ & $|cs-f^2\cdot cs|^2$ & $|cs-f^2\cdot cs|^2$ & $|s^2+f^2\cdot c^2|^2$  \\
$AB$ &  $|ck+fv\cdot sz|^2$ & $|cz-fv\cdot sk|^2$ & $|sk-fv\cdot cz|^2$ & $|sz+fv\cdot ck|^2$  \\
$BA$ &  $|ck+fv\cdot sz|^2$ & $|sk-fv\cdot cz|^2$ & $|cz-fv\cdot sk|^2$ & $|sz+fv\cdot ck|^2$  \\
$BB$ &  $|k^2+v^2\cdot z^2|^2$ & $|kz-v^2\cdot kz|^2$ & $|kz-v^2\cdot kz|^2$ & $|z^2+v^2\cdot k^2|^2$  \\
\end{tabular}
\end{center}

In most cases, all of the sections in the model of $\Phi^+$ will be in the support, in which case the state is clearly not logically contextual. However, for certain values of $c$, $f$, $v$ and $k$ (which may be chosen independently of each other) the entries of the table above may vanish, which will exclude certain sections from the support. It suffices therefore to check that the resulting possibilistic models are not logically contextual for any choices of $c$, $f$, $v$ and $k$ (and implicitly also of $s$ and $z$) which would allow one or more of the above table entries to vanish. We therefore need to consider each element in the powerset of the following set of conditions on $c$, $s$, $f$, $v$, $z$ and $k$:

$$\left\{c\vee k\! \in \! \{0,\pm 1\},\: f\vee v\! \in \! \{\pm 1,\, \pm i\},\: f\! =\pm \frac{1}{v},\: c=\pm s,\: k=\pm z,\: ck=\pm sz,\: cz=\pm sk\right\} $$

A computer can easily verify that no subset of the above set of conditions leads to a logically contextual probability model. 

Finally, using the relation between probability tables from Proposition \ref{bell}, we note that any global section of the model above can be easily extended to a global section of the corresponding dictatorship state model by adding the  assignment $+$ to the third party's outcome for the $A$ measurement, if $a_+\neq0$ and $-$ otherwise, and similarly for the third party's outcome corresponding to the $B$ measurement. We can therefore conclude that for all possible choices of measurements, the dictatorship states corresponding to $\Phi^+$ can not be logically contextual. 

For the $\Phi^-$ state note that the observables $C:=U(\phi_1,\theta_1)$ and $D:=U(\phi_2,\theta_2)$ will give the probability model
\begin{center}
\begin{tabular}{l|cccc}
& $++$ & $+-$ & $-+$ & $--$ \\ \hline
$AA$ &  $|cs-f^2\cdot cs|^2$ &  $|c^2+f^2\cdot s^2|^2$ & $|s^2+f^2\cdot c^2|^2$ &  $|cs-f^2\cdot cs|^2$   \\
$AB$ &  $|cz-fv\cdot sk|^2$ & $|ck+fv\cdot sz|^2$ & $|sz+fv\cdot ck|^2$  &  $|sk-fv\cdot cz|^2$  \\
$BA$ &  $|sk-fv\cdot cz|^2$ &  $|ck+fv\cdot sz|^2$  & $|sz+fv\cdot ck|^2$ & $|cz-fv\cdot sk|^2$ \\
$BB$ &  $|kz-v^2\cdot kz|^2$ &  $|k^2+v^2\cdot z^2|^2$ & $|z^2+v^2\cdot k^2|^2$ & $|kz-v^2\cdot kz|^2$   \\
\end{tabular}
\end{center}
where $c,k,s,z$ now take $\phi_i/2$ as arguments while $f$ and $v$ take $\theta_i$ as arguments. 

We can show that this model is also not logically contextual, using an argument completely analogous to the one used for the $\Phi^+$ state. Hence the dictatorship states corresponding to the $\Phi^-$ state are also not logically contextual. \qed

\end{proof}

\subsection{Degree two polynomials}\label{two}

There are eight balanced functionally dependent states whose corresponding polynomials have degree two. Four of these correspond to the two-variable propositional formulas $AND,$ $NAND,$ $OR$ and $NOR$. Their respective polynomials have the form $$F^a_{AND}(q_1,q_2)=a+q_1q_2$$ and $$F^a_{OR}=a+q_1+q_2+q_1q_2$$ with $a=0$ for $AND$ and $OR$ and $a=1$ for $NAND$ and $NOR$. 

The other four states correspond to logical implication and its negation. We use $L_1,\ L_{2},\ NL_{1}$ and $NL_{2}$ to denote the propositional formulas $q_1\Rightarrow q_2,\ q_2\Rightarrow q_1$ and $\overline{q_1\Rightarrow q_2}$, $\overline{q_2\Rightarrow q_1}$ respectively. The polynomials corresponding to these propositional formulas are of the form $$F^a_{NL_{i}}=a+q_i+q_1q_2$$ with $i\in\{1,2\}$, $a=0$ for $NL_i$ and $a=1$ for $L_i$.

All the eight states described above turn out to be logically contextual if we choose $Y$ and $Z$ measurements in each part.

\begin{theorem}\label{and}
 The $AND$ state is logically contextual. 
\end{theorem}

\begin{proof}
The support of the probability table for the $AND$ state is 
\begin{center}
\begin{tabular}{l|cccccccc}
& $+++$ & $++-$ & $+-+$ & $+--$ & $-++$ & $-+-$ & $--+$ & $---$ \\ \hline
$YYY$ &  $1$ & $1$ & $1$ & $1$ & $1$ & $1$ & $1$ & $1$ \\
$YYZ$ &  $1$ & $1$ & $1$ & $1$ & $1$ & $1$ & $1$ & $1$ \\
$YZY$ &  $1$ & $1$ & $0$ & $1$ & $1$ & $1$ & $1$ & $0$ \\
$ZYY$ &  $1$ & $1$ & $1$ & $1$ & $0$ & $1$ & $1$ & $0$ \\
$YZZ$ &  $1$ & $0$ & $1$ & $1$ & $1$ & $0$ & $1$ & $1$ \\
$ZYZ$ &  $1$ & $0$ & $1$ & $0$ & $1$ & $1$ & $1$ & $1$  \\
$ZZY$ &  $1$ & $1$ & $1$ & $1$ & $1$ & $1$ & $1$ & $1$ \\
$ZZZ$ & $1$ & $0$ & $1$ & $0$ & $1$ & $0$ & $0$ & $1$ \\
\end{tabular}
\end{center}

The global assignment $Z^1Z^2Z^3Y^1Y^2Y^3=++++++$ is clearly consistent with the support of the $AND$ table, so this state is not strongly contextual for $Y$ and $Z$ measurements. However, not all sections in the support can be extended to global sections. Consider for example the section $Y^1Y^2Z^3=+--$ which is in the support. The only section on the $ZZZ$ row consistent with it is $Z^1Z^2Z^3=---$. But it is now impossible to assign an outcome to $Y^3$ which will make the resulting global section restrict to sections in the support of both of the rows $YZY$ and $ZYY$. In fact, there are only two sections in the support of the $YYZ$ row which cannot be extended to global ones. These are the sections where the two $Y$ measurements are assigned different outcomes, while the $Z$ measurement is assigned the outcome $-$. \qed
\end{proof}

\begin{theorem}\label{nand}
 The $NAND$ state is logically contextual.
\end{theorem}

\begin{proof}
The support of the probability table for the $NAND$ state is 
\begin{center}
\begin{tabular}{l|cccccccc}
& $+++$ & $++-$ & $+-+$ & $+--$ & $-++$ & $-+-$ & $--+$ & $---$ \\ \hline
$YYY$ &  $1$ & $1$ & $1$ & $1$ & $1$ & $1$ & $1$ & $1$ \\
$YYZ$ &  $1$ & $1$ & $1$ & $1$ & $1$ & $1$ & $1$ & $1$ \\
$YZY$ &  $1$ & $1$ & $1$ & $0$ & $1$ & $1$ & $0$ & $1$ \\
$ZYY$ &  $1$ & $1$ & $1$ & $1$ & $1$ & $0$ & $0$ & $1$ \\
$YZZ$ &  $0$ & $1$ & $1$ & $1$ & $0$ & $1$ & $1$ & $1$ \\
$ZYZ$ &  $0$ & $1$ & $0$ & $1$ & $1$ & $1$ & $1$ & $1$  \\
$ZZY$ &  $1$ & $1$ & $1$ & $1$ & $1$ & $1$ & $1$ & $1$ \\
$ZZZ$ & $0$ & $1$ & $0$ & $1$ & $0$ & $1$ & $1$ & $0$ \\
\end{tabular}
\end{center}
Note that this table can be obtained from the $AND$ table by simply relabeling the columns. The relabeling sends the first $+$ to $+$, the second $+$ to $+$ and the third $+$ to $-$, and it sends the first two $-$s to $-$ and the third one to $+$. 

The same argument used in the proof of Theorem \ref{and} can therefore be used to prove the logical contextuality of the $NAND$ state, with the provision that the new labeling replaces the one used within the old argument's statements.\qed
\end{proof}

\begin{remark}
The notation $+++\mapsto ++-$ unambiguously describes the relabeling used in the proof of Theorem \ref{nand}, and we shall use this shorthand notation in further proofs.
\end{remark}

\begin{theorem}
The $OR$, $NOR$, $L_1$, $NL_1$, $L_2$ and $NL_2$ states are all logically contextual.
\end{theorem}

\begin{proof}
The support of the probability tables for these states are also obtained from the $AND$ table by column relabelings, so the argument used in the proof of Theorem \ref{and} can again be used to prove the logical contextuality of these states. The necessary relabelings are

\begin{itemize}
                                                                                                                                                                                                                                                                                 \item[1)] $+++\mapsto ---$ for the $OR$ state
\item[2)] $+++\mapsto --+$ for $NOR$
\item[3)] $+++\mapsto +--$ for $L_1$
\item[4)] $+++\mapsto +-+$ for $NL_1$
\item[5)] $+++\mapsto -+-$ for $L_2$ 
\item[6)] $+++\mapsto -++$ for $NL_2$
                                                                                                                                                                                                                                                                                \end{itemize}
\qed
\end{proof}

\begin{remark}
The relabelings above can also be used for the probability tables themselves, not only for their supports, but only for $Y$, $Z$ measurements. For general choices of measurements there is no simple relation between the probability tables of the balanced states with functional dependency given by degree two polynomials, nor between their supports.
\end{remark}

\section{Contextuality classification for the $n+1$-partite case, $n>2$}

We can use the results of the previous section to classify the $n+1$-partite balanced states which have a functional dependency. In the rest of this section, let $F_n$ denote a polynomial in $n$ variables. 

\subsection{Strongly contextual states}

\begin{theorem}
Given a $n+1$-partite balanced quantum state whose functional dependency is given by the polynomial $F_n(q_1,\ldots,q_n)$, the state is strongly contextual if the polynomial $F_n$ is of the form
$$F_n(q_1,\ldots,q_n)=q_i+q_j+F_{n-2}(q_1,\ldots,\hat{q_i},\ldots,\hat{q_j},\ldots,q_n)$$
for some variables $q_i$ and $q_j$ and some polynomial $F_{n-2}$. 
\end{theorem}

\begin{proof}
If $Y$ and $Z$ measurements are chosen by each party, then we can show that none of the sections in the support of the $ZZZ\ldots Z$ row can be extended to a global section. 

Consider any fixed assignment of outcomes to the $Z$ measurements performed by the first $n$ parties, except the $i^{th}$ and the $j^{th}$ party. Let $\sigma_k\in\{+,-\}$, $k\neq i,j$ denote the outcome corresponding to the measurement performed by the $k^{th}$ party. Next evaluate the polynomial $F_{n-2}$ at the values of $q_1,\ldots,\hat{q_i},\ldots,\hat{q_j},\ldots,q_n$ corresponding to the fixed assignment of outcomes, using the convention that $0$ corresponds to the $+$ outcome and $1$ corresponds to the $-$ outcome. Use $a$ to denote the result of the evaluation.

Depending on the value of $a$ we can use the argument made for the strong contextuality of either the $XOR$ or the $NXOR$ state in order to show that there is no consistent assignment of outcomes which will restrict to sections in the support for all four of the following rows:

\begin{align*}
 Z\ldots ZZ_iZ\ldots ZZ_jZ\ldots ZZ\\
 Z\ldots ZY_iZ\ldots ZY_jZ\ldots ZZ\\
 Z\ldots ZZ_iZ\ldots ZY_jZ\ldots ZY\\
 Z\ldots ZY_iZ\ldots ZZ_jZ\ldots ZY
\end{align*}

Since this can be done for all possible assignments of outcomes to the $Z$ measurements performed by the first $n$ parties, except the $i^{th}$ and the $j^{th}$ party, the quantum state we are considering must be strongly contextual.\qed
\end{proof}

\subsection{Logically contextual states}

\begin{theorem}
Any $n+1$-partite balanced quantum state whose functional dependency is given by a polynomial $F_n(q_1,\ldots,q_n)$ of degree at least two which is not of the form
$$F_n(q_1,\ldots,q_n)=q_i+q_j+F_{n-2}(q_1,\ldots,\hat{q_i},\ldots,\hat{q_j},\ldots,q_n)$$
for any choice of variables $q_i$ and $q_j$ and polynomial $F_{n-2}$ is logically contextual.
\end{theorem}

\begin{proof}
Consider any two variables $q_i$ and $q_j$ which appear in at least one of the terms with degree at least two of the polynomial $F_n$. The polynomial $F_n$ can be rewritten as 
$$F_n(q_1,\ldots,q_n)=F^1_{n-2}+q_iF^2_{n-2}+q_jF^3_{n-2}+q_iq_jF^4_{n-2}$$
where $F^i_{n-2}$ are $n-2$ variable polynomials in $q1,\ldots,\hat{q_i},\ldots,\hat{q_j},\ldots,q_n$.

Next choose any assignment of outcomes to the $Z$ measurements performed by the first $n$ parties, except the $i^{th}$ and the $j^{th}$ party, such that the polynomial $F^4$ evaluates to $1$ at the values of $q_1,\ldots,\hat{q_i},\ldots,\hat{q_j},\ldots,q_n$ corresponding to this assignment. Using this assignment, we have obtained a degree two polynomial in two variables, $q_i$ and $q_j$.

We can now use one of the arguments in Section \ref{two} in order to identify at least two sections in the support of the $$Z\ldots ZY_iZ\ldots ZY_jZ\ldots ZZ$$ row which cannot be extended to a global section consistent with the support of the rows

\begin{align*}
 Z\ldots ZZ_iZ&\ldots ZY_jZ\ldots ZY\\
 Z\ldots ZY_iZ&\ldots ZZ_jZ\ldots ZY\\
&\text{and}\\
 Z\ldots ZZ_iZ&\ldots ZZ_jZ\ldots ZZ \ \ \qed
\end{align*}  
\end{proof}

Note however that showing that at least one global section \textit{does} exist for the class of states considered in the Theorem above is not as simple as in the tripartite case, so strong contextuality cannot be immediately ruled out for these states even in the special case when one considers only $Y$ and $Z$ measurements.

\subsection{Weakly contextual states}\label{dictatorships}

\begin{theorem}
Any $n+1$-partite balanced quantum state whose functional dependency is given by a polynomial $F_n(q_1,\ldots,q_n)$ of degree one which is not of the form
$$F_n(q_1,\ldots,q_n)=q_i+q_j+F_{n-2}(q_1,\ldots,\hat{q_i},\ldots,\hat{q_j},\ldots,q_n)$$
for any choice of variables $q_i$ and $q_j$ and polynomial $F_{n-2}$ is weakly contextual.
\end{theorem}

\begin{proof}
Any degree one polynomial which is not of the above form must contain precisely one term. Thus the state we are dealing with is a dictatorship state, i.e. the value of the last qubit is dictated by the value of its $i^{th}$ qubit, and the state is either of the form
$$\Delta^+_i:=\left(\frac{|0\rangle+|1\rangle}{\sqrt{2}}\right)^{\otimes n}\otimes \frac{|0_i0_{n+1}\rangle+|1_i1_{n+1}\rangle}{\sqrt{2}}$$
or
$$\Delta^-_i:=\left(\frac{|0\rangle+|1\rangle}{\sqrt{2}}\right)^{\otimes n}\otimes \frac{|0_i1_{n+1}\rangle+|1_i0_{n+1}\rangle}{\sqrt{2}}$$
and its probability table can be expressed in terms of a suitable choice of $n-2$ constants and the probability table of either the $\Phi^+$ or of the $\Phi^-$ state.

A straightforward inductive argument based on the argument used in Proposition \ref{bell} will show that the $n+1$-partite dictatorship states are also weakly contextual for the measurements $U \left(\frac{\pi}{2},\frac{\pi}{8}\right)$, $U\left(\frac{\pi}{2},\frac{5\pi}{8}\right)$ and $U\left(\frac{\pi}{8},\frac{\pi}{2}\right)$, $U\left(\frac{5\pi}{8},\frac{\pi}{2}\right)$ respectively. 

Moreover, the generalization of the argument used in Theorem \ref{bellth} shows that the $n+1$-partite dictatorship states are not logically contextual for any possible dichotomic choice of measurements.\qed 
\end{proof}

\subsection{Non-contextual states}

Any $n+1$-partite balanced quantum state whose functional dependency is given by a constant polynomial is clearly a product state and hence non-contextual.

\end{chapter}
\begin{chapter}{Logical Contextuality is Almost Everywhere}\label{LNL}

We have seen in previous chapters that quantum states together with local observables (i.e. one-qubit measurements which can be made by each of the parties on their part of the state) give rise to probability tables which are, respectively, weakly, logically or strongly contextual.

This leads to the natural question of finding the maximum level of contextuality which a quantum state can exhibit for some choice of local observables. 

We have seen in the previous chapter that maximally entangled two-qubit states are weakly contextual, but they cannot be logically contextual, regardless of the local observables chosen. This is consistent with Hardy's results \cite{Hardy92, Hardy93}, who showed that an inequality-free proof of Bell's theorem (or equivalently, in the language of Abramsky and Brandenburger \cite{AbrBra11}, a logically contextual empirical model) can be given for all entangled two-qubit quantum states which are not maximally entangled.

We have also seen that among the multipartite states considered so far, Dycke states, GHZ states, and balanced states with functional dependencies, almost all of them have turned out to be at least logically contextual, for a suitable choice of observables. The only exceptions were given by those states which were equivalent to tensor products of maximally entangled two-qubit states and single states. In this chapter we shall see that this is not simply a coincidence, and in fact the following statement holds.

\begin{thma}
\label{mainth}
Let $\ket{\psi}$ be an $n$-qubit pure state. Then exactly one of the following two cases must hold:
\begin{enumerate}
\item $\ket{\psi}$ can be written (up to permutation of tensor factors) as a product
\begin{equation}
\label{prodform}
 \ket{\psi} \; = \; \ket{\psi_1} \otimes \cdots \otimes \ket{\psi_k} \tag{4.1}
 \end{equation}
where each $\ket{\psi_i}$ is either a 1-qubit state, or a 2-qubit maximally entangled state.
\item There are local observables such that the probability table arising from $\ket{\psi}$ and these local observables is logically contextual, that is, it admits a ``Hardy paradox'', \ie an inequality-free, probability-free proof of non-locality.
\end{enumerate}
If (2) holds, then only $n+2$ local observables are needed; two each for two of the parties, and one each for the other $n-2$ parties.
Moreover, there is an algorithm to decide which of the above cases holds, and which in case (2) explicitly computes the witnessing local observables. The complexity of this algorithm is $O(d \log^3 d)$ in the dimension $d = 2^n$.
\end{thma}

The remainder of this chapter is devoted to the detailed proof of this theorem. In particular, while the algorithm is quite simple, the arguments justifying its correctness are non-trivial.

\section{Overview of the argument}
\label{Str}

In this section, we shall state a number of main lemmas, and show how Theorem 4 is proved from these lemmas.
The proofs of the lemmas will be given in the following section.

\textbf{Notation}. We shall write $\Pn$ for the set of $n$-qubit pure states of the form~(\ref{prodform}) given in Theorem 4. For succinctness, we shall write the tensor product of a ket $\ket{\psi}$ with a 1-qubit ket $\ket{i}$, $i=0,1$, as $\ket{\psi}\ket{i}$ rather than $\ket{\psi} \otimes \ket{i}$.

We shall prove Theorem 4 by induction on $n$. The case $n=1$ is trivial.
The $n=2$ case is given by the following lemma.

\begin{restatable}[\textbf{The Base Case Lemma}]{lemma}{basecase}
\label{ground}
Every 2-qubit state $\ket{\psi}$ is either in $\PP_2$ (\ie it is either a product state, or a maximally entangled bipartite state), or there are two local observables for each party, which can be computed directly from the Schmidt decomposition of $\ket{\psi}$, and which witness the logical contextuality of $\ket{\psi}$.
\end{restatable}

The following lemmas and  corollary will be used in the induction step.

\begin{lemma}[\textbf{Going Up Lemma I}]
\label{guI}
If an $n$-qubit state $\ket{\psi}$ is logically contextual with some choice of local observables, then for any $n$-qubit state $\ket{\theta}$, and  $\alpha, \beta \in \Bbb{C}$ with $\alpha \neq 0$ and $|\alpha|^2 + |\beta|^2 = 1$, the states
\[ \alpha \ket{\psi} \ket{0} + \beta \ket{\theta} \ket{1}, \qquad \beta \ket{\theta} \ket{0} + \alpha \ket{\psi} \ket{1} \]
are also logically contextual with the same choice of observables, augmented with a single additional observable for the $n+1^{th}$ party.
\end{lemma}

\begin{restatable}[\textbf{Going Up Lemma II}]{lemma}{basecaseII}
\label{lnl}
If $|\theta\rangle=\alpha|\psi\rangle+\beta|\phi\rangle$ is a logically contextual $n$-qubit state under some choice of local observables, then for any  non-zero $x,y \in \Bbb{C}$ such that $|x\alpha|^2+|y\beta|^2=1$, the state $|\omega\rangle=x\alpha|\psi\rangle|0\rangle+y\beta|\phi\rangle|1\rangle$ is also logically contextual under the same choice of observables, augmented with a single additional observable for the $n+1^{th}$ party.
\end{restatable}

As an easy consequence of these lemmas, we have:
\begin{restatable}{corollary}{corr}
If $\ket{\theta}$ is logically contextual under some choice of local observables, so are $\ket{\theta} \otimes \ket{\eta}$ and $\ket{\eta} \otimes \ket{\theta}$ for any state $\ket{\eta}$, with the same choice of local observables for each party in $\ket{\theta}$, and a single observable for each party in $\ket{\eta}$. 
\end{restatable}

We now consider the induction step where we have an $n+1$-qubit state $\ket{\omega}$, $n>1$. We can write
\begin{equation}
\label{omeq}
\ket{\omega} \; = \; \alpha \ket{\psi} \ket{0} + \beta \ket{\phi} \ket{1} .
\end{equation}
By the Going Up Lemma I, if either $\ket{\psi}$ or $\ket{\phi}$ are logically contextual, so is $\ket{\omega}$, and we are done.

Suppose now that $\ket{\psi}$ and $\ket{\phi}$ are both in $\Pn$. We consider the parameterised family of states
\begin{equation}
\label{taudef}
\tau(a) \; = \; a \ket{\psi} + \sqrt{1 - a^2} \ket{\phi}, \qquad a \in [0, 1] . 
\end{equation}
If for some $a$, $\tau(a)$ is logically contextual, so is $\ket{\omega}$, by the Going Up Lemma II.
For the remaining case, we have the following rather remarkable result.

\begin{restatable}[\textbf{The Small Difference Lemma}]{lemma}{smalldiff}
\label{smalldifflemm}
Let $\ket{\psi}$ and $\ket{\phi}$ be states  in $\Pn$, and suppose that for all $a \in [0, 1]$, $\tau(a)$ is in $\Pn$, where $\tau(a)$ is defined by~(\ref{taudef}).
Then $\ket{\psi}$ and $\ket{\phi}$ differ in at most one qubit.
\end{restatable}

Applying this result to our decomposition~(\ref{omeq}) of $\ket{\omega}$, we have the following possibilities:
\begin{itemize}
\item $\ket{\psi} = \ket{\phi}$, in which case, from~(\ref{omeq}) and the bilinearity of the tensor product:
\[ \ket{\omega} \; = \; \alpha \ket{\psi} \ket{0} + \beta \ket{\psi} \ket{1} \; = \; \ket{\psi} \otimes (\alpha \ket{0} + \beta \ket{1}) . \]
Since by assumption $\ket{\psi}$ is in $\Pn$, $\ket{\omega}$ is in $\PP_{n+1}$.

\item $\ket{\psi} \neq \ket{\phi}$, in which case (up to permutation) we can write
\[ \begin{array}{lcl}
\ket{\psi} & = & \ket{\Psi} \otimes \ket{\eta} \\
\ket{\phi} & = & \ket{\Psi} \otimes \ket{\theta}
\end{array}
\]
where $\ket{\Psi}$ is in $\PP_{n-1}$ and $\ket{\eta}$ and $\ket{\theta}$ are 1-qubit states.
From this and~(\ref{omeq}), using the bilinearity of tensor product again we have
\[ \ket{\omega} = \ket{\Psi} \otimes \ket{\xi} \]
where $\ket{\xi}$ is a 2-qubit state.
We can apply the Base Case Lemma to $\ket{\xi}$ to conclude that $\ket{\xi}$ is either in $\PP_2$, in which case $\ket{\omega}$ is in $\PP_{n+1}$, or $\ket{\xi}$ is logically contextual, in which case $\ket{\omega}$ is logically contextual by the corollary to the Going Up Lemma.
\end{itemize}

\noindent At this point, we have established (1) and (2) of Theorem 4, but it seems that we require an infinite search to determine if there exists some $a \in [0, 1]$ for which $\tau(a)$ is logically contextual.

However, the following lemma shows that we only need to test a fixed, finite number of values for $a$ to determine this.

\begin{restatable}[\textbf{The 21 Lemma}]{lemma}{lemmtwoone}
\label{21lemm}
With the same notation as in the Small Difference Lemma, suppose that $\tau(a)$ is in $\Pn$ for 21 distinct values of $a$ in $[0, 1]$. Then $\tau(a)$ is in $\Pn$ for all $a \in [0, 1]$.
\end{restatable}

Thus this lemma allows us to determine which case applies on the basis of a finite number of tests.

In the next section, we shall give proofs of these lemmas. We shall then give an explicit algorithm in Section~5 to complete the proof of Theorem 4.

\section{Auxiliary lemmas}

Firstly, we collect a few useful basic properties.

\subsection{Background lemmas}

We consider  relations $\sim \; = \; \{ \sim_n \}_{n \in \Bbb{N}}$, where $\sim_n$ is an equivalence relation on $n$-qubit states.
We say that $\sim$ is  \emph{LC invariant} if for all $n$-qubit states $\ket{\psi}$, $\ket{\phi}$:
\begin{itemize}
\item If $\ket{\psi} \sim_n \ket{\phi}$, then $\ket{\psi} \in \Pn$ iff $\ket{\phi} \in \Pn$.
\item If $\ket{\psi} \sim_n \ket{\phi}$, then $\ket{\psi}$ is logically contextual iff  $\ket{\phi}$ is logically contextual.
\end{itemize}

\begin{lemma}
\label{permlemm}
The relation induced by permutation of tensor factors is  LC invariant.
\end{lemma}

\begin{lemma}
The relation of LU equivalence is  LC invariant. Here LU equivalence refers to the relation induced by the action of local (1-qubit) unitaries.
\end{lemma}

The following result will be used in the proof of the Small Difference lemma.
It refers to the ``partial inner product'' operation described e.g.~in \cite[p.~129]{QPSI}.\footnote{This is actually  the application of a linear map to a vector under Map-State duality \cite{abramsky2008categorical}.}

\begin{lemma}
\label{tensorthlemm}
Let $\ket{\phi}$ be a state in $\HH \otimes \KK$. 
For any states $\ket{\eta}$ in $\HH$ and $\ket{\theta}$ in $\KK$, if for all $\ket{\eta^{\bot}}$ orthogonal to $\ket{\eta}$,
$\langle \eta^{\bot} | \phi \rangle = \mathbf{0}$,
and  for all $\ket{\theta^{\bot}}$ orthogonal to $\ket{\theta}$,
$\langle \theta^{\bot} | \phi \rangle = \mathbf{0}$,
then (up to global phase) $\ket{\phi}  = \ket{\eta} \otimes \ket{\theta}$.
\end{lemma}
\begin{proof}
We extend $\ket{\eta}$ into an orthonormal basis $\ket{\eta_1}, \ldots , \ket{\eta_n}$ with $\ket{\eta} = \ket{\eta_1}$, and similarly extend $\ket{\theta}$ into $\ket{\theta_1}, \ldots , \ket{\theta_m}$ with $\ket{\theta} = \ket{\theta_1}$. 
Then $B = \{ \ket{\eta_i} \otimes \ket{\theta_j} \}_{i,j}$ forms an orthonormal basis of $\HH \otimes \KK$. Note that
\[ \HH \otimes \KK \; = \; (\ket{\eta} \otimes \ket{\theta})^{\bot \bot} \oplus S^{\bot \bot} \]
where $S = \{ \ket{\eta_i} \otimes \ket{\theta_j} \mid (i,j) \neq (1,1) \}$. Hence $S^{\bot \bot} = (\ket{\eta} \otimes \ket{\theta})^{\bot}$.
By our assumption and the defining property of the partial inner product \cite[Equation (6.47)]{QPSI}, $S \subseteq \ket{\phi}^{\bot}$. Hence
\[ \ket{\phi}^{\bot\bot} \; \subseteq \; S^{\bot} = S^{\bot \bot \bot} = (\ket{\eta} \otimes \ket{\theta})^{\bot \bot} . 
\]
Since these are one-dimensional subspaces, this implies that, up to global phase, $\ket{\phi}  = \ket{\eta} \otimes \ket{\theta}$.
\end{proof}

We now turn to detailed proofs of the main lemmas.
For convenience, we shall repeat the statements of the lemmas.

\subsection{The Base Case lemma}

\basecase*

\begin{proof}
This is essentially Hardy's construction in \cite{Hardy93}.
Using the  Schmidt decomposition, every two-particle entangled state can be written in the form
$$|\psi\rangle=\alpha|+\rangle_1|+\rangle_2 + \beta|-\rangle_1|-\rangle_2$$
for an appropriate choice of basis states $|\pm\rangle_i$ for each particle $i$, and normalized non-zero real constants $\alpha$ and $\beta$. 

The logical contextuality of $|\psi\rangle$ is witnessed by a set of four dichotomic observables, two for each of the two parties, namely $U_i$ and $D_i$, $i=1,2$. These observables can be defined as $U_i=|u_i\rangle\langle u_i|$ and $D_i=|d_i\rangle\langle d_i|$ where
\begin{align}
 |u_i\rangle&=\frac{1}{\sqrt{|\alpha|+|\beta|}} (\beta^{\frac{1}{2}}|+\rangle_i + \alpha^{\frac{1}{2}}|-\rangle_i)\\
 |d_i\rangle&=\frac{1}{\sqrt{|\alpha|^3+|\beta|^3}} (\beta^{\frac{3}{2}}|+\rangle_i - \alpha^{\frac{3}{2}}|-\rangle_i)
\end{align}
Hardy's paper also explains why it is not possible to run this particular non-locality argument when either $\alpha$ or $\beta$ are equal to zero (product states), or when $|\alpha| = |\beta|$ (maximally entangled states), that is, when the state $|\psi\rangle$ belongs to $\mathcal{P}_2$. 

There is one subtle remaining point. To show that the dichotomy asserted in the lemma is strictly disjoint, we must show that in the maximally entangled case, there is \emph{no} choice of local observables which can give rise to logical contextuality.
This is shown for the case where each party has the same two local observables as in Theorem \ref{bellth} in the previous chapter, and more generally for any finite sets of local observables as Theorem 2.6.5 in \cite{mansfield2013}.
\end{proof}

\subsection{The Going Up lemmas}

The proofs of the two Going Up lemmas are quite similar. We shall prove the second, which is somewhat harder.

\basecaseII*

\begin{proof}
Since $|\theta\rangle$ is logically contextual there must be some context $U'$ and some $s'\in S(U')$ such that the formula $\Psi=\varphi_{s'}\wedge\bigwedge_{U\in\mathcal{U}\backslash U'} \varphi_U$ is not satisfiable. We will show that it is possible to construct a similar unsatisfiable formula in order to prove the logical non-locality of $|\omega\rangle$.

The $n+1^{th}$ party is assigned a single observable $B = B(x,y)$, whose eigenvectors are $|b_+\rangle=\overline{y}|0\rangle+\overline{x}|1\rangle$ and $|b_-\rangle=x|0\rangle-y|1\rangle$. The observable $B$ is given by the self-adjoint matrix
\begin{equation}\label{Bobservable}
B(x,y)= \left( \begin{array}{cc}
             -|x|^2+|y|^2 & 2x\overline{y} \\ 2\overline{x}y & |x|^2-|y|^2)
            \end{array}\right)  
\end{equation}
For any $n$-qubit state $|\mu\rangle$ we have
\begin{align*}
\langle\mu|\langle b_+|\cdot|\omega\rangle &= x\alpha\langle\mu|\psi\rangle\langle b_+|0\rangle+ y\beta\langle\mu|\phi\rangle\langle b_+|1\rangle\\
&=xy\alpha\langle\mu|\psi\rangle + xy\beta\langle\mu|\phi\rangle\\
&=xy(\alpha\langle\mu|\psi\rangle+\beta\langle\mu|\phi\rangle)\\
&=xy\langle\mu|\theta\rangle
\end{align*}
which implies
\begin{equation}\label{theta1}
\langle\mu|\theta\rangle|^2=0 \Leftrightarrow |\langle\mu|\langle b_+|\cdot|\omega\rangle|^2=0
\end{equation}
The augmented set of allowed measurements $\widetilde{X}=X\cup{(B,n+1)}$ is then covered by the family of compatible subsets 
\[ \widetilde{\mathcal{U}}:=\{\widetilde{U}=U\cup{(B,n+1)}~|~U\in \mathcal{U}\} . \]

Let $T(\widetilde{U})$ denote the support of $\widetilde{U}\in\widetilde{\mathcal{U}}$. Equation~(\ref{theta1}) implies that $S(U)=\{s~|~s+\in T(\widetilde{U}\}$ where $s+ :\widetilde{U}\rightarrow\mathbf{2}$ extends $s$ by mapping $(B,n+1)$ to $+$. As a side remark, note that the presence of an analogously defined section $\sigma-$ in $T(\widetilde{U})$ does not necessarily imply the presence of $\sigma$ in $S(U)$.

Now let $t':=s'+$. We have $t'\in T(\widetilde{U'})$ and we can define the analogue of the proposition $\varphi_{s'}$ as
\begin{align*}
\varphi_{t'}&=\bigwedge_{x\in\widetilde{U'},\, t'(x)=+} x \AND \bigwedge_{x\in\widetilde{U'},\, t'(x)=-}\neg x\\
&=\bigwedge_{x\in U',\,s'(x)=+} x \AND \bigwedge_{x\in U',\,s'(x)=-} \neg x \AND z_{n+1}\\
&=\varphi_{s'}\AND z_{n+1}
\end{align*}
where $z_{n+1}$ denotes the boolean variable corresponding to the outcome on the $n+1^{th}$ qubit. Recall that $+$ stands for true and $-$ stands for false.

It also holds that 
\begin{align*}
&\varphi_{\widetilde{U}}=\bigvee_{t\in T(\widetilde{U})}\varphi_t=\bigvee_{t\in T(\widetilde{U}),\,t=s+} \varphi_t \vee \bigvee_{t\in T(\widetilde{U}),\, t=s-}\varphi_t\\
=&\left(\bigvee_{s\in S(U)}(\varphi_s\AND z_{n+1})\right) \vee \left(\bigvee_{t\in T(\widetilde{U}),\, t=\sigma-} \varphi_{\sigma}\AND \neg z_{n+1}\right)\\
=&\left(\left(\bigvee_{s\in S(U)}\varphi_s\right)\AND z_{n+1} \right)\vee \left(\underbrace{\left(\bigvee_{t=\sigma-\in T(\widetilde{U})} \varphi_{\sigma} \right)}_{\gamma_U} \AND \neg z_{n+1}\right)\\
=&\left(\varphi_U\AND z_{n+1}\right)\vee \left(\gamma_U \AND \neg z_{n+1}\right)
\end{align*}

We can now define the formula $\Omega$ which specifies the joint outcome $t'=s'+\in \widetilde{U'}$ as well as the joint outcomes within the supports of all compatible sets of measurements $\widetilde{U}\neq \widetilde{U'}$. This formula is just the conjunction of $\Psi$ and $z_{n+1}$. In order to show that $|\omega\rangle$ is logically contextual, it suffices to show that $\Omega$ has no satisfiable assignment. This is indeed the case, as the fact that $\Psi$ is not satisfiable implies that $\Omega$ is also not satisfiable, thus completing our proof. Indeed, we have
\begin{align*}
\Omega &=\varphi_{t'}\AND\bigwedge_{\widetilde{U}\in \mathcal{\widetilde{U}}\backslash\widetilde{U'}} \varphi_{\widetilde{U'}} \\
&=(\varphi_{s'}\AND z_{n+1})\AND \bigwedge_{\widetilde{U}\in \mathcal{\widetilde{U}}\backslash\widetilde{U'}} \left[ (\varphi_{\widetilde{U}}\AND z_{n+1})\vee (\gamma_U \AND \neg z_{n+1})\right] \\
&=\bigwedge_{\widetilde{U}\in \mathcal{\widetilde{U}}\backslash\widetilde{U'}}\left[ (\varphi_{s'}\AND z_{n+1})\AND ((\varphi_{\widetilde{U}}\AND z_{n+1})\vee(\gamma_U\AND z_{n+1}))\right]\\
&=\bigwedge_{\widetilde{U}\in \mathcal{\widetilde{U}}\backslash\widetilde{U'}}\left[ (\varphi_{s'}\AND z_{n+1})\AND(\varphi_{\widetilde{U}}\AND z_{n+1})\right]\\
&=\varphi_{s'}\AND\left(\bigwedge_{\widetilde{U}\in \mathcal{\widetilde{U}}\backslash\widetilde{U'}} \varphi_{\widetilde{U}}\right)\AND z_{n+1} = \Psi\AND z_{n+1} 
\end{align*}
\end{proof}

The proof of the Going Up Lemma I follows by a similar argument, where the observables are augmented by the  $Z$ measurement for the $n+1^{th}$ party.

The Going Up lemmas have the following useful corollary.

\corr*

\begin{proof}
We argue by induction on the number of qubits in $\ket{\eta}$. 
If $\ket{\eta} = \alpha \ket{0} + \beta\ket{1}$, then by bilinearity of the tensor product,
\[ \ket{\theta} \otimes \ket{\eta} = \alpha \ket{\theta}\ket{0} + \beta \ket{\theta} \ket{1} \]
and we can apply the Going Up Lemma I.

For the inductive case, we can write
\[ \ket{\eta} = \alpha \ket{\eta_0} \ket{0} + \beta \ket{\eta_1} \ket{1} \]
and by bilinearity
\[ \ket{\theta} \otimes \ket{\eta} = \alpha \ket{\theta}\ket{\eta_0} \ket{0} + \beta \ket{\theta} \ket{\eta_1} \ket{1} . \]
By induction hypothesis, $\ket{\theta}\ket{\eta_0}$ and $\ket{\theta}\ket{\eta_1}$ are logically contextual, and we can apply the Going Up Lemma I again to conclude.
\end{proof}

\subsection{The Small Difference lemma}

\smalldiff*

\begin{proof}
Firstly, we note that each state $\ket{\psi}$ in $\Pn$ has an \emph{entanglement type}, which can be described by a graph on $n$ vertices with an edge from $i$ to $j$ when the corresponding qubits of $\ket{\psi}$ are maximally entangled.
There are finitely many such graphs, and we can partition $\Pn$ into $P_1, \ldots , P_M$ according to the entanglement type.
All the states in each $P_l$ are LU equivalent.

As before, we can write $\ket{\psi}$, up to permutation of tensor factors, as
\begin{equation}
\label{psieq}
\ket{\psi} \; = \; \ket{\psi_1} \otimes \cdots \otimes \ket{\psi_k} 
\end{equation}
where each $\ket{\psi_i}$ is either a 1-qubit state, or a 2-qubit maximally entangled state.

We recall the definition of the parameterised family of states $\tau(a)$, $a \in [0, 1]$:
\[ \tau(a) \; = \; a \ket{\psi} + g(a) \ket{\phi}  \]
where $g(a) = \sqrt{1 - a^2}$.

Now for each $\delta \in [0, 1]$, we define a set
\[ R_{\delta}=\{  \tau(a) \mid 1-\delta < a \leq 1\} \]
Under our assumption on $\tau(a)$, each set $R_{\delta}$, which is infinite,  is partitioned among the sets $P_1, \ldots , P_M$.
Also, $\delta < \delta'$ implies $R_{\delta} \supset R_{\delta'}$. Hence, by an application of K\"onig's infinity lemma \cite{levy2012basic}, we can conclude that there is  $l$ with $1 \leq l \leq M$, and  an infinite increasing sequence $\{ a_i \}$ with supremum $1$, such that $\tau(a_i)$ is in $P_l$ for all $i$.

Since all the states $\tau(a_i)$ are LU-equivalent, we can express them in terms of a representative state $\ket{\Theta} \in P_l$ as
\begin{equation}\label{theta}
\tau(a_i) \; = \; U^1_{a_i}\otimes U^2_{a_i}\otimes\ldots\otimes U^n_{a_i} |\Theta\rangle
\end{equation}
Since $\Un(2)^n$ is compact, there is a convergent subsequence $\{U^1_{b_i}\otimes \ldots\otimes U^n_{b_i}\}_{b_i}$, whose limit as $i\rightarrow\infty$ is ${W^1}\otimes\ldots \otimes {W^n}$.
The limit of the corresponding subsequence $\{ a_{b_i} \}$ is still $1$.
Hence $\ket{\psi}$ is also a member of $P_l$, as
\begin{equation*}\label{psi}
|\psi\rangle \; = \; \lim_{i\rightarrow\infty} \tau(a_i) \; = \; \lim_{i\rightarrow\infty} \tau(a_{b_i}) \; = \;  {W^1}\otimes\ldots\otimes {W^n}|\Theta\rangle .
\end{equation*}
This, together with Equation (\ref{theta}), implies that we can express each $\tau(a_i)$ as
\[ \tau(a_i) \;  = \; [(U^1_{a_i}{W^1}^{\dagger})\otimes\ldots\otimes (U^n_{a_i}{W^n}^{\dagger})]|\psi\rangle . \]
Equivalently, using Equation \ref{psieq} and the definition of $\tau(a_i)$, we can obtain a family of equations, one for each $a_i$:
\begin{align}\label{psiphi}
a_i |\psi\rangle + g(a_i)|\phi\rangle &=[(U^1_{a_i}{W^1}^{\dagger})\otimes\ldots\otimes (U^n_{a_i}{W^n}^{\dagger})]|\psi\rangle  \nonumber \\
&= |\psi_1^{i}\rangle\otimes\ldots\otimes |\psi_{k}^{i}\rangle
\end{align}
where $|\psi_{j}^{i}\rangle$ is the state obtained after the LU transformation of $|\psi_j\rangle$, $1 \leq j \leq k$.

We shall now make use of the ``partial inner product'' operation described e.g.~in \cite[p.~129]{QPSI}.  We will use this operation to probe the components of~(\ref{psiphi}).

By Lemma~\ref{tensorthlemm}, if for all $j$, and for any state $\ket{\psi_j^{\perp}}$ orthogonal to $|\psi_j\rangle$, the application of $\ket{\psi_j^{\perp}}$ to $|\phi\rangle$ results in a null vector, we must have  $|\phi\rangle=|\psi\rangle$, and the lemma is proved. 

Otherwise, assume there is some $j$, and some $\ket{\psi_j^{\perp}}$, such that $\langle\psi_j^{\perp}|\phi\rangle$ is a non-zero vector.
For ease of notation, we take $j=k$. 

Applying $\langle\psi_k^{\perp}|$ on both sides of Equation (\ref{psiphi}), we obtain
\begin{align}
g(a_i)\langle\psi_k^{\perp}|\phi\rangle \; &= \; |\psi_1^{i}\rangle\otimes\ldots\otimes |\psi_{k-1}^{i}\rangle \underbrace{\langle\psi_k^{\perp}|\psi_k^{i}\rangle}_{\neq 0} \\
\underbrace{\langle\psi_k^{\perp}|\phi\rangle}_{constant\ vector} \; &= \; \epsilon_i |\psi_1^{i}\rangle\otimes\ldots\otimes |\psi_{k-1}^{i}\rangle \label{final}
\end{align}
The fact that the LHS of~(\ref{final}) is constant implies that for any $i$ and $j$ we must have $\epsilon_i=\epsilon_j$. We write $\epsilon$ for this common value.
We must also have $|\psi_t^{i}\rangle=|\psi_t^{j}\rangle$ for all $t<k$, and we write $\ket{\psi_t'}$ for the common value. In fact we have 
\begin{equation}\label{limit}
|\psi_t\rangle=\lim_{i\rightarrow\infty}|\psi_t^{i}\rangle=|\psi_t'\rangle, \quad 1 \leq t <k 
\end{equation}
This means that we can rewrite Equation (\ref{final}) as 
\[ \langle\psi_k^{\perp}|\phi\rangle \; = \; \epsilon |\psi_1\rangle\otimes\ldots\otimes |\psi_{k-1}\rangle . \]
A similar analysis will apply to any state $\ket{\eta}$ orthogonal to $\ket{\psi_k}$ for which $\langle \eta | \phi \rangle$ is non-zero.
Using Equation (6.48) from \cite{QPSI}, and bilinearity of the tensor product, we obtain
\begin{equation}\label{form}
|\phi\rangle \; = \; x|\phi_2\rangle\otimes|\psi_k\rangle + \epsilon|\psi_1\rangle\otimes\ldots\otimes|\psi_{k-1}\rangle\otimes \ket{\xi}\
\end{equation}
for some $\ket{\phi_2}$ and $\ket{\xi}$.
We can use Equation (\ref{limit}) to rewrite Equation (\ref{psiphi}) as
\[ a_i |\psi\rangle + g(a_i) |\phi\rangle = |\psi_1\rangle\otimes\ldots\otimes |\psi_{k-1}\rangle\otimes|\psi_{k}^{i}\rangle . \]
For any $j<k$, if we apply any state $\ket{\psi_j^{\perp}}$, orthogonal to $|\psi_j\rangle$, to the above equation we obtain $\langle\psi_j^{\perp}|\phi\rangle=0$. 
Together with Equation (\ref{form}), this implies that $|\phi_2\rangle=|\psi_1\rangle\otimes\ldots\otimes|\psi_{k-1}\rangle$. 
So $|\phi\rangle$ and $|\psi\rangle$ can differ by at most one component. 

These  components cannot be  two-qubit maximally entangled states, since by assumption all linear combinations $\tau(a)$ of $|\psi\rangle$ and $|\phi\rangle$, with $a \in [0, 1]$, belong to $\mathcal{P}_n$, and hence, using bilinearity again, the corresponding linear combinations  of these components  would also have to be maximally entangled, yielding a contradiction.

Thus we conclude that $\ket{\psi}$ and $\ket{\phi}$ can differ at most in a  one-qubit component.
\end{proof}

\subsection{The 21 lemma}

We shall need some elementary facts about partial traces (see e.g.~\cite{NieChu})):
\begin{itemize}
\item A pure state in $\HH \otimes \KK$ is a product state  if and only if tracing out over $\HH$ results in a pure state.
\item Tracing out over one party of a maximally entangled bipartite state yields a maximally mixed state.
\item A mixed state $\rho$ is pure if and only if $\Tr \rho^2 = 1$.
\end{itemize}

\lemmtwoone*

\begin{proof}

If a state belongs to $\mathcal{P}_n$ then all partial traces over $n-1$ parties result either in a pure state or in the maximally mixed state $\frac{1}{2}I_2$. 

We can express $|\psi\rangle$ and $|\phi\rangle$ as 
\begin{align*}
 |\psi\rangle \; &= \; \sum_{\sigma_i} a^0_{\sigma_i}|\sigma_i\rangle|0\rangle+ \sum_{\sigma_i} a^1_{\sigma_i}|\sigma_i\rangle|1\rangle\\
 |\phi\rangle \; &= \; \sum_{\sigma_i} b^0_{\sigma_i}|\sigma_i\rangle|0\rangle+ \sum_{\sigma_i} b^1_{\sigma_i}|\sigma_i\rangle|1\rangle
\end{align*}
where the $\sigma_i$ index the elements of the computational basis on $n-1$ qubits.

This means we can write the density matrix corresponding to $|\tau(a)\rangle=a|\phi\rangle + b |\psi\rangle$, with $b = g(a)$, as
\begin{align*}
 |\tau(a)\rangle\langle\tau(a)| \; = & \; \sum_{\sigma_i,\sigma_j}(a a^0_{\sigma_i}+b b^0_{\sigma_i})(a\overline{a^0_{\sigma_j}}+b\overline{b^0_{\sigma_j}})|\sigma_i\rangle\langle\sigma_j|\otimes |0\rangle\langle 0| \\
& \; + \sum_{\sigma_i,\sigma_j}(a a^0_{\sigma_i}+b b^0_{\sigma_i})(a\overline{a^1_{\sigma_j}}+b\overline{b^1_{\sigma_j}})|\sigma_i\rangle\langle\sigma_j|\otimes |0\rangle\langle 1|\\  
& \; + \sum_{\sigma_i,\sigma_j}(a a^1_{\sigma_i}+b b^1_{\sigma_i})(a\overline{a^0_{\sigma_j}}+b\overline{b^0_{\sigma_j}})|\sigma_i\rangle\langle\sigma_j|\otimes |1\rangle\langle 0|\\
& \; +  \sum_{\sigma_i,\sigma_j}(a a^1_{\sigma_i}+b b^1_{\sigma_i})(a\overline{a^1_{\sigma_j}}+b\overline{b^1_{\sigma_j}})|\sigma_i\rangle\langle\sigma_j|\otimes |1\rangle\langle 1|
\end{align*}

The partial trace over the first $n-1$ qubits of $|\tau(a)\rangle$ is given by
\begin{align*}
\rho_n \; = \; & \Tr_{n-1} |\tau(a)\rangle\langle\tau(a)| \; = \; \sum_{\sigma_i} \langle\sigma_i|\tau(a)\rangle\langle\tau(a)|\sigma_i\rangle \\
= \; &\sum_{\sigma_i}(a a^0_{\sigma_i}+b b^0_{\sigma_i})(a\overline{a^0_{\sigma_i}}+b\overline{b^0_{\sigma_i}})|0\rangle\langle 0| + \sum_{\sigma_i}(a a^0_{\sigma_i}+b b^0_{\sigma_i})(a\overline{a^1_{\sigma_i}}+b\overline{b^1_{\sigma_i}})|0\rangle\langle 1| \\
& + \sum_{\sigma_i}(a a^1_{\sigma_i}+b b^1_{\sigma_i})(a\overline{a^0_{\sigma_i}}+b\overline{b^0_{\sigma_i}})|1\rangle\langle 0| + \sum_{\sigma_i}(a a^1_{\sigma_i}+b b^1_{\sigma_i})(a\overline{a^1_{\sigma_i}}+b\overline{b^1_{\sigma_i}})|1\rangle\langle 1| \\
= \; &|0\rangle\langle 0|\sum_{\sigma_i}(a^2 a^0_{\sigma_i}\overline{a^0_{\sigma_i}} +(1-a^2) b^0_{\sigma_i}\overline{b^0_{\sigma_i}} + a\sqrt{1-a^2}(a^0_{\sigma_i}\overline{b^0_{\sigma_i}}+b^0_{\sigma_i}\overline{a^0_{\sigma_i}})) \\ 
\; + \; & |0\rangle\langle 1|\sum_{\sigma_i}(a^2 a^0_{\sigma_i}\overline{a^1_{\sigma_i}} +(1-a^2) b^0_{\sigma_i}\overline{b^1_{\sigma_i}} + a\sqrt{1-a^2}(a^0_{\sigma_i}\overline{b^1_{\sigma_i}}+b^0_{\sigma_i}\overline{a^1_{\sigma_i}})) \\
\; +\; & |1\rangle\langle 0|\sum_{\sigma_i}(a^2 a^1_{\sigma_i}\overline{a^0_{\sigma_i}} +(1-a^2) b^1_{\sigma_i}\overline{b^0_{\sigma_i}} + a\sqrt{1-a^2}(a^1_{\sigma_i}\overline{b^0_{\sigma_i}}+b^1_{\sigma_i}\overline{a^0_{\sigma_i}}))\\
\; +\; & |1\rangle\langle 1|\sum_{\sigma_i}(a^2 a^1_{\sigma_i}\overline{a^1_{\sigma_i}} +(1-a^2) b^1_{\sigma_i}\overline{b^1_{\sigma_i}} + a\sqrt{1-a^2}(a^1_{\sigma_i}\overline{b^1_{\sigma_i}}+b^1_{\sigma_i}\overline{a^1_{\sigma_i}})) 
\end{align*}

The partial trace $\rho_n$ is equal to the maximally mixed state if and only if 
\begin{align*}
 1/2 \; &= \; \sum_{\sigma_i}(a^2 a^0_{\sigma_i}\overline{a^0_{\sigma_i}} +(1-a^2) b^0_{\sigma_i}\overline{b^0_{\sigma_i}} + a\sqrt{1-a^2}(a^0_{\sigma_i}\overline{b^0_{\sigma_i}}+b^0_{\sigma_i}\overline{a^0_{\sigma_i}})) \\
0 \; &= \; \sum_{\sigma_i}(a^2 a^0_{\sigma_i}\overline{a^1_{\sigma_i}} +(1-a^2) b^0_{\sigma_i}\overline{b^1_{\sigma_i}} + a\sqrt{1-a^2}(a^0_{\sigma_i}\overline{b^1_{\sigma_i}}+b^0_{\sigma_i}\overline{a^1_{\sigma_i}}))\\
0\; &= \; \sum_{\sigma_i}(a^2 a^1_{\sigma_i}\overline{a^0_{\sigma_i}} +(1-a^2) b^1_{\sigma_i}\overline{b^0_{\sigma_i}} + a\sqrt{1-a^2}(a^1_{\sigma_i}\overline{b^0_{\sigma_i}}+b^1_{\sigma_i}\overline{a^0_{\sigma_i}}))\\
1/2 \; &= \; \sum_{\sigma_i}(a^2 a^1_{\sigma_i}\overline{a^1_{\sigma_i}} +(1-a^2) b^1_{\sigma_i}\overline{b^1_{\sigma_i}} + a\sqrt{1-a^2}(a^1_{\sigma_i}\overline{b^1_{\sigma_i}}+b^1_{\sigma_i}\overline{a^1_{\sigma_i}})) 
\end{align*}

Each of these equations yields a polynomial of degree 4 in $a$. Indeed, each equation has the form
\[ cab + q(a) = d \]
where $c$ and $d$ are constants, $b = \sqrt{1 - a^2}$, and $q(a)$ is a quadratic polynomial in $a$.
We can write this as
\[ cab  = -q(a) + d \]
and square both sides to obtain
\[ c^2a^2(1 - a^2) = (-q(a) + d)^2 \]
which is a quartic polynomial in $a$.
Hence, there can be at most 4 values of $a$  in $[0, 1]$ for which the partial trace $\rho_n$ is equal to a maximally mixed state.

On the other hand, $\rho_n$ is equal to a pure state if and only if $\Tr\rho_n^2=1$. By a similar analysis, this condition turns out to be equivalent to a polynomial equation of degree 16 in $a$. Unless the polynomial is degenerate, \ie the coefficients cancel so that the equation reduces to $1=1$, there can be at most 16 values of $a$ for which the partial trace $\rho_n$ is equal to a pure state.

Therefore, if there are more than $4+16=20$ values of $a$ in $[0, 1]$ for which the linear combination $\tau(a)$ belongs to $\mathcal{P}_n$, we can conclude that one of the polynomial equations above was degenerate, hence $\tau(a) \in\mathcal{P}_n$ for all values of $a$. 
\end{proof}

\paragraph{Remark} If $\ket{\psi}$ and $\ket{\phi}$ differ by one qubit, the partial trace of $\tau(a)$ for any value of $a$ will yield a pure state, and hence we will always have $\Tr\rho_n^2=1$. Thus the polynomial will indeed be degenerate in this case. This shows the necessity for the Small Difference Lemma.

\section{The algorithm}

We now give an explicit, albeit informal description of the algorithm which follows straightforwardly from our results.

We begin with a subroutine which we will use to test if a state is in $\Pn$.

\vspace{.1in}
\begin{tabular}{ll}
\textsc{subroutine} & \textsf{Test}$\Pn$  \\

\textbf{Input} & $n$-qubit quantum state $\ket{\theta}$ \\
\textbf{Output} & Either \\
& \textsf{Yes}, and entanglement type of $\ket{\theta}$, or \\
& \textsf{No}
\end{tabular}

\begin{enumerate}
\item Compute the $n-1$ partial traces $\rho_i$ over $n-1$ qubits of $\ket{\theta}$.
If any $\rho_i$ is not a maximally mixed state, compute $\Tr \rho_i^2$.
If $\Tr \rho_i^2 \neq 1$, return \textsf{No}.

We now have the list $\{ i_1, \ldots , i_k \}$ of indices for which the maximally mixed state was returned.

\item For each $i_p$ in the list, find its ``partner'' $i_q$ by computing the partial traces $\rho_{i_p,i_q}$ over $n-2$ qubits, and then testing if $\Tr \rho^2_{i_p,i_q} = 1$.\\
If we cannot find the partner for some $i_p$, return \textsf{No}.

\item Otherwise, we return \textsf{Yes}.
We also have the complete entanglement type of $\ket{\theta}$, and we have computed all the single-qubit components. $\Box$
\end{enumerate}

\begin{tabular}{ll}
\textsc{algorithm} & \\
\textbf{Input} & An $n$-qubit state $\ket{\omega}$ \\

\textbf{Output} & Either \\
& \textsf{Yes} if $\ket{\omega}$ is logically contextual, \\ 
& together with a list of $n+2$ local observables, or \\
& \textsf{No} if $\ket{\omega}$ is in $\Pn$.
\end{tabular}

\subsection*{Base Cases}

\begin{enumerate}
\item If $n=1$, output \textsf{No}.
\item If $n=2$, apply the Hardy procedure of the Base Case Lemma to the Schmidt decomposition of $\ket{\omega}$.
\end{enumerate}

\subsection*{Recursive Case: $n+1$, $n>1$}

\begin{enumerate}
\item We apply \textsf{Test}$\PP_{n+1}$ to $\ket{\omega}$. If $\ket{\omega}$ is in $\PP_{n+1}$, return \textsf{No}.

\item Otherwise,
we write
\[ \ket{\omega} \; = \; \alpha \ket{\psi} \ket{0} + \beta \ket{\phi} \ket{1} . \]
Explicitly, if $\ket{\omega}$ is represented by a $2^{n+1}$-dimensional complex vector
\[ \sum_{\sigma \in \{ 0, 1 \}^{n+1}} a_{\sigma} \ket{\sigma} \]
in the computational basis, we can define
\[ \alpha = \sqrt{\sum_{\sigma \in \{ 0, 1 \}^n} | a_{\sigma 0} |^2}, \qquad \beta = \sqrt{\sum_{\sigma \in \{ 0, 1 \}^n} | a_{\sigma 1} |^2} \]
\[ \ket{\psi} = \frac{1}{\alpha} \sum_{\sigma \in \{ 0, 1 \}^n} a_{\sigma 0} \ket{\sigma}, \qquad \ket{\phi} = \frac{1}{\beta} \sum_{\sigma \in \{ 0, 1 \}^n} a_{\sigma 1} \ket{\sigma} . \]

\item We apply \textsf{Test}$\Pn$ to $\ket{\psi}$. If $\ket{\psi}$ is not in $\Pn$, we proceed recursively with $\ket{\psi}$, and then extend the observables using the construction of the Going Up Lemma I. 

\item Otherwise, we proceed similarly with $\ket{\phi}$.

\item Otherwise, both $\ket{\psi}$ and $\ket{\phi}$ are in $\Pn$. \\
For $a$ in $(0, 1)$, we define
\[ \tau(a) \; := \; a \ket{\psi} + \sqrt{1 - a^2} \ket{\phi} . \]
For $19$ distinct values in $(0, 1)$, we assign these values to $a$, and apply \textsf{Test}$\Pn$ to $\tau(a)$.

If we find a value of $a$ for which $\tau(a)$ is not in $\Pn$, we use that value to compute the local observable $B(\frac{\alpha}{a},\frac{\beta}{\sqrt{1-a^2}})$ for the $n+1^{th}$ party, as specified in the Going Up Lemma II, and continue the recursion with the $n$-qubit state $\tau(a)$.

\item Otherwise, by the 21 Lemma and the Small Difference Lemma, the only remaining case is where $\ket{\psi}$ and $\ket{\phi}$ differ in one qubit. We have these qubits $\ket{\psi_1}$, $\ket{\phi_1}$ from our previous applications of \textsf{Test}$\Pn$.
In this final case, we can write $\ket{\omega}$ as
\[ \ket{\omega} = \ket{\Psi} \otimes \ket{\xi} \]
where $\ket{\Psi}$ is in $\PP_{n-1}$, and $\ket{\xi}$ is a 2-qubit state. 
Moreover, we have 
\[ \ket{\xi} = \alpha \ket{\psi_1}\ket{0} + \beta \ket{\phi_1}\ket{1} . \]

\item We apply the Base Case procedure to $\ket{\xi}$, which we know cannot be maximally entangled, by Step 1.
We output \textsf{Yes}, together with the two local observables for each party produced by the Hardy construction, and the $n-2$ local observables for $\ket{\Psi}$ produced by the Corollary to the Going Up lemmas. $\Box$
\end{enumerate}

The above algorithm of course involves computation over the real and complex numbers. More precisely, with the usual coding of complex numbers as pairs of reals, we require the field operations and comparison tests  on real numbers. For simplicity, we discuss the complexity of the algorithm in the Blum-Shub-Smale model of computation \cite{blum1989theory}, where we assume that arbitrary real numbers can be stored, and the above operations performed, with unit cost.
Thus the input size of an $n$-qubit state is the dimension $d = 2^n$.

The \textsf{Test}$\Pn$ subroutine performs $n-1$ partial traces over $n-1$ qubits. Each such partial trace involves computing the 4 entries of a matrix, where each entry is a sum over $2^{n-1}$ products. It also computes $O(n^2)$ partial traces over $n-2$ qubits, each of which involves computing 16 entries, each a sum over $2^{n-2}$ products. 
For an $n$-qubit input, at each level of the recursion, we call the subroutine a number of times bounded by a constant, and the recursion terminates in $O(n)$ steps.
Thus we obtain a complexity bound of $O(d \log^3 d)$ operations, in the input size $d$. Of course, in practice the limiting factor is the exponential size of the classical  representation of a quantum state.

The algorithm has been implemented in \textsf{Mathematica}\texttrademark, and has been tested on input states of up to 10 qubits.

\end{chapter}

\newpage
$ $
\newpage

\part{Second part}
\newpage
$ $
\newpage

\begin{chapter}{Overview of the Topos Approach}\label{TA}

In this chapter we shall gradually build up the necessary mathematical machinery needed to understand the main topos theoretic constructions. The abstract mathematics is presented first, and we have limited our exposition of category theoretical abstractions to those notions which are absolutely indispensable for understanding the concept of an elementary topos and the categorial definition of Gelfand duality, plus a few illuminating examples and comments. The basic category theoretical concepts introduced here can be found in any standard textbook, such as \cite{Gol, MacMoe} and, for the sake of clarity, we have chosen to follow the slightly more simplified presentation in \cite{Gol}.

We then proceed to give the technical definitions of those physical concepts which the topos approach seeks to model, again restricting ourselves to those which are absolutely necessary for the understanding of the following chapters. 

\section{The Basics}

\subsection{Categories and functors}

A category may be thought of in the first instance as a universe for a particular kind of mathematical discourse. In axiomatic terms a category $\mathscr{C}$ consists of 

\singlespacing

\begin{itemize}
\item[a)] a class $Ob(\mathscr{C})$ of \textbf{objects}
\item[b)] for each $A, B\in Ob(\mathscr{C})$ a class $\mathscr{C}(A, B)$ of \textbf{morphisms} from $A$ to $B$
\item[c)] for each $A, B, D\in Ob(\mathscr{C})$ a binary opperation (called \textbf{composition of morphisms})
\begin{equation*}
\begin{aligned}
\mathscr{C}(B, D)\times \mathscr{C}(A, B) &\longrightarrow \mathscr{C}(A, D)\\
(g, f) &\mapsto g\circ f
\end{aligned}
\end{equation*}
which satisfies the following axioms
\begin{itemize}
\item[(i)] \textbf{identity:} for each $A\in Ob(\mathscr{C})$ there exists an identity $1_A\in\mathscr{C}(A, A)$ such that 
\begin{equation*}
\begin{aligned}
f\circ 1_A&=f,\ \ \ \forall f\in \mathscr{C}(A,B)\\
1_A\circ f &=f,\ \ \ \forall f\in \mathscr{C}(B,A)
\end{aligned}
\end{equation*}
\item[(ii)] \textbf{functoriality:} for any $A,B,D,E \!\in\! Ob(\mathscr{C})$, $f\! \in\!\mathscr{C}(A,B)$, $g\! \in\! \mathscr{C}(B,D)$, $h\! \in\!\mathscr{C}(D,E)$
\begin{equation*}
h\circ(g\circ f) = (h\circ g)\circ f
\end{equation*}
\end{itemize}
\end{itemize}

\doublespacing

\noindent The opposite category of $\mathscr{C}$ is denoted by $\mathscr{C}^{op}$ and it has the same objects as $\mathscr{C}$ but the directions of the arrows are reversed. 


\begin{definition}
An arrow $m : C \rightarrow D$ is monic if for any pair of arrows $f,g:B\rightrightarrows C$ if $m\circ f = m\circ g$ then $f = g$.
\end{definition}

\noindent \textbf{Examples}

\noindent 1. One of the most familiar examples of a category is the category $\mathbf{Set}$ whose objects are sets and whose arrows are functions between sets. 

\noindent 2. If $P$ is a partially ordered set, we can regard it as a category in which the objects are the elements of $P$ and the arrows correspond to the partial order on $P$. Thus if $x\leq y$ in the poset $P$ we have an arrow $x\rightarrow y$ in the poset $P$ seen as a category.

\begin{definition}
Given two categories $\mathscr{C}$ and $\mathscr{D}$, a covariant functor $F:\mathscr{C}\longrightarrow \mathscr{D}$ is given by:

\begin{itemize}
\singlespacing
\item[a)] a map $F:Ob(\mathscr{C})\longrightarrow Ob(\mathscr{D})$
\item[b)] for all $A,B\in Ob(\mathscr{C})$, a map $F:$ $\mathscr{C}(A,B)\longrightarrow$ $\mathscr{D}(FA,FB)$ satisfying the principles of 
\begin{itemize}
\onehalfspacing
\item[i)] \textbf{identity preservation:} $F(1_A)=1_{FA}$
\item[ii)] \textbf{functoriality preservation:} $F(f\circ g)=Ff\circ Fg $
\end{itemize}
\end{itemize}

\doublespacing

\end{definition}

\noindent A contravariant functor $F:\mathscr{C}\longrightarrow \mathscr{D}$ is a covariant functor from $\mathscr{C}^{op}$ to $\mathscr{D}$. Functors can be composed in an obvious manner.

\noindent A contravariant \textbf{Set}-valued functor defined on a category $\mathscr{C}$ is called a presheaf over the base category $\mathscr{C}$. Note that this can also be regarded as a covariant functor defined on $\mathscr{C}^{op}$ with values in \textbf{Set}.

\begin{definition}
If $F$ and $G$ are (covariant) functors between the categories $\mathscr{C}$ and $\mathscr{D}$, then a natural transformation $\eta$ from $F$ to $G$ associates to every object $X$ in $\mathscr{C}$ a morphism $\eta_X : F(X) \rightarrow G(X)$ in $\mathscr{D}$ such that for every morphism $f : X \rightarrow Y$ in $\mathscr{C}$ we have $\eta_Y \circ F(f) = G(f) \circ\eta_X$; this means that the following diagram is commutative:
\[\xymatrix{F(X) \ar[r]^{F(f)} \ar[d]_{\eta_X} & F(Y) \ar[d]^{\eta_Y} \\ G(X)\ar[r]_{G(f)} & G(Y)}\]
\end{definition}

For contravariant functors the arrows $F(f)$ and $G(f)$ in the definition above will be reversed. The two functors $F$ and $G$ are called naturally isomorphic if there exists a natural transformation from $F$ to $G$ such that $\eta_X$ is an isomorphism for every object $X$ in $\mathscr{C}$.

Natural transformations act on functors just like functors act on objects of a category. Hence the collection of functors between two categories together with the natural transformations between them is a category in itself. A notable example is the collection  \textbf{Set}$^{\mathscr{C}^{op}}$ of all presheaves \textbf{Set}-valued presheaves over $\mathscr{C}$.

\begin{definition}
If $\mathscr{C}$ and $\mathscr{D}$ are categories then functors $F:\mathscr{C}\rightarrow \mathscr{D}$ and $G:\mathscr{D}\rightarrow \mathscr{C}$ form an adjunction if for any $A$ in $\mathscr{C}$ and $B$ in $\mathscr{D}$, there an isomorphism
$$\theta_{A,B} : {\mathscr{D}}(FA,B) \rightarrow {\mathscr{C}}(A,GB)$$
which is natural in both $A$ and $B$. In this case we say that $F$ is the left adjoint of $G$, or equivalently that $G$ is the
right adjoint of $F$.
\end{definition}

What is meant by `naturality in $A$' in the definition above is that given any $A'$ in $\mathscr{C}$ and any arrow $h:A'\rightarrow A$ the following diagram commutes:
\[\xymatrix{ {\mathscr{D}}(FA,B)\ar[r]^{\theta_{A,B}} \ar[d]_{(Fh)^*} &  {\mathscr{C}}(A,GB)\ar[d]^{h^*} \\  {\mathscr{D}}(FA',B) \ar[r]_{\theta_{A',B}} &  {\mathscr{C}}(A',GB)}\]
where $h^*$ denotes composition on the right by $h$. Similarly, for `naturality in $B$'.

\begin{definition}
Given two categories $\mathscr{C}$ and $\mathscr{D}$, an equivalence of categories consists of a functor $F : \mathscr{C} \rightarrow \mathscr{D}$, a functor $G : \mathscr{D} \rightarrow \mathscr{C}$, and two natural isomorphisms $\epsilon: FG \rightarrow I_{\mathscr{D}}$ and $\eta : I_{\mathscr{C}}\rightarrow GF$. Here $FG: \mathscr{D}\rightarrow \mathscr{D}$ and $GF: \mathscr{C}\rightarrow \mathscr{C}$ denote the respective compositions of $F$ and $G$, while $I_{\mathscr{C}}: \mathscr{C}\rightarrow \mathscr{C}$ and $I_{\mathscr{D}}: \mathscr{D}\rightarrow \mathscr{D}$ denote the identity functors on $\mathscr{C}$ and $\mathscr{D}$, assigning each object and morphism to itself. If $F$ and $G$ are contravariant functors one speaks of a duality of categories instead.
\end{definition}
\subsection{Limits and colimits}

It is possible to define limits and colimits in a category $\mathscr{C}$ by means of diagrams in $\mathscr{C}$. A diagram of type $J$ in $\mathscr{C}$ is a functor from $J$ to $\mathscr{C}$:
$$ F : J \rightarrow \mathscr{C}$$
The category $J$ is thought of as index category, and the diagram $F$ is thought of as indexing a collection of objects and morphisms in $\mathscr{C}$ patterned on $J$. The actual objects and morphisms in $J$ are largely irrelevant, only the way in which they are interrelated matters.

A diagram is said to be finite whenever $J$ is.

\begin{definition}\label{con}
Let $F : J \rightarrow \mathscr{C}$ be a diagram of type $J$ in a category $\mathscr{C}$. A cone to $F$ is an object $N$ of $\mathscr{C}$ together with a family $\psi_X : N \rightarrow F(X)$ of morphisms indexed by the objects of $J$, such that for every morphism $f : X \rightarrow Y$ in $J$, we have $F(f) \circ \psi_X = \psi_Y$:
\[\xymatrix{ & N \ar[ld]_{\psi_X} \ar[rd]^{\psi_Y} & \\ F(X) \ar[rr]_{F(f)} & & F(Y)}\]

\end{definition}
\begin{definition}\label{lim}

A limit of the diagram $F : J \rightarrow \mathscr{C}$ is a cone $(L, \varphi)$ to $F$ such that for any other cone $(N,\psi)$ to $F$ there exists a unique morphism $!_N : N \rightarrow L$ such that $\varphi_X \circ !_N = \psi_X$ for all $X$ and $Y$ in $J$:
\[\xymatrix{& N \ar@{.>}[d]^{!_N} \ar@/^-1pc/[ldd]_{\psi_X} \ar@/^1pc/[rdd]^{\psi_Y}  & \\ & L \ar[ld]_{\varphi_X} \ar[rd]^{\varphi_Y} & \\ F(X) \ar[rr]_{F(f)} & & F(Y)}\]
\end{definition}

The dual notions of limits and cones are colimits and co-cones. It is straightforward to obtain the definitions of these by inverting all morphisms in Definitions \ref{con} and \ref{lim}.

Important examples of limits include pullbacks, products and terminal objects, while examples of colimits include coproducts and initial objects. So, for instance, a pullback $(P,\varphi)$ is the limit of a diagram $F$ of type 
\[\xymatrix{& \CIRCLE \ar[d] \\ \CIRCLE \ar[r] & \CIRCLE }\]
in a category $\mathscr{C}$. More explicitly, if the image of $F$ in $\mathscr{C}$ is 
\[\xymatrix{& X \ar[d]^f \\ Y \ar[r]_g & Z }\]
we say that the object $P$, together with the family of arrows $\varphi=\{\varphi_X,\, \varphi_Y\}$, giving the limit of this diagram is the pullback of $f$ and $g$, and we call the diagram below
\[\xymatrix{P \ar[r]^{\varphi_X} \ar[d]_{\varphi_Y} & X \ar[d]^f \\ Y \ar[r]_g & Z }\]
a pullback square.

A product is the limit $(P,\varphi)$ of a diagram $F$ of type $\CIRCLE \ \ \CIRCLE$, while a coproduct is the colimit of a diagram of the same type. More explicitly, if the image of $F$ in $\mathscr{C}$ is $A \ \ B$ we say that the object $P$, together with the family of arrow $\varphi=\{\varphi_A,\varphi_B\}$, giving the limit of the diagram is the product of $A$ and $B$. Similarly, the object, and arrows, giving the colimit of the same diagram is called the coproduct of $A$ and $B$.

A terminal object $\mathbf{1}$ in a category $\mathscr{C}$ is the limit of the empty diagram in $C$. Thus the terminal object has the property that for any other object $N$ in $\mathscr{C}$ there is a unique arrow $!_N$ going from $N$ to $\mathbf{1}$. Similarly, an initial object $\mathbf{0}$ in a category $\mathscr{C}$ is the colimit of the empty diagram in $\mathscr{C}$ and there is a unique arrow going from $\mathbf{0}$ to each $N$ in $\mathscr{C}$.

\noindent \textbf{Examples}

\noindent 1. The terminal object in \textbf{Set} is the one-element set.\\
\noindent 2. The terminal object in \textbf{Set}$^{\mathscr{C}^{op}}$ is the functor (presheaf) which assigns the one element set $\{*\}$  to all objects in $\mathscr{C}$ and the identity on $\{*\}$ to every $\mathscr{C}$-arrow.

\subsection{Lattices and Heyting algebras}

A lattice is a partially ordered set in which any two elements have a least upper bound (also called a join) as well as a greatest lower bound (also called a meet). The symbols $\wedge$ and $\vee$ are used to denote these two operations. A lattice can also be defined in categorial language as a poset which, when seen as a category, has all finite limits and colimits. From this point of view, the lattice operations $\wedge$ and $\vee$ correspond to the categorical notions of product and coproduct.

A morphism of lattices is a monotone function which preserves both (binary) meets and (binary) joins.

A lattice is called complete if it has all limits and colimits. A complete lattice has an initial element given by the empty join, which we denote by \textbf{0}, and dually, a terminal element given by the empty meet, which we denote by \textbf{1}.

\begin{definition}
An element $x$ of a lattice $L$ with global minimum \textbf{0} is called an \textbf{atom} if any element $y$ in $L$ which is strictly smaller than $x$ is equal to \textbf{0}. A lattice is \textbf{atomic} if for every non-zero element $y\in L$ there is some atom $x$ such that $x\leq y$. An \textbf{atomistic} lattice is an atomic lattice such that each element is a join of atoms.
\end{definition}

\begin{definition}
An orthocomplemented lattice or ortholattice is a bounded lattice which is equipped with a function $\perp$ that maps each element $x$ to an orthocomplement $x^\perp$ in such a way that the following axioms are satisfied:

\begin{enumerate}
\singlespacing
 \item[(i)] complement law: $x^\perp\vee x =1$ and $x^\perp\wedge x=0$
 \item[(ii)] involution law: $x^{\perp\perp}=x$
 \item[(iii)] order-reversing: if $x\leq y$ then $y^\perp\leq x^\perp$.
\end{enumerate}

\doublespacing

\end{definition}

We call a lattice complemented if it is equipped with a (not necessarily unique) function that satisfies only the first of the three conditions above. We say that a lattice $L$ is distributive if, for all $x, y, z\in L$,
$$x \wedge (y \vee z) = (x \wedge y) \vee (x \wedge z)$$ and hence also 
$$x\vee (y\wedge z) = (x\vee y)\wedge (x\vee z)$$
We say that two elements $x$ and $y$ of $L$ are orthogonal if $x\leq y^\perp$.

\begin{definition}
A Boolean algebra is a complemented distributive lattice (and thereby also orthocomplemented).
\end{definition}

Stone's representation theorem \cite{stone} for Boolean algebras states that every Boolean algebra is isomorphic to the Boolean algebra of all clopen subsets of some (compact totally disconnected Hausdorff) topological space called the Stone space.

Boolean algebras are intimately connected with classical propositional logic. The Lindenbaum algebra given by the set of sentences in propositional calculus modulo tautology has the structure of a Boolean algebra \cite{tarski}. A truth assignment in propositional calculus is then a homomorphism from the Lindenbaum algebra to the two-element Boolean algebra.

\begin{definition}
A Heyting algebra is a lattice $H$ containing a bottom element with the property that for any two elements $x,y \in H$ there exists an exponential $x\Rightarrow y$ characterised by
$$z \leq (x \Rightarrow y) \iff z \wedge x \leq y$$
A Heyting algebra is called \textit{complete} if the underlying lattice is complete.
\end{definition}

The exponential $x\Rightarrow y$ is also referred to as the pseudo-complement of $x$ relative to $y$. When $y$ is equal to the bottom element, $x\Rightarrow \mathbf{0}$ is referred to as the pseudo-complement of $x$. In a Boolean algebra, the pseudo-complement is actually the complement of an element, so every Boolean algebra is also a Heyting algebra. However, while the definition of the exponential implies that the meet of an element and its pseudo-complement is always equal to \textbf{0}, their join may be less than \textbf{1}. This means that there are Heyting algebras which are not Boolean.

It can be shown that the underlying lattice of a Heyting algebra is distributive. Moreover, if the underlying lattice is complete, then for any family $(h_i)_{i\in I}$ of elements of $H$, the following infinite distributivity law holds:
$$\forall h \in H,\ \ h\wedge\bigvee_{i\in I} h_i= \bigvee_{i\in I} (h\wedge h_i)$$

Just like classical propositional logic is modelled by Boolean algebras, the Lindenbaum algebra of intuitionistic logic has the structure of a Heyting algebra. Intuitionistic logic can have more than two truth values, as we will see in the following section.

\subsection{Topoi}

An (elementary) topos is a category which is similar to the category of sets and functions in a sense which will be made precise in this section. The elementary topoi which we will be considering in this thesis are all Grothendieck topoi of \textbf{Set}-valued presheaves over a base category given by a partially ordered set. We will pay special attention to topoi of presheaves throughout the remainder of this section.

We start our discussion by extending the notion of subset to a general category. Thus we define, as in \cite{Gol}, a subobject of an object $A$ to be an object $S$ from the same category such that there exists a monic $m:S\rightarrow A$. In \textbf{Set} the collection of subobjects of a given set $A$ is equivalent to the power set of $A$. The set-theoretic operations of union, intersection and complement give the structure of a Boolean algebra to this collection of subobjects.

In \textbf{Set}$^{\mathscr{C}^{op}}$, a subobject of a presheaf $\underline{A}$ is another presheaf $\underline{S}$ which associates a subset $\underline{S}_C$ of $\underline{A}_C$ to every object of $\mathscr{C}$ in a way which is compatible with the restrictions $\underline{A}(f)|_{_{\underline{S}_C}}$ for all arrows $C\stackrel{f}{\longrightarrow} D$ in $\mathscr{C}^{op}$. In other words, there is a natural transformation between the presheaves $\underline{A}$ and $\underline{S}$ given by the collection of monics:
\[\xymatrix{\underline{S}_C \ar[rr]^{\underline{A}(f)|_{_{\underline{S}_C}}} \ar@{>->}[d] & & \underline{S}_D \ar@{>->}[d] \\ \underline{A}_C \ar[rr]^{\underline{A}(f)} && \underline{A}_D}\]

A point of an object $A$ in a category $\mathscr{C}$ is an arrow (which is necessarily monic) from the terminal object in $\mathscr{C}$ to $A$. For the category $\mathbf{Set}$ the terminal object is the one object set $\{*\}$ and so the categorical notion of point coincides with the familiar notion of an element of a given set. In \textbf{Set}$^{\mathscr{C}^{op}}$ a point of a presheaf is a collection of set-theoretic points which is compatible with the restriction maps. It is also called a global section.

It is a well-known fact that in \textbf{Set}, a subset $S\subseteq A$ is uniquely determined by its characteristic function $\xi_S$. This connection can be generalised to arbitrary categories, as long as they contain what is known as a subobject classifier.

\begin{definition}
In a category $\mathscr{C}$ with finite limits, a subobject classifier is an object $\Omega$ and a monic $true : \mathbf{1} \rightarrow \Omega$ going from the terminal object in $\mathscr{C}$ to $\Omega$ such that for any subobject $m : S \rightarrow A$ there is a unique arrow $\chi_S:A \rightarrow\Omega$ making the following square a pullback:

\[\xymatrix{S \ar[r]^{!_S} \ar@{>->}[d]_{m} & \mathbf{1} \ar[d]^{true} \\ A \ar[r]_{\chi_S} & \Omega}\]
\end{definition}

One can readily check that in $\mathbf{Set}$ the subobject classifier is given by $\Omega=\{0,1\}$ together with the map $true(*)=1$, and $\chi_S$ above is simply the characteristic function corresponding to $S$.

Using the fact that the pullback of a monic is itself monic one can prove that, in any category where such constructions are possible, there is a bijective correspondence between the collection of subobjects of an object $A$ and the collection of arrows between $A$ and $\Omega$. In particular
$$Sub(\mathbf{1})\cong\mathscr{C}(1,\Omega)$$

The collection of subobjects of a given object can be given the structure of a Heyting algebra using a suitable generalisation of the set-theoretic operations of union, intersection and complement (see, for example Ch. 7 in \cite{Gol}). This means that the category theoretic points of $\Omega$ are in bijective correspondence with the subobjects of the terminal object, and as such they form a Heyting algebra. The points of $\Omega$ can be thought of as the truth-values of the intuitionistic logic associated with the topos in which $\Omega$ is defined. 

We now describe the subobject classifier in \textbf{Set}$^{\mathscr{C}^{op}}$. For this we need to introduce the notion of sieve. For a $\mathscr{C}$-object $A$, let $\mathfrak{S}_A$ be the collection of all $\mathscr{C}$-arrows with co-domain $A$. Note that $\mathfrak{S}_A$ is ``closed under right composition", i.e. if an arrow $B\stackrel{g}{\longrightarrow} C$ in $\mathscr{C}^{op}$ can be composed with $f\in\mathfrak{S}_A$ then $f\circ g\in\mathfrak{S}_A$ since the arrow $f\circ g$ has co-domain $A$.

\begin{definition}
A sieve on $A$ is a subset $\mathfrak{s}$ of $\mathfrak{S}_A$ which is itself closed under right composition.
\end{definition}

Note that for every object $A$ there are at least two $A$-sieves $\mathfrak{S}_A$ and $\emptyset$, the empty sieve. 

Now we are ready to define the subobject classifier in \textbf{Set}$^{\mathscr{C}^{op}}$ as the presheaf $\Omega$ which associates the set of sieves on $A$ to every $\mathscr{A}$-object $A$. To each arrow $A\stackrel{f}{\longrightarrow} B$ in $\mathscr{C}$, the subobject classifier associates a function 
\begin{align*}
\Omega(f):\Omega(B)&\longrightarrow\Omega(A)\\
\mathfrak{s}&\longmapsto \mathfrak{s}_A=\{g\in\mathfrak{S}_A~|~ f\circ g\in \mathfrak{s}\}
\end{align*}

The set $\mathfrak{s}_A$ is closed under right composition because if $h$ can be composed with some $g\in\mathfrak{s}_A$ then $f\circ g\circ h\in\mathfrak{s}$ because sieves on $B$ are closed under right composition, and so $g\circ h\in \mathfrak{s}_A$, as required.

The arrow $true:\mathbf{1}\rightarrow\Omega$ is defined as the natural transformation that has components $true_A: \{*\}\rightarrow \Omega(A)$ given by $true_A(*)=\mathfrak{S}_A$.

\noindent \textbf{Examples}

\noindent 1. In the limiting case when the category $\mathscr{C}$ has only one object $A$ and only one arrow, namely the identity on $A$, there are precisely two $A$-sieves: $\mathfrak{S}_A=\{id_A\}$ and $\emptyset$. In this case, the definition of the subobject classifier in \textbf{Set}$^{\mathscr{C}^{op}}$ reduces to the classical one, as expected. 

\noindent 2. When the category $\mathscr{C}$ is a poset, there is an arrow from $A$ to $B$ only when $A\leq B$. We can identify arrows to $B$ with their domains. A sieve on $B$ is therefore the equivalent of a lower set on $B$ in order-theoretic terms. To see how subobjects are classified in this case, consider a subpresheaf $\underline{X}$ of some presheaf $\underline{Y}$ and, for simplicity, assume that the monic between them is given by inclusions at every stage. For every stage $A$ an element of $\underline{Y}_A$ is classified according to whether that element together with its restrictions to all possible $\underline{Y}_B$  belong to the sets $\underline{X}_A$ or $\underline{X}_B$ respectively. This classification yields a sieve on $A$ for every point because if $y$ restricts to an element of $\underline{X}_B$ for $B\leq A$ then it also restricts to an element of $\underline{X}_C$ for any $C\leq B$. The sieve, which is a lower set on $A$, possibly empty, tells us from which stage onwards we can expect to find the restriction of $y\in \underline{Y}_A$ within the subpresheaf $\underline{X}$.


We need one more concept: exponentiation, before we are able to state the formal definition of a topos. In the category of sets the exponential $B^A$ is the set of all functions from $A$ to $B$. To characterise $B^A$ by arrows note that associated with $B^A$ is a special arrow, the evaluation function:
$$ev:B^A\times A\rightarrow B$$
given by the rule $ev(f,x)=f(x)$. The evaluation function is distinguished from the set of all functions of the form $g:C\times A\rightarrow B$ because given any such $g$, there is exactly one function $\hat{g}:C\rightarrow B^A$ making the diagram below commute:
\[\xymatrix@R=2ex{B^A\times A \ar[rd]^{ev}& \\ & B \\ C\times A \ar[ru]_{g} \ar@{.>}[uu]^{\hat{g}\times id_A} & }\]

By abstraction, we say that a category $\mathscr{C}$ has exponentiation if it has all binary products and if for any $\mathscr{C}$-objects $A$ and $B$ there is a $\mathscr{C}$-object $B^A$ and a $\mathscr{C}$-arrow $ev:B^A\times A\rightarrow B$ which enjoys the universal property we have described in the diagram above. If such an arrow exists, the assignment of $\hat{g}$ to $g$ establishes an adjunction
$$\mathscr{C}(C\times A, B) \cong \mathscr{C}(C, B^A)$$

\singlespacing

\begin{definition}

An elementary topos is a category which has
\begin{enumerate}
 \item all limits and colimits taken over finite index categories
 \item exponentials
 \item a subobject classifier
\end{enumerate}

\end{definition}

\doublespacing

\section{C*-Algebras and von Neumann Algebras}

In this section we shall briefly present a selection of standard algebraic results which will allow the reader to understand the statement and implications Gelfand duality. This important duality will then be used in the following section to define the central construction of the Topos Approach - the spectral presheaf. A more detailed presentation of the concepts introduced in this section can be found in \cite{KadRin} and \cite{Dix69}.

\begin{definition}
A Banach algebra $B$ is an associative algebra over the real or complex numbers, whose underlying vector space is a Banach space. The algebra multiplication and the Banach space norm must be related by the following inequality called submultiplicativity:
$$\forall x,y\in B,\ \ \ ||xy||\leq||x||~||y||$$
This ensures that the multiplication operation, seen as a function from $B\times B$ to $B$ is continuous with respect to the norm topology.
\end{definition}

The prototypical example of a commutative Banach algebra is the space $C_0(X)$ of complex-valued continuous functions on a locally compact Hausdorff space, which vanish at infinity. The algebra operations are pointwise addition and multiplication of functions. When the space $X$ is compact, this is equivalent to the algebra $C(X)$ of continuous complex valued functions on $X$.

\begin{definition}
A $C^*$-algebra $A$ is a Banach algebra over the field of complex numbers, together with an involution map $^*:A\rightarrow A$ satisfying the following properties:
\begin{itemize}
\singlespacing
\item $(x^*)^*=x$
\item $(x+y)^*=x^*+y^*$
\item $(\lambda x)^*=\bar{\lambda}x^*$
\item $(xy)^*=y^*x^*$
\item $||x^*x||=||x^*x||=||x||\ ||x^*||$
\end{itemize}
for all elements $x,y$ of the algebra. 
\end{definition}

A $C^*$-algebra is called \textit{unital} if it contains a unit element $u$ such that 
$$ua=au=a, \ \ \forall a\in A$$

\doublespacing
The prototypical example of a non-commutative $C^*$-algebra is the complex algebra of continuous bounded linear operators on a Hilbert space $H$ with involution given by the adjoint, which we shall denote by ${\mathcal B}(H)$.  A subset of ${\mathcal B}(H)$ is a $C^*$-algebra if it is a topologically closed set in the norm topology of operators, and closed under involution, that is under the operation of taking adjoints of operators. Note that if the Hilbert space $H$ is finite-dimensional, then $\mathcal{B}(H)\simeq M_n(\mathbb{C})$.

The Banach algebra $C(X)$ is also an example of a $C^*$-algebra since the operation of pointwise complex conjugation satisfies the involution properties above.

\begin{definition}
A $^*$-homomorphism is a homomorphism between two $C^*$-algebras which is compatible with the involution maps of the two algebras i.e., $\phi(x^*)=\phi(x)^*$ for all $x\in A$.
\end{definition}

\begin{definition}
An element $x$ of a $C^*$-algebra is called self-adjoint if it satisfies $x^*=x$.
\end{definition}

\begin{definition}
If $A$ is a unital Banach algebra we define the spectrum of an element $x\in A$ to be the set of scalars $\lambda$ such that $x-\lambda 1_A$ is not invertible. We denote it by $\sigma(x)$.
\end{definition}


Note that the spectrum of a self-adjoint element of a $C^*$-algebra consists only of real numbers. As we have mentioned in the introduction, self-adjoint operators of $C^*$-algebras have a special significance from the physical point of view: they represent the observables of the system described by the respective algebra. The spectrum of a self-adjoint operator contains the possible values which the observable represented by that operator can take upon measurement.

\begin{definition}
A $C^*$-algebra $B$ is called a \textbf{$W^*$-}algebra if as a Banach space it is the dual of a Banach space (i.e. if there exists a Banach space $B_*$ such that $(B_*)^*=B$, where $(B_*)^*$ is the dual of the Banach space $B_*$). $B_*$ is called the \textbf{predual of $B$}.

A \textbf{von Neumann algebra} is a $C^*$-subalgebra of some $\mathcal B(H)$ that is closed in the weak (and equivalently, the strong) operator topology on $\mathcal B(H)$.
\end{definition}

Von Neumann algebras arise as (faithful) representations of $W^*$-algebras by linear operators on a Hilbert space, and $W^*$-algebras can be seen as abstract, representation-free von Neumann algebras.

An important example of a von Neumann algebra is the algebra $\mathcal{B}(H)$ of all bounded operators on a Hilbert space $H$. Moreover it is known that every von Neumann algebra is isomorphic to a subalgebra of some $\mathcal{B}(H)$ for a suitable Hilbert space $H$.

Recall that a self-adjoint element of a $C^*$-algebra or von Neumann algebra is called a projection if it is idempotent. We can define a partial order $\leq$ on the projections of a von Neumann algebra such that $p\leq q$ iff $pq=qp=p$. Note that the projections of a von Neumann algebra form a complete orthomodular lattice.

\begin{definition}
An \textbf{atom} of a von Neumann algebra $N$ is an atom of the projection lattice $\mathcal{P}(N)$. A von Neumann algebra is called \textbf{atomic} if every non-zero projection is greater than an atom of that algebra.
\end{definition}

\begin{definition}
A linear functional $f$ on a $C^*$-algebra is called positive if $f(x^*x)\geq 0$ for all non-zero elements $x$ of the algebra.
\end{definition}

A bounded linear functional $f$ is positive if and only if $||f||=f(1)$ (see Theorem 4.3.2 in \cite{KadRin}). If $f$ and $g$ are two linear functionals on an algebra, we say that $f$ is \textit{majorized} by $g$ if $f-g$ is positive.

Note that every positive linear functional defines an inner product on $A$:
$$\forall x,y\in A\ \ \left\langle x|y\right\rangle :=f(y^*x)$$
The Cauchy-Schwartz inequality (see Proposition 4.3.1 in \cite{KadRin}) then gives $|f(y^*x)|^2\leq f(x^*x)f(y^*y)$ for all $x,y\in A$.

\begin{definition}
A state of an algebra is a positive linear functional of norm $1$.
\end{definition}

\begin{definition}
The positive linear functionals of norm $\leq 1$ form a convex subset of the dual space of a $C^*$-algebra $A$. The non-zero extremal points of this subset are called pure states.
\end{definition}

In the case of a commutative $C^*$-algebra the pure states are precisely those positive linear functionals of norm $1$ which are multiplicative. For non-commutative algebras this does not hold in general.

\subsection{Gelfand duality and Gelfand spectra}




\begin{definition} \label{Def_GelfandSpectrum}
The \textbf{Gelfand spectrum} of a commutative $C^*$-algebra $A$ is the set of pure states of that algebra, equipped with the weak*-topology. We will denote it by $\mathrm{Spec}\,  A$.
\end{definition}

The spectrum of a unital $C^*$-algebra is a compact Hausdorff space with respect to the weak*-topology. The Gelfand spectrum of a commutative $C^*$-algebra $A$ and the spectrum of an element $x\in A$ are related in the following way: 
$$\sigma(x)=\{f(x)\ |\ f\in\mathrm{Spec}\,  A\}$$

\begin{theorem}[Gelfand and Naimark, 1943] \label{GNTh}For every unital and commutative $C^*$-algebra $A$, there exists a *-isomorphism between $A$ and $C(\mathrm{Spec}\,  A)$:
\begin{align*}
{\mathcal G}:A&\longrightarrow C(\mathrm{Spec}\,  A)\\
a&\longmapsto \ \ \overline{a}:\mathrm{Spec}\,  A\rightarrow \ \mathbb{C}\\
&\ \ \ \ \ \ \ \ \ \ \ \ \ f\ \ \mapsto f(a)
\end{align*}
called the \textbf{Gelfand representation of $A$}. The functional $\overline{a}$ is called the \textit{Gelfand transform} of the element $a$. 
\end{theorem}

Let \textbf{UComC}$^*$\textbf{Alg} denote the category of unital commutative $C^*$-algebras, and let \textbf{KHausSp}$^{op}$ denote the opposite of the category of compact Hausdorff spaces. A direct consequence of Theorem \ref{GNTh} is that these are equivalent categories. The equivalence is given by the following pair of functors:
\[\xymatrix{ \mathbf{UComC}^*\mathbf{Alg}\ar@<0.5ex>[r]^\Sigma & \mathbf{KHausSp}^{op} \ar@<0.5ex>[l]^{C(-)}}\]
\vskip -40pt
\begin{align*}
A\ \ \ \ \ \ \ \ &\longrightarrow \ \ \ \ \ \ \ \mathrm{Spec}\,  A\\
C(X)\ \ \ \ \ &\longleftarrow \ \ \ \ \ \ \ X\\
\end{align*}

Both functors are contravariant and act in the following way: if $\varphi: A\rightarrow B$ is a unital *-homomorphism between two $C^*$-algebras, $A$ and $B$, then $\Sigma(\varphi)$ is given by
\begin{align*}
\Sigma(\varphi):\mathrm{Spec}\,  B&\longrightarrow \mathrm{Spec}\,  A\\
f&\longmapsto f\circ\varphi.
\end{align*}
Similarly, if $\phi:X\rightarrow Y$ is a morphism between two compact Hausdorff spaces, $X$ and $Y$, then $C(\phi)$ is given by
\begin{align*}
C(\phi):C(Y)&\longrightarrow C(X)\\
g&\longmapsto g\circ\phi.
\end{align*}



\section{The spectral presheaf}\label{spectral}

In classical mechanics, Gelfand duality establishes a correspondence between the state space $S$ of a physical system, which at the most basic level can be seen as a set, and the commutative  $C^*$-algebra of observables, which is given by the collection of real-valued functions 
\begin{equation}\label{f_A}
f_A:S\rightarrow \Bbb{R}
\end{equation}
on state space, under the operations of pointwise addition and multiplication. Each function $f_A$ corresponds to an observable of the classical system. In this interpretation of classical mechanics, propositions about the system correspond to measurable subsets of $S$ of the form $f_A^{-1}(D)$. Such a subset represents the proposition asserting that the observable $A$ takes values in the subset $D$ of the real numbers.

Within the topos approach \cite{b1,b2,b3,b4, d1,d2,d3,d4}, the spectral presheaf associated with the non-commutative von Neumann algebra $N$ of observables of a quantum system is the analogue of the state space of a classical system. 

We denote the set of all commutative subalgebras (or contexts) of $N$, minus the trivial one $V_0=\mathbb{C}1$, by $\mathcal{V}(N)$. This is a partially ordered set under inclusion, and as such it forms a category. We can use Gelfand duality to associate a Hausdorff space to each context, and this collection of Hausdorff spaces forms the spectral presheaf.

\begin{definition}
The \textbf{spectral presheaf} $\underline{\Sigma}^N$ of a given von Neumann algebra $N$ is the following contravariant functor from the category $\mathcal{V}(N)$ to the category of sets:

\begin{enumerate}
\item[a)] on objects: for all $V\in \mathcal{V}(N)$, let $\underline{\Sigma}^N_V$ be the Gelfand spectrum of $V$, i.e. the set of multiplicative positive linear functionals of norm one, or equivalently, the set of pure states on $V$, equipped with the weak-* topology
\item[b)] on arrows: for all inclusions $i_{VV'}:V'\hookrightarrow V$, let $\underline{\Sigma}^N(i_{VV'}): \underline{\Sigma}^N_V \rightarrow \underline{\Sigma}^N_{V'}$ be the function that sends each pure state $f$ to its restriction $f|_{V'}$ to the
smaller algebra. This function is well-known to be continuous and surjective.
\end{enumerate}
\end{definition}

When no confusion arises we will simply write $\underline{\Sigma}$ instead of $\underline{\Sigma}^N$.

For the sake of obtaining a better understanding of this abstract construction, we consider a concrete von Neumann algebra $N:=M_n(\Bbb{C})$, which corresponds to a quantum system described by the Hilbert space $\mathcal{H}:=\Bbb{C}^n$, and we show what the spectral presheaf looks like in this particular case. 

Let $(\psi_1,\psi_2,\ldots,\psi_n)$ be an orthonormal basis of $\mathcal{H}$, and let $(P_1,P_2,\ldots,P_n)$ be the $n$ projections onto the one-dimensional subspaces spanned by each of the basis vectors. These projections are pairwise orthogonal, i.e., $P_iP_j=\delta_{ij}P_i$.

Using von Neumann's double commutant construction (see below), we can define the abelian subalgebra of $N$ generated by the $n$ projections considered above, $V=\{P_1,P_2,\ldots,P_n\}''$.

\begin{definition}
If $\mathcal{B}(H)$ is the algebra of bounded operators on some Hilbert space $H$ and $\mathcal{F}\subseteq \mathcal{B}(H)$, the \textbf{commutant} of $\mathcal{F}$ is the subset of $\mathcal{B}(H)$ consisting of all elements that commute with every element of $\mathcal{F}$, that is
$$\mathcal{F}'=\{T\in \mathcal{B}(H)~|~TS=ST,\ \forall S\in\mathcal{F}\}$$
The \textbf{double commutant} of $\mathcal{F}$ is just $(\mathcal{F}')'$ and is usually denoted by $\mathcal{F}''$.
\end{definition}
 

With respect to the basis $(\psi_1,\psi_2,\ldots,\psi_n)$, the abelian algebra $V$ is equal to the algebra $D_n$ of diagonal $n\times n$ matrices with complex entries on the diagonal. In this way, every orthonormal basis of $\Bbb{C}^n$ determines a maximal abelian subalgebra of $N$, and all maximal abelian subalgebras of $N$ are isomorphic to each other.  

The abelian von Neumann subalgebras of $V\simeq D_n$ above are the diagonal matrices with $k$ independent entries. A diagonal matrix with $k$ independent entries corresponds to a collection $\mathcal{S}$ of $k$ pairwise disjoint sets $S_1, S_2,\ldots, S_k$ whose union is the set $\{1,2,\ldots, n\}$. Each set in $\mathcal{S}$ contains the indices of those positions on the diagonal which have the same value. For example, the maximal algebra $D_n$ corresponds to the collection
$$\left\{\ \{1\}, \{2\}, \ldots, \{n\}\ \right\}$$

The minimal algebra $\Bbb{C} I_n$ corresponds to the collection
$$\left\{\ \{1,2,\ldots,n\} \right\}$$
Recall that this algebra is excluded by convention from the poset $\mathcal{V}(N)$.

Given a collection $\mathcal{S}$ containing $k$ sets $S_1,\ldots S_k$, for every $i$ from $1$ to $k$ we can construct the projections $$Q_i:=\sum_{j\in S_i} P_j$$ The algebra corresponding to $\mathcal{S}$ is
$$V_{\mathcal{S}}=\{Q_1,Q_2,\ldots, Q_k\}''$$
Of course, there are other projections whose double commutant would give the same algebra, but the advantage of using this particular set of projections is that they form an `orthonormal basis' for $V_{\mathcal S}$ in the sense that the projections are pairwise orthogonal and sum up to the identity $1$.

The Gelfand spectrum of $V_{\mathcal{S}}$ hence contains $k$ elements and is equipped with the discrete topology. The spectral elements are linear multiplicative functionals, and they are determined by their behaviour on the `basis' $Q_1,... Q_k$, that is
$$f_i(Q_j)=\delta_{ij},\ \  i\in \{1,2,\ldots,k\}$$
If $A\in V_{\mathcal{S}}$ is an arbitrary operator, $A=\sum_{i=1}^k a_iQ_i$ for some (unique) complex coefficients $a_i$. For each $i$, linearity gives us $f_i(A)=a_i$. The $a_i$s are the eigenvalues of $A$ and if $A$ is self-adjoint, they are real numbers.

If $V_{\mathcal{S}}\subseteq V_{\tilde{\mathcal{S}}}$, then $\tilde{\mathcal{S}}$ is a refinement of $\mathcal{S}$ in the following sense: every set $S_i\in\mathcal{S}$ can be written as the union of $t_i\geq 1$ pairwise disjoint sets $\tilde{S}^i_1, \tilde{S}^i_2\ldots,\tilde{S}^i_{t_i}$ such that 
$$\tilde{\mathcal{S}}=\left\{\ \tilde{S}^1_i, \tilde{S}^1_2,\ldots \tilde{S}^1_{t_1},\ \ \ \tilde{S}^2_i, \tilde{S}^2_2,\ldots \tilde{S}^2_{t_2},\ \ \ \ldots, \ \ \ \tilde{S}^k_i, \tilde{S}^k_2,\ldots \tilde{S}^k_{t_1} \right\}$$
where all the sets in $\tilde{\mathcal{S}}$ are pairwise disjoint. For each $i\in\{1,\ldots, k\}$ and $j\in\{1,\ldots,t_i\}$ define the projection
$$Q^i_j:=\sum_{x\in \tilde{S}^i_j}P_x$$
The collection of all projections of this form generates the algebra $V_{\tilde{\mathcal{S}}}$ via the double commutant construction. The pure states of $V_{\tilde{\mathcal{S}}}$ are therefore given as
$$\tilde{f}^i_j(Q^u_v)=\delta_{iu}\delta_{jv}$$
The restriction map going from $\Sigma_{V_{\tilde{\mathcal{S}}}}$ to $\Sigma_{V_{\mathcal{S}}}$ takes the spectral element $\tilde{f}^i_j$ to the element $f_i$.

\subsection{Daseinisation}\label{sec:CPersp}

We can use the spectral theorem to interpret projections of a von Neumann algebra $N$ as propositions of the form ``$A\varepsilon\Delta$", that is propositions of the form ``the physical quantity $A$, which is
represented by the self-adjoint operator $A\in N$, has a value in the set Borel set $\Delta$". More precisely, each projection corresponds to an equivalence class of such propositions.

If we take a commutative subalgebra $V$ of $N$, every state $f\in \mathrm{Spec}\, V$ of $V$ gives us a way to assign truth values to propositions which involve quantities
represented by self-adjoint operators from $V$. Any such $f$ can take only one of the two values $0,1$ when applied to a projection $P\in V$, since
$$f(P)=f(P^2)=f(P)f(P)$$
So we can assign to those propositions which correspond to the projection $P$ the value true if $f(P)=1$, and false if $f(P)=0$. We know from the Kochen-Specker theorem that, under certain natural conditions, it
would not be possible to make such truth-value assignments for the projections of the non-commutative algebra $N$ (unless $N$ was a type $I_2$-algebra).

The projections in a commutative von Neumann algebra $V$ correspond bijectively to clopen subsets of $\mathrm{Spec}\, V$:

\begin{proposition}\label{alpha}
If $\mathcal{P}(V)$ is the lattice of all projections in $V$ and $Cl(\mathrm{Spec}\, V)$ is the lattice of clopen subsets of $\mathrm{Spec}\, V$, then the map 
\begin{align*}
\alpha_V: \mathcal{P}(V)&\rightarrow Cl(\mathrm{Spec}\, V) \\
P&\mapsto S_{P}:=\{f\in \mathrm{Spec}\, V~|~f(P)=1\}
\end{align*}
is a lattice isomorphism.
\end{proposition}

\begin{proof}
It is easy to check that $S_{P}$ is indeed a clopen subset of $\mathrm{Spec}\, V$. We have
$$S=\overline{P}^{-1}\left(\,\left(\frac{1}{2},\infty\right)\,\right)$$
and so $S$ is open. Similarly 
$$\mathrm{Spec}\, V \backslash S= \overline{P}^{-1}\left(\,\left(-\infty,\frac{1}{2}\right)\,\right)$$
and so $\mathrm{Spec}\, V\backslash S$ is open, hence $S$ is closed.

Since the Gelfand representation is a *-isomorphism for unital commutative algebras, $\alpha_V$ must be a bijective map. \qed
\end{proof}

Note that this result is also a consequence of the Stone representation theorem, as shown by de Groote (see Theorem 3.2 in \cite{Groote}): the projections in $V$ form a complete Boolean algebra $\mathcal{P}(V)$, which is isomorphic to the complete Boolean algebra of clopen subsets of the Stone space of $\mathcal{P}(V)$. For an abelian von Neumann algebra $V$, the Stone space of $\mathcal{P}(V)$ is homeomorphic to the Gelfand spectrum of $V$.

We have seen that for a `classical part' of a quantum system described by a commutative algebra $V$ there is a correspondence between propositions, or rather the projections which
represent them, and clopen subsets of the Gelfand spectrum of the algebra $V$. Next we will see that for quantum systems (as a whole) there is an analogous correspondence between
propositions and clopen sub-objects of the spectral presheaf.

The collection of all contexts of a non-commutative von Neumann algebra $N$ can be understood as the collection of all classical perspectives on a quantum system. As we have
mentioned before, the idea behind the spectral presheaf is to characterise a quantum system by taking into account all the classical perspectives at the same time. In order to do
this, we need to adapt every proposition about the whole quantum system to each possible classical context. That is, given a proposition ``$A\varepsilon\Delta$" and its representing projection $P$, we want to choose for every context $V$ the strongest proposition implied by ``$A\varepsilon\Delta$" which can be made from the perspective of that context. For
projections, this is equivalent to taking the smallest projection in any context $V$ that is larger or equal to $P$:
$$\delta^o(P)_V:=\bigwedge\{ Q\in\mathcal{P}(V)~|~Q\geq P\}$$
If $P\in\mathcal{P}(V)$, the above approximation will simply be equal to $P$. We will call the original proposition ``$A\varepsilon\Delta$" the \textit{global proposition}, while a
proposition ``$B\varepsilon\Gamma$" corresponding to the projection $\delta^o(P)_V$ will be called a \textit{local proposition}.

From the family of projections $(\delta^o(P)_V)_{V\in\mathcal{V}(N)}$ we can obtain a family of clopen subsets of the Gelfand spectra $(\mathrm{Spec}\, V)_{V\in\mathcal{V}(N)}$ by
choosing for every $V$ the subset
$$S_{\delta^o(P)_V}=\alpha_V(\delta^o(P)_V)\subseteq\mathrm{Spec}\, V$$

These subsets form a subobject under the restriction mappings of the spectral presheaf $\underline{\Sigma}$ and so we can give the following definition, as in \cite{Doe11b}.

\begin{definition}
The \textbf{daseinisation of a projection $P$} is the subobject (or equivalently, the subpresheaf) $\underline{\delta(P)}$ of the spectral presheaf $\underline{\Sigma}$ given by the collection of clopen subsets $(S_{\delta^o(P)_V})_{V\in\mathcal{V}(N)}$, together with the restriction mappings between them.
\end{definition}

We denote the collection of all sub-objects of the spectral presheaf by $Sub(\underline{\Sigma})$. It can be seen as the analogue of the lattice of subsets of the state space of a classical system.

\begin{definition}
A subobject $\underline{S}$ of the spectral presheaf $\underline{\Sigma}$ such that for each $V\in\mathcal{V}(N)$ the component $\underline{S}_V$ is a clopen subset of
$\underline{\Sigma}_V$ is called a \textbf{clopen subobject}.
\end{definition}

Daseinisation thus gives us a map between the projections of a von Neumann algebra and the collection of clopen subobjects of the spectral presheaf
$$\underline{\delta}:\mathcal{P}(N)\rightarrow Sub_{cl}(\underline{\Sigma})$$
since all sub-objects obtained from the daseinisation of projections are clopen. The collection $Sub_{cl}(\underline{\Sigma})$ can be seen as the analogue of the collection of measurable subsets of the state space of a classical system.

The daseinisation $\underline{\delta(P)}$ of a projection $P$ representing the proposition ``$A\varepsilon\Delta$" can be seen as quantum the analogue of the measurable subset $f_A^{-1}(D)$ of the state space of a classical system, where $f_A$ represents a classical observable, as in Equation \ref{f_A}. We interpret $\underline{\delta(P)}$ as the representative of the global proposition ``$A\varepsilon\Delta$".

The daseinisation $\underline{\delta(P_\psi)}$ of a projection $P_\psi$ which projects onto the ray spanned by the vector $\psi$ is called the \textit{pseudo-state} associated to
$\psi$. It can be regarded as the analogue of a point in the state space of a classical system. It is important to note however that the pseudo-states are not global elements of $\underline{\Sigma}$. In
fact, global elements of a presheaf are the category-theoretical analogues of points. Isham and Butterfield have observed \cite{b1} that the Kochen-Specker theorem is equivalent to the fact that the spectral presheaf has no global elements. A global element $\gamma$ of $\underline\Sigma$ would pick one $\gamma_V\in\underline\Sigma_V$ for each context $V$ such that, whenever $V'\subset V$, one would have $\gamma_V|_{V'}=\gamma_{V'}$. Each $\gamma_V$ assigns values to all physical quantities described by self-adjoint operators $A$ in $V$
by evaluation, i.e., by simply forming $\gamma_V(A)$. If $A$ is contained in different commutative subalgebras $V,\widetilde V$, then it is also contained in $V':=V\cap\widetilde
V$, and $\gamma_V(A)=\gamma_{V'}(A)=\gamma_{\widetilde V}(A)$, so the defining condition of the global element $\gamma$ guarantees that $A$ is assigned the same value in every
context. Since every self-adjoint operator is contained in some commutative subalgebra $V$, a global element $\gamma$ of $\underline\Sigma$ would provide a consistent assignment
of values to all self-adjoint operators. But the Kochen-Specker theorem precisely shows that this is impossible, hence such global elements $\gamma$ cannot exist.

Pseudo-states on the other hand are minimal sub-objects in a suitable sense: they come from rank-1 projections, the smallest non-trivial projections, and daseinisation is order-preserving,
so pseudo-states are the smallest non-trivial sub-objects of $\underline\Sigma$ that can be obtained from daseinisation. Hence, pseudo-states are `as close to points as possible'.

Note also that, in analogy with the classical case where the powerset of the state space formed a Boolean algebra, the collection of all sub-objects of the spectral presheaf can be turned (see \cite{Doe11a}) into a complete Heyting algebra by defining suitable meet and join operations. 

\begin{definition}
 If $\underline{S_1}$ and $\underline{S_2}$ are two sub-objects of the spectral presheaf, their join is defined by stagewise unions in the following way:
$$(\underline{S_1}\vee\underline{S_2})_V=\underline{S_1}_V\cup\underline{S_2}_V$$
Similarly, their meet is given by stagewise intersections:
$$(\underline{S_1}\wedge\underline{S_2})_V=\underline{S_1}_V\cap\underline{S_2}_V$$
\end{definition}

In \cite{Doe11a, Doe12}, D\"{o}ring has also shown that the clopen sub-objects of the spectral presheaf $\underline{\Sigma}$ form a complete Heyting algebra under stagewise meet and join operations. For more information about the daseinisation map, its physical interpretation, its properties, and in particular its relation to the logic of the Topos Approach and to the lattice of clopen subobjects of the spectral presheaf the interested reader is directed to the clear and concise presentation exposed in \cite{Doe11a}. 

\subsection{States as measures on the spectral presheaf}\label{measures}

In general, a state on a von Neumann algebra is a positive linear functional of norm one on that algebra. Given such a state, we can associate to it a certain measure on the
corresponding spectral presheaf. This construction was explored in detail by D\"oring, who also showed that measures on the spectral presheaf can be defined without reference to
states and moreover that from each abstractly defined measure a unique state can be reconstructed \cite{ms}. We give a brief overview of these ideas below. 

In classical physics states are represented by probability measures on state space, and pure states are represented by Dirac measures. A probability measure assigns a number between $0$ and $1$ to each measurable subset of state space. Within the topos approach the role of the state space is played by the spectral presheaf, and so in analogy with classical mechanics we would like states to be represented by probability measures on (clopen subobjects of) the spectral presheaf. However, since subobjects of the spectral presheaf are not simply sets, but collections of sets, we can not expect the values taken by the measure to be given by single numbers. Instead we would expect to obtain a collection of such numbers, one for each context of the algebra which represents our system. With this in mind we give the following definitions of presheaves of real values and their global sections, which will be essential to our discussion.

\begin{definition}
Given a von Neumann algebra $N$ and its associated poset of abelian subalgebras $\mathcal{V}(N)$, let $\downarrow V:=\{W\in\mathcal{V}(N)~|~W\subseteq V\}$ denote the down-set of a context $V\in\mathcal{V}(N)$.The presheaf $\underline{\mathbb{R}^\succeq}$ is defined
\begin{itemize}
\item on objects: $\underline{\mathbb{R}^\succeq}_V=\{f:\downarrow V\rightarrow \mathbb{R}~|~ f \text{ is order reversing }\}$
\item on arrows: for $i_{V'V}: V' \hookrightarrow V$, $\underline{\mathbb{R}^\succeq}(i_{V'V}):\mathbb{R}\rightarrow\mathbb{R}$ is given by $$\underline{\mathbb{R}^\succeq}(i_{V'V})(f):=f|_{_{\downarrow V'}}$$
\end{itemize}
\end{definition}

Note that this presheaf lives in the same topos as $\underline{\Sigma}^N$. However, we do not explicitly specify this topos, by indicating the base category, when discussing this and similar presheaves of real values. It is usually clear from the context, which base category we are using.

A global section of this presheaf can be regarded as an order-reversing function from the partially ordered set $\mathcal{V}(N)$ to the real numbers equipped with the usual ordering.

The presheaf defined above plays an important role within the topos approach, and is discussed extensively in \cite{d3,disham}. However, when defining measures we will only use a sub-presheaf of this presheaf of real numbers, which we denote by $\underline{[0,1]^\succeq}$. Later on, when we will introduce the notion of entropy we will encounter a closely related presheaf, $\underline{[0,\ln n]^\preceq}$, where $n$ denotes the dimension of the algebra corresponding to our system. In this case global sections will be equivalent to order-preserving functions from $\mathcal{V}(N)$ to the real interval $[0,\ln n]$.

\begin{definition}
Given a von Neumann algebra $N$, a measure on its associated spectral presheaf $\underline{\Sigma}$ is a mapping
\begin{align*}
 \mu: \mathrm{Sub}_{\mathrm{cl}}(\underline{\Sigma}) & \longrightarrow \Gamma \underline{[0,1]^{\succeq}}\\
 \underline{S}=(\underline{S}_V)_{V\in\mathcal{V}(N)} &\longmapsto  \mu(\underline{S}):\mathcal{V}(N)\rightarrow [0,1] \\
 & \ \ \ \ \ \ \ \ \ \ \ \ \ \ \ \ V\ \ \ \mapsto  \mu(\underline{S}_V) 
\end{align*}
which satisfies the following conditions:
\begin{enumerate}
 \item $\mu(\underline{\Sigma})=1_{\mathcal{V}(N)}$
 \item for all $\underline{S}_1$, $\underline{S}_2\in\mathrm{Sub}_{\mathrm{cl}}(\underline{\Sigma})$, it holds that 
  $$\mu(\underline{S}_1\vee\underline{S}_2)+\mu(\underline{S}_1\wedge\underline{S}_2)=\mu(\underline{S}_1)+\mu(\underline{S}_2)$$
\end{enumerate}
where the addition, just like the meet and the join for sub-objects, is defined as a stagewise operation.

\end{definition}

The conditions above also imply that $\mu(\underline{0})=0$, where $\underline{0}$ is the subobject of $\underline{\Sigma}$ which assigns the empty set to each context. Note also that we have slightly abused notation by writing $\mu$ both for the measure and for its contextual components.

In this text we will mostly be concerned with a particular type of von Neumann algebras, the algebras of bounded linear operators on finite dimensional Hilbert
spaces (i.e. matrix algebras).  For these algebras the states can be identified with the density matrices: to each density matrix $\rho\in M_n$, we can associate the functional
\begin{align*}
A &\longmapsto \mathrm{Tr}(\rho A), \ \ \  \forall A\in M_n
\end{align*}
and moreover every positive linear functional of unit norm is of this form in the finite dimensional setting. With this in mind, when talking about matrix algebras we shall refer to the density matrices as states on those algebras.

\begin{definition}
Given a state $\rho$ on the matrix algebra $M_n$, it is straightforward to define its associated measure:
\begin{align*}
 \mu_\rho: \mathrm{Sub}_{\mathrm{cl}}(\underline{\Sigma}^{M_n}) & \longrightarrow \Gamma \underline{[0,1]^{\succeq}}\\
 \underline{S}=(\underline{S}_V)_{V\in\mathcal{V}(M_n)} &\longmapsto  \mu_\rho(\underline{S}):\mathcal{V}(M_n)\rightarrow [0,1] \\
 & \ \ \ \ \ \ \ \ \ \ \ \ \ \ \ \ \ \ \ V\ \ \ \ \mapsto  \mathrm{Tr}(\rho P_{\underline{S}_V})
\end{align*}
where $P_{\underline{S}_V}=\alpha_V^{-1}(\underline{S}_V)$. 
\end{definition}

One can easily check that the function $\mu_\rho(\underline{S})$ is order reversing and that $\mu_\rho$ satisfies the two properties required in the definition of a measure. This
is explicitly done in \cite{ms} for arbitrary von Neumann algebras without type $I_2$ summands.

On the other hand, an abstract measure on the spectral presheaf associated to any given algebra $N$, determines a unique state of $N$, provided $N$ contains no direct summand of type$I_2$. The proof of this rather surprising result uses a generalized version of Gleason's theorem, a review of which can be found in \cite{GGleason}. 

\begin{definition}
A finitely additive probability measure $m$ on the projections of a von Neumann algebra $N$ is a map
$$m:\mathcal{P}(N)\rightarrow [0,1]$$
such that  $m(I)=1$ and if $P$ and $Q$ are orthogonal projections then $$m(P\vee Q)=m(P+Q)=m(P)+m(Q)$$
\end{definition}

\begin{theorem}[Generalized Gleason's Theorem]\label{gleason}
Each finitely additive probability measure on the projections of a von Neumann algebra without type $I_2$ summands can be uniquely extended to a state on that algebra.
\end{theorem}

Using this powerful result we can show that each measure on the spectral presheaf uniquely determines a state on the corresponding algebra by showing that such a measure determines a unique finitely additive probability measure on the projections of the respective algebra. This has been done by D\"{o}ring in \cite{ms}, and we will reproduce his proof in the remainder of this section.

Given a measure $\mu$ on the spectral presheaf $\underline{\Sigma}$ associated to some von Neumann algebra $N$, let $\underline{S}$ be a clopen subobject of $\underline{\Sigma}$. From Proposition \ref{alpha} we know that for each context $V$ there exists an isomorphism $\alpha_V$ between $\mathcal{P}(V)$ and $Cl(\underline{\Sigma}_V)$. If $P=\alpha_V^{-1}(\underline{S}_V)$ we define 
$$m(P)=\mu(\underline{S})(V)=\mu(\underline{S}_V)$$
First note that for every projection $P$, there is some subobject $S$ and context $V$ such that $\alpha_V^{-1}(S_V) = P$: for example, the daseinisation of $P$ has this property, if we set $V = {P, 1-P}''$. But we also have to show that this does not depend on the choice of the subobject $\underline{S}$ and the context $V$, i.e. we must show that if $\underline{\tilde{S}}$ is another subobject of $\underline{\Sigma}$ and $\tilde{V}$ is a context such that 
$\alpha_V^{-1}(\underline{S}_V)=\alpha_{\tilde{V}}^{-1}(\underline{\widetilde{S}}_{\tilde{V}})$
then $\mu(\underline{S}_V)=\mu(\underline{\tilde{S}}_{\tilde{V}})$.
For this we will need two intermediate results.

\begin{lemma}
If $\underline{S}$ is a clopen subobject of $\underline{\Sigma}$ and $V'\subseteq V$ are two contexts such that $P$ is contained in both $V$ and $V'$ and $\alpha_V^{-1}(\underline{S}_V)=\alpha_{V'}^{-1}(\underline{S}_{V'})=P$, then $\mu(\underline{S}_V)=\mu(\underline{S}_{V'})$.
\end{lemma}

\begin{proof}
Since the maximal projection $I$ is contained in every context, it follows that $I-P\in V',V$. Let $\underline{S^c}$ be another clopen subobject such that 
$$\alpha_V^{-1}(\underline{S^c}_V)=\alpha_{V'}^{-1}(\underline{S^c}_{V'})=I-P$$
Such a subobject certainly exists: $\underline{\delta(I-P)}$, for example, satisfies the above property.

Since every $\alpha$ is a lattice isomorphism, we have
\begin{align*}
(\underline{S}\wedge \underline{S^c})_V=0_V=\emptyset \ \  &, \ \ \ \ (\underline{S}\wedge \underline{S^c})_{V'}=0_{V'}=\emptyset \\
(\underline{S}\vee\underline{S^c})_V=\underline{\Sigma}_V \ \ &, \ \ \ \ (\underline{S}\vee\underline{S^c})_{V'}=\underline{\Sigma}_{V'}
\end{align*}

Using the two defining properties of a measure $\mu$ we obtain
\begin{align*}
1&=\mu(\underline{\Sigma})(V)\\
&=\mu(\underline{S}\vee\underline{S^c})(V)\\
&=\mu(\underline{S})(V)+\mu(\underline{S^c})(V)-\mu(\underline{S}\wedge\underline{S^c})(V)
\end{align*}
Since the last term vanishes we obtain that $\mu(\underline{S})(V)+\mu(\underline{S^c})(V)=1$. Similarly, we can also deduce that $\mu(\underline{S})(V')+\mu(\underline{S^c}(V')=1$. But $\mu(S):\mathcal{V}(N)\rightarrow [0,1]$ is an order-reversing function, hence
\begin{align*}
&\mu(\underline{S})(V')\geq \mu(\underline{S})(V)\\
&\mu(\underline{S^c})(V')\geq \mu(\underline{S^c})(V)
\end{align*}
This implies that in fact $\mu(\underline{S})(V')=\mu(\underline{S})(V)$ and $\mu(\underline{S^c})(V')=\mu(\underline{S^c})(V)$, which completes our proof. \qed
\end{proof}

\begin{lemma}
If $\underline{S}$ and $\underline{\tilde{S}}$ are two subobjects which coincide at $V$, i.e. if $\underline{S}_V=\underline{\tilde{S}}_V$, then $\mu(\underline{S})(V)=\mu(\underline{\tilde{S}})(V)$.
\end{lemma}

\begin{proof}
From the second defining property of a measure $\mu$ we obtain that
\begin{align*}
\mu(\underline{S})(V)+\mu(\underline{\tilde{S}})(V)&=\mu(\underline{S}\vee\underline{\tilde{S}})(V)+\mu(\underline{S}\wedge\underline{\tilde{S}})(V)\\
&=\mu((\underline{S}\vee\underline{\tilde{S}})_V)+\mu((\underline{S}\wedge\underline{\tilde{S}})_V)\\
&=\mu(\underline{S}_V\cup\underline{\tilde{S}}_V)+\mu(\underline{S}_V\cap\underline{\tilde{S}}_V)\\
&=\mu(\underline{S}_V)+\mu(\underline{S}_V)\\
&=\mu(\underline{S})(V)+\mu(\underline{S})(V)
\end{align*}
which implies that $\mu(\underline{S})(V)=\mu(\underline{\tilde{S}})(V)$.\qed
\end{proof}

Now assume that $\underline{S}$ and $\underline{\tilde{S}}$ are two clopen subobjects of $\underline{\Sigma}$ and $V$ and $\tilde{V}$ are two contexts such that $\underline{S}_V$ and $\underline{\tilde{S}}_{\tilde{V}}$ correspond to the same projection $P\in V,\tilde{V}$. Then we must have that $P$ also belongs to $V\cap\tilde{V}$. We know that the clopen subobject $\underline{\delta(P)}$ coincides with $\underline{S}$ at $V$ and it also coincides with $\underline{\tilde{S}}$ at $\tilde{V}$. Moreover, $\underline{\delta(P)}_{V\cap\tilde{V}}\subseteq \underline{\Sigma}_{V\cap\tilde{V}}$ and  $\alpha_{V\cap\tilde{V}}^{-1}(\underline{\delta(P)}_{V\cap\tilde{V}})=P$. From the previous two lemmas we obtain that
\begin{align*}
\mu(\underline{S})(V)&=\mu(\underline{\delta(P)})(V)\\
&=\mu(\underline{\delta(P)})(V\cap\tilde{V})\\
&=\mu(\underline{\delta(P)})(\tilde{V})\\
&=\mu(\underline{\tilde{S}})(\tilde{V})
\end{align*}

This shows that the value $m(P)=\mu(\underline{S})(V)$ is well defined. For any $V$, the projection corresponding to $\underline{\Sigma}_V$ is the maximal projection, $I$. So from the first defining property of a measure $\mu$, we must have $$m(I)=\mu(\underline{\Sigma})(V)=1$$ 
Finally, let $P$ and $Q$ be two orthogonal projections and let $V$ be a context that contains both $P$ and $Q$. Let $\underline{S^P}$ and $\underline{S^Q}$ be two subobjects such that $\alpha_V^{-1}(\underline{S^P}_V)=P$ and $\alpha_V^{-1}(\underline{S^Q}_V)=Q$. Then $(\underline{S^P}\vee\underline{S^Q})_V$ corresponds to $P\vee Q$ and we obtain
\begin{align}
m(P\vee Q)&=\mu(\underline{S^P}\vee \underline{S^Q})(V)\\
&=\mu(\underline{S^P})(V)+\mu(\underline{S^Q})(V)-\mu(\underline{S^P}\wedge\underline{S^Q})(V)\\
&=\mu(\underline{S^P})(V)+\mu(\underline{S^Q})(V)\\
&=m(P)+m(Q)
\end{align}

This shows that the map $m:\mathcal{P}\rightarrow [0,1]$ is indeed a finitely additive probability measure, and so from the generalised version of Gleason's theorem we know that $m$ extends to a unique state $\rho_m$ of the algebra $N$.

In particular this implies that when the algebra $N$ is a finite dimensional matrix algebra, of dimension greater than $2$, there is a bijective correspondence between density matrices and measures on the corresponding spectral presheaf.

\end{chapter}


\begin{chapter}{Quantum Logic vs. the Logic of the Topos Approach}\label{LA}

Quantum logic was initially introduced in the '30s by Birkhoff and von Neumann \cite{BirkVonN}. At that time, the mathematical apparatus of quantum mechanics, mostly developed by von Neumann, was already in place.  Birkhoff and von Neumann's goal was ``to discover what logical structure one may hope to find in physical
theories which, like quantum mechanics, do not conform to classical logic". Their paper discusses the mathematical structures which may be used to describe propositions about the values of physical quantities, and how states assign truth values to these propositions.

Orthomodular lattices  play a prominent role in quantum logic \cite{DCG02,Var07}. Often, existence of atoms is required for conceptual reasons: it corresponds to the assumption that every pure state represents a `physical property'. The prototypical example is $\PH$, the lattice of projections on a Hilbert space $\cH$, which is an atomic complete orthomodular lattice.

In Chapter \ref{TA} we have seen that the topos-based form of logic for quantum systems uses certain presheaf constructions over the poset $\mathcal{V}(\cN)$ of abelian von Neumann subalgebras of a von Neumann algebra $\cN$. For clarity and simplicity, we focus here on the case $\cN=\BH$, the algebra of bounded linear operators on the Hilbert space $\cH$. Its lattice of projections is given by $\PH$. 

An abelian von Neumann subalgebra $V\subset\BH$ has a lattice of projections $\PV$ that is a complete Boolean algebra.  The abelian subalgebras $V\in\mc V(\BH)$ and their corresponding complete Boolean sublattices $\PV\in\mc B(\PH)$ are called \emph{contexts}. Conceptually, they can be thought of as `classical perspectives' on the quantum system, as discussed in Section \ref{sec:CPersp}.


We now turn our attention to an interesting question: can the orthomodular lattice $\PH$, which is traditionally used in quantum logic, be reconstructed from the poset of contexts $\mc B(\PH)$ that underlies the constructions in the topos approach? The answer is affirmative, as was shown by Harding and Navara in \cite{navara}. Their main result is that if $L$ and $M$ are OMLs and $\phi:\BL \ra \mc B(M)$ is an isomorphism of posets, then there is an isomorphism $\phi^*:L \ra M$ with $\phi(B) = \phi^*[B]$ for each Boolean subalgebra $B$ of $L$. 

Conceptually, this means that by considering the partially ordered set of contexts, one does not lose information compared to considering the whole orthomodular lattice. This also implies that the new form of presheaf- and topos-based form of logic for quantum systems is (at least) as rich as traditional quantum logic.

One motivation for the work presented in this chapter comes from a question posed in the concluding section of \cite{navara}: let $\mc S(L)$ denote the lattice of all subalgebras of an OML $L$. Give an order-theoretic construction of $\mc S(L)$ from $\BL$. As a partial solution to this problem, we present here a reasonably direct order-theoretic way of reconstructing an atomic OML $L$ from $\BL$. 

We are concerned only with a single atomic OML $L$, not with morphisms between two OMLs, so our contribution can moreover be seen as the `object counterpart' to the result by Harding and Navara. The atoms of $\BL$ have the form $\{0,P,P^{\perp},1\}$, where $P$ and $P^{\perp}$ are elements of $L$. The main task is to show how the order relations between the elements of $L$ arise from the order relations between the elements of $\BL$. Of course, $L$ determines $\mc S(L)$ in a straightforward way.

\section{Grouping and splitting} We first remark that if $L=\PH$, the projection lattice on a Hilbert space $\cH$, then the height of the poset $\VH$ of abelian subalgebras equals the dimension of $\cH$ (if we include the trivial subalgebra $V_0=\bbC\hat 1$ in $\VH$; otherwise, the dimension of $\cH$ equals the height of $\VH$ plus $1$). The dimension of $\cH$ determines the Hilbert space $\cH$ up to isomorphism, and hence determines $\PH$ up to isomorphism. This is a cheap (and rather indirect) way of `reconstructing' $\PH$ from $\VH$.

In this chapter, we will present a more explicit and generally applicable reconstruction.

\begin{remark}
It would be conceivable to have a reconstruction of an atomic OML $L$ from the poset $\mc B(L)$ of its Boolean subalgebras along the following lines:\footnote{We thank the anonymous referee for suggesting this alternative.}
\begin{itemize}
	\item[a)] First, one identifies the atoms of $L$ from $\mc B(L)$. Of course, this is not entirely straightforward, since atoms in $\mc B(L)$ are of the form $\{0,P,P^\perp,1\}$ for \emph{arbitrary} elements of $L$, not just atom - co-atom pairs, but the identification can be made using only the order-theoretic information encoded within $\mc B(L)$.
	\item[b)] Then, by considering which atoms of $L$ jointly lie in which elements of $\mc B(L)$, one identifies the compatibility (i.e., orthogonality) relations between atoms.
	\item[c)] Finally, one attempts to reconstruct $L$ in a bottom-up fashion, starting from the sets of pairwise orthogonal atoms, using the additional information about relations between elements of $L$ and atoms which is encoded within $\mc B(L)$
\end{itemize}

If one uses a bottom-up reconstruction then, even when $L$ is a Boolean algebra, one must pay attention to the fact that infinite atomic Boolean algebras may have the same set of atoms without being isomorphic. A well-known example is given by the following two Boolean algebras: on the one hand, $B_1=P\mathbb{N}$, the power set of the set of natural numbers, with intersections as meets, unions as joins, and complements (with respect to $\mathbb{N}$) as complement operation; on the other hand the Boolean algebra $B_2$, whose elements are the finite and the cofinite subsets of $\mathbb{N}$ (a subset $S$ of $\mathbb{N}$ is cofinite if $S'=\mathbb{N}\backslash S$ is finite), with intersections, unions and complements as operations.

Clearly, the atoms in both $B_1$ and $B_2$ are the singleton subsets of $\mathbb{N}$, and both algebras are atomic and infinite. Yet, $B_1$ contains more elements than $B_2$ (uncountably many, vs. countably many), so the algebras are not isomorphic. This relates to the fact that $B_1$ is complete, while $B_2$ is not complete.

Our reconstruction algorithm, to be presented below, is also based on the poset $\mc B(L)$ of Boolean subalgebras of a given atomic OML $L$. However, in our reconstruction we will follow a top-down route which can be summarized as follows: if $V=\{0,P,P^\perp,1\}$ is an atom of $\mc B(L)$, we use the order structure of $\mc B(L)$ to identify Boolean sup-algebras of $V$ that contain only $P$, or respectively $P^\perp$, and atoms of $L$ as generating elements (such algebras are called minimal spiked Boolean super-algebras below). We then form equivalence classes, corresponding to $P$, or respectively $P^\perp$, and show that the order on generic elements $P,Q$ of $L$ can be deduced from the structure of these equivalence classes (which is determined by the order on $\mc B(L)$). 

Thus, instead of singling out atoms of $L$ and proceeding from them, we use the order-theoretic information encoded within $\mc B(L)$ to identify the elements of $L$ directly -- they are in a two-to-one correspondence with the atoms of $\mc B(L)$. We then show that the order between generic elements of $L$ can also be deduced from the order structure of $\mc B(L)$. 

We now briefly consider the special case of atomic Boolean algebras in order to make clear how, for example, the two infinite atomic Boolean algebras $B_1,B_2$ are distinguished by our reconstruction procedure:\footnote{In general, by the result by Sachs \cite{Sachs}, two Boolean algebras have order-isomorphic posets of Boolean subalgebras if and only if they are isomorphic.} the poset (in fact, lattice) $\mc B(B)$ of Boolean subalgebras of a Boolean algebra $B$ has atoms itself, namely the Boolean subalgebras of the form $\{0,P,P^\perp,1\}$, where $P$ is an arbitrary element of $B$ (not necessarily an atom). Let $B_1,B_2$ be the two atomic Boolean algebras presented above. Since $B_1$ has more elements than $B_2$, the poset $\mc B(B_1)$ has more atoms than the poset $\mc B(B_2)$, so $B_1$ and $B_2$ can be distinguished at the level of their posets of Boolean subalgebras, and our reconstruction makes use of this.

\end{remark}

Let $L$ be an atomic orthomodular lattice with $0$ and $1$, and let $\BL$ be the set of Boolean subalgebras (BSAs) of $L$, partially ordered under inclusion. Two elements $P$ and $Q$ of $L$ are orthogonal if $P\leq Q^{\perp}$, where $P^{\perp}$ denotes the orthocomplement of $P$. Orthogonality also implies that the meet of $P$ and $Q$ is equal to $0$. It is clear that every element $P$ of $L$ is contained in at least one Boolean subalgebra $V$ of $L$ (for example in $V_P=\{0,P,P^{\perp},1\}$).

Let $\mc F=\{P_1,P_2,\ldots,P_n\ldots\}$ be a (possibly infinite) family of pairwise orthogonal elements in $L$ with join $1$. Then $\mc F$ generates an atomistic BSA $V\subseteq L$. The elements in $\mc F$ are the atoms of the $V$ and since each element of $V$ is a join of elements in $\mc F$, $V$ is an atomistic BSA. 

We say that a BSA $V$ generated by a family $\mc F$ as above has \textbf{dimension} $n=\#\mc F$, the cardinality of $\mc F$. In general, not every BSA $V$ of an atomic orthomodular lattice $L$ is generated by a family $\mc F$ of pairwise orthogonal elements,\footnote{An example is $\PH$, the projection lattice on an infinite-dimensional Hilbert space, which has complete Boolean sublattices that have no atoms at all, e.g. the projection lattice of the abelian von Neumann algebra generated by the position operator. There also are Boolean sublattices of $\PH$ that have some atoms, but are not generated by them.} but each element of $L$ is contained in some BSA, as we have remarked in the previous paragraph.

From now on, we will only consider those BSAs in $\BL$ which are generated by families of pairwise orthogonal elements. This allows us to describe inclusion relations within $\BL$ in terms of grouping and splitting actions. We will write $\mc F_V$ for the family of join $1$, pairwise orthogonal elements generating a BSA $V$.

\begin{definition}
If $\mc F$ and $\mc G$ are two families of pairwise orthogonal elements with join $1$, we say that $\mc G$ is obtained by \textbf{grouping} the elements in $\mc F$ if any $Q\in\mc G$ can be written as a join of elements in $\mc F$. Let $S_Q$ denote the set of elements in $\mc F$ that have join $Q$. The fact that the elements in $\mathcal{G}$ are pairwise orthogonal implies that the sets $S_Q,\;Q\in\mc G,$ are pairwise disjoint. If $\mc G$ is obtained by grouping the elements in $\mc F$, we say that $\mc F$ is obtained by \textbf{splitting} the elements in $\mc G$.
\end{definition}

The BSAs contained in a BSA $V$ are obtained from grouping the elements in $\mc F_V$ while the algebras which contain $V$, if they exist, are obtained from $V$ by splitting the elements in $\mc F_V$. 

A $2$-dimensional BSA $V\subseteq L$ is generated by two complementary elements. We can find out from the order relations within $\BL$ when one (or both) of these elements are atoms. This result will be useful later in our reconstruction of the lattice $L$. 

\begin{lemma}\label{3.2}
Given an atomistic ortholattice $L$ and a $2$-dimensional BSA $V$ of $L$, we have three possible scenarios:
\begin{itemize}
\item[i)] if $V$ is maximal in $\BL$ then its generating elements are complementary atoms. 

\item[ii)] if $V$ is included in a $3$-dimensional BSA $W$ which is maximal in $\BL$ then $V$ is generated by an atom of $L$ together with its complement which is a join of two atoms in $L$. Moreover, $W$ contains precisely two other $2$-dimensional BSAs, apart from $V$ itself.

\item[iii)]if $V$ is neither maximal, nor included in a maximal BSA, then $V$ contains an atom of $L$ if and only if all $4$-dimensional BSAs $W\subseteq L$ which contain $V$ also contain precisely three $3$-dimensional BSAs, $V_1$, $V_2$ and $V_3$, such that $V\subset V_i$, $i\in\{1,2,3\}$. 

\end{itemize}
\end{lemma}

\begin{proof}
For the first two statements, it is sufficient to observe that a BSA is maximal in $\BL$ if neither of its generating elements can be split. Since $L$ is atomistic, this implies that the generating elements of a maximal BSA must be atoms of $L$. This proves the first statement. 

For the second statement note that, $W$ being maximal, must be generated by three pairwise orthogonal atoms, call them $P$, $Q$ and $R$ which add up to the identity. The only BSAs included in $W$ are those generated either by $\{P,Q\vee R\}$ or $\{Q,P\vee R\}$ or $\{R,P\vee Q\}$, so $V$ must also be generated by one of these three families. 

For the third statement, let $\mc F_V=\{P, P^{\perp}\}$ denote the generating family of the $2$-dimensional BSA $V$. If $P$ is an atom of $L$ then any $4$-dimensional algebra $W$ which contains $V$ is obtained by splitting $P^{\perp}$ into three elements, since $P$ is an atom and cannot be split. Hence, $W$ has generating family $\mc F_W=\{P,Q_2,Q_3,Q_4\}$. There are precisely three sub-BSAs of $W$ which contain $V$. These are given by
\begin{align*}
			\mc F_{V_1} &= \{P,Q_2\vee Q_3,Q_4\},\\
			\mc F_{V_2} &= \{P,Q_2\vee Q_4,Q_3\},\\
			\mc F_{V_3} &= \{P,Q_3\vee Q_4,Q_2\}.
\end{align*}
Note that there are three other ways of grouping the elements in $W$ to obtain a $3$-dimensional BSA. The resulting $3$-dimensional BSAs $V_i$, $i=4,5,6$ do not contain $V$, since it is not possible to obtain the element $P_1$ by grouping the elements generating these other algebras.

On the other hand, consider a $2$-dimensional BSA given by $\mc F_{\tilde{V}}=\{Q, Q^{\perp}\}$ generated by two orthogonal elements which are not atoms. Since $Q$ is not an atom, it is possible to write it as a join of two orthogonal non-zero elements (in $L$), that is $Q=Q_1\join Q_2$. Similarly, it is possible to express $Q^{\perp}$ as the join of some orthogonal $Q_3$ and $Q_4$. The BSA $\mc F_{\tilde{W}}=\{Q_1,Q_2,Q_3,Q_4\}$ is a $4$-dimensional algebra which includes $\tilde{V}$, but only two of its sub-BSAs also contain $\tilde{V}$, namely
\begin{align*}
			\mc F_{\tilde{V}_1}=\{Q_1\vee Q_2,Q_3,Q_4\},\\
			\mc F_{\tilde{V}_2}=\{Q_1,Q_2,Q_3\vee Q_4\}.
\end{align*}
\qed

\end{proof}

\section{Spiked BSAs} 

Note that a family of pairwise orthogonal atoms of $L$ with join $1$ generates a \textbf{mBSA} (maximal Boolean subalgebra) of $L$.

\begin{definition}
A \textbf{sub-mBSA of $L$} is a BSA of $L$ generated by a family $\mc F$ of pairwise orthogonal elements with join $1$ with the property that only one element in $\mc F$ is the join of two atoms in $L$, while all others are atoms in $L$. 
\end{definition}

\begin{definition}\label{Spiked}
An algebra is \textbf{spiked} if it is either a mBSA, or is generated by a family $\mc F$ of pairwise orthogonal elements with join $1$ which contains precisely one non-atom of the lattice $L$ (we call this the \textbf{leading} element), while all other elements of $\mc F$ are atoms of $L$. If an algebra is spiked, we will say that its family $\mc F$ of generating elements is also spiked.
\end{definition}

\begin{remark}
The concept of a spiked Boolean subalgebra of an orthomodular lattice is related to the concept of a principal dual subalgebra of a Boolean algebra, introduced by Sachs \cite{Sachs}.\footnote{We would like to thank the anonymous referee for pointing out this connection.} Indeed, since all the atoms of a Boolean algebra $B$ must be pairwise orthogonal, the collection of atoms in any spiked family of elements generates a principal ideal of $B$, while the leading element corresponds to the dual of this ideal \cite{GivHal09}. Their union is a principal dual subalgebra of $B$ according to the definition by Sachs \cite{Sachs}. For the general case when $L$ is an orthomodular lattice, every spiked (non-maximal) BSA will also be a principal dual subalgebra of every Boolean subalgebra $B$ of $L$ that contains the given spiked BSA. However, there is no canonical choice for the Boolean algebra $B$ of which a spiked BSA of $L$ is a principal dual subalgebra, so we have used a non-relational term for denoting those BSAs of $L$ which enjoy the property described in Definition \ref{Spiked}.
\end{remark}

\begin{definition}
Given a BSA $V$ which is not a mBSA, we say that a BSA $W$ is a \textbf{successor} of $V$ if $V\subsetneq W$ and there is no BSA $W'$ such that $V\subsetneq W'\subsetneq W$. We call a successor of a successor of a BSA $V$, if it exists, a double successor of $V$.
\end{definition}

Note that in terms of generating elements, if $W$ is a successor of $V$ then the family of elements generating $W$ is obtained from the family of elements generating $V$ by splitting precisely one element into two pairwise orthogonal elements. 

Since $L$ is atomic and orthomodular, such a splitting is possible whenever $V$ is not a mBSA (i.e. when its generating family contains at least one non-atom of $L$). This is because any non-atomic element $P$ of an atomic  lattice must be larger than some atom $Q$ and the orthomodularity condition then allows us to write $P$ as the join of $Q$ and $Q^{\perp}\wedge P$, which are easily seen to be pairwise orthogonal. Moreover, note that any element of $L$ can be written as a join of pairwise orthogonal atoms of $L$.

\begin{proposition}
Let $V\in\BL$ be a BSA generated by the (possibly infinite) family of elements $\mc F_V=\{P_1,P_2,\ldots,P_k,\ldots\}$. If we assume that $V$ is neither a mBSA, nor a sub-mBSA, then $V$ is spiked if and only if all double successors of $V$ contain precisely three successors of $V$. 
\end{proposition}

\begin{proof}
Completely analogous to the proof of the third statement of Lemma \ref{3.2}.\qed
\end{proof}

This result is important because it shows that the order structure of $\BL$ allows us to decide whether a given BSA $V\in\BL$ is spiked or not. Since every spiked BSA has a distinguished leading element, one can guess that we want to somehow link the elements of the lattice $L$ to the spiked BSAs of $\BL$ using the information encoded within the order structure of $\BL$, which is what we will do in the following section.

Let $V$ be a $2$-dimensional BSA generated by $\mc F_V=\{P,P^{\perp}\}$, and let $\mc S_V$ be the set of spiked BSAs which contain $V$. The generating family $\mc F$ of an element $\tilde V$ of $\mc S_V$ is obtained either by completely splitting $P$ into pairwise orthogonal atoms and splitting $P^{\perp}$ into a spiked family of elements, or by completely splitting $P^{\perp}$ into pairwise orthogonal atoms and splitting $P$ into a spiked family of elements.

The set $\mc S_V$ of spiked BSAs which contain $V$ is partially ordered under inclusion.
\begin{itemize}
	\item [(a)] If $V$ is not spiked, the generating family $\mc F$ of a \emph{minimal} element in $\mc S_V$ with respect to this partial order is obtained by either taking $P$ as the leading element and splitting $P^{\perp}$ into atoms, or by taking $P^{\perp}$ as leading projection and splitting $P^{\perp}$ into atoms. Let $\mc M_V$ denote the set of minimal elements in $\mc S_V$. We call $\mc M_V$ the set of minimal spiked sup-BSAs of $V$ in $\BL$. 
	\item [(b)] If $V$ is spiked, the minimal element of $\mc S_V$ which contains $V$ will of course be $V$ itself. Hence for a spiked $2$-dimensional BSA $V$ we establish by convention the set $\mc M_V$ to be the set of all mBSAs which contain $V$, as these algebras correspond to keeping the atom fixed and completely splitting the co-atom, together with $V$ itself which corresponds to keeping the co-atom fixed.
		\item[(c)] if $V$ is spiked and submaximal, we again define $\mc M_V$ to be the set of all mBSAs which contain $V$ together with $V$ itself.   
	\item [(d)] If $V$ is spiked and maximal then it is generated by a pair of orthocomplementary atoms. These two atoms are  not comparable to any other elements in the lattice $L$ except for the top and bottom elements. The set $\mc M_V$ contains only one element, namely $V$ itself.
\end{itemize}

\section{Reconstructing $L$ from $\BL$} 
As we have already remarked, every $2$-dimensional BSA $V$ with $\mc F_V=\{P,P^{\perp}\}$ is generated by two complementary elements. Hence there is an obvious two-to-one mapping from $L$ to the $2$-dimensional elements of $\BL$, which of course are the atoms of the poset $\BL$. Therefore, in order to generate all the elements of the atomic orthomodular lattice $L$ from the poset $\BL$, we need to assign two elements (corresponding to the two elements $P,P^{\perp}$) to each $2$-dimensional BSA $V$ with $\mc F_V=\{P,P^{\perp}\}$. 

The minimal spiked sup-BSAs of a given $2$-dimensional BSA $V$ make good candidates for this assignment. On the one hand, they can be characterised using only information derived from the poset structure of $\BL$, on the other hand, a minimal spiked sup-BSA of $V$ can be identified with its leading element, which is one of the two generating elements of $V$. Yet, this would give us a many-to-one mapping in general, since there are many (e.g. in $\BH$ continuously many) minimal spiked sup-BSAs of $V$ with the same leading element, corresponding to the many possible ways of splitting its complement. Therefore, it will make sense to define two equivalence classes of algebras within $\mc M_V$ consisting of those algebras whose generating families of elements have the same leading element. 

In the non-degenerate cases (a-c) above, our task is to identify these two equivalence classes using the information encoded within the order structure of $\BL$. By partitioning the sets $\mc M_V$ into two equivalence classes, we are in effect identifying all pairs of complementary elements of the lattice $L$. Later we will see how the order relations between non-complementary elements can be replicated using the corresponding equivalence classes.

Of course, in the degenerate case (d) when $V$ is also maximal, we already know that $V$ is generated by two orthocomplementary atoms, and since these are not comparable to any other elements of $L$, the set $\mc M_V$ does not need any further analysis. 

For a spiked $2$-dimensional BSA $V$ with $\mc F_V=\{P,P^{\perp}\}$, where $P$ is an atom, it is easy to establish what the two equivalence classes should be. One of them, call it $\mathcal{R}_V$, ought to contain the mBSAs which contain $V$ -- this corresponds to keeping the atom $P$ as the `leading' element and completely splitting its complement $P^{\perp}$ into atoms (this is a slight abuse of terminology, since there is no leading element in a mBSA). The other equivalence class, call it $\mathcal{S}_V$, ought to contain only $V$ itself -- this corresponds to keeping the co-atom $P^{\perp}$ as the leading element.

Similarly, for a $2$-dimensional BSA $W$ generated by an element $P$ that is the join of two atoms, together with its complement $P^{\perp}$, we define one equivalence class to contain all the sub-mBSAs in $\mathcal{M}_W$ and the other one to contain all the $3$-dimensional BSAs in $\mathcal{M}_W$.

For non-spiked $2$-dimensional BSA whose (minimal) generating elements are joins of $3$ or more atoms, the two equivalence classes can be determined by considering the inclusion relations between elements belonging to different sets of minimal spiked sup-BSAs, as we will show now.
\begin{lemma}
If $V$ is a non-spiked $2$-dimensional BSA whose generating elements are joins of $3$ or more atoms, and if $A,B\in \mc M_V$, then $A$ and $B$ have the same leading element if and only if there exists some non-spiked $2$-dimensional $W\neq V$ and $C,D\in \mathcal{M}_W$ such that $A\subseteq C$ and $B\subseteq D$.
\end{lemma}

\begin{proof} Assume that $\mc F_W=\{Q,Q^{\perp}\}$ and $\mc F_V=\{P,P^{\perp}\}$ and that $A,B\in\mc M_V$ and $C,D\in\mc M_W$ such that $A\subseteq C$ and $B\subseteq D$. If $P_A, P_B, P_C$ and $P_D$ are the respective leading elements  of $A,B,C$ and $D$ (it makes sense to speak about leading elements, since $\mc M_V$ and $\mathcal{M}_W$ do not contain any mBSAs, as neither $V$ nor $W$ are spiked BSAs), the inclusion relations imply that $P_C\leq P_A$ and $P_D\leq P_B$.

Note at this point that the leading element of a minimal spiked sup-BSA of $V$ must be equal to either $P$ or $P^{\perp}$ (hence $P_A,P_B\in\{P,P^{\perp}\}$), while the leading element of a minimal spiked sup-BSA of $W$ must be equal to either $Q$ or $Q^{\perp}$ (hence $P_C,P_D\in\{Q,Q^{\perp}\}$). Assume towards a contradiction that $P_A\neq P_B$. Then $P_A$ and $P_B$ must be complementary elements. But if $P^{\perp}_A=P_B$, then $P_C$ and $P_D$ must also be complementary elements, otherwise the inclusion relations would imply that $P_A\geq P_C$ and $P^{\perp}_A\geq P_D=P_C$, which is imposible. This means that $P^{\perp}_C=P_D$. However, this leads to a contradiction, since 
\[
			P^{\perp}_C=P_D\leq P_B=P^{\perp}_A\Longrightarrow P_C\geq P_A,
\]
but $P_C=P_A$ is not possible since $W\neq V$.  Hence $P_A$ must be equal to $P_B$.

On the other hand, if $A$ and $B$ have the same leading element, then their generating families are of the form $\mc F_A=\{P, R_1,R_2,\ldots\}$ and $\mc F_B=\{P,S_1,S_2,\ldots\}$, and it is possible to write $P$ as the join of two orthogonal elements $Q$ and $Z$, where $Q$ is an atom, and $Z$ is not an atom. The BSAs $C$ and $D$ given by
\[
			\mc F_C:=\{Q,Z,S_1,S_2,\ldots\} \text{ and } \mc F_D:=\{Q,Z,R_1,R_2,\ldots\}
\]
are sup-BSAs of $A$ and $B$, respectively, and they belong to the set of minimal spiked sup-BSAs of the non-spiked $2$-dimensional BSA $W$ with $\mc F_W=\{Q,Q^{\perp}\}$. \qed
\end{proof}

Once the equivalence classes on the sets of minimal spiked sup-BSAs have been established, it is possible to define an order $\preceq$ on them which replicates the order within the lattice of elements.

\begin{definition}\label{Def1}
If $[X]$ and $[Y]$ are two equivalence classes corresponding to non-spiked BSAs, we say that $[X]\preceq [Y]$ if there exists $A\in[X]$ and $B\in [Y]$ such that $A\supseteq B$.
\end{definition}

If $A\supseteq B$ as above and $\mc F_A=\{P,R_1,R_2,\ldots\}$ with all the $R_i$ atoms while $\mc F_B=\{Q,S_1,S_2,\ldots\}$ with all the $S_i$ atoms, then the generating elements of $B$ are obtained by grouping the generating elements in $A$. This implies that the leading element of $B$ (which is the only non-atom) must be equal to a join of generating elements of $A$. This join must include the leading element of $A$ among its terms, as this is the only possible way of grouping the generating elements of $A$ into a spiked family. Hence there is some index set $I$ such that 
\[
			Q=P\vee\bigvee_{i\in I} S_i
\]
and hence $Q\geq P$.

Since all elements of $[X]$ have the same leading element and similarly, all elements of $[Y]$ have the same leading element, the order relation introduced in Definition \ref{Def1} is well defined. 

Moreover, given any two elements $Q$ and $P$ of an atomic orthomodular lattice which are neither atoms nor co-atoms, and which satisfy the order relation $Q\geq P$ within $L$, one has $Q=P\vee (P^{\perp}\wedge Q)$, and we know that $P^{\perp}\wedge Q$ can be expressed as a join of pairwise orthogonal atoms. Hence the equivalence classes corresponding to elments of $L$ which are neither atoms nor co-atoms, will always be related by $\preceq$ whenever their corresponding leading elements are related within the lattice $L$.

We have to use a different approach for defining the order relations which involve the equivalence classes corresponding to atoms and co-atoms of $L$. 

Recall first that for a spiked $2$-dimensional BSA $V$, the set $\mathcal{R}_V$ denotes the mBSAs which contain $V$, and this is the equivalence class which corresponds to the atom of $V$, while $\mathcal{S}_V$ denotes the one member equivalence class (containing only $V$ itself) which corresponds to the co-atom of $V$.

\begin{definition}\label{Def2}
If $V$ and $W$ are $2$-dimensional BSAs such that $V$ is spiked and $W$ is not spiked, and $[X]\subset\mathcal{M}_W$, then $\mathcal{R}_V\prec [X]$ if there exists $A\in\mathcal{M}_W-[X]$ such that $V\subseteq A$. If this is the case, then also $\mathcal{S}_V\succ \mathcal{M}_W-[X]$. 
\end{definition}

Note that if $V$ with $\mc F_V=\{P,P^{\perp}\}$ is a spiked BSA with $P$ an atom, and $W$ with $\mc F_W=\{Q,Q^{\perp}\}$ is a non-spiked BSA. And if moreover $P<Q$, then $P^{\perp}>Q^{\perp}$ and  there is some way of decomposing $Q$ into a join over a set of atoms which includes $P$. Hence $V$ will be contained in some minimal spiked sup-BSA of $W$ which has $Q^{\perp}$ as its leading element. So $\preceq$ is again well-defined and it captures all the relations between atoms (or co-atoms) and other elements of $L$.

The only relations from $L$ we have not yet captured are those between the atoms and co-atoms themselves. We do this with the following definition.

\begin{definition}\label{Def3}
If both $V$ and $W$ are spiked, their generating elements will contain either equal atoms (if $V=W$) or pairwise orthogonal atoms (if $V$ and $W$ are both contained in some maximal BSA of $L$) or incomparable atoms. Between equivalence classes we then either have $\mathcal{R}_V=\mathcal{R}_W$ and $\mathcal{S}_V=\mathcal{S}_W$, or we define $\mathcal{R}_V\preceq \mathcal{S}_W$ and $\mathcal{R}_W\preceq \mathcal{S}_V$, or they are incomparable. 
\end{definition}

\begin{theorem}			\label{Thm_MainResult}
Let $\mc B_2(L)$ denote the set of $2$-dimensional BSAs of an atomistic ortholattice $L$ which are not mBSAs. Let $\mc M_2(L)$ denote the set of $2$-dimensional mBSAs of $L$. The set $$\mc C(L):=\{\mc M_V/_\sim \}_{V\in \mc B_2(L)}\cup \{A_W^1, A_W^2\}_{W\in \mc M_2(L)} \cup \{0,1\}$$ together with the order $\preceq$ defined in $\ref{Def1}$, $\ref{Def2}$ and $\ref{Def3}$ above, and the additional conventions that $0$ and $1$ stand for the top and the bottom elements of $\mc C(L)$, while the $A_W^1$s and $A_V^2$s are pairs of orthocomplementary atoms which are only comparable with $0$ and $1$, is isomorphic to $L$.
\end{theorem}

\begin{proof}

The lattice isomorphism can easily be constructed using the results presented so far. It sends the top and bottom elements of $L$ to the top and bottom elements of $\mc C(L)$. The orthocomplementary atoms of $L$ are identified with the elements of the pairs of the form $\{A_W^1, A_W^2\}$. And for all the other elements $P\in L$, if $\{P,R_1,R_2,\ldots\}$ is a spiked family of elements, we have the assignment
\[
			P\mapsto [\{P,R_1,R_2,\ldots,\}]\in\mathcal{M}_{\{P,P^{\perp}\}}
\]

\qed
\end{proof}

\end{chapter}
\begin{chapter}{Entropy within the Topos Approach}\label{CE}

It has been argued \cite{1,2,3,5} that Quantum Mechanics can be understood in a more natural way as a theory about the possibilities and impossibilities of information transfer and processing as opposed to a theory about the mechanics of nonclassical waves or particles. Understanding the representation and manipulation of information can help us shed light on the fundamental structure of both classical and quantum theories and it can lead to fresh insights about the essential differences between these two.

Shannon entropy \cite{Sha48} is used in classical mechanics as a measure of the unpredictability of a physical system. Its analogue in quantum mechanics is the von Neumann entropy \cite{vN55}. Several generalizations of these entropies have already been considered \cite{stephanie,barnum} and it is interesting to ask how much information about a quantum state can be encoded using Shannon and von Neumann entropies. This led to the task of finding a formulation of the notion of entropy within the framework of the topos approach. It turns out that this can be done in a natural way, via the contextual entropy construction, which unifies the concepts of Shannon and von Neumann entropy. It should also be noted that throughout this section we only consider quantum theory on finite-dimensional Hilbert spaces.


In classical physics entropy is a real-valued function defined on the set of probability distributions. In the quantum case, entropy is a real-valued function on the set of density matrices. In this chapter we shall define contextual entropy within the topos approach. There the role of the real numbers is played by a slightly more complicated presheaf based on real numbers, whose `points' are given by its global sections, while the set of states is given by a set of measures on a certain non-commutative space (the spectral presheaf). In analogy with the classical case, our entropy will be a map from this set of measures to the set of global sections of our real-number object.

In particular, in section \ref{ctxt} we show how a measure on the spectral presheaf (i.e. a state) gives a canonical probability distribution in each classical context, and how it is therefore possible to associate a Shannon entropy to each classical `perspective' on a state. Contextual entropy is defined in terms of this collection of Shannon entropies, which are shown to form a global section of a certain presheaf of real values (which can differ from context to context). We also show how one can retrieve the von Neumann entropy of a state from such a global section. This confirms our expectation that entropy within the topos approach should `look' like Shannon entropy from each classical perspective, but one can also retrieve the quantum mechanical von Neumann entropy by taking into account all perspectives at the same time. 

In fact, one can do even more than this, and we show that contextual entropy encodes enough information to explicitly reconstruct the quantum state from which it originated. This argument relies on a powerful result known as the Schur-Horn Lemma. In Sections \ref{r1} and \ref{arbitrary} we show how pure quantum states and general quantum states respectively can be reconstructed from the contextual entropy map.


In Section \ref{other} we show that it is possible to adapt other classical entropies within the formalism of the topos approach. In particular, we show how Renyi entropies can be defined within the topos formalism, and moreover we will see that contextual Renyi entropies also encode sufficient information to allow for state reconstruction.

\section{Entropy in Classical and Quantum Mechanics}\label{bckg}

\subsection{Majorization order}

In classical physics states can be interpreted as probability measures on the classical state space. In the discrete case, these measures are simply probability distributions over a set of $n$ outcomes. Given two such distributions we would like to tell which one is more `uniform'. For this purpose one can compare the two distributions in the so called majorization order (for more details, see \cite{BenZyc06}. If $\overrightarrow{x}$ is a vector with $n$ components representing a probability distribution, let $x^{\downarrow}$ denote the vector with the same components as $\overrightarrow{x}$ but arranged in decreasing order.

\begin{definition}
Given two $n$-dimensional probability distributions, we say that $\overrightarrow{x}$ is majorized by $\overrightarrow{y}$, and write $\overrightarrow{x}\prec\overrightarrow{y}$, if and only if 
$$\sum_{i=1}^k x_i^{\downarrow}\leq \sum_{i=1}^k y_i^{\downarrow},\ \ \forall~k=1,\ldots,n$$
\end{definition}

Roughly speaking, this would mean that $\overrightarrow{x}$ is a more  `uniform' probability distribution. The smallest probability distribution with respect to the majorization order is given by the totally mixed distribution $\overrightarrow{x_*}=(1/n,\ldots,1/n)$, while the largest probability distribution with respect to this order is given by any distribution $\overrightarrow{y}$ such that $y^{\downarrow}=(1,0,\ldots,0)$.

Because the passage of time tends to make things more uniform, many processes in physics occur in the direction of the majorization arrow.

\begin{definition}
A bistochastic matrix is a matrix $B$ with positive entries such that the entries of each row and of each column add up to one.
\end{definition}

Note that a bistochastic matrix preserves positivity, and as the sum of the entries of each column add up to one, it also preserves the $l_1$ norm, when acting on positive vectors. The fact that its rows also add up to one implies that it leaves the totally mixed distribution , $\overrightarrow{x_*}$, invariant. Hence a bistochastic matrix acting on the set of probability distributions will cause some kind of contraction of the probability simplex towards its centre. More precisely, we have the following lemma which we state without proof:

\begin{lemma}[Hardy, Littlewood and Polya \cite{HLP}]
 $\overrightarrow{x}\prec\overrightarrow{y}$ if and only if there exists a bistochastic matrix $B$ such that $\overrightarrow{x}=B\overrightarrow{y}$.
\end{lemma}

\subsection{Shannon entropy}

Entropy was initially developed within the framework of thermodynamics, where it was introduced in order to explain the loss of energy within thermodynamic systems. Later, Claude Shannon attempted to mathematically quantify the statistical nature of lost information in phone-line signals. To do this, Shannon developed the very general concept of information entropy, a fundamental cornerstone of information theory \cite{Sha48}.

\begin{definition}
 If $\overrightarrow{x}$ is an $n$-dimensional probability distribution, we define its Shannon entropy to be
$$\mathrm{Sh}(\overrightarrow{x})=\mathrm{Sh}(x_1,\ldots,x_n)=-k\sum_{i=1}^n x_i\ln x_i$$
where $k$ is a positive real number that we usually set equal to $1$, and with the convention that $0\ln 0=0$.
\end{definition}

The Shannon entropy has many useful properties which follow more or less directly from its definition, and we state some of them below, following the standard presentation made in \cite{BenZyc06}.

\begin{itemize}
 \item \textbf{Positivity:} Clearly, $\mathrm{Sh}(\overrightarrow{x})\geq 0$ for all discrete probability distributions.
 \item \textbf{Continuity:} Shannon entropy is a continuous function of the distribution, where the topology of the probability simplex is the natural one inherited from $\mathbb{R}^n$.
 \item \textbf{Expansibility:} $\mathrm{Sh}(x_1,\ldots,x_n)=\mathrm{Sh}(x_1,\ldots,x_n,0)$.
 \item \textbf{Concavity:} $\mathrm{Sh}(p\overrightarrow{x}+(1-p)\overrightarrow{y})\geq p\mathrm{Sh}(\overrightarrow{x})+(1-p)\mathrm{Sh}(\overrightarrow{y})$ for any $p\in[0,1]$.
 \item \textbf{Additivity:} If we have a joint probability distribution of two independent random variables described by probability distributions $\overrightarrow{x}$ and $\overrightarrow{y}$, so that the joint probabilities are products of the individual probabilities, then 
 $$\mathrm{Sh}(x_1y_1,\ldots,x_1y_m,\ \ldots,\ x_ny_1,\ldots,x_ny_m)=\mathrm{Sh}(\overrightarrow{x})+\mathrm{Sh}(\overrightarrow{y})$$
 \item \textbf{Subadditivity:} If we have a joint probability distribution $$\overrightarrow{z}=(z_{11},\ldots,z_{1m},\ \ldots,\ z_{n1},\ldots,z_{nm})$$ of two random variables given by the probability distributions $$\overrightarrow{x}=\left(\sum_{i=1}^m z_{1i},\ldots,\sum_{i=1}^m z_{ni}\right) \text{ and } \overrightarrow{y}=\left(\sum_{j=1}^n z_{j1},\ldots,\sum_{j=1}^n z_{jm}\right)$$ then
  $$\mathrm{Sh}(\overrightarrow{z})\leq \mathrm{Sh}(\overrightarrow{x})+\mathrm{Sh}(\overrightarrow{y})$$
with equality if and only if the two random variables are independent.
 \item \textbf{Monotonicity:} If $\overrightarrow{z}$, $\overrightarrow{x}$ and $\overrightarrow{y}$ are defined as for the subadditivity property, then the Shannon entropy of the joint probability distribution is larger than the Shannon entropy of each of its parts.
 \item \textbf{Recursion property:} If we coarse grain our probability distribution in the sense that we do not distinguish between all the outcomes then we are dealing with a new probability distribution with components
 $$p_1=\sum_{i=1}^{k_1} x_i, \ \ p_2=\sum_{i=k_1+1}^{k_2} x_i, \ \ldots, \ p_r=\sum_{i=k_{r-1}+1}^{k_r} p_i$$
for some $0<k_1<k_2<\ldots<k_r=n$. One can easily show that 
$$\mathrm{Sh}(x_1,\ldots,x_n)=\mathrm{Sh}(p_1,\ldots,p_r)+ \sum_{i=1}^r p_i \mathrm{Sh}\left(\frac{x_1}{p_i},\ldots,\frac{x_{k_1}}{p_i}\right)$$
\item \textbf{Schur concavity:} Shannon entropy is majorization reversing (or Schur concave). This implies that the maximum is attained for the totally mixed probability distribution $\overrightarrow{x_*}=(1/n,\ldots,1/n)$, when $\mathrm{Sh}(\overrightarrow{x_*})=\ln n$, while the minimum, $0$, is attained for any of the probability distributions of random variables with one certain outcome.
\end{itemize}  

One can interpret Shannon entropy as a measure of the uncertainty about the outcome of an experiment that is known to occur according to a given probability distribution, or as the expected length of communication needed to specify the outcome that actually occurs. When using the later interpretation we usually set the constant $k$ to be $1/\ln 2$, which simply means that we use logarithms to the base $2$ instead of natural logarithms. With this choice the entropy is said to be measured in units of bits: if we have a $n=d^a$-dimensional distribution then the maximum value of the Shannon entropy is $\mathrm{log}_2n=a$ bits, which is the length of the string of binary digits one can use to label the outcomes. 

To make this interpretation more precise, consider a source that produces outcomes of an infinite sequence of independent and identically distributed random variables. We want to represent each possible outcome by a code word (i.e. a string of binary numbers) such that given any sequence of code words, it can be read in an unambiguous way.  The expected length of string needed to code one outcome is defined as 
$$L=\sum_i p_il_i$$
where  $p_i$ is the probability that the $i^{th}$ possible outcome will occur and $l_i$ is the length (in bits) of the code word used to represent that outcome. Given the probability distribution, we would like to find optimal codes which minimize the expected length. While we will not describe their construction, we state the following important theorem concerning such codes:

\begin{theorem}[Shannon's noiseless coding theorem]
Given a source distribution $\overrightarrow{x}$, let $L_*$ denote the expected length of a code word used in an optimal code. Then
$$\mathrm{Sh}(\overrightarrow{x})\leq L_*\leq\mathrm{Sh}(\overrightarrow{x})+1$$
\end{theorem}

For the proof of this theorem and how to find an optimal code we refer the reader to the literature \cite{CovTho91}.

\subsection{Von Neumann entropy}

\noindent The quantum mechanical analogue of Shannon entropy, which is actually an older concept than Shannon entropy, is called the von Neumann entropy \cite{vN55}.  In finite dimensional quantum mechanics states are usually represented by density matrices. Any given $n$-dimensional density matrix $\rho$ can be decomposed as 
$$\rho=\sum_{i=1}^n\lambda_i|e_i\left>\right<e_i|$$
where $\lambda_i$ are the eigenvalues of $\rho$ and $|e_i\left>\right.$ are the corresponding eigenvectors. This decomposition is unique if $\rho$ has distinct eigenvalues. Since the eigenvalues of a density matrix are non-negative and they sum to $1$, we can give the following definition:

\begin{definition}
The von Neumann entropy of the state $\rho$ is defined as the Shannon entropy of the spectrum of $\rho$:
$$\mathrm{VN}(\rho)=-\mathrm{Tr}(\rho\ln\rho)=-\sum_{i=1}^n \lambda_i\ln \lambda_i$$
\end{definition}
Note that unitarily equivalent states have the same von Neumann entropy.

We shall end this section by stating and proving some of the most important properties of the von Neumann entropy. The proofs will rely on several well-known results which extend inequalities that hold for functions defined on $\mathbb{R}$ to functions of operators. We summarize these results below and direct the reader to \cite{BenZyc06} for more details. 

\begin{definition}
Hermitian operators admit a partial order: $B\geq A$ if and only if $B-A$ is a positive operator. An operator monotone function is a function $f$ defined on Hermitian operators such that $f(A)\leq f(B)$ whenever $A\leq B$.
\end{definition}

\begin{theorem}[L\"{o}wner]
A function $f(t)$ on an open interval is operator monotone if and only if it can be extended analytically to the upper half plane and transforms the upper half plane into itself.

In particular, $f(t)=-t\ln t$ is operator monotone.
\end{theorem}

\begin{definition}
 An operator concave function is a function $f$ such that
$$f(pA+(1-p)B)\leq pf(A)+(1-p)f(B), \ \ \forall p\in[0,1]$$
\end{definition}

A continuous function mapping $[0,\infty)$ to itself is operator concave if and only if $f$ is operator monotone. An operator convex function $f$ is a function such that $-f$ is operator concave.

\begin{lemma}[Klein's inequality]
If $f$ is an operator convex function and $A$ and $B$ are Hermitian operators then
$$\mathrm{Tr}[f(A)-f(B)]\geq \mathrm{Tr}[(A-B)f'(B)]$$
with equality if and only if $A=B$.

In particular, if we restrict ourselves to $f(t)=t\ln t$ we have
$$\mathrm{Tr}(A\ln A - A\ln B)\geq \mathrm{Tr}(A-B)$$
with equality if and only if $A=B$.
\end{lemma}

\noindent Now we are ready to have a look at some of the properties of von Neumann entropy.

\begin{itemize}
\item \textbf{Positivity:} It is clear that von Neumann entropy is positive. Moreover $\mathrm{VN}(\rho)$ vanishes if and only if $\rho=|\psi\left>\right<\psi|$, for some unit vector $|\psi\left>\right.$ (i.e. if and only if $\rho$ is a pure state).
\item \textbf{Continuity:} Von Neumann entropy is a continuous function of the eigenvalues of $\rho$ seen as a vector in $\Bbb{R}^n$ with its standard topology..
\item \textbf{Concavity:} Is a direct consequence of Klein's inequality: let $\rho=p\sigma+(1-p)\omega$, $0\leq p\leq 1$. We can use the particular case of the inequality, with $B=\rho$ and $A=\sigma$ and then $A=\omega$. This gives us two sets of inequalities:
\begin{align*}
\mathrm{Tr}(\sigma\ln\rho)&\leq \mathrm{Tr}(\sigma\ln\sigma+\rho-\sigma)\\
\mathrm{Tr}(\omega\ln\rho)&\leq\mathrm{Tr}(\omega\ln\omega+\rho-\omega)
\end{align*}
Multiplying the first inequality by $p$, the second one by $(1-p)$ and adding, we obtain
$$\mathrm{Tr}(\rho\ln\rho)=p\mathrm{Tr}(\sigma\ln\rho)+(1-p)\mathrm{Tr}(\omega\ln\rho) \leq p\mathrm{Tr}(\sigma\ln\sigma)+(1-p)\mathrm{Tr}(\omega\ln\omega)+\rho-\rho$$
Reversing the sign gives us $\mathrm{VN}(\rho)\geq p\mathrm{VN}(\sigma)+(1-p)\mathrm{VN}(\omega)$.
\item \textbf{Subadditivity:} This is a property concerning composite systems. Let $\rho$ be a state defined on a Hilbert space which is isomorphic to a tensor product of two Hilbert spaces. If $\rho_1$ and $\rho_2$ are the reduced density matrices obtained by taking the partial trace of $\rho$ over the second and the first subsystem respectively, then subadditivity is expressed by the following equation
$$\mathrm{VN}(\rho)\leq\mathrm{VN}(\rho_1)+\mathrm{VN}(\rho_2)$$
Again, this can be proved using Klein's inequality, this time with  $A=\rho$ and $B=\rho_1\otimes\rho_2=(\rho_1\otimes I)(I\otimes\rho_2)$, with the observation that $\mathrm{Tr}(A-B)=0$ since both are density matrices. We then have
\begin{align*}
\mathrm{Tr}(\rho\ln\rho)&\geq\mathrm{Tr}(\rho\ln\rho_1\otimes\rho_2)\\
&=\mathrm{Tr}(\rho(\ln\rho_1\otimes I+\ln I\otimes\rho_2))\\
&=\mathrm{Tr}(\rho_1\ln\rho_1)+\mathrm{Tr}(\rho_2\ln\rho_2)
\end{align*}
which becomes subadditivity when we reverse the sign.
\item \textbf{Additivity:} It is not hard to see that if (and only if) $\rho=\rho_1\otimes\rho_2$ we have equality in the above and $$\mathrm{VN}(\rho)=\mathrm{VN}(\rho_1)+\mathrm{VN}(\rho_2)$$
\item \textbf{Recursion:} If  the density matrices $\rho_i$ are defined on orthogonal subspaces $H_i$ of a Hilbert space $H=\oplus_i H_i$, then the density matrix $\rho=\sum_i p_i\rho_i$ has the von Neumann entropy
$$\mathrm{VN}(\rho)=\mathrm{Sh}(\overrightarrow{p})+\sum_i p_i\mathrm{VN}(\rho_i)$$
This follows from the recursion property of Shannon entropy and the fact that if the matrix $\rho_i$ has eigenvalues $\lambda^i_j$ then the eigenvalues of $\rho$ will be of the form $p_i\lambda^i_j$.
\end{itemize}

\textbf{Monotonicity:} This is a property concerning composite systems, and it can be expressed by the following equation: $\mathrm{VN}(\rho)\geq\mathrm{VN}(\rho_1)$. However, this does not hold in general for von Neumann entropy. It is known that a composite system can be in a pure state (and so its von Neumann entropy will vanish) while its subsystems can be mixed (and so have positive von Neumann entropy).

\subsection{Renyi entropy}

Renyi entropies form a one parameter family of Schur concave, additive entropies defined by 
$$R_q(p_1,\ldots,p_n)=\frac{1}{1-q}\ln\left[\sum_{i=1}^n p_i^q\right], \ \forall q\geq 0$$
Special cases of the Renyi entropies include $q=0$, which is the logarithm of the number of non-zero components of the distribution and is known as the \textit{Hartley entropy}. When $q\rightarrow 1$, we have the Shannon entropy, and when $q\rightarrow \infty$ the \textit{Chebyshev entropy} $R_\infty=-\ln p_{max}$, a function of the largest component $p_{max}$. 

For any given probability vector $\overrightarrow{p}$ the Renyi entropy is a continuous, non-increasing function of its parameter:
$$R_t(\overrightarrow{p})\leq R_q(\overrightarrow{q}), \ \forall t>q$$
To illustrate this, we have plotted in Figure \ref{Renyi} several Renyi entropies as functions of a probability distribution with two variables. Note that since Renyi entropies are Schur concave, their maximum value is attained for the totally mixed probability distribution, in which case $R_q(1/n,\ldots,1/n)=\ln n$.

\begin{figure}[ht!]
\centering
\includegraphics[width=10cm]{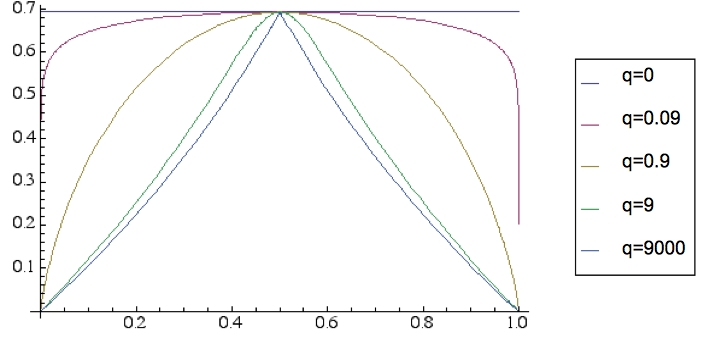}
\caption{$R_q(x,1-x)$ for $q=0, 0.5, 3$, $q\rightarrow 1$ and $q\rightarrow \infty$}\label{Renyi}
\end{figure}

At each parameter $q$, the quantum Renyi entropy can be defined on the set of density matrices as the classical Renyi entropy of the corresponding spectra:
$$\mathrm{R}_q(\rho)=\frac{1}{1-q}\ln\mathrm{Tr}(\rho^q)=\frac{1}{1-q}\ln\left[\sum_{i=1}^n \lambda_i^q\right]=R(\lambda_1,\ldots,\lambda_n)$$
Quantum Renyi entropy assigns the value $0$ to pure states exclusively, and $\ln n$ to the maximally mixed state $\rho_*=\frac{1}{n} I$.

Note that Renyi entropies are not concave in general, nor subadditive, but they are Schur-concave, additive and monotone \cite{BenZyc06}.

\end{chapter}

\section{Contextual entropy}\label{ctxt}

\subsection{Measures and partial traces}

We saw that, given a probability measure $\mu$ on the clopen subobjects of a spectral presheaf, if we fix a subobject $\underline{S}$ of $\underline{\Sigma}$ we obtain a map from
$\mathcal{V}(N)$ to $[0,1]$. We assume from now on that we are given a certain probability measure $\mu$, representing a quantum state. Then we can adopt a different perspective and instead of looking at a fixed subobject we can look at a fixed context $V$. There is a lattice isomorphism $\alpha_V$ between the projections in $V$ and the clopen subsets of $\underline{\Sigma}_V$. Hence from $\mu$ we can also obtain a map 
\begin{align*}
\mu|_{_V}:\mathcal{P}(V)&\longrightarrow[0,1]\\
 P&\longmapsto \mu(S_P)
\end{align*}
where $S_P=\alpha_V(P)\subseteq \underline{\Sigma}_V$.

Using this new perspective, we can show that measures on the spectral presheaf associated to a matrix algebra behave well with respect to the partial trace. This result has a certain physical significance. We have already seen that there is a bijective correspondence between states and probability measures, and we now show that moreover these measures capture the essential information theoretic property of the partial trace in a natural way. Thus, if we are given a measure corresponding to a composite state, we can obtain its partial traces in a direct way by simply considering its restrictions to contexts of a particular form. Intuitively, we would expect these contexts to be precisely those which only encode information related to the first subsystem (if we want to trace out the second one) or vice versa, and we will see that this will indeed be the case.

Note also that this result will be useful for us later on, when discussing the subadditivity property of our contextual entropy.

\begin{proposition}
Consider a state $\rho$ on the matrix algebra $\mathcal{M}_{nm}\simeq \mathcal{M}_n\otimes \mathcal{M}_m$. Let $\rho_1=\mathrm{Tr}_2(\rho)\in\mathcal{M}_n$ and $\rho_2=\mathrm{Tr}_1(\rho)\in\mathcal{M}_m$ be the partial traces of $\rho$. Then if $V\in \mathcal{V(M}_n)$ and $\mathbb{C}I_m$ denotes the trivial subalgebra of $\mathcal{M}_m$ we have
$$\mu_\rho|_{_{V\otimes \mathbb{C}I_m}}=\mu_{\rho_1}|_{_V}$$
Similarly, if $W\in \mathcal{V(M}_m)$ and $\mathbb{C}I_n$ denotes the trivial subalgebra of $\mathcal{M}_n$ we have
$$\mu_\rho|_{_{\mathbb{C}I_n\otimes W}}=\mu_{\rho_2}|_{_W}$$
\end{proposition}

\begin{proof}
Note first that there is a lattice isomorphism between the domains of definition of $\mu_\rho|_{_{V\otimes \mathbb{C}I_m}}$ and $\mu_{\rho_1}|_{_V}$ which takes $P\in\mathcal{P}(V)$ to $P\otimes I_m\in\mathcal{P}(V\otimes \mathbb{C}I_m)$. Then using the definition of measures for states on matrix algebras and the defining property of the partial trace, we have that 
$$\mu_\rho|_{_{V\otimes \mathbb{C}I_m}}(P\otimes I_m)=Tr(\rho\cdot P\otimes I_m)=Tr(\rho_1\cdot P)=\mu_{\rho_1}|_{_V}(P),\ \ \forall P\in\mathcal{P}(V\otimes \mathbb{C}I_m)$$
and similarly for the second statement.\qed 
\end{proof}

Finally, the fact that $\mu$ is a measure implies several properties for $\mu|_{_V}$ which hold for all contexts $V\in\mathcal{V}(N)$, and which we shall state below:
\begin{enumerate}
 \item $\mu|_{_V}(I)=1$ and $\mu|_{_V}(0)=0$
 \item $\mu|_{_V}(P\vee Q)+\mu|_{_V}(P\wedge Q)=\mu|_{_V}(P)+\mu|_{_V}(Q)$
 \item in particular, if $P$ and $Q$ are orthogonal then $P\wedge Q=0$ and $P\vee Q=P+Q$ and hence $$\mu|_{_V}(P+Q)=\mu|_{_V}(P)+\mu|_{_V}(Q)$$
 \item if $P\leq Q$ then $\mu|_{_V}(P)\leq \mu|_{_V}(Q)$
\end{enumerate}

These properties show that $\mu|_{_V}$ is a finitely additive probability measure on the lattice of projections of $V$.

\subsection{The entropy of a measure}

We saw that in classical probability theory Shannon entropy assigns a real number to every discrete probability distribution.  It is known that a von Neumann algebra on a finite-dimensional Hilbert space is simply a matrix algebra, or a finite direct sum of matrix algebras. We will see further on how to associate a distinguished probability distribution to each commutative subalgebra (or context) of a von Neumann algebra of bounded operators on finite dimensional Hilbert space, given a state on the system described by that algebra in the form of a measure on its associated spectral presheaf. Once this is done, we will be able to associate to each context its corresponding Shannon entropy, and moreover we will see that this collection of Shannon entropies fits together in a nice way and gives a global section of a certain real-number presheaf. This is consistent with the basic idea of the topos approach, that of putting together the information obtained from each classical perspective on a quantum system. We will see in later sections that by keeping track of all classical entropies associated to a quantum state we can not only retrieve that state's von Neumann entropy, but also reconstruct the state itself.

\begin{definition}
Let $H$ be an Hilbert Space, $\mathcal{B}(H)$ the algebra of bounded operators in $H$ and $\mathcal{F}\subseteq \mathcal{B}(H)$. The von Neumann commutant of $\mathcal{F}$, usually denoted by $\mathcal{F}'$, is the subset of $\mathcal{B}(H)$ consisting of all elements that commute with every element of $\mathcal{F}$, that is
$$\mathcal{F}'=\{T\in \mathcal{B}(H)~|~TS=ST,\  \forall S\in\mathcal{F}\}$$

The von Neumann double commutant $\mathcal{F}$ of is just $(\mathcal{F}')'$ and is usually denoted by $\mathcal{F}''$.
\end{definition}

If we consider a set of orthogonal rank-one projections $\{P_1,\ldots,P_n\}''$, their double commutant can be shown to be simply $\mathbb{C}P_1+\ldots+\mathbb{C}P_n$.

This shows that in finite dimensions each context $V$ can be generated via the von Neumann double commutant construction in a unique way from a set of pairwise orthogonal projections which add up to the identity. If we denote this canonical set of projections by $\{P_1, P_2,\ldots,P_k\}$ then $(\mu|_{_V}(P_1), \mu|_{_V}(P_2),\ldots,\mu|_{_V}(P_k))$ is a probability distribution. Hence to each context $V$ we can assign the Shannon entropy of its associated probability distribution:
$$\mathrm{Sh}(\mu|_{_V}(P_1), \mu|_{_V}(P_2),\ldots,\mu|_{_V}(P_k))=-\sum_{i=1}^k \mu|_{_V}(P_i)\ln \mu|_{_V}(P_i)$$
If $V'\supseteq V$ then $V' = \{Q^1_1,\ldots,Q^1_{l_1}, Q^2_1,,\ldots Q^2_{l_2},\ \ldots,\ Q^k_1,\ldots,Q^k_{l_k}\}''$, where the $Q^j_i$s are pairwise orthogonal and
$$\sum_{i=1}^{k_j} Q^j_i=P_j$$

The Shannon entropy associated to $V'$ is related to the Shannon entropy associated to $V$ via the recursion formula:
$$\mathrm{Sh}(V')=\mathrm{Sh}(V)+\sum_{i=1}^k \mu|_{_V}(P_i)\cdot \mathrm{Sh}\left(\frac{\mu|_{_{V'}}(Q^i_1)}{\mu|_{_V}(P_i)}, \frac{\mu|_{_{V'}}(Q^i_2)}{\mu|_{_V}(P_i)},\ldots,
\frac{\mu|_{_{V'}}(Q^i_{l_i})}{\mu|_{_V}(P_i)}\right)$$

Since Shannon entropy is non-negative, it follows that $\mathrm{Sh}(V')\geq\mathrm{Sh}(V)$ and this enables us to give the following definition for the entropy of a measure (and hence of a quantum state).

We call a context $k$-dimensional if it is generated by $k$ pairwise orthogonal projections which add up to the identity. 

\begin{definition}
If $\mu$ is a measure on the clopen subobjects of the spectral presheaf $\underline{\Sigma}$ then the entropy $E(\mu)$ associated to $\mu$ is a global section of the presheaf $\underline{[0,\ln n]^\preceq}$ which at a context $V=\{P_1, P_2,\ldots,P_k\}''$ has the value $$E(\mu)|_{_V}=\mathrm{Sh}(\mu|_{_V}(P_1),
\mu|_{_V}(P_2),\ldots,\mu|_{_V}(P_k))=-\sum_{i=1}^k \mu|_{_V}(P_i)\ln \mu|_{_V}(P_i)$$
Note that if the $V$ is a $k$-dimensional context then the value taken by $E(\mu)$ at $V$ is less then or equal to $\ln k$, and hence for an $n$-dimensional matrix algebra, the maximal value taken by $E(\mu)$ at any context is $\ln n$. Therefore contextual entropy can be seen as a mapping defined on the set of measures associated to a spectral presheaf:
$$E:\mathcal{M}(\underline{\Sigma})\longrightarrow \Gamma \underline{[0,\ln n]^{\preceq}}\ \  .$$
\end{definition}

Notice that although there is a bijective correspondence between states of a von Neumann algebra and measures on the spectral presheaf associated to it, the above definition does
not make any direct reference to the quantum state which the measure corresponds to.

\subsection{Properties of the contextual entropy}

\subsubsection{Retrieving the von Neumann entropy}\label{VNeu}
Given a density matrix $\rho$, there exists at least one orthonormal basis of Hilbert space with respect to which $\rho$ is diagonal. Such a basis corresponds to a set of
one-dimensional pairwise orthogonal projections $\{P_1,\ldots,P_n\}$, which in turn determine a maximal context $V_\rho$ via the double commutant construction. It is easy to check
that the eigenvalues $\{\lambda_i\}_{i=1}^n$ of $\rho$ satisfy $\lambda_i=Tr(\rho P_i)$. Hence the value assigned to the entropy of the measure $\mu_\rho$ at any context $V_\rho$
obtained through the above procedure, is just the von Neumann entropy of the state $\rho$:
\begin{align*}
E(\mu_\rho)_{_{V_\rho}} &=-\sum_{i=1}^n \mu_\rho|_{_{V_\rho}}(P_i)\ln \mu_\rho|_{_{V_\rho}}(P_i)\\ 
&=-\sum_{i=1}^n \mathrm{Tr}(\rho P_i)\ln\mathrm{Tr}(\rho P_i)\\
&=-\sum_{i=1}^n \lambda_i\ln\lambda_i=\mathrm{VN}(\rho)
\end{align*}

The natural question to ask at this stage is whether there is any way of determining the von Neumann entropy of a state if we are given an arbitrarily defined measure $\mu$ without being explicitly told which state it corresponds to. It turns out that the answer is yes, since it can be showed that the von Neumann entropy is the minimal value amongst the
numbers assigned to the maximal contexts of a von Neumann algebra by our generalized notion of entropy. Proving this result requires the Schur-Horn Lemma, which we state below.

\begin{theorem}[Schur-Horn Lemma]
Let $\rho$ be a Hermitian matrix, and let the vector $(\lambda_1,\ldots,\lambda_n)$ denote its spectrum. Let $(\delta_1,\ldots,\delta_n)$ denote its diagonal elements in a given basis. Then 
$$(\delta_1,\ldots,\delta_n)\preceq (\lambda_1,\ldots,\lambda_n)\ \ .$$
Conversely, if this equation holds, there exists a Hermitian matrix with spectrum $(\lambda_1,\ldots,\lambda_n)$ whose diagonal elements are given by
$(\delta_1,\ldots,\delta_n)$.
\end{theorem}

The proof of this useful result can be found in \cite{horn}. 

\begin{theorem}\label{VN}
 Given a state $\rho$, a maximal context $V_\rho$ in which $\rho$ is diagonal, and any other maximal context $V$, we have
 $$\mathrm{Sh}(\delta_1,\ldots,\delta_n)=E(\mu_\rho)_{_V}\geq E(\mu_\rho)_{_{V_\rho}}=\mathrm{Sh}(\lambda_1,\ldots,\lambda_n)=\mathrm{VN}(\rho)$$
where we have denoted by $\delta_i$ the diagonal entries of $\rho$ in the basis (unique up to phases) in which the elements of V are diagonal matrices.
\end{theorem}

\begin{proof}
The key part of this proof is showing that $E(\mu_\rho)_{_V}$ is equal to the Shannon entropy of the diagonal elements of $\rho$ in some given basis. 

Let $V=\{Q_1,\ldots,Q_n\}''$. Then $$E(\mu_\rho)_{_V}=-\sum_{i=1}^n \mathrm{Tr}(\rho Q_i)\ln\mathrm{Tr}(\rho Q_i)$$
Since the $Q_i$'s are pairwise orthogonal, there exists a unitary $U$ which simultaneously diagonalises them. Using the fact that $\mathrm{Tr}(AB)=\mathrm{Tr}(BA)$ for any two
matrices $A$ and $B$, we can write
$$E(\mu_\rho)_{_V}=-\sum_{i=1}^n \mathrm{Tr}(U\rho U^{-1}UQ_iU^{-1})\ln\mathrm{Tr}(U\rho U^{-1}UQ_iU^{-1})$$
Since for all $i$, $UQ_iU^{-1}$ is a diagonal rank-one projection, the collection of numbers given by the traces $\mathrm{Tr}(U\rho U^{-1}UQ_iU^{-1})$ is he same as that which consists of the diagonal entries of $U\rho U^{-1}$, i.e. it is simply the collection of diagonal entries of $\rho$ when expressed in the basis determined by the column vectors of the unitary matrix $U$. Since these entries are positive and add up to unity they form a probability distribution, and we can say that
$$E(\mu_\rho)_{_V}=\mathrm{Sh}(\delta_1,\ldots,\delta_n)$$

The Schur-Horn Lemma together with the fact that Shannon entropy is majorization reversing give us 
$$E(\mu_\rho)_{_V}=\mathrm{Sh}(\delta_1,\ldots,\delta_n)\geq \mathrm{Sh}(\lambda_1,\ldots,\lambda_n)=\mathrm{VN}(\rho)$$
and we have already seen that $ E(\mu_\rho)_{_{V_\rho}}=\mathrm{VN}(\rho)$. \qed
\end{proof}

\begin{remark}\label{unitarily equivalent global sections}
From the proof of the last theorem we can extract an important observation: if we evaluate the contextual entropy of a state $\rho$ at some  context $V$ (not necessarily maximal), this will be equal to the contextual entropy of any unitarily equivalent state as long as we evaluate it at a context which is obtained from $V$ through rotation by the same unitary. That is,
$$E(\mu_\rho)_{_V}=E(\mu_{U\rho U^{-1}})_{_{U VU^{-1}}}$$

This equation expresses a covariance property: it does not matter if we consider the contextual entropy of a state $\rho$ (expressed in some basis) or of the same state expressed with respect to some other basis, $U\rho U^{-1}$, as long as we also adapt the context $V$ that we are considering accordingly, i.e. to $UVU^{-1}$. As functions, $E(\mu_\rho)$ and $E(\mu_{U\rho U^{-1}})$ are not the same, but they are the same up to a `rotation by U of contexts'.
\end{remark}

Given the contextual entropy map, the problem of finding a maximal context for which the minimum discussed above is attained is equivalent to the problem of finding the point at which a real-valued function on the group of unitaries $\mathcal{U}(n)$ attains its minimal value. To see why this is the case, let $C:=(E_1,\ldots,E_n)$ denote the maximal context determined by projections which are diagonal with respect to the computational basis. Any other maximal context $V=( P_1,\ldots,  P_n)$ can be written as $UC U^{-1}:=( UE_1 U^{-1},\ldots, UE_n U^{-1})$ for some unitary $ U$. Hence we can construct a real-valued function on the group of unitaries by considering the values which the contextual entropy map takes when it is evaluated on the set of maximal contexts. Explicitly, this map is
\begin{align*}
 \mathcal{W}_{\mu}:\mathcal{U}(n)&\longrightarrow \mathbb{R}\\
 U&\longmapsto E(\mu)_{_{UC U^{-1}}}
\end{align*}

It is possible to use existing optimization algorithms \cite{traian1,traian2} in order to determine the point at which this function attains its global minimum.

\subsubsection{Contextual vs. Shannon and von Neumann entropies}\label{cns}

We will now consider which of the properties of Shannon and von Neumann entropies have counterparts for contextual entropy. An immediate difficulty is posed by the fact
that the values of the contextual entropies are not real numbers but global sections of certain presheaves of real numbers, which may live in different topoi, i.e. they may be defined over different base categories. In some cases it is possible to work around this difficulty by adapting the definitions of order relations and algebraic operations on $\mathbb{R}$ to suit our more general framework. 

\noindent \textbf{1) Positivity}

Both von Neumann and Shannon entropies are positive. Shannon entropy is zero for any probability distribution in which one outcome occurs with $100$\% certainty and strictly
positive otherwise. Similarly, von Neumann entropy is zero for all pure states, and strictly positive for the others.

The contextual entropy does assign non-negative values to all contexts, hence the resulting global section can be thought of as non-negative. The minimum of contextual entropy on maximal contexts is 0 if and only the state is pure. This is entirely analogous to Shannon and von Neumann entropy. 


\noindent \textbf{2) Concavity}

Shannon entropy is concave: if $\vec{p}$ and $\vec{q}$ are two probability distributions then $$\mathrm{Sh}(r\cdot \vec{p}+(1-r)\cdot \vec{q})\geq
r\mathrm{Sh}(\vec{p})+(1-r)\mathrm{Sh}(\vec{q})$$ For von Neumann entropy concavity is defined by a similar formula: $$\mathrm{VN}(r\rho+(1-r)\sigma)\geq
r\mathrm{VN}(\rho)+(1-r)\mathrm{VN}(\sigma)$$

The contextual entropy satisfies a similar property. If $\rho$ and $\sigma$ are defined on the same Hilbert space $\mathcal{H}$ then for every context $V\in\mathcal{B(H)}$, if $V$
is generated by the projections $\{P_1,\ldots,P_k\}$, we have
\begin{align*}
E(\mu_{r\rho+(1-r)\sigma})_{_V}&=\mathrm{Sh}(~ \mathrm{Tr}[(r\rho+(1-r)\sigma)P_1],\, \ldots,\mathrm{Tr}[(r\rho+(1-r)\sigma)P_k]~ )\\
 &=\mathrm{Sh}(~ [r\mathrm{Tr}(\rho P_1)+(1-r)\mathrm{Tr}(\sigma P_1)],\, \ldots,r\mathrm{Tr}(\rho P_k)+(1-r)\mathrm{Tr}(\sigma P_k)~)\\
 &\geq r\mathrm{Sh}(\mathrm{Tr}(\rho P_1),\, \ldots,\mathrm{Tr}(\rho P_k))~+~(1-r)\mathrm{Sh}(\mathrm{Tr}(\sigma P_1),\, \ldots,\mathrm{Tr}(\sigma P_k))\\
 &=r\cdot E(\mu_\rho)_{_V}+(1-r)E(\mu_\sigma)_{_V}
\end{align*}

Hence contextual entropy is globally concave:
$$E(\mu_{r\rho+(1-r)\sigma})\geq r\cdot E(\mu_\rho)+(1-r)E(\mu_\sigma), \ \ \forall r\in[0,1]$$

\noindent \textbf{3) Additivity and Subadditivity}

Subadditivity a property concerning composite systems. Recall that an entropy is called subadditive if the entropy of a composite system is smaller than the sum of the entropies of its parts. Both von Neumann and Shannon entropies are subadditive. We would like to obtain an inequality of the form
$$E(\mu_\rho)\leq E(\mu_{\rho_1})+E(\mu_{\rho_2})$$
where $\rho$ is the density matrix representing a composite state and $\rho_1$ and $\rho_2$ are the partial traces of $\rho$. It is not immediately clear how one could define such an inequality, since this time the terms involved are global sections of presheaves over three different base categories.
Hence in order to talk about subadditivity in a meaningful way, we must first define a suitable notion of addition between the global sections $E(\mu_{\rho_1})$ and
$E(\mu_{\rho_2})$.

In order to see how this might be done, we start by considering some context $V$ of the first subsystem and some other context $W$ of the second subsystem. If
$V=\{P_1,\ldots,P_k\}''$ and $W=\{Q_1,\ldots,Q_r\}''$, from the definition of the entropy we have
$$ E(\mu_{\rho_1})|_{_V}=\sum_{i=1}^k \mathrm{Tr}(\rho_1 P_i) \ln  \mathrm{Tr}(\rho_1 P_i),$$
$$ E(\mu_{\rho_2})|_{_W}=\sum_{j=1}^r \mathrm{Tr}(\rho_2 Q_j) \ln  \mathrm{Tr}(\rho_2 Q_j)$$
We can add these two numbers together, and we can use the fact that Shannon entropy is additive for independent probability distributions (i.e. $\sum_{i=1}^k p_i\ln p_i+
\sum_{j=1}^r q_j \ln q_j = \sum_{i,j} p_iq_j \ln p_iq_j$) and the fact that $ \mathrm{Tr}(\rho_1 P_i) \mathrm{Tr}(\rho_2 Q_j)= \mathrm{Tr}(\rho_1\otimes \rho_2 P_i\otimes Q_j)$ to
obtain
$$E(\mu_{\rho_1})|_{_V}+E(\mu_{\rho_2})|_{_W} = \sum_{i=1,j} \mathrm{Tr}(\rho_1\otimes \rho_2 P_i\otimes Q_j) \ln \mathrm{Tr}(\rho_1\otimes \rho_2 P_i\otimes Q_j) =
E(\mu_{\rho_1\otimes\rho_2})|_{_{V\otimes W}}$$

Hence we could use the following requirement for the definition of subadditivity: $E(\mu_\rho)$ should be less than or equal to $E(\mu_{\rho_1\otimes\rho_2})$ at each
context $\widetilde{V}$ of the composite system. This definition enables us to say, for instance, that the contextual entropy is additive when $\rho=\rho_1\otimes\rho_2$. Note that this is a direct consequence of the additivity property of Shannon entropy.

Even when $\rho$ is not equal to $\rho_1\otimes\rho_2$ the subadditivity property holds in split contexts (i.e. contexts of the form $V\otimes W$) as a consequence of Shannon subadditivity. Consider $\widetilde{V}=V\otimes W$, with $V$ and $W$ as above. We know that 
$$\mu_{\rho_1}(P_i) = \mu_\rho(P_i\otimes I) = \sum_{j=1}^r \mu_\rho(P_i \otimes Q_j)$$
for all $i\in\{1,\ldots,k\}$ and 
$$\mu_{\rho_2}(Q_j) = \mu_\rho(I\otimes Q_j)= \sum_{j=1}^r \mu_\rho(P_i \otimes Q_j)$$
for all $j\in\{1,\ldots,r\}$. 

Using the subadditivity property of Shannon entropy we obtain 
\begin{align*}
 E(\mu_\rho)|_{_{V\otimes W}} &= \mathrm{Sh}\left(P_1\otimes Q_1,\ldots,P_1\otimes Q_r,\ \ldots \ ,P_k\otimes Q_1,\ldots, P_k\otimes Q_r\right) \\
&\leq \mathrm{Sh}\left(\sum_{i=1}^k \mu_\rho(P_i\otimes Q_1), \sum_{i=1}^k \mu_\rho(P_i\otimes Q_2),\ldots, \sum_{i=1}^k \mu_\rho(P_i\otimes Q_r)\right) + \\
& \ \ ~ \ \mathrm{Sh}\left(\sum_{j=1}^r \mu_\rho(P_1\otimes Q_j), \sum_{j=1}^r \mu_\rho(P_2\otimes Q_j),\ldots, \sum_{j=1}^r \mu_\rho(P_k\otimes Q_j)\right)\\
&=\mathrm{Sh}\left(\mu_{\rho_1}(P_1),\mu_{\rho_1}(P_2),\ldots, \mu_{\rho_1}(P_k)\right)+\mathrm{Sh}(\mu_{\rho_2}(Q_1),\mu_{\rho_2}(Q_2),\ldots,\mu_{\rho_2}(Q_r))\\
&=E(\mu_{\rho_1})|_{_V}+E(\mu_{\rho_2})|_W=E(\mu_{\rho_1\otimes \rho_2})|_{_{V\otimes W}} 
\end{align*}

\begin{remark}
The fact that the contextual entropy is subadditive in all split contexts can be used to give a more direct proof of the subadditivity property of von Neumann entropy, which avoids using Klein's inequality: if we choose the split context $\widetilde{V}$ such that $\rho_1$ is diagonal in $W$ and $\rho_2$ is diagonal in $W$ we have from Theorem \ref{VN} that
$$\mathrm{VN}(\rho)\leq E(\mu_\rho)|_{_{V\otimes W}} \leq E(\mu_{\rho_1})|_{_V}+E(\mu_{\rho_2})|_W = \mathrm{VN}(\rho_1)+\mathrm{VN}(\rho_2)$$
\end{remark}

For contexts which are not split (which we usually call entangled contexts), the subadditivity property does not necessarily hold. Intuitively, we can understand why this happens: the converse of the Schur-Horn lemma implies that for
any density matrix $\rho$, there is some unitary $U$ for which the diagonal of $U\rho U^{-1}$ is the maximally mixed vector $(\frac{1}{n},\frac{1}{n},\ldots,\frac{1}{n})$. Let $D_n$ denote the context generated by the set of projections $\{E_{11},\ldots,E_{nn}\}$, where we have fixed our basis such that $E_{ii}$ is the projection with the $i^{th}$ diagonal entry equal to one and all other entries equal to zero. Then

$$E(\mu_\rho)|_{_{U^{-1}\cdot D_n\cdot U}}=E(\mu_{U\rho U^{-1}})|_{_ {D_n}}=\mathrm{Sh}\left((U\rho U^{-1})_{11},\ldots,(U\rho U^{-1})_{nn}\right)=\ln n$$

There is however no guarantee that the diagonal of $U\rho_1\otimes\rho_2U^{-1}$ will also be the maximally mixed vector. Hence the contextual entropy map will assign a smaller value to the state $U\rho_1\otimes\rho_2U^{-1}$ at the context $D_n$. This in turn implies that 
$$E(\mu_\rho)|_{_{U^{-1}\cdot D_n\cdot U}}>E(\mu_{\rho_1\otimes\rho_2})|_{_{U^{-1}\cdot D_n\cdot U}}$$

Note however that explicitly finding the unitary matrix $U$ for which the diagonal of $U\rho U^{-1}$ is the maximally mixed vector is an example of an inverse eigenvalue problem, and is not at all trivial. Algorithms for finding such unitaries do exist, see \cite{DHST03}, but they will not help us find a counter-example of the subadditivity property - we may end up finding a unitary which takes both the state $\rho$ and the state $\rho_1\otimes\rho_2$ to a basis where each has the maximally mixed vector on the diagonal. Instead we can look at particular density matrices and construct general examples for which the subadditivity property fails to hold. The details of these computations can be found in Appendix \ref{AppendixA}.

The conclusion is that subadditivity is not a global property of our contextual entropy. In particular, it may fail to hold at entangled contexts, a behaviour which is not captured either by the classical Shannon entropy, or by its quantum-mechanical counterpart.

\noindent \textbf{4) Monotonicity}

Monotonicity is a property which refers to composite systems. Recall that an entropy is called monotone if the entropy of a composite system is larger than each of the entropies of its parts. Shannon entropy is monotone: let $(p_1,\ldots,p_n)$ and $(q_1,\ldots,q_m)$ be two probability distributions, and let $(p_1q_1,\ldots,p_1q_m,\ \ldots,\ p_nq_1,\ldots,p_nq_m)$ be the probability distribution obtained by composing them. The Shannon entropy of the latter is at least as large as the Shannon entropies of the each of the former distributions.

Von Neumann entropy on the other hand is not monotone. Let $\rho$ be defined on a Hilbert space $\mathcal{H}=\mathcal{H}_1\otimes\mathcal{H}_2$, and let $\rho_1=\mathrm{Tr}_2\rho$
and $\rho_2=\mathrm{Tr}_1\rho$ be defined on the Hilbert spaces $\mathcal{H}_1$ and $\mathcal{H}_2$ respectively. We know that a composite system can be in a pure state, in which
case $\mathrm{VN}(\rho)=0$, but its subsystems might be mixed, and then $\mathrm{VN}(\rho_i)>0$.

When we consider the contextual entropy, we would like to obtain an equation of the form
$E(\mu_\rho)\geq E(\mu_{\rho_i})$ for each subsystem $i$. However, we are again faced with the problem that the entropy of a composite system and the entropies of its parts are global sections of presheaves which live in different topoi. We try to solve this problem by giving a general method of comparing such global sections.

\begin{definition}
 Let $\mathcal{H}=\mathcal{H}_1\otimes\mathcal{H}_2$ and let $N=\mathcal{B(H)}$, $N_i=\mathcal{B(H}_i)$, for $i=1,2$. If $\gamma\in\Gamma\underline{[0,\ln n]^{\preceq}}$ is a global section of the presheaf of real values which lives in the topos $\mathrm{Sets}^{{\mathcal{V}(N)}^{op}}$ and $\gamma_i\in\Gamma\underline{[0,\ln
n_i]^{\preceq}}$ are global sections of the presheaves living in the topoi $\mathrm{Sets}^{{\mathcal{V}(N_i)}^{op}}$, we say that $\gamma\geq\gamma_1$ if for all contexts $V\in\mathcal{V}(N_1),\ W\in\mathcal{V}(N_2)$ and all unitaries $U$:
 $$\gamma_{_{U\cdot V\otimes W\cdot U^{-1}}}\geq {\gamma_1}_{_V}$$
Similarly, one can define a binary order relation between $\gamma$ and $\gamma_2$.
\end{definition}

Note that the order relation we have just defined is not a partial order since it is not reflexive. However, with respect to this order we can see that for split contexts our
contextual entropy behaves like Shannon entropy: let $\{P_1,\ldots,P_n\}$ be the projections generating $V$ and $\{Q_1,\ldots,Q_m\}$ the projections generating $W$. Then
$\{P_1\otimes Q_1,P_1\otimes Q2,\ldots,P_n\otimes Q_m\}$ are the canonical projections generating $\widetilde{V}$. For any $\rho,\ \rho_1$ and $\rho_2$ defined as above we have
$$p_i:=\mathrm{Tr}(\rho_1 P_i)=\mathrm{Tr}(\rho P_i\otimes I)=\mathrm{Tr}(\rho P_i\otimes \sum_{j=1}^m Q_j)=\sum_{j=1}^m \mathrm{Tr}(\rho P_i\otimes Q_j)=:\sum_{j=1}^m r_{ij}$$
If we now use the recursion relation for Shannon entropy we have:
$$\mathrm{Sh}(r_{11},\ldots,r_{nm})=\mathrm{Sh}(p_1,\ldots,p_n)+\sum_{i=1}^n p_i\cdot \mathrm{Sh}(\frac{r_{i1}}{p_i},\frac{r_{i2}}{p_i},\ldots,\frac{r_{im}}{p_i})$$
and since all the terms in the sum on the right hand side of the equation are positive we must have $$\gamma_{_{V\otimes
W}}=\mathrm{Sh}(r_{11},\ldots,r_{nm})\geq\mathrm{Sh}(p_1,\ldots,p_n)={\gamma_1}_{_V}$$
for all contexts $V$ and $W$.

For entangled contexts our contextual entropy behaves like von Neumann entropy in the sense that it does not satisfy the monotonicity property. In order to see this, we can look at the same situation in which monotonicity failed for von Neumann entropy: let $\rho_1$ be a mixed state and $\rho$ a pure one such that $\mathrm{Tr}_2\rho=\rho_1$. Let
$V_\rho$ be a maximal context in which $\gamma_{_{V_\rho}}=VN(\rho)=0$. We know that $V_\rho$ must be of the form $U\cdot V\otimes W\cdot U^{-1}$ for some maximal contexts $V$ and
$W$ and unitary $U$. From Theorem \ref{VN} we have ${\gamma_1}_{_V}\geq \mathrm{VN}(\rho_1)>0$. Hence it is not true that $\gamma\geq \gamma_1$.

Clearly if this happens, $U\cdot V\otimes W \cdot U^{-1}$ must be an entangled context, if it would be split we would be in the situation analyzed previously.

\section{Reconstructing pure states from global sections}\label{r1}

A direct implication of Remark \ref{unitarily equivalent global sections} is that unlike von Neumann entropy, which gives the same value for unitarily equivalent states, our contextual entropy gives different (though in a sense unitarily equivalent) global sections of the presheaf $\underline{[0,\ln n]^\preceq}$. This enables us not only to distinguish which global sections come from measures associated to pure states but also to explicitly reconstruct those pure states. We explain this method in more detail.

Recall that the von Neumann entropy of a state vanishes if and only if that state is a pure one. Given a global section $\gamma\in \Gamma\underline{[0,\ln n]^\preceq}$ if $\gamma$
is in the image of the contextual entropy mapping $E$ then it comes from a measure associated to a pure state if and only if there exists a maximal context $V$ such that
$\gamma|_{_V}=0$. This means that if $V$ is generated by the set of rank one projections $$\{P_1,\ldots,P_n\}=\{\left|\psi_1\right>\left<\psi_1\right|,\ldots,
\left|\psi_n\right>\left<\psi_n\right|\}$$ our state must equal one of these projections and our only task is to determine which one. For this, consider unitaries $U_1,\ldots,U_n$
which have the property that $U_iP_iU^{-1}_i=P_i$ and $$\{U_iP_jU^{-1}_i~|~1\leq j\leq n, j\neq i\}\neq\{P_1,\ldots,\widehat{P_i},\ldots,P_n\}$$
Think of this as taking $n$ rotations in Hilbert space, each of which preserves one axis of the orthonormal basis $\{\left|\psi_1\right>,\ldots,\left|\psi_n\right>\}$ and rotates
the others, but without permuting them.

If we consider the contexts $V_i= \{U_iP_1U_i^{-1},\ldots,U_iP_nU_i^{-1}\}''$ then $\rho$ will be diagonal only in one of the orthonormal bases which correspond to these contexts.
This means the contextual entropy will assign the value zero to precisely one of the contexts $V_i$, and hence our state is
$$\rho= \{U_iP_1U_i^{-1},\ldots,U_iP_nU_i^{-1}\}\cap \{P_1,\ldots,P_n\}$$

\section{Reconstructing arbitrary quantum states from global sections}\label{arbitrary}

Consider a global section $\gamma\in \Gamma\underline{[0,\ln n]^\preceq}$. We present here an algorithm for reconstructing the state $\rho$ for which $E(\mu_\rho)=\gamma$. We
assume for now that $\gamma$ is in the image of the contextual entropy mapping. If our algorithm will fail to find a solution we will know that our initial assumption was false.
Otherwise we must perform one final check at the end of our algorithm to make sure that this assumption was correct.

\noindent\textbf{Step 1.} Start by identifying one maximal context $V$ such that $\gamma|_{_V}\leq\gamma|_{_W}$ for all maximal contexts $W$. This amounts to retrieving the von Neumann entropy of the state $\rho$ from the contextual entropy. Note that in general, the minimal value of $E$ will be attained in many different maximal contexts, but any of them can be used in our reconstruction.  

If the minimal value among maximal contexts equals zero we must have a pure state, and we have already discussed these. Otherwise, we know from Section \ref{VNeu} that
$\rho$ must be diagonal in the context $V$. If we consider the canonical projections $\{P_1,\ldots,P_n\}$ which generate $V$, the fact that $\rho$ is diagonal at $V$ implies that
it is of the form
$$\rho=\lambda_1 P_1+\ldots+\lambda_n P_n$$
where the $\lambda_i$'s are the eigenvalues of $\rho$. 

Of course, if we have access to the probability distribution on the spectral presheaf from which $E_\rho$ was constructed, just identifying this context immediately yields the quantum state: $$\rho=\sum_i \mu|_{_V} (P_i) P_i$$ 

However, we are seeking to reconstruct the state mathematically using only information encoded by the contextual entropy map.

\noindent\textbf{Step 2.} We are now left with the task of determining the eigenvalues of $\rho$. For this assume that the dimension $n$ of our Hilbert space
is greater or equal to $3$. For each $i\in\{1,\ldots,n\}$ let $$W_i:=\{P_i,I-P_i\}''$$
Then $\textrm{Sh}(\lambda_i,1-\lambda_i)$ must equal $\gamma|_{_{W_i}}$ for all $i$. If $$\gamma|_{_{W_i}}>\ln 2$$ then the global section $\gamma$ cannot be in the image of the
contextual entropy mapping, and our algorithm stops. Otherwise, the transcendental equation $\textrm{Sh}(x_1,x_2)=k$ has two solutions which are symmetric around $\frac{1}{2}$ as
indicated in Figure \ref{Sh}.

\begin{figure}[ht!]
\centering
\includegraphics[width=5.2cm]{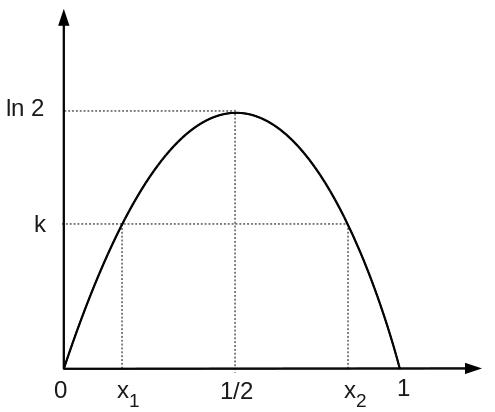}
\caption{The Shannon entropy of a two-variable probability distribution}
\label{Sh}
\end{figure}

\noindent Let $p_i$ and $1-p_i$ be the solutions of $\textrm{Sh}(x_1,x_2)=\gamma|_{_{W_i}}$ and assume without loss of generality that $p_i\leq \frac{1}{2}$. For each $i$ we have
at most two choices for the value of the $i^{th}$ eigenvalue of $\rho$: we can either set $\lambda_i=p_i$ or $\lambda_i=1-p_i$. Since
$$\lambda_1+\ldots+\lambda_n=1$$
there can be at most one $j$ such that $p_j < \frac{1}{2}$ and $\lambda_j=1-p_j$, while for all $i\neq j$ we must have $\lambda_i=p_i$. Let 
$$S=\sum_{i=1}^n p_i$$
Clearly $\sum_{i=1}^n\lambda_i\geq S$. We are now faced with three possible scenarios:

\noindent\textbf{(a)} If $S>1$ we obtain a contradiction, hence $\gamma$ can not be in the image of the contextual entropy mapping.

\noindent\textbf{(b)} If $S=1$ then the assignment $\lambda_i=p_i$ gives one possible solution for the set of eigenvalues of our state $\rho$. This solution is clearly unique: any other choice of
values will make the total sum of the eigenvalues of $\rho$ greater than $1$.

\noindent\textbf{(c)} If $S<1$ then we must determine the $j$ for which $p_j<\frac{1}{2}$ and $\lambda_j=1-p_j$. If such a $j$ exists then
$$1=\sum_{i=1}^n\lambda_i=S-p_j+(1-p_j)$$
hence $p_j$ should equal $\frac{S}{2}$. Now
\begin{itemize}
{\setlength\itemindent{25pt}  \item[\textbf{(c1)}] if the value $\frac{S}{2}$ does not appear amongst $\{p_1,\ldots,p_n\}$ then we have no solution }
{\setlength\itemindent{25pt}  \item[\textbf{(c2)}] if $\frac{S}{2}$ appears once, we have a unique solution }
{\setlength\itemindent{25pt} \item[\textbf{(c3)}] if it appears more than once, let $\{j_1,\ldots,j_m\}$ be the set of indices for which $p_{j_k}=\frac{S}{2}$. If we set $\lambda_{j_k}=1-p_{j_k}$ and take another
$l\in\{1,\ldots,m\}$, $l\neq k$. Then
$$\sum_{i=1}^n\lambda_i\geq \lambda_{j_k}+ \lambda_{j_l} = 1-p_{j_k}+p_{j_k}=1$$
In order to have equality we must have $m=2$ and $p_i=0$ for all $i\notin \{j_1,j_2\}$. Unless this happens we cannot find a solution. On the other hand, for $m=2$ we have two
possible solutions. These correspond to the two states
\vspace{0.5cm}

\begin{equation*}
\rho_1=p_{j_1} P_{j_1}+ (1-p_{j_1})P_{j_2}
\end{equation*} 

and 

\begin{equation*}
\rho_2=(1-p_{j_1}) P_{j_1}+ p_{j_1} P_{j_2}
\end{equation*}
}

\end{itemize}

\noindent In order to decide between the two states in \textbf{(c3)} we need to run our algorithm again but with a slight modification: instead of considering two-dimensional subalgebras of $V$, we take a
unitary $U$ which rotates all the canonical projections generating $V$, except $P_{j_1}$, which it leaves unchanged, and we consider the two dimensional subalgebras of $U\cdot
V\cdot U^{-1}$ of the form
$$\widetilde{W_i}=\{UP_iU^{-1}, I- UP_iU^{-1}\}''$$
We solve the equations $\mathrm{Sh}(x_i,1-x_i)=\gamma|_{_{\widetilde{W_i}}}$ and choose as before $n$ numbers from these solutions, such that they add up to one. These numbers
represent the diagonal entries of the matrix $U^{-1}\rho U$. We will not encounter any problems when retrieving these entries (unless of course, our initial assumption about
$\gamma$ being in the image of the contextual entropy mapping was false) because unlike the eigenvalues of $\rho$, these diagonal entries must contain more than three non-zero
elements, so the \textbf{(c3)} branch is not accessible this time. Moreover, the $j_1^{th}$ entry on the diagonal of $U^{-1}\rho U$ will be the same as the $j_1^{th}$ eigenvalue of $\rho$, which allows us to choose between $\rho_1$ and $\rho_2$.

\noindent\textbf{Step 3.} We have now reached the end of our algorithm. If it has failed to retrieve a solution, we conclude that we have considered a global section $\gamma$ which was not in the image of
the contextual entropy mapping. Otherwise, our reconstructed state is
$$\rho=\lambda_1 P_1+\ldots+\lambda_n P_n$$
In order to obtain $\rho$ we have taken into account only a finite number of contexts, and it might happen that when all contexts are taken into account $E(\mu_\rho)\neq\gamma$. In this case we also conclude that $\gamma$ was not in the image of the contextual entropy mapping, and discard the state $\rho$.

\subsection{Two-dimensional Hilbert spaces}\label{r3}

For two dimensional Hilbert spaces the contextual entropy is a two-to-one mapping. We will justify this statement below.

First, it is easy to check that for any one dimensional projection $P$ the states $\rho_1=\lambda P+(1-\lambda)(I-P)$ and $\rho_2=(1-\lambda)P+\lambda (I-P)$ are mapped to the same global section of $\underline{[0,\ln 2]}^\preceq$: note that $\rho_1=I-\rho_2$. Hence for every context $W=\{Q,I-Q\}''$
$$E(\mu_{\rho_1})|_{_W}=\mathrm{Sh}(\mathrm{Tr}\rho_1 Q, 1-\mathrm{Tr}\rho_1 Q)$$
while
\begin{align*}
E(\mu_{I-\rho_1})|_{_W}&=\mathrm{Sh}(\ \mathrm{Tr}(I-\rho_1) Q , \  1-\mathrm{Tr}(I-\rho_1) Q)\\
&=\mathrm{Sh}(\ \mathrm{Tr}(I-\rho_1)(I- Q), \ 1-\mathrm{Tr}(I-\rho_1)(I-Q)) 
\end{align*}
And since every one dimensional projection $Q$ has trace equal to unity,
$$\mathrm{Tr}(I-\rho_1)(I- Q)=\mathrm{Tr}I-\rho_1 - Q+\rho_1Q=\mathrm{Tr}\rho_1 Q$$
and so also $E(\mu_{\rho_1})|_{_W}=E(\mu_{I-\rho_1})|_{_W}$.

On the other hand, given a global section of $\underline{[0,\ln 2]}^\preceq$, the poset $\mathcal{V}(M_2)$ consists only of two-dimensional subalgebras. We can identify a context
$V=\{P, 1-P\}''$ for which $\gamma|_{_V}$ is minimal, and solve the equation $\mathrm{Sh}(x,1-x)=\gamma|_{_V}$ to find the eigenvalues of $\rho$. Since we have no further
information available, we cannot say which eigenvalue corresponds to which of the two projections generating $V$.

Note however that we are not far from reconstructing $\rho$: we would need to encode only one extra bit of information in order to fully reconstruct a two-dimensional quantum
state.

\section{Renyi entropy}\label{other}

We have seen how Shannon entropy can be encoded in the topos approach, and how one can afterwards retrieve its quantum analogue, the von Neumann entropy. It is natural to ask at this point whether a similar encoding can be found for other classical entropies, and whether such an encoding would still enable us to retrieve their quantum analogues. We will look here at Renyi entropies, and show that it is possible to obtain their topos theoretic equivalent.

We first define contextual Renyi entropy locally as
$$\mathcal{R}_q(\mu)_{_V}=R_q(\mu|_{_V}(P_1),\ldots,\mu|_{_V}(P_n), \  \forall V=\{P_1,\ldots,P_n\}''$$
Of course, we would like these local components to fit together nicely as before, and to form a global section of some presheaf of real numbers. For Shannon entropy, the fact that a global section could be formed was a consequence of the recursion property. Renyi entropies are in general not recursive, but they do satisfy a property which we shall call \textit{weak recursivity}, and we shall see that this is enough for our purposes. 

\begin{definition}
Let $\mathrm{S}$ be some function defined on the set of all probability distributions. If we coarse grain a probability distribution $(x_1,\ldots,x_n)$ by not distinguishing between all the outcomes, we obtain a new probability distribution with components
 $$p_1=\sum_{i=1}^{k_1} x_i, \  \ldots, \ p_r=\sum_{i=k_{r-1}+1}^{k_r} x_i$$
for some $0<k_1<k_2<\ldots<k_r=n$. We say that $\mathrm{S}$ is \textbf{weakly recursive} if
$$\mathrm{S}(x_1,\ldots,x_n)\geq \mathrm{S}(p_1,\ldots,p_r)$$
\end{definition}

One can easily check that Renyi entropies indeed satisfy this property, and hence for any two contexts $V'\supseteq V$
$$\mathcal{R}_q(\mu)_{_{V'}}\geq\mathcal{R}_q(\mu)_{_V}, \ \forall \mu\in\mathcal{M}(\underline{\Sigma})$$
This means it is possible to define contextual Renyi entropy as a mapping
$$\mathcal{R}_q:\mathcal{M}(\underline{\Sigma})\longrightarrow \Gamma \underline{[0,\ln n]^\preceq}$$

Since Renyi entropies are Schur concave, their quantum counterparts can be retrieved from the contextual Renyi entropies by finding the minimum over the set of values assigned to all maximal contexts. This is justified by the Schur-Horn lemma and similar arguments to those that were already used in Section \ref{VNeu}.

\subsection{Properties}

We will now briefly discuss some of the properties of Renyi entropies and their contextual analogues. 

\noindent \textbf{Concavity}

We saw in Section \ref{cns} that the global concavity of the contextual entropy was expressed as the concavity of each of its local components, and hence it was a direct consequence of the concavity property of Shannon entropy. Renyi entropies however are only concave for $0<q\leq 1$. In fact, it is known that concavity is lost for $q>q_*>1$, where $q_*$ depends on the dimension of the probability distribution. Concavity of the contextual Renyi entropies is then going to hold under the same conditions.

\noindent\textbf{Additivity and Subadditivity}

Renyi entropies are additive, so we can use the same justification as in Section \ref{cns} to define subadditivity for contextual Renyi entropies as the following condition:
$$\mathcal{R}_q(\mu_\rho)_{_V}\leq\mathcal{R}_q(\mu_{\rho_1\otimes\rho_2})_{_V}, \ \forall V\in \mathcal{V}(\mathcal{B}(H))$$
This allows us to say that contextual Renyi entropies are also additive. On the other hand, since neither classical nor quantum Renyi entropies are subadditive (except for $q=0$ and $q=1$),  contextual Renyi entropy also doesn't have this property. 

\noindent\textbf{Monotonicity}

Classical Renyi entropies are monotone, just like Shannon entropy. However, their quantum counterparts do not necessarily have this property since, just like von Neumann entropy, they assign the value $0$ to pure states and it is possible to have a composite system in a pure state and both its components in mixed states. For contextual Renyi entropies, this means that we will not have a global property of monotonicity. Instead this property is going to hold only for split contexts. 

\noindent \textbf{State reconstruction}

Finally, recall that the reconstruction algorithms described in Sections \ref{r1}-\ref{r3} relied on Gleason's theorem, the Schur-Horn lemma, and two extra ingredients: one was the fact that von Neumann entropy vanished only for pure states, and the second was the fact that for probability distributions with two variables one could find precisely two solutions (symmetric around $1/2$) for which Shannon entropy would take any given value within its image. Both of these ingredients are present when we consider Renyi entropy, for positive parameter $q0$, as Figure \ref{Renyi} clearly indicates. This means that the reconstruction algorithms can also be applied to contextual Renyi entropies, with the exception of $\mathcal{R}_0$.

\newpage
$ $
\newpage

\begin{chapter}{Conclusions and Outlook}\label{gls}

In this thesis we have investigated two distinct but related research strands involving sheaf and category theory, within the fields of quantum information theory and foundations of physics. In the following paragraphs we will outline the key results of this work and discuss its potentially interesting extensions to future research.

In the first half of the thesis we have surveyed the landscape of multipartite entangled states through the lens of the contextuality classification developed by Abramsky and Brandenburger in \cite{AbrBra11}. 

In Chapter \ref{CC} we have made a complete classification of a particular class of non-symmetric multipartite entangled states, the balanced states with functional dependencies. We have showed that each of the three strengths of contextuality distinguished by Abramsky and Brandenburger is exhibited by a certain subclass of these functionally dependent multipartite quantum states. This classification has recently been shown to relate to the study of violations of local realism in quantum hypergraph states conducted by G\"{u}hne et al.~in \cite{Mari}.

One potential future use of the results of this investigation could fructify the fact that the study carried out by Abramsky and Brandenburger has been conducted at a very high level of generality, without any presupposition of quantum mechanics. Although their methods are readily applicable in quantum mechanical scenarios, there is also the possibility of making further connections between the study of contextuality and non-locality in physics and ideas arising in other fields. Since the functional dependency in the states considered in Chapter \ref{CC} is given by Boolean functions, it would be interesting to see if our classification could yield any insights within other areas of research where Boolean functions play a prominent role.

We have also noted in Sections \ref{dict} and \ref{dictatorships} that the states belonging to the subclass with the weakest degree of contextuality, the dictatorship states, are closely related to maximally entangled two-qubit states, being equivalent to tensor products of pure states and maximally entangled two-qubit states. We have shown that such states can not exhibit a stronger degree of contextuality if the parties which share such a state each have a choice between the same two local observables. 

This observation naturally leads to the question of quantifying the maximum degree of contextuality which can be exhibited by a quantum state for some choice of local observables. Given the current scarcity of results related to the classification and quantification of multipartite entanglement, the perspective offered by the question of the strongest degree of contextuality which a state may achieve could become very useful, particularly in the light of recent results \cite{Raussendorf,howard} showing that contextuality is a key feature enabling quantum computation. 

A first step in this direction has already been taken in Chapter \ref{LNL}, where we demonstrate that all $n$-qubit quantum states, with the exception of tensor products of maximally entangled bipartite states and single quantum states, are at least logically contextual. That is to say, they occupy at least the middle level of the three-level hierarchy of contextuality. This implies that the probability- and inequality-free logical formulation of contextuality and non-locality specific to logically, and in particular also to strongly contextual quantum states, is not a rare occurrence, but in fact arises for almost all states. 

If we further consider the notion of strong contextuality, a natural next challenge follows from the results presented in the first half of this thesis. Namely, to characterise those quantum states for which local observables giving rise to a strongly contextual empirical model can be found.

This question remains open, and appears difficult. It has an interesting relation to the question of ``All-versus-Nothing'' arguments first used by Mermin \cite{mermin1990simple}, which have recently been studied in the sheaf-theoretic approach by Abramsky et al. \cite{abramsky2015contextuality}, and shown to be related to the cohomological witnesses for contextuality previously introduced by Abramsky, Mansfield and Soares Barbossa in \cite{AbrShaneRui11}. All currently known examples of strong contextuality arising in quantum mechanics come from All-versus-Nothing arguments. Determining whether this is true in general is another challenging problem, which may hold the key to the main question.

The qualitative arguments of Chapters \ref{CC} and \ref{LNL} can easily be turned into quantitative ones. Abramsky and Hardy have shown in \cite{AbrHar12} how any instance of logical contextuality gives rise to a Bell inequality based on logical consistency conditions, which allows for  quantitative, robust experimental tests. While the structure of the argument in Chapter \ref{LNL} can be easily modified to yield a lower bound for the violations of the logical Bell inequalities corresponding to each logically non-local quantum state, it would also be interesting to seek to obtain concrete violations of Bell inequalities for each particular type among the functionally dependent entangled states. 

In the second half of our thesis, we have explored two strands of research within the Topos Approach to the formulation of physical theories.

First, in Chapter \ref{LA}, we have showed that it is possible to give a fairly direct reconstruction of an atomic orthomodular lattice from its associated poset of distributive sub-lattices. Harding and Navara had shown in \cite{navara} that an isomorphism of the posets of distributive sub-lattices of two orthomodular latices implies the existence of an isomorphism between the respective orthomodular lattices themselves. This result has an important implication, as far as the Topos Approach is concerned, since it establishes that the orthomodular lattice $\PH$, which is traditionally used in quantum logic, can be reconstructed from the poset of contexts $\mc B(\PH)$ that underlies the constructions in the Topos Approach. In other words, this means that the new form of presheaf- and topos-based form of logic for quantum systems is (at least) as rich as traditional quantum logic.

Our reconstruction result has provided a partial answer to one of the open questions posed by Harding and Navara towards the end of their paper. In future research it would be interesting to investigate whether the reconstruction can be extended to the case of non-atomic orthomodular lattices in order to provide a complete answer to Harding and Navara's question. 

Finally, in Chapter \ref{CE}, we have considered the information theoretic concept of entropy from the perspective of the Topos Approach. This new perspective allowed us to treat classical and quantum notions of entropy, such as Shannon and von Neumann entropies, or Renyi entropies, in a unified setting, via the innovative construction of contextual entropy. Within this construction a classical Shannon entropy is associated to each commuting subalgebra of observables of the non-commutative algebra corresponding to a given quantum system. This assignment of Shannon entropies is based on the reformulation of quantum states within the Topos Approach described by D\"{o}ring in \cite{ms}. Given a quantum state, the classical Shannon entropies which build up the contextual entropy map are the entropies of the quantum state seen through the `classical windows' given by the commuting subalgebras. We have further analysed how the state's von Neumann entropy, which is the quantum counterpart of Shannon entropy, is associated to a distinguished maximal commuting subalgebra. In Section \ref{other} we have showed that the same principles can be applied to construct the Renyi contextual entropy, which similarly treats classical and quantum Renyi entropies in a unified setting.

Perhaps the most striking feature of the contextual entropy map is the fact that it is rich enough to allow for a reconstruction algorithm which takes in a contextual entropy map and outputs a quantum state. This implies that contextual entropy can be seen as a mathematical equivalent of quantum states based on the information theoretical concept of entropy. As such it is a step towards an information-theoretic characterisation of quantum states. 

In future work we propose checking whether the concept of contextual entropy could also be defined for infinte dimensional quantum systems. Since both Gleason's theorem and the Schur-Horn lemma, which the reconstruction algorithm essentially relies on, can be used in infinite dimensions (as proved by Kaftal and Weiss in \cite{schur}), it would be particularly interesting to see if our reconstruction results can also be applied in the infinite dimensional case.

Finding an axiomatic characterisation of the contextual entropy map would allow us to relate our reconstruction results to a generalized version of Gleason's theorem. This theorem states that every finitely additive probability measure $m$ on the projections of a von Neumann algebra with no type $I_2$ summand can be uniquely extended to a state on that algebra. 

Since every $m$ as above uniquely determines a measure $\mu$ on the spectral presheaf, we could easily construct its associated contextual entropy map $E_\mu$. Using our reconstruction algorithm, we could in principle retrieve the unique quantum state $\rho$ associated to $E_\mu$, and hence also the probability measure $\m$. However, the caveat is that in our algorithm we had to assume in the first place that we had started from a probability measure $\mu$ on projections. Therefore having an \emph{axiomatic} characterisation of those real-valued maps on contexts which are contextual entropy maps (and hence come from quantum states) would allow us to reconstruct quantum states directly. This would be an important step towards giving a new structural proof of Gleason's theorem, and as such it promises to be both an interesting and non-trivial task.

Last, but not least, perhaps the most advanced and important application of this topos-theoretic construction during this exciting era of Quantum Information Theory would be to obtain new quantitative and qualitative insights into the nature of multipartite entanglement starting from the rich structure of the contextual entropy map.

\end{chapter}


\begin{appendices}
\newpage
$ $
\newpage
\chapter{The Subadditivity Property of Contextual Entropy}\label{AppendixA}

In this Appendix we seek to find general examples of density matrices and entangled contexts, for which the subadditivity property of contextual entropy fails to hold.

Let us assume we are looking at a composite system of dimension $n=n_1n_2$, with $n_1,n_2\geq 2$. Consider the following diagonal density matrix:

$$\rho=\left(\begin{array}{cccc}
	\frac{1}{2}&&\\
	&\mathbf{0}_{n-2}&\\
	&&\frac{1}{2}
	\end{array}\right)
$$
Its partial traces will  be of the same form, and their tensor product will also be a diagonal matrix:

$$\rho_1\otimes\rho_2=\left(\begin{array}{ccccccc}
	\frac{1}{4}&&&&&&\\
	&\mathbf{0}_{n_2-2}&&&&&\\
	&&\frac{1}{4}&&&&\\
	&&&\mathbf{0}_{(n_1-2)n_2}&&&\\
	&&&&\frac{1}{4}&&\\
	&&&&&\mathbf{0}_{n_2-2}&\\
	&&&&&&\frac{1}{4}
	\end{array}\right)
$$
If $U=(U_{i,j})$ is a $n$-dimensional unitary matrix, then we are only interested in the diagonal entries of $U\rho U^{-1}$ and of $U\rho_1\otimes\rho_2 U^{-1}$ and these can be easily calculated:

$$U\rho U^{-1}=\frac{1}{2}\left(\begin{array}{cccc}
	|U_{1,1}|^2+ |U_{1,n}|^2&&&\\
	&\ddots&&\\
	&&|U_{n-1,1}|^2+|U_{n-1,n}|^2&\\
	&&&|U_{n,1}|^2+|U_{n,n}|^2
	\end{array}\right)
$$
while
$$U\rho_1\otimes\rho_2 U^{-1}=\frac{1}{4}\left(\begin{array}{l}
	|U_{1,1}|^2+ |U_{1,n}|^2 + |U_{1,n_2}|^2+ |U_{1,n-n_2}|^2\\
	  \ \ \ \ \ \  \ddots\\
    	  \ \ \ \ \ \ \ \ \ \ \ \ \ |U_{n-1,1}|^2+|U_{n-1,n}|^2+ |U_{n-1,n_2}|^2+ |U_{n-1,n-n_2}|^2\\
	  \ \ \ \ \ \ \ \ \ \ \ \ \ \ \ \ \ \ \ \ \ \ \ \ \ \ \ \ \   |U_{n,1}|^2+|U_{n,n}|^2+ |U_{n,n_2}|^2+ |U_{n,n-n_2}|^2
	\end{array}\right)
$$

We want to calculate
\begin{align*}
A&=E(\mu_\rho)|_{_{U^{-1}\cdot D_{n_1}\otimes D_{n_2}\cdot U}}=E(\mu_\rho)|_{_{U^{-1}\cdot D_n\cdot U}}\\
&=E(\mu_{U\rho U^{-1}})|_{_{D_n}}=\mathrm{Sh}\left((U\rho U^{-1})_{11},\ldots,(U\rho U^{-1})_{nn}\right)
\end{align*}
and show that it is strictly greater than
$$B=E(\mu_{\rho_1\otimes\rho_2})|_{_{U^{-1}\cdot D_{n_1}\otimes D_{n_2} \cdot U}}=\mathrm{Sh}\left((U\rho_1\otimes\rho_2 U^{-1})_{11},\ldots,(U\rho_1\otimes\rho_2 U^{-1})_{nn}\right)$$

In order to compute these we need to specify only four columns of the unitary matrix $U$, namely columns $1$, $n_2$, $n-n_2$ and $n$. That is, we must specify four orthonormal vectors of dimension $n$.  We first look at the situation when $n$ is even and greater or equal to six. In this case the four vectors can be taken as follows:

\begin{align*}
U_1^T&= \frac{1}{\sqrt n} (1,1,\ldots,1,1,1,\ldots,1) \\
U_n^T&= \frac{1}{\sqrt n}\left(\underbrace{1,1,\ldots,1}_{n/2},\underbrace{ -1,-1,\ldots,-1}_{n/2}\right)\\
U_{n_2}^T&=\frac{1}{2}\left(-1,1,\underbrace{0,\ldots,0}_{n-4},-1,1\right)\\
U_{n-n_2}^T&=\frac{1}{2} \left(1,-1,\underbrace{0,\ldots,0}_{n-4},-1,1\right)
\end{align*}
and so 
$$A=\mathrm{Sh}(1/n,\ldots,1/n)=\ln n$$
which is strictly greater than
$$B=\mathrm{Sh}\left(\frac{1}{2n}+\frac{1}{8}, \frac{1}{2n}+\frac{1}{8}, \underbrace{\frac{1}{2n},\ldots,\frac{1}{2n}}_{n-4},\frac{1}{2n}+ \frac{1}{8}, \frac{1}{2n}+\frac{1}{8}\right)$$
as desired. For $n=4$ we can see that $\frac{1}{2n}+\frac{1}{8}=1/n$, and this approach will fail to produce a counterexample. However, by randomly generating four-dimensional density matrices and unitaries on a computer, it is easy to come across an example for which subadditivity fails. For instance if we consider the state
$$\rho = \left(\begin{array}{cccc}
  0.089            &-0.107 - 0.038i  & 0.070 - 0.009i   &0.116 - 0.056i\\
 -0.107 + 0.038i  & 0.328            &-0.150 - 0.048i  &-0.226 + 0.053i\\
  0.070 + 0.009i  &-0.150 + 0.048i  & 0.205            & 0.117 - 0.073i\\
  0.116 + 0.056i  &-0.226 - 0.053i  & 0.117 + 0.073i  & 0.376          
               \end{array}\right)$$
and the unitary matrix

$$U = \left(\begin{array}{cccc}

  0.662 + 0.163i   &0.027 + 0.130i   &0.411 - 0.532i   &0.234 - 0.098i\\
 -0.250 - 0.232i  & 0.459 - 0.526i   &0.290 - 0.229i   &0.261 + 0.435i\\
 -0.011 + 0.106i  & 0.490 - 0.268i   &0.147 + 0.242i   &0.017 - 0.772i\\
  0.601 - 0.213i & -0.161 - 0.394i  &-0.117 + 0.562i   &0.238 + 0.155i
      \end{array}\right)$$
it is only a matter of straightforward computations to check that 

$$E(\mu_\rho)|_{_{U^{-1}\cdot D_4\cdot U}}=E(\mu_{U\rho U^{-1}})|_{_ {D_4}}>E(\mu_{U\rho_1\otimes\rho_2 U^{-1}})|_{_ {D_4}}=E(\mu_{\rho_1\otimes\rho_2})|_{_{U^{-1}\cdot D_4\cdot U}}$$
and hence the subadditivity property does not hold at the entangled context $U^{-1}\cdot D_4\cdot U$.











The last case to consider is the one when the dimension of our system is odd. Then we must have both $n_1$ and $n_2$ odd and greater or equal to three, so $n\geq 9$. In this case the four vectors can be taken as follows:
\begin{align*}
U_1^T&= \frac{1}{\sqrt n} \left(1,1,\ldots,1,1,1,\ldots,1,3,0,0\right) \\
U_n^T&= \frac{1}{\sqrt n}\left(\underbrace{1,1,\ldots,1}_{(n-3)/2},\underbrace{ -1,-1,\ldots,-1}_{(n-3)/2},0,3,0\right)\\
U_{n_2}^T&=\frac{1}{2} \left(-1,1,\underbrace{0,\ldots,0}_{n-7},-1,1,0,0,0\right)\\
U_{n-n_2}^T&=\frac{1}{2} \left(1,-1,\underbrace{0,\ldots,0}_{n-7},-1,1,0,0,0\right)
\end{align*}
in which case
$$A=\mathrm{Sh}\left(1/n,\ldots,1/n,\frac{3}{2n},\frac{3}{2n},0\right)$$
which is less than $\ln n$, but it is still greater than
$$B=\mathrm{Sh}\left(\frac{1}{2n}+\frac{1}{8}, \frac{1}{2n}+\frac{1}{8}, \underbrace{\frac{1}{2n},\ldots, \frac{1}{2n}}_{n-7}, \frac{1}{2n}+\frac{1}{8},\frac{1}{2n}+\frac{1}{8}, \frac{3}{4n},\frac{3}{4n},0\right)$$

\end{appendices}

\nocite{*}

\bibliographystyle{dashamsplain}
\bibliography{newreferences}

\providecommand{\bysame}{\leavevmode\hbox to3em{\hrulefill}\thinspace}
\providecommand{\MR}{\relax\ifhmode\unskip\space\fi MR }
\providecommand{\MRhref}[2]{%
  \href{http://www.ams.org/mathscinet-getitem?mr=#1}{#2}
}
\providecommand{\href}[2]{#2}
\begin{thebibliography}{100}

\bibitem{abramsky2015contextuality}
Samson Abramsky, Rui~Soares Barbosa, Kohei Kishida, Raymond Lal, and Shane
  Mansfield, \emph{{Contextuality, Cohomology and Paradox}},  (2015),
  \href{http://arXiv.org/abs/arXiv:1502.03097}{{\tt arXiv:1502.03097}}.

\bibitem{AbrBra11}
Samson Abramsky and Adam Brandenburger, \emph{The sheaf-theoretic structure of
  non-locality and contextuality}, New Journal of Physics \textbf{13} (2011),
  \href{http://arXiv.org/abs/arXiv:1102.0264}{{\tt arXiv:1102.0264}}.

\bibitem{AbrBra14}
Samson Abramsky and Adam Brandenburger, \emph{An operational interpretation of
  negative probabilities and no-signalling models},  (2014),
  \href{http://arXiv.org/abs/arXiv:1401.2561}{{\tt arXiv:1401.2561}}.

\bibitem{abramsky2008categorical}
Samson Abramsky and Bob Coecke, \emph{Categorical quantum mechanics}, {Handbook
  of Quantum Logic and Quantum Structures: Quantum Logic} (Kurt Engesser, Dov~M
  Gabbay, and Daniel Lehmann, eds.), Elsevier, 2008, pp.~261--324.

\bibitem{AbrCon}
Samson Abramsky and Carmen~M. Constantin, \emph{A classification of
  multipartite states by degree of non-locality}, Electronic Proceedings in
  Theoretical Computer Science (2014),
  \href{http://arXiv.org/abs/arXiv:1412.5213}{{\tt arXiv:1412.5213}}.

\bibitem{AbrConYing}
Samson Abramsky, Carmen~M. Constantin, and Shenggang Ying, \emph{Hardy is
  (almost) everywhere: nonlocality without inequalities for almost all
  entangled multipartite states}, Electronic Proceedings in Theoretical
  Computer Science (2015), \href{http://arXiv.org/abs/arXiv:1506.01365}{{\tt
  arXiv:1506.01365}}.

\bibitem{AbrHar12}
Samson Abramsky and Lucien Hardy, \emph{Logical {B}ell inequalities}, Physical
  Review A \textbf{85} (2012), \href{http://arXiv.org/abs/arXiv:1203.1352}{{\tt
  arXiv:1203.1352}}.

\bibitem{AbrShaneRui11}
Samson Abramsky, Shane Mansfield, and Rui Soares~Barbosa, \emph{The cohomology
  of non-locality and contextuality}, EPTCS \textbf{95} (2012), 1--14,
  \href{http://arXiv.org/abs/arXiv:1111.3620}{{\tt arXiv:1111.3620}}.

\bibitem{traian2}
Traian~E. Abrudan, Jan Eriksson, and Visa Koivunen, \emph{Steepest descent
  algorithm for optimization under unitary matrix constraint}, IEEE
  Transactions on Signal Processing \textbf{56} (2008), no.~3, 1134.

\bibitem{traian1}
Traian~E. Abrudan, Jan Eriksson, and Visa Koivunen, \emph{Conjugate gradient
  algorithm for optimization under unitary matrix constraint}, Signal
  Processing \textbf{89} (2009), no.~3, 1704.

\bibitem{AciGisMas}
Antonio Acin, Nicolas Gisin, and Lluis Masanes, \emph{From {B}ell's theorem to
  secure quantum key distribution}, Physical Review Letters \textbf{97} (2006),
  120405, \href{http://arXiv.org/abs/arXiv:quant-ph/0510094}{{\tt
  arXiv:quant-ph/0510094}}.

\bibitem{Ake}
Charles~A. Akemann, \emph{Left ideal structure of {C}$^*$-algebras}, Journal of
  Functional Analysis \textbf{6} (1970), 305.

\bibitem{barnum}
Howard Barnum, Jonathan Barrett, Lisa~O. Clark, Matthew Leifer, Robert
  Spekkens, Nicholas Stepanik, Alex Wilce, and Robin Wilke, \emph{Entropy and
  information causality in general probabilistic theories}, New Journal of
  Physics \textbf{12} (2010), 033024,
  \href{http://arXiv.org/abs/arXiv:0909.5075}{{\tt arXiv:0909.5075}}.

\bibitem{BarHarKen}
Jonathan Barrett, Lucien Hardy, and Adrian Kent, \emph{No signaling and quantum
  key distribution}, Physical Review Letters \textbf{95} (2005), 10503,
  \href{http://arXiv.org/abs/arXiv:quant-ph/0405101}{{\tt
  arXiv:quant-ph/0405101}}.

\bibitem{Bell:1964kc}
John~S. Bell, \emph{On the {E}instein-{P}odolsky-{R}osen paradox}, Physics
  \textbf{1} (1964), 195.

\bibitem{BenZyc06}
Ingemar Bengtsson and Karol Zyczkowski, \emph{Geometry of quantum states: An
  introduction to quantum entanglement}, Cambridge University Press, 2006.

\bibitem{BenBra}
Charles~H. Bennett and Gilles Brassard, \emph{Quantum cryptography: Public key
  distribution and coin tossing}, Proceedings of IEEE International Conference
  on Computers, Systems, and Signal Processing (1984), 175--179,
  \href{http://arXiv.org/abs/arXiv:1102.0264}{{\tt arXiv:1102.0264}}.

\bibitem{Bennett:1992tv}
Charles~H. Bennett, Gilles Brassard, Claude Crepeau, Richard Jozsa, Asher
  Peres, and William~K. Wootters, \emph{Teleporting an unknown quantum state
  via dual classical and {E}instein-{P}odolsky-{R}osen channels}, Physical
  Review Letters \textbf{70} (1993), 1895--1899.

\bibitem{Bennett:1992zzb}
Charles~H. Bennett and Stephen~J. Wiesner, \emph{Communication via one- and
  two-particle operators on {E}instein-{P}odolsky-{R}osen states}, Physical
  Review Letters \textbf{69} (1992), 2881--2884.

\bibitem{BirkVonN}
Garrett Birkhoff and John von Neumann, \emph{The logic of quantum mechanics},
  Annals of Mathematics \textbf{37} (1936), no.~4, 823--843.

\bibitem{blum1989theory}
Lenore Blum, Mike Shub, and Steve Smale, \emph{{On a theory of computation and
  complexity over the real numbers: NP-completeness, recursive functions and
  universal machines}}, Bulletin (New Series) of the American Mathematical
  Society \textbf{21} (1989), no.~1, 1--46.

\bibitem{Bor}
Francis Borceux and Gilberte van~den Bossche, \emph{An essay on noncommutative
  topology}, Topology and its Applications \textbf{31} (1989), 202.

\bibitem{1}
Gilles Brassard, \emph{Is information the key?}, Nature Physics \textbf{1}
  (2005), no.~2.

\bibitem{Brukner}
Caslav Brukner, Marek Zukowski, Jian-Wei Pan, and Anton Zeilinger,
  \emph{{B}ell's inequalities and quantum communication complexity}, Physical
  Review Letters \textbf{92} (2004), 127901.

\bibitem{5}
Jeffrey Bub, \emph{Quantum mechanics is about quantum information},  (2004),
  \href{http://arXiv.org/abs/arXiv:quant-ph/0408020v2}{{\tt
  arXiv:quant-ph/0408020v2}}.

\bibitem{3}
Jeffrey Bub, Rob Clifton, and Hans Halvorson, \emph{Characterizing quantum
  theory in terms of information-theoretic constraints}, Foundations of Physics
  \textbf{33} (2003), 1561,
  \href{http://arXiv.org/abs/arXiv:quant-ph/0211089}{{\tt
  arXiv:quant-ph/0211089}}.

\bibitem{Mari}
Constantino Budroni, Mariami Gachechiladze, and Otfried G{\" u}hne,
  \emph{Extreme violation of local realism in quantum hypergraph states},
  (2015), \href{http://arXiv.org/abs/arXiv:1507.03570}{{\tt arXiv:1507.03570}}.

\bibitem{b1}
Jeremy Butterfield and Chris~J. Isham, \emph{A topos perspective on the
  {K}ochen-{S}pecker theorem: {I}. {Q}uantum states as generalized
  valuations.}, International Journal of Theoretical Physics \textbf{37}
  (1998), 2669, \href{http://arXiv.org/abs/arXiv:quant-ph/9803055}{{\tt
  arXiv:quant-ph/9803055}}.

\bibitem{b2}
Jeremy Butterfield and Chris~J. Isham, \emph{A topos perspective on the
  {K}ochen-{S}pecker theorem: {II}. {C}onceptual aspects, and classical
  analogues.}, International Journal of Theoretical Physics \textbf{38} (1999),
  827, \href{http://arXiv.org/abs/arXiv:quant-ph/9808067}{{\tt
  arXiv:quant-ph/9808067}}.

\bibitem{b3}
Jeremy Butterfield and Chris~J. Isham, \emph{A topos perspective on the
  {K}ochen-{S}pecker theorem: {III}. {V}on {N}eumann algebras as the base
  category.}, International Journal of Theoretical Physics \textbf{39} (2000),
  1413, \href{http://arXiv.org/abs/arXiv:quant-ph/9911020}{{\tt
  arXiv:quant-ph/9911020}}.

\bibitem{b4}
Jeremy Butterfield and Chris~J. Isham, \emph{A topos perspective on the
  {K}ochen-{S}pecker theorem: {IV}. {I}nterval valuations.}, International
  Journal of Theoretical Physics \textbf{41} (2002), 613,
  \href{http://arXiv.org/abs/arXiv:quant-ph/0107123}{{\tt
  arXiv:quant-ph/0107123}}.

\bibitem{CSW}
Adan Cabello, Simone Severini, and Andreas Winter, \emph{({N}on-)contextuality
  of physical theories as an axiom},  (2010),
  \href{http://arXiv.org/abs/arXiv:1010.2163}{{\tt arXiv:1010.2163}}.

\bibitem{DCG02}
Maria L.~Dalla Chiara and Roberto Giuntini, \emph{Quantum logics}, Handbook of
  Philosophical Logic (Dov~M. Gabbay and Franz Guenthner, eds.), vol.~6, Kluwer
  Academic Publishers, 2002, p.~129.

\bibitem{Clauser:1974tg}
John~F. Clauser and Michael~A. Horne, \emph{Experimental consequences of
  objective local theories}, Physical Review \textbf{D10} (1974), 526.

\bibitem{Clauser:1969ny}
John~F. Clauser, Michael~A. Horne, Abner Shimony, and Richard~A. Holt,
  \emph{Proposed experiment to test local hidden variable theories}, Physical
  Review Letters \textbf{23} (1969), 880--884.

\bibitem{ConDoe}
Carmen~M. Constantin and Andreas D{\" o}ring, \emph{Reconstructing an atomic
  orthomodular lattice from the poset of its {B}oolean sublattices}, Houston
  Journal of Mathematics (2014),
  \href{http://arXiv.org/abs/arXiv:1306.1950}{{\tt arXiv:1306.1950}}.

\bibitem{CovTho91}
Thomas~M. Cover and Joy~A. Thomas, \emph{Elements of information theory}, Wiley
  Series in Telecommunications and Signal Processing, 1991.

\bibitem{CraHam11}
Yves Crama and Peter~L. Hammer, \emph{{B}oolean functions: Theory, algorithms,
  and applications}, Encyclopedia of Mathematics and its Applications, vol.
  142, Cambridge University Press, 2011.

\bibitem{Groote}
Hans~F. de~Groote, \emph{Observables {I}: {S}tone spectra},  (2005),
  \href{http://arXiv.org/abs/arXiv:math-ph/0509020}{{\tt
  arXiv:math-ph/0509020}}.

\bibitem{Nadish}
Nadish de~Silva, \emph{From topology to noncommutative geometry: K-theory},
  (2014), \href{http://arXiv.org/abs/arXiv:1408.1170}{{\tt arXiv:1408.1170}}.

\bibitem{DHST03}
Inderjit~S. Dhillon, Robert~W. Heath, Matyas~A. Sustik, and Joel~A. Tropp,
  \emph{Generalized finite algorithms for constructing {H}ermitian matrices
  with prescribed diagonal and spectrum}, SIAM Journal on Matrix Analysis and
  Applications (2003).

\bibitem{dicke1954coherence}
Robert~H. Dicke, \emph{Coherence in spontaneous radiation processes}, Physical
  Review \textbf{93} (1954), no.~1, 99.

\bibitem{Dix69}
Jacques Dixmier, \emph{Les {C}$^*$-alg\`{e}bres et leurs repr\'{e}sentations},
  Gauthier-Villars, 1969.

\bibitem{ms}
Andreas D{\" o}ring, \emph{Quantum states and measures on the spectral
  presheaf}, Advanced Scientific Letters \textbf{2} (2009), no.~2, 291,
  \href{http://arXiv.org/abs/arXiv:0809.4847}{{\tt arXiv:0809.4847}}, Special
  Issue on `Quantum Gravity, Cosmology and Black Holes', ed. Martin Bojowald.

\bibitem{Doe09b}
Andreas D{\" o}ring, \emph{Topos theory and `neo-realist' quantum theory},
  Quantum Field Theory, Competitive Models (Bertfried Fauser, Jürgen Tolksdorf,
  and Eberhard Zeidler, eds.), Birkh{\" a}user, 2009.

\bibitem{Doe11b}
Andreas D{\" o}ring, \emph{The physical interpretation of daseinisation}, Deep
  Beauty (Hans Halvorson, ed.), Cambridge University Press, 2011, p.~207.

\bibitem{Doe11a}
Andreas D{\" o}ring, \emph{Topos quantum logic and mixed states}, Electronic
  Notes in Theoretical Computer Science \textbf{270} (2011), no.~2,
  \href{http://arXiv.org/abs/arXiv:1004.3561}{{\tt arXiv:1004.3561}}, in
  Proceedings of the 6th International Workshop on Quantum Physics and Logic
  (QPL 2009).

\bibitem{Doe12c}
Andreas D{\" o}ring, \emph{Flows on generalised {G}elfand spectra of nonabelian
  unital {C}$^*$-algebras and time evolution of quantum systems},  (2012),
  \href{http://arXiv.org/abs/arXiv:1212.4882}{{\tt arXiv:1212.4882}}.

\bibitem{Doe12e}
Andreas D{\" o}ring, \emph{Generalised {G}elfand spectra of nonabelian unital
  {C}$^*$-algebras},  (2012), \href{http://arXiv.org/abs/arXiv:1212.2613}{{\tt
  arXiv:1212.2613}}.

\bibitem{Doe13}
Andreas D{\" o}ring, \emph{Some remarks on the logic of quantum gravity},
  (2013), \href{http://arXiv.org/abs/arXiv:1306.3076}{{\tt arXiv:1306.3076}}.

\bibitem{Doe14}
Andreas D{\" o}ring, \emph{Two new complete invariants of von {N}eumann
  algebras},  (2014), \href{http://arXiv.org/abs/arXiv:1411.5558}{{\tt
  arXiv:1411.5558}}.

\bibitem{Doe12}
Andreas D{\" o}ring, \emph{Topos-based logic for quantum systems and
  bi-{H}eyting algebras}, Logic \& Algebra in Quantum Computing, Lecture Notes
  in Logic, Association for Symbolic Logic and Cambridge University Press,
  2015.

\bibitem{DB11}
Andreas D{\" o}ring and Rui~Soares Barbosa, \emph{Unsharp values, domains and
  topoi}, Quantum Field Theory and Gravity, Conceptual and Mathematical
  Advances in the Search for a Unified Framework (Felix Finster, Olaf Müller,
  Marc Nardmann, Jürgen Tolksdorf, and Eberhard Zeidler, eds.), Birkh{\"
  a}user, 2012, p.~65.

\bibitem{ConDoe1}
Andreas D{\" o}ring and Carmen~M. Constantin, \emph{Contextual entropy and
  reconstruction of quantum states},  (2012),
  \href{http://arXiv.org/abs/arXiv:1208.2046}{{\tt arXiv:1208.2046}}.

\bibitem{DoeDew2}
Andreas D{\" o}ring and Barry Dewitt, \emph{Self-adjoint operators as functions
  {II}: Quantum probability},  (2012),
  \href{http://arXiv.org/abs/arXiv:1210.5747v2}{{\tt arXiv:1210.5747v2}}.

\bibitem{DoeDew1}
Andreas D{\" o}ring and Barry Dewitt, \emph{Self-adjoint operators as functions
  {I}: Lattices, {G}alois connections, and the spectral order}, Communications
  in Mathematical Physics \textbf{328} (2014), no.~2, 499--525,
  \href{http://arXiv.org/abs/arXiv:1208.4724}{{\tt arXiv:1208.4724}}.

\bibitem{harding}
Andreas D{\" o}ring and John Harding, \emph{Abelian subalgebras and the
  {J}ordan structure of a von {N}eumann algebra}, Houston Journal of
  Mathematics (2013), to appear,
  \href{http://arXiv.org/abs/arXiv:1009.4945}{{\tt arXiv:1009.4945}}.

\bibitem{d1}
Andreas D{\" o}ring and Chris~J. Isham, \emph{A topos foundation for theories
  of physics: {I}. {F}ormal languages for physics}, Journal of Mathematical
  Physics \textbf{49} (2008), no.~053515,
  \href{http://arXiv.org/abs/arXiv:quant-ph/0703060}{{\tt
  arXiv:quant-ph/0703060}}.

\bibitem{d2}
Andreas D{\" o}ring and Chris~J. Isham, \emph{A topos foundation for theories
  of physics: {II}. {D}aseinisation and the liberation of quantum theory},
  Journal of Mathematical Physics \textbf{49} (2008), no.~053516,
  \href{http://arXiv.org/abs/arXiv:quant-ph/0703062}{{\tt
  arXiv:quant-ph/0703062}}.

\bibitem{d3}
Andreas D{\" o}ring and Chris~J. Isham, \emph{A topos foundation for theories
  of physics: {III}. {T}he representation of physical quantities with arrows},
  Journal of Mathematical Physics \textbf{49} (2008), no.~053517,
  \href{http://arXiv.org/abs/arXiv:quant-ph/0703064}{{\tt
  arXiv:quant-ph/0703064}}.

\bibitem{d4}
Andreas D{\" o}ring and Chris~J. Isham, \emph{A topos foundation for theories
  of physics: {IV}. {C}ategories of systems}, Journal of Mathematical Physics
  \textbf{49} (2008), no.~053518,
  \href{http://arXiv.org/abs/arXiv:quant-ph/0703066}{{\tt
  arXiv:quant-ph/0703066}}.

\bibitem{disham}
Andreas D{\" o}ring and Chris~J. Isham, \emph{`{W}hat is a thing?': Topos
  theory in the foundations of physics}, New Structures in Physics (Bob Coecke,
  ed.), Springer, 2008, p.~753.

\bibitem{DoeIsh12}
Andreas D{\" o}ring and Chris~J. Isham, \emph{Classical and quantum
  probabilities as truth values}, Journal of Mathematical Physics \textbf{53}
  (2012), 032101, \href{http://arXiv.org/abs/arXiv:1102.2213}{{\tt
  arXiv:1102.2213}}.

\bibitem{Dur:2000zz}
Wolfgang Dur, Guifre Vidal, and Ignacio~J. Cirac, \emph{Three qubits can be
  entangled in two inequivalent ways}, Physical Review \textbf{A62} (2000),
  062314, \href{http://arXiv.org/abs/arXiv:quant-ph/0005115}{{\tt
  arXiv:quant-ph/0005115}}.

\bibitem{Einstein}
Albert Einstein, \emph{{Quanten-Mechanik Und Wirklichkeit}}, Dialectica
  \textbf{2} (1948), 320--324.

\bibitem{Einstein:1935rr}
Albert Einstein, Boris Podolsky, and Nathan Rosen, \emph{Can quantum mechanical
  description of physical reality be considered complete?}, Physical Review
  \textbf{47} (1935), 777--780.

\bibitem{Ekert:1991zz}
Artur~K. Ekert, \emph{Quantum cryptography based on {B}ell's theorem}, Physical
  Review Letters \textbf{67} (1991), 661--663.

\bibitem{FRV}
Bertfried Fauser, Guillaume Raynaud, and Steven Vickers, \emph{The {B}orn rule
  as structure of spectral bundles (extended abstract)},  (2012),
  \href{http://arXiv.org/abs/arXiv:1210.0615}{{\tt arXiv:1210.0615}}.

\bibitem{fine1982hidden}
Arthur Fine, \emph{{Hidden variables, joint probability, and the {B}ell
  inequalities}}, Physical Review Letters \textbf{48} (1982), no.~5, 291--295.

\bibitem{Flo11}
Cecilia Flori, \emph{Group action in topos quantum physics},  (2011),
  \href{http://arXiv.org/abs/arXiv:1110.1650v3}{{\tt arXiv:1110.1650v3}}.

\bibitem{BreFlo12}
Cecilia Flori and Wilson Brenna, \emph{Complex numbers, one-parameter of
  unitary transformations and {S}tone's theorem in topos quantum theory},
  (2012), \href{http://arXiv.org/abs/arXiv:1206.0809v2}{{\tt
  arXiv:1206.0809v2}}.

\bibitem{BGeFlo12}
Cecilia Flori and Joseph~Ben Geloun, \emph{Topos analogues of the {KMS} state},
   (2012), \href{http://arXiv.org/abs/arXiv:1207.0227v2}{{\tt
  arXiv:1207.0227v2}}.

\bibitem{2}
Christopher~A. Fuchs, \emph{Quantum mechanics as quantum information (and only
  a little more)},  (2002),
  \href{http://arXiv.org/abs/arXiv:quant-ph/0205039v1}{{\tt
  arXiv:quant-ph/0205039v1}}.

\bibitem{GilKum}
Robin Giles and Hans Kummer, \emph{A non-commutative generalization of
  topology}, Indiana University Mathematics Journal \textbf{21} (1971), 91.

\bibitem{GivHal09}
Steven Givant and Paul Halmos, \emph{Introduction to {B}oolean algebras},
  Springer, 2009.

\bibitem{Gle57}
Andrew~M. Gleason, \emph{Measures on the closed subspaces of {H}ilbert space},
  Journal of Mathematics and Mechanics \textbf{6} (1957), 885.

\bibitem{Gol}
Robert Goldblatt, \emph{Topoi - the categorial analysis of logic}, Studies in
  Logic and the Foundations of Mathematics (J.~Barwise, D.~Kaplan, H.~J.
  Keisler, P.~Suppes, and A.~S. Troelstra, eds.), vol.~98, North-Holland
  Publishing Company, 1979.

\bibitem{GHZ90}
Daniel~M. Greenberger, Michael~A. Horne, Abner Shimony, and Anton Zeilinger,
  \emph{{{B}ell's theorem without inequalities}}, American Journal of Physics
  \textbf{58} (1990), 1131.

\bibitem{GHZ89}
Daniel~M Greenberger, Michael~A Horne, and Anton Zeilinger, \emph{{Going beyond
  {B}ell's theorem}}, {{B}ell's theorem, quantum theory and conceptions of the
  universe}, Springer, 1989, pp.~69--72.

\bibitem{navara}
John Harding and Mirko Navara, \emph{Subalgebras of orthomodular lattices},
  Order \textbf{3} (2011), 549,
  \href{http://arXiv.org/abs/arXiv:1009.4433}{{\tt arXiv:1009.4433}}.

\bibitem{HLP}
Godfrey~H. Hardy, John~E. Littlewood, and George Polya, \emph{Some simple
  inequalities satisfied by convex functions}, Messenger of Mathematics
  \textbf{58} (1929), 145--152.

\bibitem{Hardy92}
Lucien Hardy, \emph{Quantum mechanics, local realistic theories and
  {L}orentz-invariant realistic theories}, Physical Review Letters \textbf{68}
  (1992), 2981.

\bibitem{Hardy93}
Lucien Hardy, \emph{Nonlocality for two particles without inequalities for
  almost all entangled states}, Physical Review Letters \textbf{71} (1993),
  1665--1668.

\bibitem{HLS09a}
Chris Heunen, Nicolaas~P. Landsman, and Bas Spitters, \emph{A topos for
  algebraic quantum theory}, Communications in Mathematical Physics
  \textbf{291} (2009), 63, \href{http://arXiv.org/abs/arXiv:0709.4364}{{\tt
  arXiv:0709.4364}}.

\bibitem{HLS11}
Chris Heunen, Nicolaas~P. Landsman, and Bas Spitters, \emph{Bohrification},
  Deep Beauty (Hans Halvorson, ed.), Cambridge University Press, 2011, p.~271.

\bibitem{HLS09b}
Chris Heunen, Nicolaas~P. Landsman, and Bas Spitters, \emph{Bohrification of
  von {N}eumann algebras and quantum logic}, Synthese \textbf{186} (2012),
  no.~3, 719, \href{http://arXiv.org/abs/arXiv:0905.2275}{{\tt
  arXiv:0905.2275}}.

\bibitem{horn}
Alfred Horn, \emph{Doubly stochastic matrices and the diagonal of a rotation
  matrix}, American Journal of Mathematics \textbf{76} (1954), 620.

\bibitem{howard}
Mark Howard, Joel Wallman, Victor Veitch, and Joseph Emerson,
  \emph{Contextuality supplies the {`}magic{'} for quantum computation}, Nature
  \textbf{510} (2014), no.~7505, 351--355,
  \href{http://arXiv.org/abs/arXiv:1401.4174}{{\tt arXiv:1401.4174}}.

\bibitem{Ish97}
Chris~J. Isham, \emph{Topos theory and consistent histories: The internal logic
  of the set of all consistent sets}, International Journal of Theoretical
  Physics \textbf{36} (1997), 785,
  \href{http://arXiv.org/abs/arXiv:gr-qc/9607069}{{\tt arXiv:gr-qc/9607069}}.

\bibitem{Ist}
Carmen~M. Istrate, \emph{Quantales and locales - comparing two notions of
  spectra for non-commutative {C}$^*$-algebras}, Master's thesis, University of
  Oxford, 2010.

\bibitem{Joh1}
Peter~T. Johnstone, \emph{Sketches of an elephant - a topos theory compendium
  {I}}, Oxford Logic Guides (Dov~M. Gabay, Angus MacIntyre, and Dana Scott,
  eds.), vol.~44, Oxford University Press, 2002.

\bibitem{Joh2}
Peter~T. Johnstone, \emph{Sketches of an elephant - a topos theory compendium
  {II}}, Oxford Logic Guides (Dov~M. Gabay, Angus MacIntyre, and Dana Scott,
  eds.), vol.~44, Oxford University Press, 2003.

\bibitem{KadRin}
Richard~V. Kadison and John~R. Ringrose, \emph{Fundamentals of the theory of
  operator algebras, volume {I}: Elementary theory}, American Mathematical
  Society, 2002.

\bibitem{schur}
Victor Kaftal and Gary Weiss, \emph{An infinite dimensional {S}chur-{H}orn
  theorem and majorization theory with application to operator ideals},
  (2007), \href{http://arXiv.org/abs/arXiv:0710.5566v2}{{\tt
  arXiv:0710.5566v2}}.

\bibitem{klein}
Oskar~Benjamin Klein, \emph{Zur {Q}uantenmechanischen {B}egr{\" u}ndung des
  zweiten {H}auptsatzes der {W}{\" a}rmelehre}, Zeitschrift f{\" u}r Physik
  \textbf{72} (1931), no.~11, 767.

\bibitem{KocSpe}
Simon~B. Kochen and Ernst~P. Specker, \emph{The problem of hidden variables in
  quantum mechanics}, Journal of Mathematics and Mechanics \textbf{17} (1967),
  59.

\bibitem{33}
Vipin Kumar, \emph{Algorithms for constraint-satisfaction problems: A survey.},
  AI magazine \textbf{13} (1992), no.~1, 32--44.

\bibitem{levy2012basic}
Azriel Levy, \emph{{Basic Set Theory}}, Dover Books, 2012.

\bibitem{MacMoe}
Saunders MacLane and Ieke Moerdjik, \emph{Sheaves in geometry and logic: A
  first introduction to topos theory}, Springer, 1992.

\bibitem{GGleason}
Shuichiro Maeda, \emph{Probability measures on projections in von {N}eumann
  algebras}, Reviews in Mathematical Physics \textbf{1} (1989), no.~2/3, 235.

\bibitem{mansfield2013}
Shane Mansfield, \emph{The mathematical structure of non-locality and
  contextuality}, Ph.D. thesis, University of Oxford, 2013.

\bibitem{ManFri}
Shane Mansfield and Tobias Fritz, \emph{Hardy's non-locality paradox and
  possibilistic conditions for non-locality}, Foundations of Physics
  \textbf{42} (2012), no.~5, 709--719,
  \href{http://arXiv.org/abs/arXiv:1105.1819}{{\tt arXiv:1105.1819}}.

\bibitem{mermin1990simple}
David~N. Mermin, \emph{Simple unified form for the major no-hidden-variables
  theorems}, Physical Review Letters \textbf{65} (1990), no.~27, 3373--3376.

\bibitem{32}
Ugo Montanari, \emph{Networks of constraints: Fundamental properties and
  applications to picture processing.}, Information Sciences \textbf{7} (1974),
  95--132.

\bibitem{Mul}
Cristopher~J. Mulvey, \emph{`$\&$'}, Supplemento ai Rendiconti del Circolo
  Matematico di Palermo \textbf{12} (1986), 99.

\bibitem{Nak11}
Kunji Nakayama, \emph{Sheaves in quantum topos induced by quantization},
  (2011), \href{http://arXiv.org/abs/arXiv:1109.1192v2}{{\tt
  arXiv:1109.1192v2}}.

\bibitem{NieChu}
Michael~A. Nielsen and Isaac~L. Chuang, \emph{Quantum computation and quantum
  information}, Cambridge University Press, 2000.

\bibitem{Nui11}
Joost Nuiten, \emph{Bohrification of local nets of observables},  (2011),
  \href{http://arXiv.org/abs/arXiv:1109:1397}{{\tt arXiv:1109:1397}}.

\bibitem{Plenio:2007zz}
Martin~B. Plenio and Shashank Virmani, \emph{An introduction to entanglement
  measures}, Quantum Information and Computation \textbf{7} (2007), 1--51,
  \href{http://arXiv.org/abs/arXiv:quant-ph/0504163}{{\tt
  arXiv:quant-ph/0504163}}.

\bibitem{popescu1994quantum}
Sandu Popescu and Daniel Rohrlich, \emph{Quantum nonlocality as an axiom},
  Foundations of Physics \textbf{24} (1994), no.~3, 379--385.

\bibitem{Raussendorf}
Robert Raussendorf, \emph{Contextuality in measurement-based quantum
  computation}, Physical Review A \textbf{88} (2013), no.~2, 022322,
  \href{http://arXiv.org/abs/arXiv:0907.5449}{{\tt arXiv:0907.5449}}.

\bibitem{Ren}
Alfr{\' e}d R{\' e}nyi, \emph{On measures of information and entropy}, Berkeley
  Symposium on Mathematical Statistics and Probability, vol.~1, University of
  California Press, 1961, p.~547.

\bibitem{Sachs}
David Sachs, \emph{The lattice of subalgebras of a {B}oolean algebra}, Canadian
  Journal of Mathematics \textbf{14} (1962), 451--460.

\bibitem{Salart}
Daniel Salart, Augustin Baas, Jeroen A.~W. van Houwelingen, Nicolas Gisin, and
  Hugo Zbinden, \emph{Spacelike separation in a {B}ell test assuming
  gravitationally induced collapses}, Physical Review Letters \textbf{100}
  (2008), \href{http://arXiv.org/abs/arXiv:0803.2425}{{\tt arXiv:0803.2425}}.

\bibitem{Santos}
Emilio Santos, \emph{{B}ell's theorem and the experiments: Increasing empirical
  support to local realism?}, Studies in History and Philosophy of Modern
  Physics \textbf{2} (2005), 179--206,
  \href{http://arXiv.org/abs/arXiv:quant-ph/0410193}{{\tt
  arXiv:quant-ph/0410193}}.

\bibitem{QPSI}
Benjamin Schumacher and Michael Westmoreland, \emph{Quantum processes, systems
  and information}, Cambridge University Press, 2010.

\bibitem{Sha48}
Claude~E. Shannon, \emph{A mathematical theory of communication}, {B}ell System
  Technical Journal \textbf{27} (1948), no.~3, 379.

\bibitem{Shor:1994jg}
Peter~W. Shor, \emph{Polynomial time algorithms for prime factorization and
  discrete logarithms on a quantum computer}, SIAM Journal on Scientific and
  Statistical Computing \textbf{26} (1997), 1484,
  \href{http://arXiv.org/abs/arXiv:quant-ph/9508027}{{\tt
  arXiv:quant-ph/9508027}}.

\bibitem{stephanie}
Anthony~J. Short and Stephanie Wehner, \emph{Entropy in general physical
  theories}, New Journal of Physics \textbf{12} (2010), 033023,
  \href{http://arXiv.org/abs/arXiv:0909.4801}{{\tt arXiv:0909.4801}}.

\bibitem{stone}
Marshall~H. Stone, \emph{The theory of representations of {B}oolean algebras},
  Transactions of the American Mathematical Society \textbf{40} (1936), 37.

\bibitem{tarski}
Alfred Tarski, \emph{Logic, semantics, and metamathematics - papers from 1923
  to 1938}, Hackett Publishing Company, 1983.

\bibitem{Terhal}
Barbara Terhal, \emph{Detecting quantum entanglement}, Theoretical Computer
  Science \textbf{287} (2002), 313,
  \href{http://arXiv.org/abs/arXiv:quant-ph/0101032}{{\tt
  arXiv:quant-ph/0101032}}.

\bibitem{Var07}
Veeravalli~S. Varadarajan, \emph{Geometry of quantum theory}, 2 ed., Springer,
  2007.

\bibitem{vN55}
John von Neumann, \emph{Mathematische {G}rundlagen der {Q}uantenmechanik},
  Springer, 1955.

\bibitem{wang2012nonlocality}
Zizhu Wang and Damian Markham, \emph{Nonlocality of symmetric states}, Physical
  Review Letters \textbf{108} (2012), no.~21, 210407.

\bibitem{Wol10}
Sander Wolters, \emph{A comparison of two topos-theoretic approaches to quantum
  theory},  (2010), \href{http://arXiv.org/abs/arXiv:1010.2031v2}{{\tt
  arXiv:1010.2031v2}}.

\bibitem{Woo11}
Thomas Woodhouse, \emph{Time evolution in quantum theory and quantum
  information, a topos theoretic perspective}, Master's thesis, University of
  Oxford, 2011.

\end{thebibliography}

\end{document}